\documentclass[aos]{imsart}

\RequirePackage[numbers,sort&compress]{natbib}
\RequirePackage[colorlinks,citecolor=blue,urlcolor=blue]{hyperref}
\RequirePackage{graphicx}
\usepackage{graphicx}
\usepackage{amsthm,amsmath,amsfonts,amssymb}
\usepackage{bigints}
\usepackage{bm}
\usepackage{color}

\startlocaldefs
\theoremstyle{plain}

\newtheorem{thm}{Theorem}[section]
\newtheorem{lem}[thm]{Lemma}
\newtheorem{cor}[thm]{Corollary}
\newtheorem{prop}[thm]{Proposition}
\newtheorem*{thm*}{Theorem}
\newtheorem{rem}[thm]{Remark}
\newtheorem{cond}{Condition}

\theoremstyle{remark}
\newtheorem{defn}[thm]{Definition}
\newtheorem{ex}{Example}

\endlocaldefs

\usepackage{graphicx}
\usepackage{bigints}
\usepackage{bm}
\usepackage{color}

\newcommand{\p}{\mathbb{P}}
\newcommand{\e}{\mathbb{E}}
\newcommand{\lk}{\left[ }
\newcommand{\rk}{\right] }
\newcommand{\lc}{\left(}
\newcommand{\rc}{\right)}
\newcommand{\R}{\mathbb{R}}

\newcommand{\N}{\mathbb{N}}

\newcommand{\bc}{\mathcal{B}}
\newcommand{\id}{\mathbf{1}}

\newcommand{\sumn}{\sum_{i=1}^n}
\newcommand{\sumd}{\sum_{i=1}^d}
\newcommand{\prodd}{\prod_{i=1}^d}

\newcommand{\suminf}{\sum_{i=1}^\infty}

\newcommand{\iinn}{{i\in\N}}
\newcommand{\ijdn}{{(i,j)\in d\times\N}}
\newcommand{\leqn}{{1\leq i\leq n}}
\newcommand{\leqd}{{1\leq i\leq d}}
\newcommand{\rmd}{\mathrm{d}}

\newcommand{\bmx}{{\bm{x}}}

\newcommand{\bmy}{{\bm{y}}}
\newcommand{\bmY}{{\bm{Y}}}

\newcommand{\ubmY}{{\underline{\bm{Y}}}}
\newcommand{\ubmy}{{\underline{\bm{y}}}}
\newcommand{\ubmZ}{{\underline{\bm{Z}}}}

\newcommand{\bmb}{{\bm{b}}}
\newcommand{\bmX}{{\bm{X}}}
\newcommand{\bmxij}{{\bm{x}_{IJ}}}
\newcommand{\bmXIJ}{{\bm{X}_{IJ}}}
\newcommand{\ubmX}{{\underline{\bm{X}}}}
\newcommand{\ubmx}{{\underline{\bm{x}}}}

\newcommand{\bmZ}{{\bm{Z}}}

\newcommand{\binfty}{{\bm{\infty}}}
\newcommand{\ee}{\mathcal{E}}

\newcommand{\bmt}{{\bm{t}}}

\DeclareMathOperator*{\cov}{Cov}
\DeclareMathOperator*{\IDRM}{IDRM}
\DeclareMathOperator*{\minid}{min-id}
\DeclareMathOperator{\crm}{CRM}

\allowdisplaybreaks

\begin{document}

\begin{frontmatter}
\title{Infinitely divisible priors for multivariate survival functions}
\runtitle{Infinitely divisible priors for multivariate survival functions}

\begin{aug}
\author[A]{\fnms{Florian}~\snm{Brück} \ead[label=e1]{florian.brueck.edu@gmail.com}},

\address[A]{Research Institute for Statistics and Information Science, University of Geneva \printead[presep={ ,\ }]{e1}}

\end{aug}

\begin{abstract}
    This article introduces a novel framework for nonparametric priors on real-valued random vectors, which can be viewed as a multivariate generalization of neutral-to-the-right priors. It is based on randomizing the exponent measure of a minimum-infinitely divisible random vector by an infinitely divisible random measure and naturally incorporates partially exchangeable data as well as exchangeable random vectors. We show how to construct hierarchical priors from simple building blocks and embed many models from Bayesian nonparametric survival analysis into our framework. The prior can concentrate on discrete or continuous distributions and other properties such as dependence, moments and moments of mean functionals are characterized. The posterior predictive distribution is derived in a general framework and is refined under some regularity conditions. In addition, a framework for the simulation from the posterior predictive distribution is provided, which is illustrated in a simulation study.
    As a byproduct, the construction of tractable infinitely divisible random measures is studied and the concept of subordination of homogeneous completely random measures by homogeneous completely random measures is extended to the subordination of homogeneous completely random measures by infinitely divisible random measures. This technique allows to create vectors of dependent infinitely divisible random measures with tractable Laplace transforms and serves as a general tool for the construction of tractable infinitely divisible random measures. 
\end{abstract}

\begin{keyword}[class=MSC]
\kwd[Primary ]{62G05}
\kwd[; secondary ]{60G57}
\end{keyword}

\begin{keyword}
\kwd{Bayesian nonparametrics}
\kwd{Infinitely divisible random measure}
\kwd{Min-id distribution}
\kwd{Exchangeable sequences}
\end{keyword}

\end{frontmatter}

\section{Introduction}
The Dirichlet process serves as the cornerstone of Bayesian nonparametrics. Since its introduction in \cite{fergusonbayesiananalysis1973}, significant efforts have been made to extend its underlying construction methods, leading to more flexible models for random probability measures and, consequently, more flexible frameworks for Bayesian analysis.
Essentially, these extensions can be categorized into three categories. The first type of extensions is based on the construction of the Dirichlet process as $\mu(\cdot)/\mu(\mathbb{Y})$, where $\mu$ is a Gamma distributed completely random measure (CRM) on a Polish space $\mathbb{Y}$. Dropping the assumption of the Gamma distribution, the resulting random probability measures are known as normalized CRMs, which were introduced in \cite{jamespruensterlijoi2009}. This framework has been leveraged in many works due to its flexibility and analytical tractability of the posterior, see e.g.\ \cite{chen2013,barrioslijoinietobarajaspreunster2013,lijoinipotipruenster2014,argiento2020} for examples in clustering, topic modeling and density estimation.
The second type of extensions of the Dirichlet process is based on its stick-breaking representation \cite{sethuraman1994} of the form $\sum_{i\in\N} \pi_i \delta_{Y_i}$, where $\pi_i=V_i \prod_{j=1}^{i-1} (1-V_j)$ with $(V_i)_{i\in\N}\overset{i.i.d.}{\sim} Beta(1,\alpha)$ independent of an i.i.d.\ sequence $(Y_i)_{i\in\N}$. This construction principle is very popular due to its simplicity and flexibility and it has been generalized in many ways, see notably \cite{kingman1967completely,pitmanyor1997,ishwaranjames2001}, and applications are found in species sampling, clustering or hierarchical mixture models.

The third construction principle of the Dirichlet process is only valid in the univariate and real-valued case, in which the Dirichlet process can be viewed as a neutral-to-the-right (NTR) process \cite{doksumtailfree1974}. An NTR process is a random probability measure constructed via a random survival function of the form $\exp\lc -\mu(-\infty,t])\rc$, where $\mu$ is a CRM on $\R$. Extensions of this construction principle to the multivariate framework are less investigated than in the other two cases, since it is not obvious how a multivariate analog of $\exp\lc -\mu(-\infty,t])\rc$ should be constructed to yield a valid multivariate survival function. Most works thus have focused on creating vectors of dependent univariate NTR processes, see e.g.\ \cite{epifanilijoi2010,lijoinipoti2014,rivapalacioleisen2018}, with the exception of \cite{jamesspntr2006} who constructs a multivariate random probability measure with one specific margin being an NTR process.
The content of this paper can be viewed as the natural generalization of an NTR prior to the multivariate setting. Consider the random survival function of a real-valued random vector $\bmX\in\R^d$ given by
\begin{align}
  S_\mu(\bmx)=P\lc X_1>x_1,\ldots,X_d>x_d\mid\mu\rc=\exp\lc - \mu\lc  (\bmx,\binfty]^\complement \rc \rc,  \label{introminid}
\end{align}
where $\mu$ denotes a random measure on $(-\infty,\infty]^d$ which is finite on every set of the form $ (\bmx,\binfty]^\complement:= (-\infty,\infty]^d\setminus (\bmx,\binfty]$. The measure $\mu$ is called exponent measure due to its appearance in the exponent of $S_\mu$ and the corresponding random probability measure is denoted as $\minid(\mu)$. Since it is not trivial to see that $S_\mu$ is a well-defined multivariate survival function for every $\mu$, we remark that survival functions of the form (\ref{introminid}) naturally appear in the study of extremes of independent random vectors, where they are called minimum infinitely divisible (min-id) survival functions, explaining the notation $\minid(\mu)$ and referring to \cite{Husler1989,resnickextreme2013} for more background information. It is easy to see that this construction principle is identical to the construction of an NTR prior when $\mu$ is a univariate CRM. When $\mu$ is a multivariate CRM, each margin of $S_\mu$ corresponds to an NTR process and it can be checked that $S_\mu(\bmx_1),(S_\mu(\bmx_1)-S_\mu(\bmx_2))/S_\mu(\bmx_1),\ldots,(S_\mu(\bmx_{K-1})-S_\mu(\bmx_K))/S_\mu(\bmx_{K-1})$ is a collection of independent random variables for all $\bmx_1<\ldots<\bmx_K$; extending the classical neutrality concept for univariate distributions from \cite{doksumtailfree1974} to the multivariate setting. Thus, random probability measures which are constructed by choosing a CRM $\mu$ in (\ref{introminid}) may be viewed as the natural multivariate generalizations of NTR processes. 

Restricting $\mu$ to the class of CRMs in (\ref{introminid}) should however be seen as a limitation, since every realization of $\mu$ is a purely discrete measure and a corresponding random vector inherits this discrete structure. In this paper, we will go far beyond CRMs by allowing the random measure $\mu$ to be infinitely divisible, i.e.\ for every $n\in\N$ there exist i.i.d.\ random measures $\lc \mu^{(i,n)}\rc_{\leqn}$ such that $\mu\sim \sum_{i=1}^n \mu^{(i,n)}$, see \cite[Chapter 3.3]{kallenberg2017} for an account on infinitely divisible random measures (IDRMs). The framework of IDRMs is much more flexible than that of CRMs, since IDRMs naturally include continuous as well as discrete random measures. Additionally, the class of IDRMs is closed under many natural operations on measures such as integration. For example, the univariate models of \cite{Dykstralaud1981,loweng1989} for the cumulative hazard of a random variable are of the form $\int_0^t g(x)\mu(\rmd x)$ for some diffuse CRM $\mu$ and thus they are naturally included in the framework of IDRMs. Furthermore, many models for partially exchangeable (survival) data may also be embedded in this framework, such as \cite{epifanilijoi2010,lijoinipoti2014,rivapalacioleisen2018,camerlengilijoipruenster2021}. In essence, this is always a consequence of these models being based on transformations of vectors of dependent CRMs, which often result in IDRMs.  

At the heart of many multivariate Bayesian models is the construction of vectors of dependent CRMs, which are conditionally independent given a baseline CRM, see \cite{catalano2023} for an overview of the topic. These vectors of dependent CRMs may or may not be smoothed at a higher level in the construction of the final model, but the discrete nature of the baseline CRM implies certain tie-like dependencies in the resulting vector of dependent random measures. Our framework allows to avoid these tie-like dependencies stemming from discrete baseline CRMs, as it allows for the construction of vectors of dependent IDRMs, which are conditionally independent given a continuous baseline IDRM. This type of construction may be seen as a useful tool beyond the scope of this paper as it enables ``smoother'' dependence structures between  dependent random measures.

Our specification of the prior can be summarized by the following hierarchical construction of the random probability distribution for a random vector $\bmX$:
\begin{align} 
\bmX\sim \minid(\mu),\ \mu\sim \IDRM . \label{intrhierarchconstr}
\end{align}
The rest of this paper will be devoted to formalizing (\ref{intrhierarchconstr}) in a sound mathematical framework for nonparametric Bayesian analysis. We already want to mention that the derivation of the posterior distribution will essentially be based on applications of de Finetti's theorem, since we will show that the law of the exchangeable sequence $(\bmX_j)_{j\in\N}$, $\bmX_j\sim\minid(\mu)$, is again in the class of min-id distributions. This might be interpreted as a conjugacy property of our class of prior distributions and also implies positive dependence between the margins since min-id distributions are associated, see \cite[Section 5.6]{resnickextreme2013}. The conjugacy property will allow to represent our posterior without the need for an auxiliary augmentation, since results about the conditional distributions of min-id sequences can be leveraged to identify the posterior of the random probability measure $\minid(\mu)$.

The main text is organized as follows. Section \ref{secexminididrm} introduces (exchangeable) sequences of min-id random vectors and IDRMs and derives a one-to-one correspondence \textcolor{black}{between suitable subclasses of these objects}. Section \ref{secidempriors} formally introduces our framework for prior distributions \textcolor{black}{(IDEM priors)} and discusses how to construct tractable models from simple building blocks. Further, several models from the literature are embedded into our framework. Section \ref{secpostdistr} derives the posterior distribution in the most general setting and Section \ref{seccclosedformreppost} provides a framework to construct continuous priors with analytically tractable posterior distribution. Section \ref{secmultivmodel} demonstrates the flexibility of our framework by introducing a tractable model for partially exchangeable data based on a shared continuous dependence structure and discusses and illustrates its simulation.
The paper is complemented by supplementary material \cite{brueck2026}. It contains an introduction to min-id processes, which we strongly recommend to readers unfamiliar with the topic before proceeding to the main text. Furthermore, it contains more information about simulation algorithms, additional examples and properties of IDEM priors as well as the posterior of multivariate NTR priors. It also contains all proofs, technical details on random measures and an overview of the notation.

\section{Exchangeable min-id sequences and infinitely divisible random measures}
\label{secexminididrm}
This section introduces exchangeable sequences of min-id random vectors, which naturally arise by imposing a prior on the exponent measure of a random vector according to (\ref{intrhierarchconstr}). Further, we introduce infinitely divisible random measures on $(-\infty,\infty]^d$ which will serve as the prior on the exponent measures in the construction of min-id random vectors. Finally, we prove that a sequence of min-id random vectors whose associated exponent measure is randomized by an infinitely divisible random measure stays in the class of exchangeable min-id sequences, which will be the essential technical ingredient to derive the posterior distribution.

In the following, we will always assume that our random objects are defined on a common probability space $(\Omega,\mathcal{F},P)$, if not explicitly mentioned otherwise. Moreover, to differentiate a sequence of random vectors, the first $m$ random vectors in a sequence of random vectors and a random variable in the sequence we use the respective notation $\ubmX:=(\bmX_j)_{j\in\N}:=\lc\lc X_{i,j}\rc_{1\leq i\leq d}\rc_{j\in\N}$, $\ubmX_m:=(\bmX_j)_{1\leq j\leq m}$ and $X_{i,j}$. Similarly, a sequence of vectors, the first $m$ vectors in a sequence of vectors and an entry in the sequence are denoted as $\ubmx:=(\bmx_j)_{j\in\N}:=\lc\lc x_{i,j}\rc_{\leqd}\rc_{j\in\N}$, $\ubmx_m:=(\bmx_j)_{1\leq j\leq m}$ and $x_{i,j}$, respectively.  Moreover, the operators on vectors/sequences $+,-,*,/,\min,\max,\sup,\inf$  are always applied componentwise.

\subsection{Min-id sequences}
Central objects of this section are exchangeable sequences of random vectors. The following definition formalizes what we mean by exchangeable.
\begin{defn}
\label{defnexrandvec}
    A sequence of random vectors $\ubmX:=\lc \bmX_1,\bmX_2,\ldots\rc$ is called exchangeable (exchangeable sequence) if for every deterministic permutation $\tau$ of $\N$ we have $\ubmX\sim \lc \bmX_{\tau(1)},\bmX_{\tau(2)},\ldots\rc$.
\end{defn}
This definition should not be confused with partial exchangeability, which often occurs in hierarchical models in Bayesian nonparametrics. Partial exchangeability allows different permutations in every margin of the sequence of random vectors and thus is a stronger condition.
We focus on sequences and random vectors on the space 
$$E_p:=(-\binfty,\binfty]= (-\infty,\infty]^p,$$ 
where $p\in\N\cup\{\infty\}$. $E_p$ inherits the standard topology and corresponding Borel $\sigma$-algebra from $[-\infty,\infty]^p$. Frequently, the subspace $E^\prime_p :=E_p\setminus\{\binfty\}$, which inherits its topology and $\sigma$-algebra from $E_p$, will be important in theoretical considerations. 

Min-id sequences are at the heart of this paper. Formally, they can be defined in various equivalent ways, which are presented jointly, because we will frequently use all of these definitions in our derivations. To ease the notation, denote the index set of a sequence of (random) vectors as $d\times \N:= \{ (i,j) \mid 1\leq i\leq d,j\in \N\}$. Further, we will also use the notation $IJ:=\lc (i_k,j_k)\rc_{1\leq k\leq p}$ to refer to a $p$-dimensional margin of $d\times \N$ and implicitly use the important convention $\min_\emptyset =\binfty$. Finally, to emphasize that a random sequence (resp.\ vector) $\ubmX$ (resp.\ $\bmX_{IJ}$) takes values in a product space with index set $d\times\N$ (resp.\ $IJ$), we denote $E_{IJ}=(-\infty,\infty]^{\vert IJ\vert}$ for any $IJ\subset d\times\N$; and similarly for $E^\prime_{IJ}$.

\begin{defn}[\cite{vatanmaxidinfinitedim,resnickextreme2013}]
    \label{defnexminseq}
A sequence of random vectors $\ubmX\in  E_{d\times \N}$ such that $P(X_{i,j}>x)>0$ for all $(i,j)\in d\times\N$ and $x\in\R$ is called a min-id sequence if it satisfies one of the following equivalent conditions:
\begin{enumerate}
    \item[$(i)$] For every $n\in\N$ there exist i.i.d.\ sequences of random vectors $\lc\ubmX^{(i,n)}\rc_{\leqn}$ such that $\ubmX\sim \min_{\leqn} \ubmX^{(i,n)}$. 
    \item[$(ii)$] There exists a Poisson random measure (PRM) $N=\sum_{k\in\N} \delta_{\ubmx^{(k)}}$ on $E^\prime_{d\times \N}$ with unique intensity measure $\lambda$, which satisfies $\lambda\lc \{ \ubmy\in E^\prime_{d\times \N}\mid y_{i,j}\leq x\} \rc<\infty$ for all $x\in\R$ and $\ijdn$, such that 
    \begin{align}
       \ubmX \sim \min_{k\in \N}\ubmx^{(k)}. \label{prmrepminid} 
    \end{align}
    \item[$(iii)$] For every $n\in\N$, all $p$-dimensional margins $IJ$ of $\ubmX$ satisfy that $S_{IJ}\lc \bmx_{IJ}\rc^{1/n}:=P\lc \bmX_{IJ} >\bmx_{IJ}\rc^{1/n}$ defines a proper survival function of an $E_{IJ}$-valued random vector. 
\end{enumerate}
The finite dimensional margins of a min-id sequence are called min-id (random) vectors. Further, a sequence $\ubmX$ is called exchangeable min-id sequence if it is an exchangeable sequence which is also min-id.
\end{defn}

Note that every measure $\lambda$ on $E^\prime_{d\times \N}$ which satisfies the conditions of Definition \ref{defnexminseq} $(ii)$ uniquely specifies the law of a min-id sequence. Since $\lambda$ appears in the exponent of the survival function of $\ubmX$, it is called the exponent measure of $\ubmX$. Let us denote $\lambda^{(IJ)}$ as the marginalization of $\lambda$ to margin $IJ$, i.e.\ the intensity of a PRM on $E^\prime_p$ such that the analog of (\ref{prmrepminid}) holds for $\bmX_{IJ}$, which we also call exponent measure of $\bmX_{IJ}$. One can easily check that when extending $\lambda$ ($\lambda^{(IJ)}$) to $E_{d\times \N}$ ($E_{IJ}$) the mass of $\lambda$ ($\lambda^{(IJ)}$) on $\binfty$ is irrelevant for the distribution of $\ubmX$ ($\bmX_{IJ}$), which is why the point $\binfty$ had to be disregarded to obtain uniqueness in Definition \ref{defnexminseq}$(ii)$. Since it later will turn out to be convenient, we frequently view $\lambda^{(IJ)}$ and $\lambda$ as measures on $E_{IJ}$ and $E_{d\times \N}$, where we employ the convention $\lambda^{(IJ)}(\{\binfty\})=\lambda(\{\binfty\})=\infty$. With this convention, the exponent measure $\lambda^{(IJ)}$ of $\bmX_{IJ}$ can be viewed as the projection of the exponent measure $\lambda$ of $\ubmX$ to the components $IJ$, i.e.\ $\lambda^{(IJ)}(A)=\lambda(\{\ubmx\mid \bmx_{IJ}\in A\})$, and for all $\ubmx\in [-\infty,\infty)^{d\times\N}$ we can simply write 
$$P(\ubmX>\ubmx)=\exp\Big( -\lambda\big( \{\ubmy\in E_{d\times \N}\mid y_{i,j}\leq x_{i,j} \text{ for some } \ijdn \} \big) \Big),$$
while the analogous expression holds for $\bmX_{IJ}$ and $\lambda^{(IJ)}$. To simplify the notation we will denote $\ubmX\sim\minid(\lambda)$ ($\bmX_{IJ}\sim \minid(\lambda^{(IJ)})$) to denote that $\ubmX$ ($\bmX_{IJ}$) is a min-id sequence (vector) with exponent measure $\lambda$ ($\lambda^{(IJ)}$).

\begin{ex}[Min-id random vectors with independent components]
    \cite[5.3.1. (i)]{resnickextreme2013} shows that $\bmX_{IJ}\sim\minid(\lambda^{(IJ)})$ has independent components if and only if $\lambda^{(IJ)}$ concentrates on lines through $\binfty$, i.e.\ $\lambda^{(IJ)}$ is supported on the set $\cup_{i=1}^{p} \ee_{i,j}$, where 
    $$\ee_{i,j}:=\big\{ \bmx\in E_{p}^\prime \mid x_{i,j}<\infty ,\ x_{k,l}=\infty\ \forall (k,l)\in IJ\text{ s.t. } (i,j)\not=(k,l)  \big\}.$$
    To illustrate this, note that when $\lambda^{(IJ)}$ is concentrated on $\cup_{i=1}^{p} \ee_{i,j}$ we have that
    $$P(\bmX_{IJ}>\bmx)=\exp\lc - \sum_{(i,j)\in IJ} \lambda^{(IJ)} \big( \{ \bmy\in \ee_{i,j} \mid y_{i,j}\leq x_{i,j}\}\big)\rc.$$
    Thus, the survival function of $\bmXIJ$ is the product of the terms $\exp\big( -\lambda^{(IJ)} \big( \{ \bmy\in \ee_{i,j} \mid$ $ y_{i,j}\leq x_{i,j}\}\big)\big)$, which are the marginal survival functions of the $X_{i,j}$. 
\end{ex}
    
    The example emphasizes the importance of allowing for mass of $\lambda$ (resp.\ $\lambda^{(IJ)}$) on $E_{d\times\N}\setminus \R^{d\times \N}$ (resp.\ $E_{IJ}\setminus \R^{\vert IJ\vert}$), since this is needed to allow for independent components in min-id sequences (resp.\ random vectors). Moreover, as we will see later in Examples \ref{exhierdepmixhaz} and \ref{exhazratmix}, many Bayesian models from the literature use priors which concentrate on min-id random vectors with (conditionally) independent components and thus such random vectors will naturally also play an important role in our framework.

\begin{rem}[The role of $\binfty$]
As mentioned above, the mass of an exponent measure on $\binfty$ is not relevant for the distribution of $\ubmX\sim \minid( \lambda)$. However, a careful treatment of the point $\binfty$ is necessary when embedding $\bmXIJ$ in $\ubmX$. To illustrate this, assume that $\ubmX$ has independent entries and recall that the exponent measure of $\bmX_{IJ}$ is given by $\lambda^{(IJ)}=\lambda\lc \{\ubmx\in E_{d\times \N} \mid \bmx_{IJ}\in  \cdot \cap E_{p}^\prime  \}\rc$, which is the projection of $\lambda$ to the components $IJ$. Thus, the mass of $\lambda$ on $\{ \bmx_{IJ}=\binfty\}$, which corresponds to the mass of the marginalized components $d\times \N\setminus IJ$, is disregarded in $\lambda^{(IJ)}$ as it is fully projected to $\binfty$ when considering the exponent measure of $\bmX_{IJ}$.  However, when trying to reverse-engineer $\lambda$ from $\lambda^{(IJ)}$ (resp.\ $\ubmX$ from $\bmX_{IJ}$) it is now crucial to take into account the mass that was projected to $\binfty$. Thus, for reconstructing $\lambda$ from $\lambda^{(IJ)}$, one should view $\lambda^{(IJ)}$ as a measure on $E_{p}$, since there is potentially mass of $\lambda$ on $\{\ubmx\in E_{d\times \N}^\prime\mid \bmx_{IJ}=\binfty\}$, which is relevant to determine $\lambda$ and the distribution of $\bmX_{(d\times \N)\setminus IJ}$. This peculiarity will turn out to be important for many derivations and calculations in the paper, as we frequently have to embed $\bmX_{IJ}$ in a sequence, which is why this point is emphasized already at this stage.
    
\end{rem}

\begin{rem}
    We remark that $P(X_{i,j}>x)>0$ for all $(i,j)\in d\times \N$ and $x\in \R$ should be viewed as a standardization of our framework, since, in general, min-id distributions can have an arbitrary upper endpoint of their support. It can be easily checked that $\lc f_{i,j}(X_{i,j})\rc_{(i,j)\in d\times\N}$ remains a min-id sequence in the sense of Definition \ref{defnexminseq}$(i)$ for every collection of non-decreasing transformations $(f_{i,j})_{(i,j)\in d\times\N}$ from $[-\infty,\infty]\to[-\infty,\infty]$. Thus w.l.o.g.\ we can conduct our analysis for min-id sequences which satisfy $P(X_{i,j}>x)>0$ for all $(i,j)\in d\times \N$ and $x\in \R$, whereas results about min-id sequences with finite upper endpoint of their support may be deduced from a suitable non-decreasing marginal transformation of a min-id sequence with $\infty$ as upper endpoint of their support.
\end{rem}

Definition \ref{defnexminseq}$(i)$ entails the name-giving stochastic representation of a min-id sequence, since it shows that every min-id sequence can be represented as the componentwise minimum of arbitrarily many i.i.d.\ sequences. On the other hand, the stochastic representation (\ref{prmrepminid}) is the most useful representation for our purposes, since it nicely illustrates that every $\ubmX$ can simply be constructed as the componentwise minimum of countably many sequences $\ubmx^{(k)}$, which are the atoms of a PRM. A natural question that arises from this representation is whether for every $(i,j)\in d\times\N$ there is exactly one $k(i,j)$ such that $X_{i,j}=x^{(k(i,j))}_{i,j}$, i.e.\ whether the value $X_{i,j}$ can be uniquely associated to a single atom $\ubmx^{(k)}$ of the PRM. In general, the answer to this is negative, since it is rather easy to see that when the intensity $\lambda$ has atoms there can be several $x^{(k)}_{i,j}$ that take the same value. However, when the marginal distributions of $\ubmX$ are continuous, the answer is positive.

\begin{lem}[\text{\cite[Proposition 2.5]{dombryeyiminko2013regular}}]
\label{lemuniquemaxfct}
    Let $\ubmX$ denote a min-id sequence with stochastic representation (\ref{prmrepminid}) such that for every $(i,j)\in d\times \N $ the random variable $X_{i,j}$ has a continuous distribution function. Then, with probability $1$, there exists exactly one $k(i,j)$ such that $X_{i,j}=x^{(k(i,j))}_{i,j}$ for all $(i,j)\in d\times \N$.
\end{lem}

Lemma \ref{lemuniquemaxfct} allows to define random equivalence classes for the index set $d\times \N$ ($IJ$) of a min-id sequence $\ubmX$ ($\bmX_{IJ}$) with continuous margins by collecting all indices $(i,j)$ which have the same minimizing sequence $\ubmx^{(k)}$, i.e.\ by
\begin{align} (i,j)\sim (\Tilde{i},\Tilde{j}) \Leftrightarrow X_{i,j}=x^{(k)}_{i,j} \text{ and } X_{\Tilde{i},\Tilde{j}}=x^{(k)}_{\Tilde{i},\Tilde{j}}.  \label{equivrelhitt}
\end{align}
This equivalence relation is the basis of the so-called hitting scenario of a min-id sequence, describing the latent structure of the atoms of the PRM which are responsible for seeing the realization $\ubmX$ ($\bmX_{IJ}$). It was first introduced in \cite{dombryeyiminko2013regular} and will play an important role in the representation of the posterior in later sections.
\begin{defn}
    \label{defhittingscen}
    Consider a margin $\bmX_{IJ}$ of a min-id sequence $\ubmX$. A hitting scenario $\Theta(IJ)=(\Theta_1,\ldots,\Theta_{L})$ is a random partition of $IJ$, where the partition is generated according to the equivalence relation (\ref{equivrelhitt}) and $L:=L(\Theta(IJ))$ denotes the length of the partition. 
\end{defn}

\subsection{Infinitely divisible exponent measures}

In this paper, we focus on random measures $\mu$ on $E_d$ that can be interpreted as random exponent measures.
To define random measures on $E_d$, we follow the technical elaborations in \cite[Section 1+2]{kallenberg2017} and consider only $d<\infty$, while more details concerning the well-definedness of the following definitions can be found in Section 15 of the Supplementary Material \cite{brueck2026}. We have to take extra care in the definition of $\sigma$-finite random measures on $E_d$ since we later want to allow for measures that explode towards $\binfty$.
First, we fix a measurable partition $(U_i)_{i\in\N}$ of $E_d^\prime$ given by 
\begin{align}
U_i:= \lc E_d\setminus \lc (i,\infty]^d \rc \rc \setminus  \cup_{j=1}^{i-1} U_j \label{deflocseq} 
\end{align}
and let 
$$M_d:=\big\{ \eta \mid \eta \text{ is a measure on }E_d \text{ such that } \eta(U_i)<\infty\text{ and }\eta(\binfty)=\infty \big\}$$ 
denote the subset of $\sigma$-finite measures on $E_d$ which are finite on $(U_i)_{i\in\N}$. Note that $\cup_{j=1}^{i-1} U_j=\lc (i-1,\infty]^d\rc^\complement$. Therefore, every measure $\eta\in M_d$ is an exponent measure of a min-id random vector on $E_d$ as $\eta\lc (\bmx,\binfty]^\complement\rc<\infty$ for all $\bmx\in\R^d$ and we can equivalently write $M_d=\{\eta\mid \eta\text{ is the exponent measure of a random vector on }E_d\}$.
Next, we equip $M_d$ with the $\sigma$-algebra $\mathcal{G}$ that is generated by the evaluations of the measure on finite collections of measurable subsets of $E_d$, which allows us to define a random measure.

\begin{defn}
\label{defnrandommeasure}
    A random measure is a random element $\mu\in M_d$, i.e.\ a measurable map from $\lc\Omega,\mathcal{F},P\rc\to (M_d,\mathcal{G})$.
\end{defn}

Thus, a random measure $\mu$ is the sum of countably many almost-surely finite random measures $(\mu(\cdot\cap U_i))_{i\in\N}$ and an infinite atom at $\binfty$.
The introduction of the ``localizing sets'' $(U_i)_{i\in\N}$ is necessary, since the definition of a random measure on the space of $\sigma$-finite random measures on $E_d$ additionally requires to determine where the measure is almost surely finite, see the technical elaborations in \cite[Section 1+2]{kallenberg2017}. In our case, the definition ensures that a random measure is a $\sigma$-finite exponent measure for the same exhausting sequence $(U_i)_{i\in\N}$, i.e.\ the sets on which a random measure may explode are non-random.
Next, we restrict ourselves to random measures that can be decomposed into sums of arbitrarily many i.i.d.\ random measures, called infinitely divisible random measures. The formal definition is as follows.
\begin{defn}
\label{defnidrandommeasure}
 A random measure $\mu$ is called an infinitely divisible random measure (IDRM) if for every $n\in\N$ there exist i.i.d.\ random measures $\mu^{(i,n)}$ such that $\mu\sim\sumn \mu^{(i,n)}$.
\end{defn}

Every IDRM can be characterized in terms of its Laplace transform, which follows a specific representation, called the Lévy-Khintchine representation. We recall this result here, since it will be important for the further developments in this paper.

\begin{thm}[\text{\cite[Theorem 3.20]{kallenberg2017}}]
\label{thmidrandmeasurelevykhintchine}
Every IDRM $\mu$ on $E_d$ has the representation
\begin{align}
    \mu\sim \alpha+ \int_{M^0_{d}}\eta  N(\rmd\eta)=\alpha+ \sum_{j\in\N}\eta_j, \label{prmrepidrm}
\end{align} 
where $\alpha\in M_d$ is a unique deterministic exponent measure and $N=\sum_{j\in\N}\delta_{\eta_j}$ is a PRM on $M^0_{d}:=\{ \eta\in M_d \mid  \eta( E_d^\prime)\not=0\}$ with unique intensity measure $\nu$ such that $\int_{M^0_d} \min\{\eta(U_i);1\}$ $\nu(\rmd \eta)<\infty$ for all $i\in\N$.
    The Laplace transform of $\mu$ has the form
\begin{align*}
    &\e\lk \exp\lc -\int_{E_d} f(\bmx) \mu(d\bmx) \rc \rk \\
    &=\exp \bigg( -\int_{E_d} f(\bmx) \alpha(d\bmx)  -\int_{M^0_d} 1-\exp\lc -\int_{E_d} f(\bmx) \eta(d\bmx) \rc  \nu(d\eta)  \bigg),
\end{align*}
for every measurable function $f\geq 0$ on $E_d$.
Conversely, every tuple $(\alpha,\nu)$ satisfying the stated conditions specifies the law of an IDRM that is unique in law on $E_d$.
\end{thm}

The measure $\alpha$ will be called base measure, the measure $\nu$ will be called Lévy measure and the tuple $(\alpha,\nu)$ will be called Lévy-Khintchine characteristics of an IDRM $\mu$.
For our choice of localizing sets $(U_i)_{i\in\N}$ we have that every IDRM is actually an infinitely divisible random exponent measure (IDEM) on $E_d$ with $\alpha(\binfty)=\infty$.

According to the definition of an IDEM $\mu$ there exist i.i.d.\ random measures $\mu^{(i,n)}$ such that $\mu\sim\sumn \mu^{(i,n)}$. An immediate consequence is that each $\mu^{(i,n)}$ must also have localizing sets $(U_i)_{i\in\N}$. Thus, each $\mu^{(i,n)}$ may be identified with a corresponding random exponent measure as well to obtain a decomposition of $\mu$ into sums of i.i.d.\ random exponent measures. Further, one can show that each $\lc \mu^{(i,n)}\rc_{\leqn}$ is infinitely divisible. The results are summarized in the following corollary. 

\begin{cor}
\label{coridexpmisidbyidexpm}
    For every IDEM $\mu$ and $n\in\N$ there exist i.i.d.\ IDEMs $\lc \mu^{(i,n)}\rc_\leqn$  such that $\mu\sim\sumn \mu^{(i,n)}$. Moreover, $\mu\sim \alpha +\sum_{j\in\N} \eta_j$, where $\alpha$ is a deterministic exponent measure on $E_d$ and $\lc \eta_j\rc_{j\in\N}$ are random exponent measures on $E_d$, which constitute the atoms of a PRM on $M_d^0$ with unique intensity $\nu$.
\end{cor}

\begin{ex}[Completely random measure]
\label{exCRM}
Let us illustrate the simplest choice of IDEMs. Consider a completely random measure (CRM) $\mu$ on $E_d$ without fixed atoms on $E^\prime_d$, i.e.\ a random measure with the representation $\alpha+\sum_{k\in \N} a_k \delta_{\bm b_k}$, where $\lc (a_k, \bm b_k)\rc_{k\in\N}$ are the atoms of a PRM on $(0,\infty)\times E_d^\prime$ with intensity $\upsilon(\rmd (a,\bm b))$ such that $\int_0^\infty\int_{ (\bmx,\binfty]^\complement } \min\{1,a\}\linebreak\upsilon(\rmd (a,\bm b))$ $<\infty$ for every $\bmx \in\R^d$. Then $\mu(A_1),\ldots,\mu(A_n)$ are independent random variables for every collection of disjoint measurable sets $(A_i)_\leqn\subset E_d$, which justifies the term ``completely random''. It is easy to see that $\mu$ is an IDEM with base measure $\alpha$ and Lévy measure $\nu(A)=\int_{(0,\infty)\times E_d^\prime }  \id_{\{a\delta_{\bm b}+\infty\delta_\binfty\in A\}} \upsilon(\rmd  (a,\bm b))$. Thus, the Lévy measure of a completely random measure essentially concentrates on weighted Dirac measures.
The Laplace transform of a completely random measure $\mu$ takes the well-known form
\begin{align*}
 &\e\lk \exp\lc- \int_{E_d} f(\bmx)\mu(\rmd\bmx) \rc\rk\\
 &=\exp\lc - \int_{E_d} f(\bmx)\alpha(\rmd\bmx)
 -\int_{(0,\infty)\times E_d} 1-\exp\lc- af(\bm b) \rc \upsilon (\rmd(a,\bm b))\rc.    
\end{align*} 

\end{ex}

\subsection{Exchangeable min-id sequences generated by infinitely divisible exponent measures}

Let us end this section by \textcolor{black}{describing the connection between IDEMs and the exchangeable min-id sequences generated by them}. First, we formulate a condition which ensures that $\bmX\sim\minid(\mu)$ is real-valued.
\begin{cond}
\label{condsuppminid}
  The IDEM $\mu$ on $E_d $ satisfies $\lim_{n\to\infty}\mu\lc \{ \bmx \in E_d\mid x_i\leq n\} \rc=\infty$ almost surely for all $1\leq i\leq d$. 
\end{cond}

The following theorem is one of the main technical ingredients for the derivation of the posterior distribution. It shows that randomizing the exponent measure of a min-id random vector by an IDEM yields an exchangeable min-id sequence, conditionally on $\mu$ as well as unconditionally (i.e.\ integrating out the randomness of $\mu$). 

\begin{thm}
\label{thmidexpmimpliesexminid}
    Let $\mu$ denote an IDEM with Lévy-Khintchine characteristics $(\alpha,\nu)$ that satisfies Condition \ref{condsuppminid}. Then, we can define a sequence of conditionally i.i.d.\ random vectors $\ubmX\in  \R^{d\times\N} $ via
    \begin{align}
        (\bmX_j)_{j\in\N}\overset{i.i.d.}{\sim} \minid(\mu). \label{eqnconstrexminid}
    \end{align} 
    Further, $\ubmX$ is an exchangeable min-id sequence, whose law is uniquely determined by $(\alpha,\nu)$. 
    The exponent measure $\lambda$ of $\ubmX$ is determined by the Lévy-Khintchine characteristics of $\mu$ via
    \begin{align}
        \lambda(A)=\lambda_\alpha(A) +\lambda_{\nu}(A)  \label{thmreplambda},
    \end{align}
    where $\lambda_\alpha$ is a measure concentrated on $\mathcal{W}:=\cup_{i\in\N} \{ \ubmx \in E_{d\times \N} \mid \bmx_j=\binfty \text{ for all } i\not=j\in\N\}$ given by
    $$ \lambda_\alpha(A) :=\sum_{j\in\N} \alpha\lc \big\{ \bmx \in E_d\mid  \times_{k=1}^{j-1} \{ \binfty\} \times \{\bmx\} \times_{k=j+1}^\infty \{\binfty\}  \in  A \big\}\rc $$
    and 
    $$\lambda_{\nu}(A):= \int_{M^0_d}  \otimes_{j\in\N} \minid(\eta) \lc \{ \ubmx \in A\}\rc \nu(\rmd\eta) ,$$
    where $\otimes_{j\in\N} \minid(\eta) $ denotes the law of an i.i.d.\ sequence of min-id random vectors on $E_d$ with exponent measure $\eta$.  
    
    Conversely, for every real-valued exchangeable min-id sequence whose exponent measure is of the form $\lambda_\alpha+\lambda_\nu$ for some Lévy characteristics $(\alpha,\nu)$ there exists an IDEM, unique in law and satisfying Condition \ref{condsuppminid}, which generates the exchangeable min-id sequence according to (\ref{eqnconstrexminid}).

\end{thm}
The theorem could be called de Finetti's theorem for exponent measures of exchangeable min-id sequences generated by IDEMs, since it shows that the exponent measure associated to these sequences is a (possibly infinite) mixture of the law of i.i.d.\ min-id sequences plus the exponent measure of a min-id sequence with i.i.d.\ components. Notice that our requirement of $\mu$ taking values in $M_d$ allows to define $\bmX_j$ as a min-id random vector conditionally on every realization of $\mu$.
It should be emphasized that, even though $(\bmX_j)_{j\in\N}$ is real-valued, $\eta$ might be the exponent measure of a min-id random vector whose components can take the value $\infty$, i.e.\ $\minid(\eta)$ is the distribution of a random vector on $E_d$ and $\eta$ might not satisfy Condition \ref{condsuppminid}; see e.g.\ Example \ref{exCRM}.
An immediate consequence of Theorem \ref{thmidexpmimpliesexminid} is the following stochastic decomposition of $\ubmX$.

\begin{cor}
\label{corstochdecompminid}
    An exchangeable min-id sequence $\ubmX$ \textcolor{black}{with exponent measure of the form (\ref{thmreplambda})} has the following stochastic decomposition:
    $$ \ubmX\sim \min \big\{ \ubmX_1,\ubmX_2  \big\},$$
    where $\ubmX_1$ is a sequence of i.i.d.\ min-id random vectors with exponent measure $\alpha$ and $\ubmX_2$ is an exchangeable min-id sequence with exponent measure $\lambda_\nu$. Moreover, there is no i.i.d.\ sequence of \textcolor{black}{real-valued} random vectors $\underline{\bmY}_1$ independent of a min-id sequence $\underline{\bmY}_2$ such that $\ubmX_2\sim\min\{\underline{\bmY}_1,\underline{\bmY}_2\}$.
\end{cor}

Thus, intuitively, an exchangeable min-id sequence from Theorem \ref{thmidexpmimpliesexminid} can be divided into a trivial i.i.d.\ part $\ubmX_1$ and a non-trivial part $\ubmX_2$, which is responsible for the dependence of the sequence. For our purposes of imposing a prior on $\bmX$ this result says that $\alpha$ drives an i.i.d.\ "censoring-like" part of $\ubmX$ and $\nu$ is the quantity that purely governs the dependence between the $\bmX_j$ and allows to learn about the distribution of $\bmX$.

\section{IDEM priors}
\label{secidempriors}
Having collected all auxiliary results, we are ready to formally introduce our prior distribution on the space of probability measures on $\R^d$. 

\begin{defn}[IDEM prior]
Let $\mu$ denote an IDEM on $E_d$ satisfying Condition \ref{condsuppminid}. Then, the random probability measure $\minid(\mu)$ is called an IDEM prior.
\end{defn}
In abuse of notation we will also call an IDEM $\mu$ a prior in the following, which is justified by the fact that every IDEM prior has a unique associated IDEM due to Theorem \ref{thmidexpmimpliesexminid}.

\begin{ex}[Partial exchangeability and independent components]
It is easy to see that $\bmX\sim\minid(\mu)$ having conditionally (on $\mu$) independent components almost surely implies partial exchangeability of the associated exchangeable sequence $\ubmX$. Thus, we have that $\mu$ is an IDEM for partially exchangeable data whenever $\mu$ almost surely concentrates on $\cup_{i=1}^d\ee_{i}$, where 
$$\ee_i:=\{\bmx\in E_d\mid x_i<\infty, x_j=\infty\ \forall\ 1\leq j\leq d,\ j\neq i\}.$$
\end{ex}

\subsection{Building tractable IDEM priors from simple building blocks}
\label{sectrafosIDRM}
Since IDRMs can be complicated objects, it is natural to ask whether one can construct such random measures based on simpler building blocks. Theoretically, the answer is simple. Due to Theorem \ref{thmidrandmeasurelevykhintchine} we simply need to look at PRMs on $M_d^0$ to construct all IDRMs. However, in practice, PRMs on $M_d^0$ are not very convenient to model and one might ask the question whether one can obtain IDRMs by simple transformations of tractable subclasses of IDRMs, e.g.\ CRMs. 
It is not hard to see that the property of infinite divisibility is preserved under many natural operations on measures, such as integration. The difficult part is to retain tractable Lévy characteristics, which is especially important, since we will later see that tractable Lévy characteristics are the central ingredient for obtaining tractable formulas for the posterior. Here, we extend results on transformations of univariate infinitely divisible processes from \cite{brueckexactsim2022,brueckmaischerer2023} to multivariate IDRMs to retain infinite divisibility together with tractable Lévy characteristics. In practice one might concatenate several of these operations to obtain a suitable IDEM prior with tractable Lévy characteristics. 

\begin{prop}
\label{proptrafos}
Let $\mu$ denote an IDEM on $E_d$ with Lévy characteristics $(\alpha,\nu)$ and assume that $\mu$ satisfies Condition \ref{condsuppminid}. 
\begin{enumerate}
    \item[$(i)$] Let $g:\ E_d\to E_d$ be measurable and monotone increasing in each component with $g(\binfty)=\binfty$ \textcolor{black}{and $\lim_{t\to\infty}g(t,\ldots,t)=\binfty$}. Define $\mu_g(A):=\int_{\{g(\bmx)\in A\}}\rmd\mu$. Then, $\mu_g$ is an IDEM with Lévy characteristics $(\alpha_g,\nu_g)$, where $\alpha_g(A):=\int_{\{g(\bmx)\in A\}}\rmd\alpha$ and
    $\nu_g(A):=\nu\big(\{\eta \in M_d^0 \mid \eta_g\in A,\ \eta_g(E_d^\prime)$ $\not=0 \}\big)$.
    \item[$(ii)$] Let $g:E_d\to [0,\infty)$ be measurable and bounded on $(\bmy,\binfty]^\complement$ for every $\bmy\in\R^d$ with $g(\binfty)>0$. Then, $\mu^{(g)}(A):=\int_{A} g(\bmx) \mu(\rmd \bmx)$ is an IDEM with Lévy characteristics $(\alpha^{(g)},\nu^{(g)})$, where $\nu^{(g)}(A)=\nu\lc \{ \eta \in M_d^0\mid \eta^{(g)}\in A,\eta^{(g)}(E_d^\prime)\not=0\}  \rc$
    \item[$(iii)$] Let $\beta$ denote a measure on $E_d$ such that $\beta\lc (\bmy,\binfty]^\complement\rc<\infty$ for all $\bmy\in \R^d$ and $\beta(\binfty)>0$. Then, the measure $\mu^{(\beta)} \lc A\rc :=\int_{A}\mu\lc (-\binfty,\bmx]\rc  \beta(\rmd\bmx)$ is an IDEM with Lévy characteristics $(\alpha^{(\beta)},\nu^{(\beta)})$, where $\nu^{(\beta)}(A)=\nu\lc \{ \eta \in M_d^0\mid \eta^{(\beta)}\in A,\eta^{(\beta)}(E_d^\prime)\not=0\}  \rc$.
\end{enumerate}
\end{prop}

The transformation in Proposition \ref{proptrafos}$(i)$ is useful to adjust the margins of an IDEM prior to some predefined margins. For $g(\bmx)=(g_1(x_1),\ldots,g_d(x_d))$ the distribution of the $i$-th margin of $\minid(\mu_g)$ only depends on $g_i$ and thus it may be chosen such that $\e\lk\exp\lc-\mu_g\lc \{\bmx \in E^\prime_d\mid x_i\leq \cdot\}\rc\rc\rk$ follows an arbitrary prespecified survival function.
The intuition behind the definition of $\mu^{(\beta)}$ in Proposition \ref{proptrafos}$(iii)$ is that it can be seen as a natural generalization of the popular smoothing of a univariate CRM given by $\int_A \mu\lc(-\infty,s]\rc\rmd s$.

Additional properties of IDEM priors such as prior moments, moments of mean functionals and their induced dependence structure can be found in Section 10 of the Supplementary Material \cite{brueck2026}.

\subsubsection{Subordination of independent CRMs by IDRM}
\label{subsecsubordcrm}

A popular construction scheme for multivariate models in Bayesian nonparametrics is to build vectors of dependent random measures via a hierarchical construction, see e.g.\ \cite{catalano2023} for an overview. Often, these constructions are based on simple transformations of (vectors of) CRMs. For example, in the context of survival analysis, \cite{camerlengilijoipruenster2021} consider a diffuse homogeneous CRM $\mu_0$ on $\R$ such that 
\begin{align}
    \int_{-\infty}^t \int_0^\infty \min\{1,a\}\rho_i(a)\rmd a\mu_0(\rmd y)<\infty\text{ for every } t\in\R \label{condcrmpolishpsace}
\end{align}
and define $d$ conditionally independent homogeneous completely random measures $\lc \mu_i\rc_{1\leq i\leq d}$ with respective base measures $\alpha_i$ and Lévy measures $\nu_i$ given as the image measure of $(a,y)\mapsto a\delta_y+\infty\delta_\binfty$ under $\rho_i(a)\rmd a \mu_0(\rmd y)$, i.e.\ $\nu_i(A)=\int_{-\infty}^\infty\int_0^\infty \id_{\{ a\delta_y +\infty\delta_\binfty \in A\}}\rho_i(a)\linebreak\rmd a \mu_0(\rmd y)$. 
The resulting $\lc \mu_i\rc_{1\leq i\leq d}$ are dependent random measures, which are conditionally independent CRMs given $\mu_0$. \cite{camerlengilijoipruenster2021} define a model for partially exchangeable data by setting the marginal random survival function of margin $i$ to $\exp( -\mu_i^{(\kappa_i)}\lc (-\infty,t] \rc)$, where for every measurable function $\kappa: \R\times\R\to [0,\infty)$ and every measure $\eta$ on $\R$ we define 
$$ \eta^{(\kappa)}\lc (-\infty,t] \rc:=\int_{-\infty}^t \int_\R \kappa(s,y)\eta(\rmd y)\rmd s \ \forall\  t\in\R$$
and $\eta^{(\kappa)}\lc \infty \rc=\infty$. Now, their model can be embedded into our framework by defining an exponent measure $\mu_{\mu_1^{(\kappa_1)},\ldots,\mu_d^{(\kappa_d)}}$ concentrated on $\cup_{i=1}^d \mathcal{E}_i$ via
\begin{align}
    \mu_{\mu_1^{(\kappa_1)},\ldots,\mu_d^{(\kappa_d)}} \lc   \ee_i(t) \rc:=\mu_i^{(\kappa_i)}\lc  (-\infty,t]\rc  \ \forall\ t\in\R \text{ and } 1\leq i \leq d \label{defsubordinatedCRMexpmeasure}
\end{align}
where $\ee_i(t):=\big\{ \bmx\in E_d \mid x_i\leq t, x_j=\infty\ \forall\ 1\leq j\leq d,\ j\neq i \big\} $ and $\mu_{\mu_1^{(\kappa_1)},\ldots,\mu_d^{(\kappa_d)}}(\binfty)=\infty$. It was already implicitly derived in \cite{camerlengilijoipruenster2021} that $\mu_{\mu_1^{(\kappa_1)},\ldots,\mu_d^{(\kappa_d)}}$ is infinitely divisible, however its Lévy characteristics were not tractable. In Proposition  \ref{propsubordinationcrm} $(ii)$ below we show that the vector of random measures $\lc \mu_i\rc_{1\leq i\leq d}$ is not only a vector of CRMs conditionally on $\mu_0$, but also unconditionally. This result extends the classical subordination of Lévy processes to homogeneous CRMs and was concurrently and independently established in \cite[Theorem 2.3]{catalano2026}. Further, we extend the framework to the case where $\mu_0$ can be an IDRM instead of a homogeneous CRM, which allows for a continuous shared structure. To the best of the author's knowledge, such continuous structures at the root could not be considered to date since their Laplace transform has not been tractable. The following result provides the infinite divisibility and Lévy characteristics of such ``subordinated'' CRMs and is of independent interest beyond the scope of this paper. To ease the notation in the following let $\otimes_{i=1}^d \crm_{\eta_i(\rmd (a,\bm b))}$ denote the law of $d$ independent CRMs with base measure $\infty \delta_\binfty$ and Lévy measures given as the image measure of $(a,\bm b)\mapsto a\delta_{\bm b}\textcolor{black}{+\infty\delta_\binfty}$ under $\eta_i$, respectively, dropping $\otimes_{i=1}^d$ when $d=1$.

\begin{prop}
\label{propsubordinationcrm}
Assume that $\mu_0$ is an IDRM on $\R$ with Lévy characteristics $(\alpha_0,\nu_0)$ such that $\mu_0\lc(-\infty,t]\rc<\infty$ almost surely for all $t\in\R$. Then, conditionally on $\mu_0$, let $\lc\mu_i\rc_{1\leq i\leq d}$ denote $d$ independent homogeneous CRMs on $\R$ with base measures $\alpha_i$ and Lévy measures $\nu_i$ given as the image measures of $(a,b)\mapsto a\delta_b\textcolor{black}{+\infty\delta_\infty}$ under $\rho_i(a)\rmd a \mu_0(\rmd b)$, respectively. The following is true:
\begin{enumerate}
    \item[$(i)$]  $\mu_i((-\infty,t])<\infty$ almost surely for all $\leqd$ and $t\in\R$. Moreover, \textcolor{black}{the vector of random measures $\lc \mu_i\rc_{1\leq i\leq d}$ is jointly infinitely divisible,} i.e.\ for every $n\in\N$ one can find i.i.d.\ vectors of IDRMs $\lc \lc \mu_i^{(j,n)}\rc_{1\leq i \leq d}\rc_{1\leq j\leq n}$ s.t.\ $\lc \mu_i\rc_{1\leq i\leq d}\sim \lc  \sum_{j=1}^n \mu_i^{(j,n)}\rc_{1\leq i\leq d}$. The Lévy characteristics of $\mu_i$ are given by $(\alpha_i,\nu_{1,i}+\nu_{2,i})$ where
     $$ \nu_{1,i}(A)=\int_\R \int_0^\infty  \id_{\{ a\delta_b +\infty\delta_\binfty\in A\}} \rho_i(a)\rmd a \alpha_0(\rmd b)$$
     and 
     $$ \nu_{2,i}(A)=\int_{\textcolor{black}{M^0_1}} \int_{\textcolor{black}{M_1}}\id_{\{\infty\delta_{\infty}\neq \eta_i\in A\}} CRM_{\rho_i(a)\rmd a \eta(\rmd b)}(\textcolor{black}{\rmd\eta_i}) \nu_0(\rmd\eta) .$$
 \item[$(ii)$]\textcolor{black}{If $\mu_0$ is a CRM each $\lc \mu_i\rc_\leqd$ is a CRM as well.}
 \item[$(iii)$] When $\mu^{(\kappa_i)}_i\lc (-\infty, t]\rc<\infty$ almost surely for every $t\in \R$ and $1\leq i\leq d$ then $\mu_{\mu_1^{(\kappa_1)},\ldots,\mu_d^{(\kappa_d)}}$ as defined in (\ref{defsubordinatedCRMexpmeasure}) is an IDEM with base measure $\mu_{\alpha_1^{(\kappa_1)},\ldots,\alpha_d^{(\kappa_d)}}$ and Lévy measure $\nu=\nu_1+\nu_2$, where 
    $$  \nu_1(A)= \sumd \int_\R \int_{0}^\infty  \id_{\{ \infty\delta_{\infty}\neq \otimes_{j=1}^{i-1} \delta_\infty \times (a\delta_b)^{(\kappa_i)}\otimes_{j=i+1}^d \delta_\infty \in A \}} \rho_i(a)\rmd a \alpha_0(\rmd b) $$
 and 
$$ \nu_2( A)=\int_{\textcolor{black}{M^0_1}}\int_{\times_{i=1}^d \textcolor{black}{M_1}}  \id_{\{ \infty\delta_{\infty}\neq  \mu_{\eta^{(\kappa_1)}_1,\ldots,\eta^{(\kappa_d)}_d} \in A \}} \otimes_{i=1}^d \crm_{\rho_i(a)\rmd a \eta(\rmd b)}(\rmd \eta_i) \nu_0(\rmd\eta).$$

\end{enumerate}

\end{prop}

\begin{rem}
    The assumption that the $(\mu_i)_{0\leq i\leq d}$ in Proposition \ref{propsubordinationcrm} are random measures on $\R$ may be relaxed to allow the $\mu_i$ to be random measures on an arbitrary Polish space $\mathbb{Y}$. The results of Proposition \ref{propsubordinationcrm} remain valid, where $\kappa_i(s,y)$ needs to be defined on $\R\times\mathbb{Y}$, implicitly assuming that the $(\mu_i)_{0\leq i \leq d}$ are well-defined. Similarly, Proposition \ref{proptrafos} ($i$)-($ii$) may be straightforwardly extended to IDRMs on $\mathbb{Y}$. To avoid excessive notation and technical regularity conditions we have focused on the most important case $\mathbb{Y}=\R$ and leave spelling out the details to the reader. 
\end{rem}

\subsection{Examples of IDEM priors from the literature}
\label{secexamples}
Before proceeding with the posterior distribution of an IDEM prior, we embed several well-known models into our framework to illustrate how IDEM priors are constructed in the literature.

\begin{ex}[Neutral-to-the-right-priors \cite{fergusonprioronprob1974,doksumtailfree1974,hjort1990}]
\label{exNTRpriors}
Continuing Example \ref{exCRM} in the case $d=1$, assume we are given a diffuse CRM $\mu$ on $(0,\infty)$. In the spirit of Theorem \ref{thmidexpmimpliesexminid} we can then define a prior in terms of $X\sim\minid(\mu)$. The corresponding prior is well known in Bayesian nonparametrics as an NTR prior \cite{fergusonprioronprob1974,doksumtailfree1974,hjort1990}. 
We obtain that $\ubmX$ is a min-id sequence with exponent measure $\lambda(A)=\lambda_\alpha(A)+\lambda_\nu(A)$, where $\lambda_\nu(A)=\int_{(0,\infty)^2} \otimes_{j\in\N} \minid (a\delta_b+\infty\delta_\binfty) \upsilon(\rmd (a,b)).$
Thus, $\lambda_\nu$ is the mixture of the law of i.i.d.\ sequences, each of them being concentrated on the two points $\{b,\infty\}$. Further, by Corollary \ref{corstochdecompminid} it is easy to see that $\ubmX$ has the stochastic representation 
\begin{align}
    \ubmX\sim \min\Big\{\underline{\bmX}_1; \ubmX_2\Big\} ,\label{stochrepntrseq}
\end{align}\
where $\ubmX_2:=\min_{k\in\N} \ubmx^{(k)}$ is an exchangeable min-id sequence such that $\ubmx^{(k)}\in \{b_k,\infty\}^\N$ is an exchangeable sequence with marginal distribution $(1-\exp(-a_k))\delta_{b_k}+\exp(-a_k)\delta_{\infty}$ and $\ubmX_1= (y_k)_{k\in\N}$ is an i.i.d.\ sequence with $y_k\sim\minid(\alpha)$. Analogous statements hold when $\mu$ is a multivariate CRM, generalizing NTR priors to the multivariate setting. A detailed treatment of multivariate NTR priors can be found in Section 11 of the Supplementary Material \cite{brueck2026}.
\end{ex}

\begin{ex}[Hierarchically dependent mixture hazards \cite{camerlengilijoipruenster2021}]
\label{exhierdepmixhaz}
As already explained at the beginning of Section \ref{subsecsubordcrm}, \cite{camerlengilijoipruenster2021} construct vectors of dependent CRMs $(\mu_i)_{1\leq i\leq d}$ on a Polish space $\mathbb{Y}$, which are conditionally i.i.d.\ given a CRM $\mu_0$. Then, they define the respective survival functions of $d$ groups of observations by $S_i(t):=\exp\lc -\mu_i^{(\kappa_i)}\lc(-\infty,t]\rc\rc$, which is equivalent to assuming $\bmX\sim \minid \lc\mu_{\mu_1^{(\kappa_1)},\ldots,\mu_d^{(\kappa_d)} }\rc$. Thus, the model concentrates on exponent measures of random vectors with independent components, implying partial exchangeability of the resulting sequence $\ubmX$. By Proposition \ref{propsubordinationcrm}, $\mu_{\mu_1^{(\kappa)},\ldots,\mu_d^{(\kappa)}}$ is an IDEM prior with tractable Lévy characteristics. \textcolor{black}{In the case $d=1$, which also includes the models of \cite{Dykstralaud1981,loweng1989} as special cases, Example 12 in the Supplementary Material illustrates that there is actually one canonical choice of $\mathbb{Y}$ and $\kappa$ which allows to construct all possible models \textcolor{black}{$\mu_1^{(\kappa_1)}$}.}

\end{ex}

\begin{ex}[Hazard rate mixtures \cite{lijoinipoti2014}]
\label{exhazratmix}
  Two completely random measures $\mu_1$ and $\mu_2$ on a Polish space $\mathbb{Y}$ are coupled by introducing a latent variable $Z\in[0,1]$ such that, conditionally on $Z$, the CRMs can be represented as $\mu_1(A)=\Tilde{\mu}_1(A)+\Tilde{\mu}_0(A)$ and $\mu_2(A)=\Tilde{\mu}_2(A)+\Tilde{\mu}_0(A)$, where $\lc\Tilde{\mu}_i\rc_{0\leq i\leq 2}$ are independent CRMs on $\mathbb{Y}$ with base measure $\infty\delta_{\binfty}$ and intensities $Z\tilde{\nu}$, $(1-Z)\tilde{\nu}$ and $(1-Z)\tilde{\nu}$ for some Lévy measure $\tilde{\nu}$. Thus, $Z$ steers the dependence of $\mu_1$ and $\mu_2$ by governing the magnitude of $\Tilde{\mu}_0$ relative to $\Tilde{\mu}_1$ and $\Tilde{\mu}_2$. The authors define a corresponding random vector $\bmX\in(0,\infty)^2$ via the survival function $S(t_1,t_2)=\exp\lc -\int_0^{t_1} \int_{\mathbb{Y}} \kappa(s,y)\mu_1(\rmd y)\rmd s-\int_0^{t_2} \int_{\mathbb{Y}} \kappa(s,y)\mu_2(\rmd y)\rmd s \rc$, where $\kappa:\R\times\mathbb{Y}\to[0,\infty)$. Conditionally on $Z$ and $(\mu_i)_{i=1,2}$, $\bmX$ is min-id with independent components and the associated exponent measure is given by $\mu_{\mu_1^{(\kappa)},\mu_2^{(\kappa)}}$, which is an IDEM conditionally on $Z$. We obtain that the Lévy measure of $\mu$ is given by $\nu(A)=\nu_0(A)+\nu_1(A)+\nu_2(A)$, where $\nu_0(A)=Z\tilde{\nu}\lc 
 \{(a,y)\in (0,\infty)\times \mathbb{Y} \mid \mu_{(a\delta_y)^{(\kappa)}, (a\delta_y)^{(\kappa)}} \in A\}\rc$, 
 $\nu_1(A)=(1-Z)\tilde{\nu}\lc \{
  (a,y)\in (0,\infty)\times \mathbb{Y}\mid \mu_{(a\delta_y)^{(\kappa)}, \delta_\infty} \in A  \}\rc$ and similarly for $\nu_2$.
\end{ex}

An additional example embedding the frameworks of \cite{epifanilijoi2010,rivapalacioleisen2018} into the framework of IDEM priors can be found in Section 12 of the Supplementary Material \cite{brueck2026}.
\section{Posterior distribution of IDEM priors}
\label{secpostdistr}
In this section, we derive the posterior distribution of an IDEM prior by deriving the conditional distribution of $\ubmX$ given $\ubmX_n$ and applying de Finetti's theorem to $(\bmX_{n+1},\bmX_{n+2},\ldots)$ to identify the posterior of an IDEM $\mu$. 

\subsection{Conditional distribution of min-id sequences}
A central ingredient to derive the posterior distribution of an IDEM $\mu$ is the conditional distribution of $\ubmX$ given $\ubmX_n$. In this subsection, we exploit the results of \cite{dombryeyiminko2013regular} to describe these conditional distributions.
Before formally introducing the conditional distribution of an exchangeable min-id sequence, we must introduce some further terminology. Recall the hitting scenario from Definition \ref{defhittingscen}, which determines how the minimum over all atoms of the PRM in (\ref{prmrepminid}) is attained. We will see that the conditional hitting scenario, i.e.\ the conditional distribution of a hitting scenario given $\ubmX_n$, will play a major role in the representation of the posterior. Formally, the definition of a conditional hitting scenario is as follows.

\begin{defn}
    Let $\ubmX$ be an exchangeable min-id sequence. The conditional hitting scenario $\Tilde{\Theta}=\lc \Tilde{\Theta}_1,\ldots,\Tilde{\Theta}_L  \rc$ of the $IJ$-margin $\bmX_{IJ}$ of $\ubmX$ is the random partition of $IJ$ which follows the distribution of the hitting scenario $\Theta(IJ)$ conditionally on $\bmX_{IJ}$, where $L=L(\Tilde{\Theta})$ denotes the length of the partition. Its probability distribution is denoted as $\tau(\bmX_{IJ},\theta):= P(\Tilde{\Theta}=\theta\mid \bmX_{IJ})$.  
\end{defn}

\begin{ex}[Conditional hitting scenarios of NTR process]
\label{excondhitscenntr}
Surprisingly, the conditional hitting scenarios of NTR processes without fixed jumps are very simple. There is only one conditional hitting scenario with positive probability, which can be directly read off from the sample $\ubmX_n$. Simply create a partition of $\{1,\ldots,n\}$ by assigning $i$ and $j$ to the same subset of $\{1,\ldots,n\}$ if $X_{i}=X_{j}$, i.e.\ if random variables $i$ and $j$ form a tie in the data.
This can be verified using the stochastic representation of a min-id sequence associated to an NTR process given by (\ref{stochrepntrseq}). First, recall that the associated non-negative and non-decreasing additive process $\mu\lc(-\infty,t]\rc=\alpha((-\infty,t])+\sum_{b_k\leq t} a_k\delta_{b_k}$ cannot have any fixed points of discontinuity, which implies that the $(b_k)_{k\in\N}$ and $(y_k)_{k\in\N}$ are almost surely distinct. Due to (\ref{stochrepntrseq}), we have $\ubmX\sim \min\{ \min_{k\in\N}  \ubmx^{(k)}, (y_k)_{k\in\N}\}$, where $\ubmx^{(k)}\in\{b_k,\infty\}^\N$. Thus we can only have that $X_i=X_j$ if and only if there exists a $k$ such that $x^{(k)}_i=X_i=b_k=X_j=x^{(k)}_j$, which proves the claim. 
\end{ex}

The second ingredient that will be needed to formulate the posterior distribution of an IDEM prior is the conditional probability distribution of the exponent measure $\lambda$.
Recall that $\lambda$ is generally an infinite measure, which is why the definition of a conditional probability distribution of $\lambda$ requires some care. In analogy to the definition of a regular conditional probability we define the conditional probability distribution of the measure $\lambda$ given $\bmx_{IJ}:=(x_{i_k,j_k})_{1\leq k\leq p}$ via the integral equation
\begin{align}
 &\int_{E_{d\times \N}} \mathsf{G}\big( \bmx_{IJ}, \ubmx_{\setminus IJ} \big) \lambda(\rmd \ubmx)\nonumber \\
 &=  \int_{E^\prime_p
 }\int_{(-\infty,\infty]^{d\times \N\setminus IJ}} \mathsf{G}\big( \bmx_{IJ}, \ubmx_{\setminus IJ}\big) K_{IJ}\big( \bmx_{IJ},  \rmd \ubmx_{\setminus IJ} \big) \lambda^{(IJ)}\big( \rmd \bmx_{IJ} \big), \label{defconddistexpmeasure}
\end{align}
where $\ubmx_{\setminus IJ}:=(x_{i,j})_{ (i,j)\in d\times \N\setminus IJ}$,  $\mathsf{G}(\cdot,\cdot)$ is a measurable function that vanishes on $\{\binfty\}\times (-\infty,\infty]^{d\times \N\setminus IJ}$ and $K_{IJ}(\bmx_{IJ},\cdot)$ denotes a probability kernel on $(-\infty,\infty]^{d\times \N\setminus IJ}$ for every $\bmx_{IJ}\in E^\prime_p$. Formally, the existence of such probability kernels $K_{IJ}$ has been verified for min-id sequences $\ubmX$ with continuous marginal distribution and compact index set in \cite[Appendix A.2]{dombryeyiminko2013regular}, but it can be checked that their definition only relies on the assumption of working in Polish spaces, which is trivially satisfied in our setting. Further, \cite{dombryeyiminko2013regular} provide formulas for the conditional distribution of continuous min-stable processes with continuous marginal distributions, which allows to derive the distribution of $\ubmX_m$ given $\bmXIJ$ under the following strengthening of Condition \ref{condsuppminid}, which implies continuous marginal distributions of $\ubmX$ and well-definedness of the conditional hitting scenario.

\begin{cond}
\label{condcontmargins}
    For all $1\leq i\leq d$ the IDEM $\mu$ is such that the random functions $h_i(t):=\mu\lc \{\bmx \in E_d^\prime \mid x_i\leq t\}\rc$ are stochastically continuous with $\lim_{t\to\infty} h_i(t)=\infty$ almost surely.
\end{cond}
Note that the condition does not exclude that a realization of $\mu$ can be a discrete measure, it just excludes that a margin can have a fixed point of discontinuity with positive probability. A proof of the fact that Condition \ref{condcontmargins} is equivalent to continuous margins of $\ubmX$ is given in \cite[Corollary 2.3]{brueckmaischerer2023}. Under Condition \ref{condcontmargins}, an application of \cite[Theorem 3.3]{dombryeyiminko2013regular} yields the  posterior predictive distributions, i.e.\ the conditional distribution of $\ubmX$ given $\ubmX_n$.

\begin{thm}
\label{thmdombryeyiminkcond}
\textcolor{black}{Let $\mu$ denote an IDEM satisfying Condition \ref{condcontmargins} and let $\ubmX$ denote the corresponding exchangeable min-id sequence}. Then, the regular conditional distribution of $\ubmX$ given $\ubmX_n$ is given by
\begin{align}
    &P\lc \bmX_{n+1}>\bmx_{1},\ldots,\bmX_{n+k}>\bmx_{k} \ \big\vert\ \ubmX_n\rc \nonumber \\
    &= \exp\Big(   -\lambda \lc \Big\{  \ubmy \ \big\vert\ \lc \bmy_{n+j}\rc_{1\leq j\leq k} \in \big( \ubmx_{k},\binfty\big]^\complement   \text{ , } \ubmy_n>\ubmX_n  \Big\}\Big) \rc \label{eqnpostcomp1}\\
    &\sum_{\tilde{\theta} \in \mathcal{P}_n}  \tau(\ubmX_n, \tilde{\theta}) 
    \prod_{l=1}^{L(\tilde{\theta})} \frac{K_{\tilde{\theta}_l}\lc \bmX_{\tilde{\theta}_l}, \big\{ \ubmy_{\setminus \Tilde{\theta}_l} \mid  \ubmy_{n\setminus \Tilde{\theta}_l}>\ubmX_{n\setminus \Tilde{\theta}_l}, \lc\bmy_{n+j}\rc_{1\leq j\leq k}>\ubmx_k  \big\} \rc }{ K_{ \tilde{\theta}_l}\lc \bmX_{\Tilde{\theta}_l} ,  \big\{ \ubmy_{\setminus \tilde{\Theta}_l} \mid  \ubmy_{n\setminus \tilde{\theta}_l}>\ubmX_{n\setminus \Tilde{\theta}_l}\big\}\rc } ,\label{eqnpostcomp2}
\end{align}
where $\mathcal{P}_n$ denotes the set of all possible partitions of $d\times n$.
\end{thm}

\subsection{A general representation of the posterior}

The formula from Theorem \ref{thmdombryeyiminkcond} splits into the product of two terms: the term (\ref{eqnpostcomp1}) and the remaining term (\ref{eqnpostcomp2}) that involves a sum over all partitions of $d\times n$. This implies that, conditionally on $\ubmX_n$,
\begin{align}
    \lc \bmX_{n+j}\rc_{j\in\N}\sim \min\big\{ \underline{\bm Y}^{(n)},\underline{\bm Z}^{(n)}\big\}, \label{stochrepposterior}
\end{align} 
where $\underline{\bm Y}^{(n)}$ is distributed according to (\ref{eqnpostcomp1}) and independent of $\underline{\bm Z}^{(n)}$, which is distributed according to  (\ref{eqnpostcomp2}).
The term (\ref{eqnpostcomp1}) is easily identified as the restriction of the exponent measure $\lambda$ to the set $A(\ubmX_n):= \big\{ \underline{\bmy}\in E_{d\times \N} \mid   \ubmy_{n}>\ubmX_{n}  \big\}$,
which due to Theorem \ref{thmidexpmimpliesexminid} can be represented as
\begin{align*}
    \lambda(\cdot\cap A(\ubmX_n))&= \lambda_\alpha(\cdot \cap A(\ubmX_n)) + \int_{M_d^0}  \otimes_{j\in\N} \minid(\eta) \lc \cdot \cap A(\ubmX_n)\rc \nu(\rmd\eta) .
\end{align*}
Therefore, for every measurable set $B \subset E^\prime_{d\times \N}$ we have
\begin{align*}
      &\lambda(\{ \ubmy\in E_{d\times \N}\mid \ubmy_{\setminus d\times n}\in B, \ubmy_n\in A(\ubmX_n)\})\\
      &= \lambda_\alpha(B) +\int_{M_d^0}   \otimes_{j\in\N} \minid(\eta) \lc \{ \lc \bmy_{n+j}\rc_{j\in\N}\in B \} \rc \prod_{j=1}^n \exp\lc -\eta\lc (\bmX_j,\binfty]^\complement \rc\rc \nu(\rmd\eta) 
\end{align*}
Thus, conditionally on $\ubmX_n$, $\ubmY^{(n)}$ is distributed as an exchangeable min-id sequence with exponent measure given by
\begin{align}
    \lambda_\alpha(\cdot)+ \int_{M_d^0}  \otimes_{j\in\N} \minid(\eta) \lc \{ \lc \bmy_{n+j}\rc_{j\in\N}\in \cdot \} \rc\overline{\nu}_n(\rmd\eta), \label{defpostexpmeasure}
\end{align}   
where
\begin{align}  \overline{\nu}_n(\rmd\eta):=  \prod_{j=1}^n \exp\lc -\eta\lc (\bmX_j,\binfty]^\complement \rc\rc \nu(\rmd\eta) \label{defpostlevymeasure}
\end{align}
defines a valid Lévy measure of an IDEM, since it satisfies the conditions from Theorem \ref{thmidrandmeasurelevykhintchine}. Thus, $\underline{\bm Y}^{(n)}$ is
an exchangeable min-id sequence associated to the IDEM $\overline{\mu}_{n}$ with base measure $\alpha$ and Lévy measure $\overline{\nu}_n$.

The remaining term (\ref{eqnpostcomp2}) is more complicated to identify, but it is easily seen that $\underline{\bm Z}^{(n)}$ is generated by a hierarchical structure:
\begin{enumerate}
    \item Draw a conditional hitting scenario $\Tilde{\Theta}$ according to $\tau(\ubmX_n,\cdot)$.
    \item Conditionally on $\Tilde{\Theta}=(\Tilde{\Theta}_1,\ldots,\Tilde{\Theta}_L)$, draw $L$ independent sequences $\lc\ubmZ^{(\Tilde{\Theta}_l)}\rc_{1\leq l\leq L}$ with distribution $K_{ \Tilde{\Theta}_l}(\bmX_{\Tilde{\Theta}_l},\cdot)$ conditionally on $\{  \ubmy_{n\setminus\Tilde{\Theta}_l}>\ubmX_{n\setminus\Tilde{\Theta}_l}\}$ and set  
    \begin{align}
       \underline{\bm Z}^{(n,l)}:=\lc \bmZ_{n+j}^{(\Tilde{\Theta}_l)}\rc_{j\in\N} ,\  1\leq l\leq L\label{defcondextrseq} .
    \end{align}
    \item Define $\underline{\bm Z}^{(n)}=\min_{1\leq l \leq L}  \underline{\bm Z}^{(n,l)} $.
\end{enumerate}
It should be mentioned that, even though $\lc \ubmX_{n+j}\rc_{j\in\N}$ is real-valued, $\underline{\bmZ}^{(n,l)}$ may also assume the value $\infty$ in its entries as its distribution is determined by $K_{\tilde{\Theta}_l}(\bmX_{\Tilde{\Theta}_l},\cdot)$, which are probability distributions on $E_{d\times \N\setminus\tilde{\theta}_l}$. One can show that each $\ubmZ^{(n.l)}$ is also exchangeable.

\begin{lem}
\label{lemcondhitseqex}
    \textcolor{black}{In the setting of Theorem \ref{thmdombryeyiminkcond}} and conditionally \textcolor{black}{on $\ubmX_n$ and} the conditional hitting scenario $\Tilde{\Theta}=(\Tilde{\Theta}_1,\ldots,\Tilde{\Theta}_L)$, the sequences $\lc \lc \bmZ^{(n,l)}_{j}\rc_{j\in\N}\rc_{1\leq l\leq L}$ from (\ref{defcondextrseq}) are exchangeable \textcolor{black}{for almost every $(\ubmX_n,\Tilde{\Theta})$}.
\end{lem}

A consequence of Lemma \ref{lemcondhitseqex} is that, by de Finetti's theorem, each sequence of random vectors $ \underline{\bmZ}^{(n,l)}$  can be associated to a random multivariate survival function $S^{(n,l)}$ such that $ \ubmZ^{(n,l)}$ is conditionally i.i.d.\ given $S^{(n,l)}$. Thus, we can define the random survival function associated to $\ubmZ^{(n)}$ as
 $S_n(\bmx):=\prod_{l=1}^L  S^{(n,l)} (\bmx) $,
which immediately implies that $\ubmZ^{(n)}$ is exchangeable as well.

\begin{thm}[Posterior of IDEM prior]
\label{thmmainresult}
    Let $\mu$ be an IDEM and let $\ubmX$ denote the corresponding exchangeable min-id sequence. Under Condition \ref{condcontmargins}, the posterior predictive distribution of an IDEM prior can be characterized in terms of de Finetti's theorem as follows. Conditionally on $\ubmX_n$, the random survival function associated to the exchangeable sequence $(\bmX_{n+j})_{j\in\N}$ is given by
    $$  \exp\lc-\overline{\mu}_{n}\lc (\bmx,\infty]^\complement\rc\rc S_n(\bmx), $$
    where $\overline{\mu}_{n}$ denotes an IDEM with Lévy-Khintchine characteristics $(\alpha,\overline{\nu}_n)$ independent of $S_n=\prod_{1\leq l\leq L} S^{(n,l)}$, where the $S^{(n,l)}$ are the random survival functions associated to $\ubmZ^{(n,l)}$. 
\end{thm}

It should be mentioned that we had to formulate Theorem \ref{thmmainresult} in terms of de Finetti's theorem, since the author does not know whether or not $S^{(n,l)}$ is always the survival function of a min-id distribution. When imposing an additional regularity condition, one can show that each $S^{(n,l)}$ may be identified with a random exponent measure $\overline{\mu}_n^{(l)}$ and one can phrase Theorem \ref{thmmainresult} directly in terms of the posterior of $\mu$, see Theorem \ref{thmconddistmixtureofminid} below.
In the univariate case $\overline{\mu}_n^{(l)}\lc(x,\infty]^\complement\rc:=-\log(S^{(n,l)}(x))$ always defines a measure and one can always express the posterior of $\lc\bmX_{n+j}\rc_{j\in\N}$ directly in terms of the posterior of $\mu$.

\begin{cor}[Univariate case]
    In the setting of Theorem \ref{thmmainresult} with $d=1$ the posterior distribution of an IDEM prior $\mu$ is given by
    $$ \mu \sim \overline{\mu}_{n} +\sum_{1\leq l\leq L} \overline{\mu}_n^{(l)}, $$
    where $\overline{\mu}_{n}$ denotes an IDEM with Lévy-Khintchine characteristics $(\alpha,\overline{\nu}_n)$ independent of the random exponent measures $\lc\overline{\mu}_n^{(l)}\rc_{1\leq l\leq L}$ defined via $\overline{\mu}_n^{(l)}((-\infty,x]):=-\log(S^{(n,l)}(x)) $.
\end{cor}

\begin{rem}[Partial exchangeability and partial observations]
In the case of partial exchangeability, i.e.\ $\mu$ concentrating on $\cup_{i=1}^d \ee_i$, it is often the case that one observes different numbers of observations per component of $\bmX$, i.e.\ one observes $\bmX_{IJ}$ instead of $\ubmX_n$. Theorem \ref{thmmainresult} can easily be adapted to partial observations of $\ubmX_n$ by considering the conditional hitting scenario of $IJ$ and the respective restriction of $\lambda$. A thorough formulation of Theorem \ref{thmmainresult} under this more general framework would require additional notation and interpretation, which is why we have refrained from spelling out the details and leave this exercise to the reader. A concrete example where we have derived the posterior for partial observations is the case of the multivariate NTR prior, see Section 11 in the Supplementary Material \cite{brueck2026}.
\end{rem}

Clearly, the most complex quantities appearing in the posterior are the conditional distributions $K_{\Tilde{\Theta}_l}$ and the distribution of the conditional hitting scenarios $\tau(\ubmX_n,\cdot)$. Providing universal tractable formulas for these quantities seems rather difficult as they are dependent on the specific choice of $\mu$. Section \ref{seccclosedformreppost} below shows that, under some simple regularity condition, one can obtain tractable formulas for a large class of IDEM priors which concentrate on continuous distributions, which also allow to represent $S^{(n,l)}$ via a min-id distribution. On the other hand,  IDEM priors which concentrate on discrete distributions usually cannot satisfy the regularity conditions of Section \ref{seccclosedformreppost}. Nevertheless, we have derived a closed-form analytical representation of the posterior of a multivariate NTR prior in Section 11 of the Supplementary Material \cite{brueck2026}, which can serve as a template for deriving the posterior of other discrete IDEM priors.

\section{A large class of tractable models for continuous distributions}
\label{seccclosedformreppost}

Since the results of Section \ref{secpostdistr} are rather theoretical and difficult to grasp, we here provide closed-form representations of the posterior distribution of continuous IDEMs which follow some regularity condition that is often met in practice. The strategy is to impose a regularity condition on $\mu$ via $(\alpha,\nu)$ such that the resulting exponent measure $\lambda$ retains this regularity condition. This is also convenient in practice since an IDEM prior will usually be specified in terms of its corresponding Lévy characteristics $(\alpha,\nu)$.

Our idea can be viewed as a generalization of the ideas of \cite[Section 4.3]{dombryeyiminko2013regular} to the Bayesian framework, as they provide formulas for $K_{IJ}(\bmx_{IJ},\cdot)$ and $\tau(\bmx_{IJ},\cdot)$ when $\lambda^{(IJ)}$ has a density w.r.t.\ the Lebesgue measure. However, conditions on the level of $\lambda^{(IJ)}$ are not suitable in our framework, since the existence of a density of $\lambda^{(IJ)}$ w.r.t.\ the Lebesgue measure is usually not realistic and difficult to check when imposing a prior on $\mu$. For example, when $\mu$ concentrates on $\cup_{i=1}^d \ee_i$, $\lambda^{(IJ)}$ might or might not have a Lebesgue density, depending on the specific choice of $\mu$. Therefore, a condition on the level of $\lambda^{(IJ)}$ does not seem natural in a Bayesian setting and we take a different approach by imposing a regularity condition on the level of $(\alpha,\nu)$. To be specific, we assume that for $\nu$-almost every $\eta$ $\minid(\eta)$ has a density which can be represented as a sum of lower-dimensional Lebesgue densities w.r.t.\ the disjoint sets 
$$\ee_D=\{ \bmx\in E_d \mid x_i<\infty \text{ for all }i\in D \text{ and }x_j=\infty\text{ for all } 1\leq j\not\in D \leq d\},$$ 
where $D\subseteq\{1,\ldots,d\}$. 
Denoting $\Lambda$ as the Lebesgue measure on $\R$ we define a measure on $(-\infty,\infty]$ by $\Lambda^{(\infty)}:=\delta_\infty+\Lambda$ with the corresponding measure on $E_p$ defined as
$$\Lambda^{(\infty)}_p(\rmd\bmx):= \otimes_{i=1}^p \Lambda^{(\infty)} (\rmd x_i), $$
which will act as a dominating measure for $\minid(\eta)$.
The definition implies that $\Lambda^{(\infty)}_d(\cdot\cap \ee_D)$ corresponds to the $\vert D\vert$-dimensional Lebesgue measure on the $\vert D\vert$ finite components in $\ee_D$. Now, the idea is to deduce that when $\nu$-almost every $\minid(\eta)$ is absolutely continuous w.r.t.\ $\Lambda^{(\infty)}_d$ this implies that $\lambda^{(IJ)}$ is absolutely continuous w.r.t. $\Lambda^{(\infty)}_{\vert IJ\vert}$ as well. Formally, this translates to the following simplifying assumption, which is a strengthened version of Condition \ref{condcontmargins} .

\begin{cond}
    \label{condabsolutecont}
     The Lévy measure $\nu$ of the IDEM $\mu$ satisfies $\lim_{t\to\infty} \int_{M_d^0}  \minid(\eta)\big( \{\bmx\in E_d\mid $ $ x_i\leq t\}\big)\nu(\rmd\eta)=\infty$ for all $1\leq i\leq d$. Moreover, for $\nu$-almost every $\eta$, the probability measure $\minid(\eta)$ and the measure $\alpha$ admit densities $f_\eta(\cdot)$ and $g_\alpha(\cdot)$ w.r.t.\ $\Lambda^{(\infty)}_d$.
\end{cond}

It is worth mentioning that $f_\eta$ is a probability density whereas $g_\alpha$ is generally not, and $g_\alpha(\binfty)=\infty$. Further, $f_\eta(\binfty)>0$ is possible, since $\minid(\eta)$ is a probability distribution on $E_d$ even though $\minid(\mu)$ concentrates on $\R^d$. The following example illustrates Condition \ref{condabsolutecont}.

\begin{ex}
Let us illustrate the bivariate case by a continuation of Example \ref{exhazratmix}, assuming $\lim_{s\to\infty}\kappa(s,y)>0$ for all $y\in\mathbb{Y}$. We know that $g_\alpha=\infty\id_{\{\bmx=\binfty\}}$ and that $\nu$ concentrates on measures of the form $\mu_{(a\delta_y)^{(\kappa)},\delta_\infty}$, $\mu_{\delta_\infty,(a\delta_y)^{(\kappa)}}$, $\mu_{(a\delta_y)^{(\kappa)},(a\delta_y)^{(\kappa)}}$ for some $a\in(0,\infty)$ and $y\in\mathbb{Y}$. Thus, $\minid(\eta)$ has corresponding densities w.r.t.\ $\Lambda^{(\infty)}_2$ of the form $f_{\mu_{(a\delta_y)^{(\kappa)},(a\delta_y)^{(\kappa)}}}(\bmx)=  \lc\prod_{i=1}^2 a\kappa(x_i, y) \rc \exp\lc -a\sum_{i=1}^2 \int_0^{x_i} \kappa(s, y)\rmd s  \rc \id_{\{\bmx\in\R^2\}}$, $f_{\mu_{(a\delta_y)^{(\kappa)},\delta_\infty}}(\bmx)=a\kappa(x_1, y)\exp\lc -a\int_0^{x_1} \kappa(s, y)\rmd s\rc  \id_{\{\bmx\in\ee_1\}}$ and similarly for  $f_{\mu_{\delta_\infty,(a\delta_y)^{(\kappa)}}}(\bmx)$. 
\end{ex}

\begin{rem}
    In general, it is hard to check that a given multivariate probability density is the density of a min-id distribution, see \cite[Exercise 5.1.4]{resnickextreme2013}. Therefore it is difficult to directly specify a multivariate model in terms of $f_\eta$ and it is advisable to first specify $\eta$ and then to derive the corresponding density of $\minid(\eta)$ from its survival function. For example, one may specify $\eta=\sum_{D\subseteq\{1,\ldots,d\}} \eta_D$, where $\eta_D$ has a Lebesgue density on $\ee_D$. Then, the corresponding density of $\minid(\eta_D)$ might be derived by applying standard density formulas, accounting for the potential mass on $\infty$ of other coordinates. Note that non-zero $\eta_D$ implies that $(X_i)_{i\in D}$ are dependent for $\bmX\sim \minid(\eta)$. When $\eta=\sum_{i=1}^d \eta_{\{i\}}$ for $\nu$-almost every $\eta$ and $\alpha$ also concentrates on $\cup_{i=1}^d\ee_i$ then $\minid(\mu)$ is a model for partially exchangeable data.
\end{rem}

According to Condition \ref{condabsolutecont}, for every set $ A\in \mathcal{B}\lc E_{d\times m}\rc$ we get 
\begin{align*} 
\lambda_\nu^{(d\times m)}(A)&=\int_{M_d^0} \otimes_{i=1}^m\minid(\eta) (A)\nu(\rmd\eta)=\int_{M_d^0} \int_{ A} \prod_{i=1}^m f_{\eta}(\bmx_i)    \Lambda_{dm}^{(\infty)} (\rmd\ubmx_m) \nu(\rmd\eta) \\
&= \int_{ A_i}  \int_{M_d^0} \prod_{i=1}^m f_{\eta}(\bmx_i) \nu(\rmd\eta)   \Lambda_{md}^{(\infty)} (\rmd\ubmx_m) ,
\end{align*}
where we can apply Fubini because $\nu$ and $\Lambda^{(\infty)}_{dm}$ are both $\sigma$-finite measures, due to Lemma 15.4 in the Supplementary Material \cite{brueck2026}. This implies that the $md$-dimensional margin of $\lambda_\nu$ corresponding to $\ubmX_m$ has a density w.r.t.\ $\Lambda^{(\infty)}_{md}$ given by
\begin{align}
   f^{(d\times m)}(\ubmx_m):=\int_{M_d^0}\prod_{i=1}^m f_{\eta}(\bmx_i)\nu(\rmd\eta) \label{def_f_m} .
\end{align} 
Note that $f^{(d\times m)}$ is not a probability density anymore, but  $f^{(d\times m)}(\ubmx_m)$  is finite for $\lambda_\nu^{(d\times m)}$-almost every $\ubmx_m\in E_{md}^\prime$, since $\lambda^{(d\times m)}$ is finite in every neighborhood of $\ubmx_m\in E_{d\times m}^\prime$. Moreover, it can easily be observed that for every $IJ\subset d\times m$ the $IJ$-margin of $f^{(d\times m)}$ is given by
\begin{align}
f^{(IJ)}\lc \bmx_{IJ}\rc&=\int_{ E_{md-\vert IJ\vert}} f^{(d\times m)}(\ubmx_m)  \Lambda^{(\infty)}_{d\times m\setminus\vert IJ\vert} \lc\rmd \ubmx_{m\setminus IJ}\rc . \label{def_f_IJ}
\end{align}
which is the density of $\lambda_{\nu}^{(IJ)}$ w.r.t.\ $\Lambda^{(\infty)}_{\vert IJ\vert}$.
By similar arguments, the density of $\lambda_\alpha$ w.r.t.\ $\Lambda^{(\infty)}_{md}$ corresponding to $\ubmX_m$ is given by
\begin{align} 
 g^{(d\times m)}\lc \ubmx_m\rc := \sum_{\substack{1\leq j\leq m}}  g_\alpha(\bmx_j) \id_{\{ \bmx_k=\binfty \text{ for all }k\not=j, 1\leq k\leq m \}}   \label{def_g_m}
 \end{align}
and the density corresponding to an $IJ\subset d\times m$ margin is therefore given by
\begin{align*}
  g^{(IJ)}\lc \bmx_{IJ}\rc &:=\int_{  E_{md-\vert IJ\vert}} g^{(d\times m)}(\ubmx_m) \Lambda^{(\infty)}_{d\times m\setminus\vert IJ\vert}\lc\rmd \ubmx_{m\setminus IJ} \rc. 
  \end{align*}
Now, we are able to derive the regular conditional distributions of $\lambda^{(d\times m)}$ via (\ref{defconddistexpmeasure}).

\begin{lem}
\label{lemconddistexpregcond}
     Let $\mu$ be an IDEM which satisfies Condition \ref{condabsolutecont} with corresponding densities $f^{(d\times m)}$ and $g^{(d\times m)}$ as defined in (\ref{def_f_m}) and (\ref{def_g_m}). Further, let $m\in\N$ and let $IJ\subset d\times m$. Then, for every $\bmx_{IJ}\in E^\prime_{\vert IJ\vert}$, the conditional probability distribution of $\lambda^{(d\times m)}$ is given by
     $$K_{IJ}\lc \bmx_{IJ},d\ubmx_{m\setminus IJ}\rc=\frac{f^{(d\times m)}\lc \ubmx_m \rc+g^{(d\times m)}\lc \ubmx_m\rc}{f^{(IJ)}\lc \bmx_{IJ}\rc+g^{(IJ)}\lc \bmx_{IJ} \rc} \lc\Lambda^{(\infty)}_{md-\vert IJ\vert}\rc (\rmd\ubmx_{m\setminus IJ}) \text{ a.s.}. $$
\end{lem}

A consequence of Lemma \ref{lemconddistexpregcond} is that we can obtain a simpler representation of $\ubmZ^{(n,l)}$ in the posterior representation of $\ubmX$.

\begin{lem}
\label{lemconddistextrfctregcond}
Under the same conditions as in Lemma \ref{lemconddistexpregcond} and conditionally on the conditional hitting scenario $\Tilde{\Theta}$ and $\ubmX_n$, the survival function of $\lc\bmZ^{(n,l)}_{j}\rc_{1\leq  j\leq k}$ from (\ref{defcondextrseq}) takes the form
     \begin{align*}
    &P\lc  \bm Z^{(n,l)}_{j} >\bmx_{j} \text{ for all }1\leq j\leq k \mid \Tilde{\Theta},\ubmX_n\rc  \\
    &= \frac{ K(\ubmX_n,\Tilde{\Theta}_l) + \int_{M_d^0} \underset{1\leq  j\leq k}{\otimes} \minid(\eta) \lc  \big \{\ubmy_k\mid \ubmy_k>\ubmx_{k} \big\}  \rc  h(\eta,\ubmX_n,\Tilde{\Theta}_l) \nu(\rmd\eta)  }{C(\ubmX_n,\Tilde{\Theta}_l) },  
\end{align*}
where $C(\ubmX_n,\Tilde{\Theta}_l)$ denotes the corresponding normalizing constant and
$$h(\eta,\ubmX_n,\Tilde{\Theta}_l) :=  \int_{\{ \ubmy_{n\setminus \Tilde{\Theta}_l} >\ubmX_{n\setminus\Tilde{\Theta}_l}\}} \prod_{j=1}^n f_\eta\lc (\ubmX_n,\underline{\bmy}_n)(\Tilde{\Theta}_l)\rc \lc  \Lambda^{(\infty)}_{nd-\vert \Tilde{\Theta}_l\vert} \rc \lc\rmd \ubmy_{n\setminus \Tilde{\Theta}_l}\rc,$$
$$ K(\ubmX_n,\Tilde{\Theta}_l):=\int_{\{  \ubmy_{n\setminus \Tilde{\Theta}_l}>\ubmX_{n\setminus \Tilde{\Theta}_l}\}} g^{(d\times n)}\lc (\ubmX_n,\ubmy_n)(\Tilde{\Theta}_l)\rc \Lambda^{(\infty)}_{nd-\vert \Tilde{\Theta}_l\vert}  \lc\rmd \ubmy_{n\setminus \Tilde{\Theta}_l} \rc $$
with $(\ubmX_n,\underline{\bmy}_n)(\Tilde{\Theta}_l):=\lc \id_{\{ (i,j)\in \Tilde{\Theta}_{l} \}}X_{i,j}+ y_{i,j}\id_{\{(i,j)\not \in \Tilde{\Theta}_{l}\}}\rc_{(i,j)\in d\times n}.$
\end{lem}

Thus, the distribution of $\lc \bmZ^{(n,l)}_{j}\rc_{j\in\N}$ is the mixture of min-id distributions and it can be generated as follows. First, draw a Bernoulli random variable $B:=B(\ubmX_n,\Tilde{\Theta}_l)$ with success probability $1-K(\ubmX_n,\Tilde{\Theta}_l)/C(\ubmX_n,\Tilde{\Theta}_l)$. If $B=0$, set $\ubmZ^{(n,l)}_j=\binfty$ for all $1\leq j\leq k$, which corresponds to $\bmZ^{(n,l)}_j\sim\minid(\overline{\mu}^{(l)}_n)$ with $\overline{\mu}^{(l)}_n=\infty\delta_{\binfty}$. Otherwise, simulate a measure $\overline{\mu}^{(l)}_n$ from the appropriately normalized version of $h(\eta,\ubmX_n,\Tilde{\Theta}_l) \nu(\rmd\eta)$ and then simulate sequences of i.i.d.\ random vectors $\lc \bmZ^{(n,l)}_{j}\rc_{1\leq j\leq k}$ with distribution $\minid\big(\overline{\mu}_n^{(l)}\big)$. It should be mentioned that $K(\ubmX_n,\Tilde{\Theta}_l)=0$ whenever $\Tilde{\Theta}_l$ contains more than one distinct $j$ index or $\alpha=\infty \delta_\binfty$, which implies $B=1$ almost surely.
For a full representation of the posterior it remains to derive the distribution of the conditional hitting scenario.
\begin{lem}
\label{lemcondhitscenregcond}
    Under the same conditions as in Lemma \ref{lemconddistexpregcond}, the distribution of the conditional hitting scenario is given by
    \begin{align}
       &\tau(\ubmX_n,\Tilde{\Theta})= \label{conddenshitscen}\\
       &\frac{ \prod_{l=1}^{L(\Tilde{\Theta})}  \int_{\{  \ubmy_{n\setminus\Tilde{\Theta}_l}>\ubmX_{n\setminus\Tilde{\Theta}_l}\}} \lc f^{(d\times n)}+ g^{(d\times n)}\rc\lc (\ubmX_n,\underline{\bmy}_n)(\Tilde{\Theta}_l) \rc  \lc   \Lambda^{(\infty)}_{nd-\vert \Tilde{\Theta}_l\vert}  \rc  \rmd \lc \ubmy_{n\setminus\Tilde{\Theta}_l}\rc     }
    {  \sum_{\hat{\Theta}\in\mathcal{P}_n}  \prod_{l=1}^{L(\hat{\Theta})} \int_{\{  \ubmy_{n\setminus\hat{\Theta}_l}>\ubmX_{n\setminus\hat{\Theta}_l}\}} \lc f^{(d\times n)}+ g^{(d\times n)}\rc\lc (\ubmX_n,\underline{\bmy}_n)(\hat{\Theta}_l) \rc   \lc \Lambda^{(\infty)}_{nd-\vert \hat{\Theta}_l\vert}  \rc   \rmd \lc \ubmy_{n\setminus\hat{\Theta}_l}\rc     }. \nonumber
    \end{align}
\end{lem}

The results of this section are summarized in the following theorem, which provides a tractable formula for the posterior of an IDEM prior under Condition \ref{condabsolutecont}, assuming $\alpha =\infty\delta_\binfty$ for simplicity.
\begin{thm}
    \label{thmconddistmixtureofminid}
    Let $\mu$ denote an IDEM which satisfies Condition \ref{condabsolutecont} and assume that $\alpha=\infty\delta_{\binfty}$. Then, the posterior of the corresponding exchangeable min-id sequence $\ubmX$ and $\mu$ is given by
    $$ \lc\bmX_{n+j}\rc_{j\in\N} \mid \ubmX_n \sim \minid\lc \overline{\mu} \rc,\ \ \ \mu \mid \ubmX_n \sim\overline{\mu}=\overline{\mu}_n+\sum_{l=1}^{L(\Tilde{\Theta})} \overline{\mu}^{(l)}_n .$$
    where $\overline{\mu}_n$ is an IDEM with Lévy characteristics $\lc \alpha,\overline{\nu}_n\rc$,  $\lc\Tilde{\Theta},\lc\overline{\mu}_n^{(l)}\rc_{1\leq l\leq L(\tilde{\Theta})}\rc$ are independent of $\overline{\mu}_n$, the distribution of $\Tilde{\Theta}$ is determined by Lemma \ref{lemcondhitscenregcond} and, conditionally on $\Tilde{\Theta}$, the $\lc\overline{\mu}_n^{(l)}\rc_{1\leq l\leq L(\tilde{\Theta})}$ are independent random measures with distribution $C(\ubmX_n,\Tilde{\Theta}_l)^{-1}$ $\lc  h(\eta,\ubmX_n,\Tilde{\Theta}_l) \nu(\rmd\eta) +K(\ubmX_n,\Tilde{\Theta}_l)\delta_{(\infty\delta_{\binfty})}\rc$ as defined in Lemma \ref{lemconddistextrfctregcond}.
\end{thm}

\section{A multivariate model for partially exchangeable data with continuous dependence at the root}
\label{secmultivmodel}

In this section, we illustrate the flexibility of our framework by constructing a model for multivariate survival analysis. The setting is as follows. Assume that one observes survival data from different groups and that the goal is to learn the survival function of individuals in the different groups. We assume that the observations are homogeneous within a group and heterogeneous between the groups, while assuming that the groups share some common, but unknown, structure. In a Bayesian setting, such data is commonly modeled by a partially exchangeable framework, where the dimension $d$ represents the number of groups. The shared structure between the groups then allows for the so-called ``borrowing of information'', i.e.\ to learn about the survival function in one group from observations in other groups.

Mathematically, the assumption of partial exchangeability translates to the prior distribution concentrating on the space of product probability measures, see e.g.\ \cite{epifanilijoi2010,lijoinipoti2014,rivapalacioleisen2018,camerlengilijoipruenster2021} for examples of such frameworks. All the latter models have in common that their prior is constructed in such a way that the shared structure between the groups is identical, i.e.\ if the prior has a (smoothed) atom at a certain location then its occurrence is either present in all groups or its occurrence is completely independent from the other groups. The technical reason for this is that the shared structure is usually discrete, since this is technically more convenient to handle. Let us exemplarily explain this for the framework of \cite{camerlengilijoipruenster2021}, which was already discussed at the beginning of Section \ref{subsecsubordcrm} and in Example \ref{exhierdepmixhaz}, while similar comments apply to the frameworks of \cite{epifanilijoi2010,lijoinipoti2014,rivapalacioleisen2018}. The completely random measure $\mu_0$ at the root, which represents the shared structure, is a discrete measure. This implies that the hazard rates of the individual groups are sums of terms of the form $a_i\kappa(s,y_i)$, where the $(y_i)_{i\in\N}$ are the same for all groups. The standard choice $\mathbb{Y}=[0,\infty)$ and $\kappa(s,y)=a \id_{\{s\geq y\}}$ implies that the prior assumes that each group has a jump in the hazard rate starting at the same locations $y_i$. This assumption might be seen as restrictive for survival data, since, even though the hazard rates should be ``similar'', there is no particular reason why the locations $y_i$ should be identical across groups.

Here, we propose a model which solely assumes that an increase in the hazard rate at a certain location in a certain group implies that the probability of an increase in the hazard rate in the other groups close to this location is higher, which reflects the intuition that increases in hazard rates should be ``similar'' but not identical. On a technical level, this will be achieved by choosing a continuous IDEM $\mu_0$  at the root in the construction of the subordinated CRMs in Proposition \ref{propsubordinationcrm} as follows. Choose a CRM $\mu_0$ on $[0,\infty)$ with base measure $\alpha_0$ and Lévy measure $\nu_0$ given by the image measure of $(a,b)\mapsto a\delta_b\textcolor{black}{+\infty\delta_\infty}$ under $l(a,b)\rmd a \rmd b$. Further, choose a kernel $\kappa_0$ and define $\mu^{(\kappa_0)}_0([0,t]):=\int_0^t \int_0^\infty \kappa_0(s,y) \mu_0(\rmd y) \rmd s$. Now, $\mu_0^{(\kappa_0)}$ is used as the common intensity measure for the construction of $d$ conditionally independent CRMs without base measure, i.e.\ we use the continuous $\mu_0^{(\kappa_0)}$ as the shared structure for the $d$ different groups. Choosing $d$ kernels $(\kappa_i)_{\leqd}$ and marginal intensities $\lc\rho_i\rc_{1\leq i\leq d}$ we can define the corresponding IDEM $\mu_{\mu^{(\kappa_1)}_1,\ldots,\mu^{(\kappa_d)}_d}$ as in (\ref{defsubordinatedCRMexpmeasure}).
Intuitively speaking, the construction implies that the $\mu_i$ will be dependent but with distinct atoms almost surely. Consequently, $\minid\lc \mu_{\mu^{(\kappa_1)}_1,\ldots,\mu^{(\kappa_d)}_d}\rc$ is a model for partially exchangeable data with marginal survival function $\exp\lc -\mu^{(\kappa_i)}_i\big( (-\infty,x_i]\big)\rc$, i.e.\ for $\bmX\sim\minid(\mu)$ we have $ \p\lc \bmX>\bmx\mid \mu \rc=\prodd \exp\lc-\mu_i^{(\kappa_i)}\big( (-\infty,x_i]\big) \rc$.
We need to impose some regularity conditions on the $\kappa_i$ and $\mu_i$ such that the proposed model is well-defined. 

\begin{cond}
\label{condkernel}
    The kernels $(\kappa_i)_{0\leq i\leq d}$ satisfy 
    \begin{enumerate}
        \item[(K1)]  There exists a $C>0$ such that $\kappa_i(s,b)\leq C$ for all $s,b>0$. Moreover, there exists a $\delta,\tilde{C}>0$ such that for every $b>0$ and all $\tilde{s}\in (b,b+\delta)$ we have $\kappa_i(\tilde{s},b)>\tilde{C}$.
        \item[(K2)] For every $t>0$ there exists a $C>0$ such that $\kappa_i(s,b)=0$ for all $s\leq t$ and $b>C$.
    \end{enumerate}

    In addition, we require
    \begin{enumerate}
        \item[(C1)] $l(a,b)>0$ for all $a,b>0$, $\int_0^\infty l(a,b)\rmd a=\infty$ for every $b>0$ and $\int_0^t \int_0^\infty \min\{a;1\} $ $ l(a,b)\rmd a \rmd b<\infty$ for all $t>0$. Moreover, for every $1\leq i\leq d$ we have $\int_0^\infty\rho_i(a)\rmd a=\infty$.
        \item[(C2)] Either $\int_0^\infty \kappa_i(s,b)\rmd s=\infty$ for all $1\leq i\leq d$ and $b>0$ or $\mu_0(\R)=\infty$ almost surely.
    \end{enumerate}
     
\end{cond}

Examples of kernels satisfying Condition \ref{condkernel} (K1)-(K2) are e.g.\ the Dykstra-Laud kernel $\kappa(s,b)=\tau\id_{\{s\geq b\}}$, the rectangular kernel $\kappa(s,b)=a\id_{\{ \vert s-b \vert\leq \tau\}}$ and the Ornstein-Uhlenbeck kernel $\kappa(s,b)=\sqrt{2\tau}\exp\lc -\tau(s-b)\rc \id_{\{ s\geq b\}}$, where $a,\tau>0$. The conditions on the $\mu_i$ are very mild and most likely satisfied for all models that are used in practice.
Under Condition \ref{condkernel} one can verify that $\mu_{\mu^{(\kappa_1)}_1,\ldots,\mu^{(\kappa_d)}_d}$ is a well-defined IDEM and that Condition \ref{condabsolutecont} is satisfied, which is summarized in the following Lemma. To simplify the notation in the following for any exponent measure $\eta$ and kernel $\kappa$ we denote   
$$ e_{\eta,\kappa}(x) =  \lc \id_{\{x<\infty \} }\int_0^\infty \kappa(x,y) \eta(\rmd y)  +\id_{\{x=\infty \} } \rc \exp\lc -\int_0^{x} \int_0^\infty \kappa(s,y) \eta(\rmd y)\rmd s \rc . $$
\begin{lem}
\label{lemsurvivmodel}
    Let $\mu_{\mu^{(\kappa_1)}_1,\ldots,\mu^{(\kappa_d)}_d}$ be constructed as in Proposition \ref{propsubordinationcrm}, where $\mu_0$ is a CRM with base measure $\alpha_0$ and intensity $l(a,b)\rmd a\rmd b$ and \textcolor{black}{, conditionally on $\mu_0^{(\kappa_0)}$,} the $\mu_i$ are \textcolor{black}{independent} CRMs without base measure and intensities $\rho_i(a)\rmd a \mu_0^{(\kappa_0)}(\rmd b)$.  Additionally assuming that Condition \ref{condkernel} is satisfied, we have that $\mu_i^{(\kappa_i)}([0,t])\in(0,\infty)$ for all $t>0$, $\mu_{\mu^{(\kappa_1)}_1,\ldots,\mu^{(\kappa_d)}_d}$ is an IDEM without base measure and Lévy measure $\nu=\nu_1+\nu_2$ where
    $$  \nu_1(A)= \sumd \int_{0}^\infty \int_{0}^\infty  \id_{\{ \infty\delta_\binfty\neq  \otimes_{j=1}^{i-1} \delta_\infty \times (a\delta_b)^{(\kappa_i)}\otimes_{j=i+1}^d \delta_\infty \in A \}} \rho_i(a)\rmd a \alpha^{(\kappa_0)}_0(\rmd b) $$
    and 
    \begin{align*}
       & \nu_2( A)\\
       &=\int_{0}^\infty \int_{0}^\infty \int_{\times_{i=1}^d \textcolor{black}{M_1}}  \id_{\big\{  \mu_{\eta^{(\kappa_1)}_1,\ldots,\eta^{(\kappa_d)}_d} \in A \big\}} \otimes_{i=1}^d \crm_{\rho_i(a)\rmd a (a_0\delta_{b_0})^{(\kappa_0)}(\rmd b)}(\rmd \eta_i) l(a_0,b_0)\rmd a_0\rmd b_0.
    \end{align*} 
    
    Moreover, $\mu_{\mu^{(\kappa_1)}_1,\ldots,\mu^{(\kappa_d)}_d}$ satisfies Condition \ref{condabsolutecont} and the corresponding densities w.r.t.\ $\Lambda^{(\infty)}_d$ are given by $g_\alpha=\infty \id_{\{\binfty\}}$,
        $$  f_{\otimes_{j=1}^{i-1} \delta_\infty \times (a\delta_b)^{(\kappa_i)}\otimes_{j=i+1}^d \delta_\infty}(\bmx)=  e_{a\delta_b,\kappa_i}(x_i)\id_{\{\bmx\in\ee_i\cup\{\binfty\}\}}$$
    and 
        $$f_{\mu_{\eta^{(\kappa_1)}_1,\ldots,\eta^{(\kappa_d)}_d}}(\bmx)= \prodd e_{\eta_i,\kappa_i}(x_i).$$
    
\end{lem}

Theorem \ref{thmconddistmixtureofminid} now directly provides the representation of the posterior of $\mu_{\mu^{(\kappa_1)}_1,\ldots,\mu^{(\kappa_d)}_d}$. We remark that the model is structurally conjugate, as every random measure $\overline{\mu}_n^{(l)}$ appearing in the posterior in Theorem \ref{thmconddistmixtureofminid} is constructed as in (\ref{defsubordinatedCRMexpmeasure}). To see this, observe that $\overline{\mu}_n^{(l)}$ is again constructed via a vector of conditionally independent random measures as follows, assuming $\alpha_0=\infty\delta_\infty$ for simplicity. Draw a baseline random measure 
$(a_0\delta_{b_0})^{(\kappa_0)}$, where $(a_0,b_0)\sim  K(a_0,b_0,\ubmX_n,\Tilde{\Theta}) C(\ubmX_n,\Tilde{\Theta})^{-1} l(a_0,b_0)\rmd a_0\rmd b_0$ and
$$ K(a_0,b_0,\ubmX_n,\Tilde{\Theta}):= \e_{\otimes_{i=1}^d \crm_{\rho_i(a)\rmd a (a_0\delta_{b_0})^{(\kappa_0)}(\rmd b)}}\lk h\lc\mu_{\eta_1^{(\kappa_1)},\ldots,\eta_d^{(\kappa_d)}},\ubmX_n,\Tilde{\Theta}\rc\rk. $$
Construct a vector of univariate random measures $\eta_i$ whose law is given by 
$$ P\lc (\eta_i)_{\leqd}\in A \rc= \frac{\e_{\otimes_{i=1}^d \crm_{\rho_i(a)\rmd a (a_0\delta_{b_0})^{(\kappa_0)}(\rmd b)}}\lk \id_{\{(\eta_i)_{\leqd}\in A\}}h\lc\mu_{\eta_1^{(\kappa_1)},\ldots,\eta_d^{(\kappa_d)}},\ubmX_n,\Tilde{\Theta}\rc\rk }{K(a_0,b_0,\ubmX_n,\Tilde{\Theta})}.$$
Then, set $\overline{\mu}_n^{(l)}=\mu_{\eta_1^{(\kappa_1)},\ldots,\eta_d^{(\kappa_d)} }$ according to (\ref{defsubordinatedCRMexpmeasure}).

It remains to find tractable formulas for the conditional hitting scenarios, which by Lemma \ref{lemcondhitscenregcond} is the case if we can express $f^{(IJ)}$ from (\ref{def_f_IJ}) in a tractable way.

\begin{lem}
    \label{lemdenshierachicalmodel}
    In the setting of Lemma \ref{lemsurvivmodel}, the density $f^{(IJ)}$ w.r.t.\ $ \Lambda^{(\infty)}_{\vert IJ\vert}$ associated to $\mu_{\mu^{(\kappa_1)}_1,\ldots,\mu^{(\kappa_d)}_d}$ is given by
    \begin{align}
        f^{(IJ)}\lc \bmx_{IJ} \rc&=\sumd   \id_{\{ x_{k,j}=\infty \forall (k,j)\in IJ, k\not=i\}} \int_0^\infty \int_0^\infty  \lc  \prod_{j:(i,j)\in IJ}  e_{a\delta_b,\kappa_i}(x_{i,j})  \rc  \rho_i(a) \rmd a  \alpha_0^{(\kappa_0)}(\rmd b) \nonumber \\
        &+  \int_{0}^\infty \int_{0}^\infty  \prodd   \e_{\eta_i\sim CRM_{\rho_i(a)\rmd a (a_0\delta_{b_0})^{(\kappa_0)}(\rmd b)}}\Bigg[   \prod_{j:(i,j)\in IJ}  e_{\eta_i,\kappa_i}(x_{i,j})  \Bigg]    l(a_0,b_0) \rmd a_0\rmd b_0 \label{eqnpostconc} 
    \end{align} 
\end{lem}

Finally, let us verify that our model allows a group-dependent increase in the hazard rate around a shared center $b_0$. Assume for simplicity that $\alpha_0=\infty\delta_\infty$, $d=2$, $\kappa_0(s,y)=\id_{\{\vert x-y\vert\leq \tau\}}$ and $\kappa_1(s,y):=\kappa_2(s,y):=\id_{\{s\geq y\}}$.
Then, the hazard rate corresponding to $e_{\eta_i,\kappa_i}$ is simply $\eta_i([0,s])$, which increases in the interval $[b_0-\tau,b_0+\tau]$ and is constant thereafter, while the increase is determined by the realization of $\eta_i$. If we would have chosen to use $\mu_0$ instead of $\mu_0^{(\kappa_0)}$ as our shared structure between the groups $\eta_i$ would have only had one atom at $b_0$ and the hazard rate corresponding to $e_{\eta_i,\kappa_i}$ would jump at $b_0$ and be constant thereafter. Thus, combining this with the discussion above Lemma \ref{lemdenshierachicalmodel}, the $\overline{\mu}_n^{(l)}$ induce similar, but not identical, increases in the hazard rates in the $d$ groups, which reflects the intuition we have announced for the motivation of this model.

\subsection{Simulation from the posterior}
\label{secsimpartiallyexdatamaintext}
 To complement the theoretical derivations in this section, let us illustrate the results in a simulation study. As we have access to the posterior predictive distribution, we develop a simulation algorithm that allows us to draw random vectors which, conditionally on $\ubmX_n$ and $\overline{\mu}$, follow the distribution
$$ \lc\bmX_{n+j}\rc_{1\leq j\leq k} \mid \ubmX_n \overset{i.i.d.}{\sim} \minid\lc \overline{\mu} \rc,\ \ \ \mu \mid \ubmX_n \sim\overline{\mu}=\overline{\mu}_n+\sum_{l=1}^{L(\Tilde{\Theta})} \overline{\mu}^{(l)}_n, $$
where $k\in\N$ and the distributions of $\bar{\mu}_n$ and $\lc\bar{\mu}_n^{(l)}\rc_{1\leq l\leq L(\Tilde{\Theta})}$ are described in Theorem \ref{thmconddistmixtureofminid} and further specified in Lemmas \ref{lemsurvivmodel} and \ref{lemdenshierachicalmodel} above. Since the simulation from this model is rather complex we will only give a rough summary of the main ideas here and refer to Sections 7.3 and 9 of the
Supplementary Material \cite{brueck2026} for a detailed discussion of a general framework for the simulation from the posterior of an IDEM prior. Moreover, the simulation algorithm for the posterior of the specific model from this section is discussed in Section 9 of the Supplementary Material \cite{brueck2026}, where we also provide practical comments on the implementation. The code to reproduce the simulation study can be found at \url{https://github.com/florianbrueck/IDEM}.

Rather than first simulating the random exponent measure and then its associated min-id distribution, we use simulation methods from extreme value theory \cite{ribatet2013conditionalsimmaxstable,brueckexactsim2022}. This avoids adapting the simulation procedure to each realization of the exponent measure. As described in Section \ref{secpostdistr}, the posterior predictive distribution
of $\lc\bmX_{n+j}\rc_{1\leq j\leq k}$ can be represented as
$$ \lc\bmX_{n+j}\rc_{1\leq j\leq k} \mid \ubmX_n \sim \min\big\{ \ubmY^{(n)}_k,\ubmZ^{(n)}_k \big\},$$
where, by (\ref{defpostexpmeasure}) and (\ref{defpostlevymeasure}), $\ubmY^{(n)}_k$ is distributed as a min-id random vector with exponent measure
$ \int_{M_d^0}  \otimes_{1\leq j\leq k} \minid(\eta)(\cdot) \overline{\nu}_n(\rmd\eta)  =:\bar{\lambda}^{(n)}_k$ 
and $\ubmZ^{(n)}_k$ is distributed as the first $k$ vectors in the exchangeable sequence $\ubmZ^{(n)}$ from (\ref{defcondextrseq}). From (\ref{defcondextrseq}), we recall that $\ubmZ^{(n)}_k$ can be represented as $\min_{1\leq l\leq L(\tilde{\Theta})} \ubmZ^{(n,l)}_k$ where each $\ubmZ^{(n,l)}_k$ is distributed as the first $k$ vectors in an exchangeable sequence of random vectors whose distribution is determined by the conditional distributions associated with the exponent-measure $\lambda_\nu$ and the conditional hitting scenario $\tilde{\Theta}$, already invoking that $\lambda_\alpha$ vanishes in our current setting. 
The key point to observe is that this representation of $\lc\bmX_{n+j}\rc_{1\leq j\leq k}$ allows simulation from a min-id distribution with fixed exponent measure $\bar\lambda_k^{(n)}$ and, as we show in Section 9 of the Supplementary Material \cite{brueck2026}, the simulation of each $\ubmZ_k^{(n,l)}$ essentially requires the same ``ingredients'' as the simulation of $\ubmY^{(n)}_k$. 
Therefore, the suggested simulation procedure for the posterior is to simulate $\ubmY^{(n)}_k$ and $\ubmZ^{(n)}_k$ separately and to take the componentwise minimum of the two resulting vectors to obtain $\lc\bmX_{n+j}\rc_{1\leq j\leq k}$.

The simulation procedure thus consists of three main steps:
simulating from a fixed min-id distribution with exponent
measure $\bar{\lambda}^{(n)}_k$, sampling the conditional
hitting scenario $\tilde{\Theta}$ and simulating random
vectors $\ubmZ^{(n,l)}_k$ whose distribution is determined
by the conditional distribution of $\lambda_\nu$.
Although this strategy follows the same principles as
existing methods from extreme value theory, the complexity
of the underlying objects requires substantial theoretical
and practical developments. We refer to Section 7.3, 8 and 9 of the Supplementary
Material \cite{brueck2026} for more details.

The main aim of this section is to provide a proof-of-concept of the practical feasibility of this simulation strategy by considering a toy example with $d=2$ groups and $6$ observations per group, which are given by $\{0.47,0.16,3.01,1.32,0.91,0.17\}$ and $\{0.24,1.42,0.96,1.11,0.14,1.87\}$, respectively. We have chosen the same hyperparameters for each group of observations such that the two groups share an a-priori identical survival function, while allowing for group-specific deviations in the posterior. For example, this type of data could represent survival times in two related patient populations or two slightly different treatment variants. We fix $k=1$ and simulate $1000$ samples from the posterior predictive distribution of $\bmX_{n+1}$ according to the MCMC simulation scheme specified in Supplementary Material \cite{brueck2026} \ref{appsimpartiallyexdata}. Due to space constraints, we present here only the plot of the posterior predictive survival function in Figure \ref{figsurvfct}, which is obtained as the empirical survival function of the samples from the posterior. Figure \ref{figsurvfct} clearly shows that the posterior shifts probability mass towards the mostly smaller observed event times. Therefore, it showcases that the posterior is able to learn from the data, even for a small sample size. Additional explanations of the setup, simulation results and discussions can be found in Section 9.3 of the Supplementary Material \cite{brueck2026}.

\begin{figure}[!htbp]
    \centering
    \includegraphics[width=0.8\textwidth]{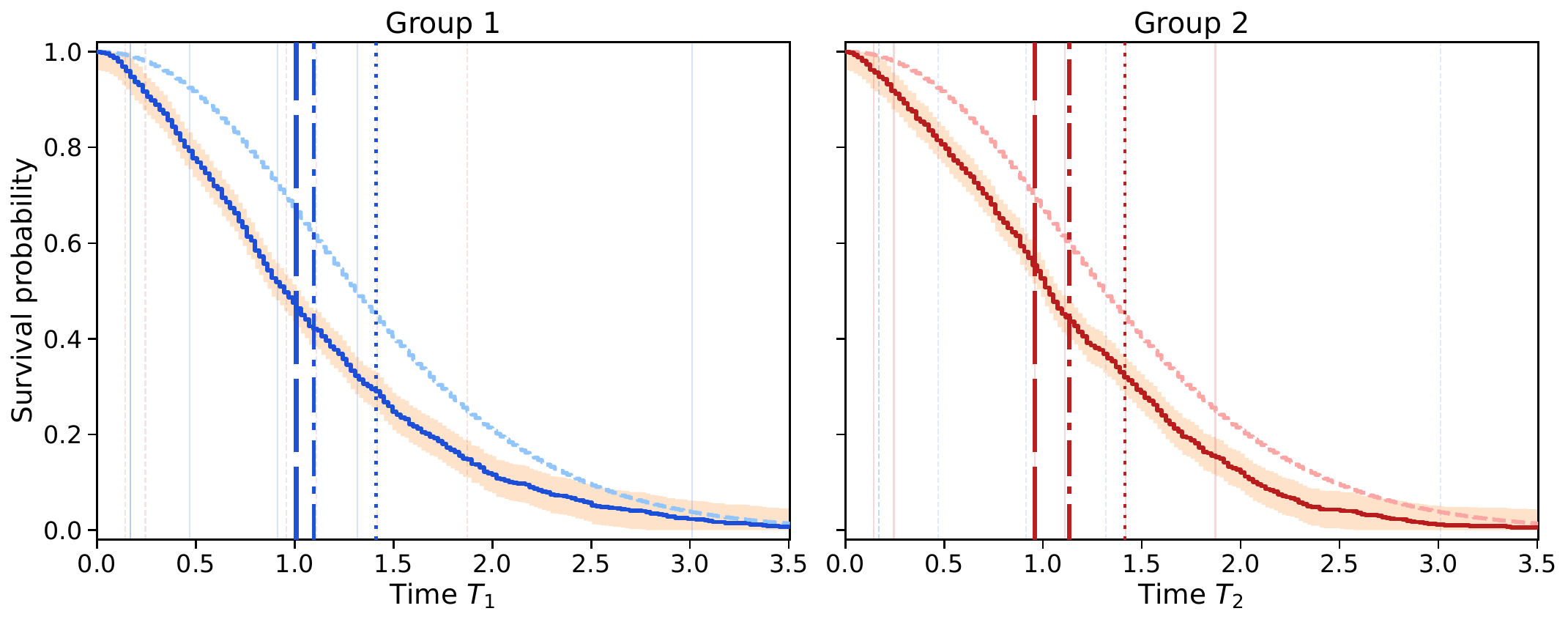}
    \caption{Prior (dotted) and posterior predictive (solid) survival functions. Opaque vertical lines mark the observed event times within the groups (blue,red). The dashed, dashed-dotted and dotted vertical lines represent the empirical mean of the observations within the group, the posterior mean and the prior mean, respectively. \textcolor{black}{The orange shaded regions are approximative Dvoretzky-Kiefer-Wolfowitz bands quantifying the Monte Carlo uncertainty of the estimated predictive survival functions.}}
    \label{figsurvfct}
\end{figure}

\begin{acks}[Acknowledgments]
I would like to thank Matthias Scherer and Jan-Frederik Mai who have encouraged me to explore this line of research while doing my PhD at the Technical University Munich. \textcolor{black}{Moreover, I would like to thank an anonymous referee for his very constructive feedback on an earlier version of my manuscript. His valuable suggestions and careful proofreading have substantially improved the quality of the paper.}
\end{acks}

\begin{funding}
This work was supported by Grant 186858 of the SNSF.
\end{funding}




\bibliographystyle{imsart-number} 
\bibliography{references_doi_commented_fixed}  

\newpage
\begin{supplement}
\stitle{Supplementary material for Infinitely divisible priors for multivariate survival functions }
\sdescription{This document contains the supplementary material for the paper \textbf{\textit{Infinitely divisible priors for multivariate survival functions}}. It contains additional theoretical results, examples and the proofs for the main paper.}
\end{supplement}

\tableofcontents

\section{Introduction to min-id processes}
\label{appintrominprocess}
This section provides a more thorough introduction to min-id processes $\lc X_t\rc_{t\in T}$, where $T$ \textcolor{black}{ is a locally compact metric space}. The purpose of this section is to foster the understanding of their stochastic representation and the concepts of extremal functions, hitting scenarios and their conditional distributions. We present these concepts for min-id processes indexed by an abstract index set $T$ since the author firmly believes that the process-view simplifies the understanding of the key concepts, as their full beauty and elegance is only unveiled on this level. Nevertheless, sequences and finite-dimensional min-id vectors are also treated by the two specific cases $T=d\times\N$ and $T=\{1,\ldots,d\}$. This section is heavily based on \cite{brueckexactsim2022,dombryeyiminkostrongmixing2012,dombryeyiminko2013regular} and the author refers the interested reader to those publications for the technical details. As a general remark, we want to mention that most of the results in the literature are expressed for max-id processes, the negative of min-id processes, as this is more convenient in the field of extreme value theory. Therefore, many of the results presented below cannot be found in the exact same form elsewhere, as we had to translate the results for max-id processes to min-id processes, which is however the only discrepancy to be found in the following.

\subsection{Stochastic representation: The extremal functions and the hitting scenario}

\textcolor{white}{1}\linebreak Minimum-infinitely divisible (min-id) random vectors $\bmX$ naturally arise as weak limits of scaled (componentwise) minima of independent random vectors of the form
$$ \bm a_n\lc\min_{\leqn} \bmY^{(i)}-\bm b_n\rc \to \bmX , $$
where $\bm a_n \in(0,\infty)^d$,$\bm b_n\in \R^d$ and $\lc\bmY^{(i)}\rc_{i\in\N}$ are independent random vectors, see \cite{Husler1989} for the details. The most important subclass of min-id distributions is the subclass of min-stable distributions, which comprises all possible limit laws of i.i.d.\ sequences $\lc\bmY^{(i)}\rc_{i\in\N}$. Min-stable distributions are well-studied as they constitute the main objects of interest in extreme-value theory, see \cite{resnickextreme2013} for a thorough textbook account. 

Min-id processes $\ubmX=(X_t)_{t\in T}$ are the stochastic-process analogue of min-id random vectors and arise as the possible limits of scaled minima of independent stochastic processes of the form
$$  \underline{\bm a}_n\lc\min_{\leqn} \ubmY^{(i)}-\underline{\bm b}_n\rc \to \ubmX , $$
where $\underline{\bm a}_n \in(0,\infty)^T$,$\underline{\bm b}_n\in \R^T$ are deterministic sequences and $\lc\ubmY^{(i)}\rc_{i\in\N}$ are independent stochastic processes, see 
\cite{Ginecontmaxidprocess1990,baalkemamaxidprocess1993} for details.
Equivalently, a stochastic process $\ubmX$ is minimum-infinitely divisible (min-id) if for every given $n\in\N$ there exist i.i.d.\ stochastic processes $\lc\ubmX^{(i,n)}\rc_{1\leq i\leq n }$ such that $\ubmX$ is distributed as their componentwise minima, i.e.\
\begin{align}
    \ubmX=\lc X_t\rc_{t\in T}\sim \lc \min_\leqn X^{(i,n)}_t\rc_{t\in T}=:\min_{\leqn} \ubmX^{(i,n)} ,\label{eqnmaxidrep}
\end{align}
explaining the nomenclature minimum-infinitely divisible. 

Under the assumption that $t\mapsto X_t$ is continuous, real-valued and that $\sup \{ x\in\R\mid \p\lc X_t \leq x\rc\textcolor{black}{<}1\}=\infty$ for all $t\in T$, which we also assume for the remainder of this section, \cite{Ginecontmaxidprocess1990,baalkemamaxidprocess1993} show that $\ubmX$ can be represented as the componentwise minimum of the atoms of a (usually infinite) PRM $N=\sum_{\iinn} \delta_{f_i}$ on the space of continuous functions 
$$C(T):=\{f\mid f:T\to (-\infty,\infty] \text{ is continuous}\}$$
i.e.,
\begin{align}
    \ubmX\sim \min_{\iinn} f_i \label{eqnmaxidrepasprm}.
\end{align}
The intensity measure $\lambda(\cdot):=\e\lk N(\cdot) \rk$ of the PRM $N$ is also called the exponent measure of $\ubmX$ and it uniquely characterizes its distribution via
$$ P\lc X_{t_i}> x_i \text{ for all } \leqn \rc=\exp\lc -\lambda\lc \{ f\mid f(t_i)\leq x_i\text{ for some }\leqn \}\rc\rc$$
Conversely, for every PRM $N=\sum_{i\in\N} \delta_{f_i}$ on $C(T)$ such that $\lambda\lc \{ f\mid f(t)<\infty \}\rc=\infty$ and $\lambda\lc \{ f\mid \inf_{t\in K}f(t)\leq a \}\rc<\infty$ for all $t\in T$, compact sets $K\subset T$ and $a\in \R$ one can define a corresponding continuous min-id process via
$$ \ubmX:=\lc X_t\rc_{t\in T}:=\lc\min_{i\in\N} f_i(t) \rc_{t\in T}.$$
One should observe that for $T=d\times\N$ or $T=\{1,\ldots,d\}$ every min-id process is continuous and thus the stochastic representation above holds for all min-id processes with those index sets.

When working with min-id processes it is very useful to introduce the concept of extremal functions and hitting scenarios. To be able to understand these concepts it is crucial to observe that for $\bmt:=(t_1,\ldots,t_k)$ the value of $\bmX_\bmt:=(X_{t_1},\ldots,X_{t_k})$ is fully determined by the atoms of the so-called extremal point measure at $\bmt$  given by
\begin{align}
    N^+_\bmt:=\sum_{f\in N} \delta_{f}\id_{\big\{ f(t_i)=X_{t_i} \text{ for some } 1\leq i\leq k \big\}}. \label{extremalmeasure}
\end{align} 
Thus, all atoms of the so-called subextremal point measure at $\bmt$ given by
\begin{align}
  N^-_\bmt:=\sum_{f\in N} \delta_{f}\id_{\big\{ f(t_i)>X_{t_i} \text{ for all } 1\leq i\leq k \big\}} \label{subextremalmeasure}
\end{align}
are irrelevant when we are solely interested in $\bmX_\bmt$. The atoms of $N^+_\bmt$, resp.\ $N^-_\bmt$, are called the extremal, resp.\ subextremal, functions at $\bmt$. \cite[Section 2]{dombryeyiminko2013regular} formally analyze the extremal and subextremal point measures and show that they are indeed well-defined random measures. Moreover, they show that $N_\bmt^+$ has finitely many atoms almost surely and that there is almost surely one extremal function per location $t_i$ \textcolor{black}{whenever $X_{t_i}$ has a continuous distribution function.}

Collecting those locations in $\bmt$ for which the same atom of $N^+_\bmt$ is responsible for seeing the observation $\bmX_\bmt$ one obtains a partition $\Theta=\lc\Theta_l\rc_{1\leq l\leq L(\Theta)}$ of $\bmt$ of length $L(\Theta)$, which is called the hitting scenario of $\bmt$. More precisely, we can write
$$ N^+_\bmt=\sum_{l=1}^{L(\Theta)} \delta_{e_l}$$
where $e_l$ are the extremal functions at locations $\Theta_l$, i.e.\ those functions for which
$$ X_t =e_l(t) \text{ for all } t\in\Theta_l. $$
Thus, the hitting scenario encodes the locations at which an extremal function determines the observation $\bmX_\bmt$.

Fig.\ \ref{figprmfunction} illustrates the extremal and subextremal functions of a continuous min-id process on $\R$ with $\bmt=(0,1,\ldots,5)$ and hitting scenario $\Theta(\{0,\ldots ,5\})=\big\{ \{0,1,2,5\},\{3,4\}\big\}$.

\begin{figure}
    \centering
   \includegraphics[scale=0.5]{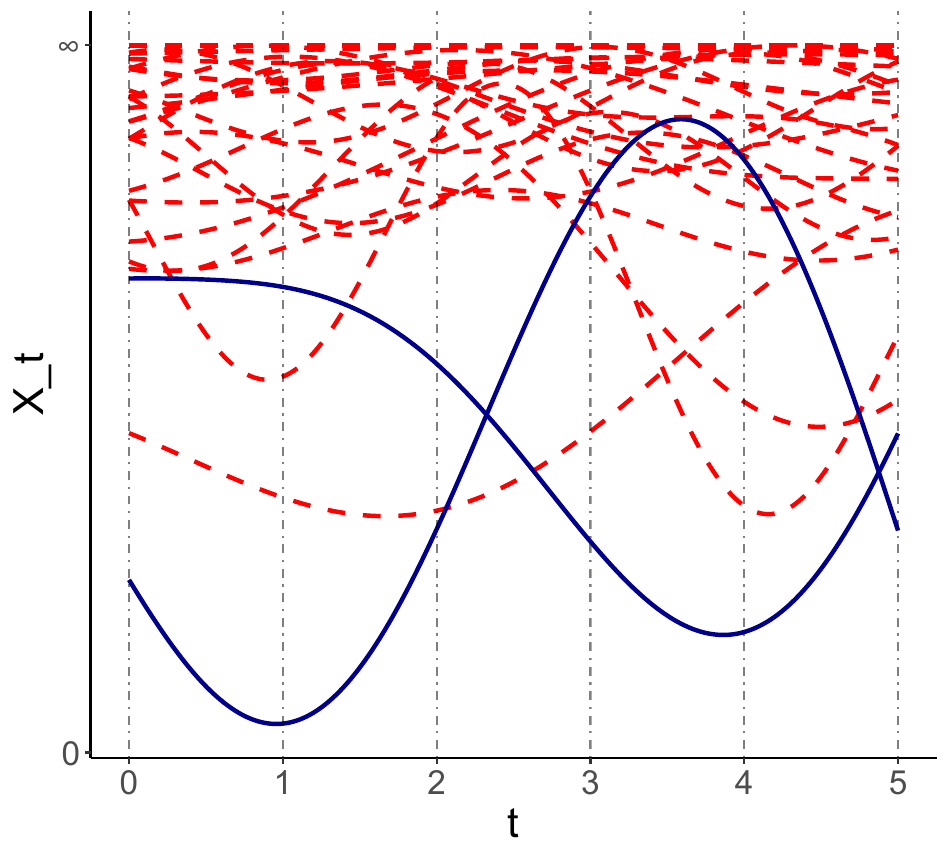}
\caption{\footnotesize  Illustration from \cite{brueckexactsim2022} of extremal and subextremal functions of a PRM $N$ on $C(\R)$. Functions in solid-blue belong to $N^+_{(0,\ldots,5)}$, functions in dashed-red belong to $N^-_{(0,\ldots,5)}$. The hitting scenario is given by $\Theta(\{0,\ldots ,5\})=\big\{ \{0,1,2,5\},\{3,4\}\big\}$}
\label{figprmfunction}
\end{figure}

\subsection{Conditional distribution of min-id processes}

It turns out that the concept of extremal functions and hitting scenarios is crucial to understand the conditional distribution of min-id processes. \textcolor{black}{To ease the presentation in the following we from now on additionally assume that $X_t$ has a continuous distribution function for every $t\in T$, essentially ensuring that there exists a unique extremal function per location.}
The core idea to derive the conditional distributions of a min-id process is to split the conditional distribution of $\ubmX$ given $\bmX_\bmt$ into two parts:
\begin{enumerate}
    \item The conditional distribution of the extremal point measure $N^+_\bmt$, encoded in the conditional hitting scenario with the corresponding extremal functions given by $\lc \Tilde{\Theta},(e_l)_{1\leq l\leq L(\Tilde{\Theta})}\rc$.
    \item The conditional distribution of the subextremal point measure $N_\bmt^-$.
\end{enumerate}

From the discussion above it is obvious that the conditional hitting scenario with the corresponding extremal functions determines the conditional distribution of the extremal point measure and together with the conditional distribution of the subextremal point measure one can fully describe the conditional distribution of $\ubmX$. Intuitively, one first reconstructs the random vector $\bmX_\bmt$ via the conditional hitting scenario $\Tilde{\Theta}$ and the extremal functions $(e_l)_{1\leq l\leq L(\Tilde\Theta)}$ and then reconstructs the remaining part of $\ubmX$ via the conditional distribution of the subextremal point measure. The corresponding stochastic representation of $\ubmX$ given $\bmX_\bmt$ is then given by
$$\ubmX\mid\bmX_\bmt \sim \min\big\{\ubmZ,\ubmY\big\}$$
where 
\begin{align}
   \ubmY\sim \min_{f\in N^-_\bmt} f \text{ and } \ubmZ \sim \min_{1\leq l \leq L(\Tilde{\Theta})}e_l, \label{defconddistdecomp} 
\end{align}
and we use the notation $f\in N$ to say that $N$ has an atom at $f$.

Let us start with the conditional distribution of $N^-_\bmt$ given $N^+_\bmt$, which is derived in \cite[Lemma 3.2]{dombryeyiminkostrongmixing2012}. The authors show that the conditional distribution of $N^-_\bmt$ given $N^+_\bmt$ is equal to the distribution of a PRM with intensity 
$$\id_{\{ f(t_i)>X_{t_i} , 1\leq i\leq k \}} \lambda(\rmd f).$$
Thus, the conditional distribution of the subextremal point measure $N^-_\bmt$ given $N^+_\bmt$ actually solely depends on $\bmX_\bmt$ and thus it is also the conditional distribution of $N^-_\bmt$ given $\bmX_\bmt$. In other words, the conditional distribution of $N^-_\bmt$ given $\bmX_\bmt$ is equal to the distribution of the PRM underlying the stochastic representation of $\ubmX$ subject to the restriction that neither of its atoms changes the value $\bmX_\bmt$, which is encoded by $\id_{\{ f(t_i)>X_{t_i} , 1\leq i\leq k \}}$. This directly implies that $\ubmY$ is again a min-id process whose exponent measure is given by $\id_{\{ f(t_i)>X_{t_i} , 1\leq i\leq k \}}\lambda(\rmd f)$.

It remains to describe the conditional distribution of $N^+_\bmt$ in terms of the conditional hitting scenario and the conditional distribution of the extremal functions. We start with the conditional distribution of the extremal functions given the  conditional hitting scenario, since the distribution of the conditional hitting scenario is more involved. Intuitively, conditionally on the conditional hitting scenario $\Tilde{\Theta}$, one needs to find functions $e_l$ such that 
$$e_l(t)=X_t \text{ for all }t\in\Tilde\Theta_l \text{ and that }e_l(t)>X_t \text{ for all } t\in \bmt\setminus\Tilde\Theta_l .$$
It turns out that the distribution of those functions can be described by the conditional distribution of the exponent measure $\lambda$. As $\lambda$ is usually an infinite measure, let us describe how we define the conditional distribution of an infinite measure, relying on the description in \cite[Appendix A2]{dombryeyiminko2013regular}. Formally, the conditional distribution of $\lambda$ is defined as follows. For every $\bmt\subset T$ and function $\mathsf{G}:(-\infty,\infty]^k\times C(T\setminus \bmt)\to [0,\infty)$ which vanishes on $\{\binfty\}\times  C(T\setminus \bmt) $ there is a disintegration formula for $\lambda$ given by
$$ \int \mathsf{G}\lc (f(t))_{t\in\bmt}, (f(t))_{t\in T\setminus \bmt}\rc\lambda(\rmd f) =\int \mathsf{G}\lc \bmx,(f(t))_{t\in T\setminus \bmt}\rc K_\bmt\lc\bmx,\rmd(f(t))_{t\in T\setminus \bmt}\rc \lambda^{(\bmt)}(\rmd\bmx) ,$$
where $K_\bmt(\bmx,\cdot)$ is a probability kernel on $C(T\setminus \bmt)$ for every $\bmx\in (-\infty,\infty]^k\setminus \binfty$ and $\lambda^{(\bmt)}$ denotes the $\bmt$-margin of $\lambda$ given by 
$$\lambda^{(\bmt)}(A)=\lambda\lc\{f\mid (f(t))_{t\in \bmt}\in A\}\rc.$$
The existence and well-definedness of $K_\bmt(\bmx,\cdot)$ is provided in \cite{dombryeyiminko2013regular}. Crucially, they also show that for a suitable normalizing constant $C(\bmX_\bmt,\Tilde{\Theta})$ we have
\begin{align}
&P\lc e_l\in A_l, 1\leq l\leq L(\Tilde{\Theta}) \mid \Tilde{\Theta}, \bmX_\bmt\rc \nonumber \\
&=C(\bmX_\bmt,\Tilde{\Theta})^{-1}\prod_{l=1}^{L(\Tilde\Theta)}K_{\Tilde{\Theta}_l}\lc\bmX_{\Tilde{\Theta}_l},\{ f \in A_l;\ (f(t))_{t\in \bmt \setminus \tilde{\Theta}_l}>\bmX_{\bmt\setminus\Tilde{\Theta}_l}\}\rc, \label{extrfctindep}
\end{align}
i.e., conditionally on the hitting scenario $\Tilde{\Theta}$ and $\bmX_\bmt$, the extremal functions $(e_l)_{1\leq l\leq L(\Tilde{\Theta})}$ are independent with distribution determined by the conditional distributions of $\lambda$. 

To fully describe the conditional distribution of a min-id process it solely remains to describe the distribution of the conditional hitting scenario. The crucial ingredient is a family of finite measures 
$\lc\beta_\theta\rc_{\theta\in\mathcal{P}_{\bmt}}$ on $\R^{k}$,
which are defined via 
$$\beta_\theta (A):=P\lc\bmX_\bmt\in A, \Theta=\theta\rc,$$
where $\mathcal{P}_{\bmt}$ denotes the set of all possible partitions of $\bmt$. Intuitively,  $\beta_\theta$ can be viewed as the (non-normalized) distribution of $\bmX_{\bmt}$ when the hitting scenario $\Theta$ is fixed to $\theta$. Therefore, we also have that $\beta:=\sum_{\theta\in\mathcal{P}_{\bmt}} \beta_\theta$ is equal to the distribution of $\bmX_\bmt$ since
$$ \beta(A):=\sum_{\theta\in\mathcal{P}_\bmt} \beta_\theta(A)=P(\bmX_\bmt\in A) $$
Trivially, each $\beta_\theta$ is absolutely continuous w.r.t.\ $\beta$ with a density $d\beta_\theta/d\beta$ on $\R^{k}$. \cite[Theorem 3.2.1]{dombryeyiminko2013regular} provides that the distribution of the conditional hitting scenario $\Tilde{\Theta}$ is given by
    $$ \tau(\bmX_\bmt,\theta):=P\lc \Tilde{\Theta}=\theta\rc= P\lc \Theta=\theta\ \big\vert\ \bmX_{\bmt}\rc=\frac{d\beta_{\theta}}{d\beta}(\bmX_\bmt). $$
Thus, one ``solely'' has to find the Radon-Nikodym derivative $d\beta_\theta/d\beta$ to determine the distribution of the conditional hitting scenario. However, since the relations among the $\beta_\theta$ are generally rather complex, providing more intuitive and explicit expressions of $d\beta_\theta/d\beta$ is difficult. A starting point can be that, also due to \cite{dombryeyiminko2013regular}, $\beta_\theta$ has the representation
\begin{align}
    \beta_{\theta}(\rmd\bmx)= \exp\lc -\lambda^{(\bmt)}\lc (\bmx,\binfty]^\complement\rc\rc \prod_{l=1}^{L(\theta)} K_{ \theta_l}\lc \bmx_{\theta_l},\big\{ (f(t))_{t\in \bmt\setminus \theta_l}>\bmx_{\bmt\setminus \theta_l} \big\}\rc   \otimes_{l=1}^{L(\theta)}\lambda^{(\theta_l)}(\rmd \bmx_{\theta_l}). \label{defbetameasures}
\end{align}
For example, when each $\lambda^{(\theta_l)}$ has a $\vert\theta_l\vert$-dimensional Lebesgue density, it is relatively easy to see that standard conditional density arguments can be employed to find $d\beta_\theta/d\beta$.

\subsection{A high-level overview of the simulation of the conditional distribution of a min-id process} 
\label{appsiminprocess}
Lastly, we want to provide some intuition about the simulation from the conditional distribution of a min-id process, which is one potential approach to simulate from the posterior of an exchangeable min-id process. However, we take a rather high-level approach as the details strongly depend on the specific model at hand and a schematic overview will help the reader to better grasp the general ideas. From the previous subsection it is clear that a simulation from the conditional distribution of a min-id process requires the simulation from the two stochastic processes $\ubmY$ and $\ubmZ$ which were defined in (\ref{defconddistdecomp}). 

We already know that $\ubmY$ is a min-id process and therefore simulation algorithms for min-id processes can be employed, such as the one proposed in \cite{brueckexactsim2022}. The general idea of the simulation algorithm is simple. One starts with an initialization of the extremal function $e_1$ corresponding to $Y_t$, where the location $t$ is arbitrary. Then one exploits that, conditionally on $Y_t$, $N_t^-$ is again distributed as a PRM with intensity measure $\id_{\{f(t)>e_1(t)\}}\lambda(\rmd f)$ to obtain $Y_{s}$ at another location $s$. This is possible since the only atoms of $N_t^-$ that are relevant to determine $Y_{s}$ are those atoms $f$ which satisfy $f(s)<e_1(s)$, as the remaining atoms are subextremal at location $s$. Due to the fact that $\lambda\lc \{ f\mid f(s)<a\}\rc<\infty$ for all $a\in\R$, this only requires the simulation of finitely many atoms of the corresponding PRM and thus is feasible in practice. Iterating this procedure for the remaining locations of interest allows to obtain an exact simulation of $\bmY_\bmt$ for arbitrarily many locations $\bmt$. The possibility of obtaining an initialization of the extremal function $e_1$ that we implicitly assumed above is not a complication. The results of the previous subsection show that one only needs to be able to simulate from the marginal distribution of $Y_t$ and then to simulate from $K_t(Y_t,\cdot)$ to obtain the extremal function. For the technical details of this algorithm we refer to \cite{brueckexactsim2022}.

Assuming that the conditional hitting scenario $\Tilde\Theta$ has already been simulated, the simulation from $\ubmZ\sim\min_{1\leq l\leq L(\Tilde{\Theta})} \ubmZ^{(n,l)}$ boils down to simulating each $\ubmZ^{(n,l)}$. Their distribution is given by $\ubmZ^{(n,l)}\sim c_l^{-1}K_{\Tilde{\Theta}_l}(\bmX_{\Tilde{\Theta}_l},\cdot \cap \{ (f(t))_{t\in \bmt\setminus \Tilde{\Theta}_l}>\bmX_{\bmt\setminus\Tilde{\Theta}_l} \})$ where $c_l:=K_{\Tilde{\Theta}_l}(\bmX_{\Tilde{\Theta}_l},\{(f(t))_{t\in\bmt\setminus\Tilde{\Theta}_l}>\bmX_{\bmt\setminus\Tilde{\Theta}_l}\})$. These functions are independent conditionally on $(\bmX_\bmt,\Tilde{\Theta})$, as follows from (\ref{extrfctindep}). The form of $K_{\tilde{\Theta}_l}$ is highly model dependent and the feasibility of simulation from this restricted and normalized kernel must be assessed case-by-case. The same is true for simulation from the conditional hitting scenario. In general, it is not easy to derive its distribution and, even if one is able to do so, it is still highly challenging to provide exact simulation algorithms as the size of its state space is the $\vert\bmt\vert$-th Bell number, which grows superexponentially. Thus, one usually must resort to an MCMC scheme to simulate from the conditional hitting scenario, as is described in Section \ref{appgibbssampler} below.


\section{A general recipe for simulation from the posterior of IDEM priors}
\label{appsimfromposterior}

In this section, we want to provide a discussion of the potential simulation approaches from the posterior of an IDEM prior. Even though the posterior distribution of an IDEM prior is known analytically, certain statistical functionals of interest are usually only accessible by simulating from the posterior, e.g.\ credible intervals of $S_\mu(\bm t)$, which is why access to a posterior simulation algorithm is usually essential in applications. Here, we focus the discussion on the simulation from the posterior predictive distribution, i.e.\ on the simulation of the conditional distribution of an exchangeable min-id sequence.
As a consequence, in contrast to the standard approach in Bayesian nonparametrics, this means that we suggest to not try to simulate from $\overline{\mu}\sim\mu\mid \ubmX_n$, but to simulate from $\lc\bmX_{n+i}\rc_{1\leq i\leq k}\mid \ubmX_n$ to derive such quantities. The reason that we follow this approach is that direct simulation from the random measure $ \overline{\mu}$ can be quite involved, since it usually does not have a simple generating process. Further, even when $\overline{\mu}$ could be simulated, one still needs to simulate from $\bmX_{n+i}\sim\minid(\overline{\mu})$ in many cases, which is a non-trivial problem itself as simulation from a min-id distribution corresponding to a given exponent measure is non-trivial. It is to be expected that for each realization of $\overline{\mu}$ one needs to adapt the simulation scheme of $\minid(\overline{\mu})$, which seems unrealistic in practice.

On the contrary, a simulation algorithm for $\lc\bmX_{n+i}\rc_{1\leq i\leq k}\mid \ubmX_n$ can be based on two different approaches. First, one can extend existing simulation algorithms for conditional distributions of min-id processes. This approach will, in theory, always work and we present it as a general recipe. Second, when Condition \ref{condabsolutecont} is satisfied, one can resort to more classical MCMC algorithms to simulate the conditional distributions of min-id processes.

First, let us discuss a general simulation algorithm, which only requires Condition \ref{condcontmargins} to be satisfied. To simplify the notation, we define $m=n+k$ and note that we need to simulate from the conditional distribution of $\ubmX_m$ given $\ubmX_n$. Then, recalling the stochastic representation (\ref{stochrepposterior}) of $\lc \bmX_{n+j}\rc_{j\in\N}\mid \ubmX_n\sim \min\{ \ubmY^{(n)},\ubmZ^{(n)}\}$, the simulation from $\ubmX_m$ boils down to simulating from three objects. First, one needs to simulate the first $k$ components from the exchangeable min-id sequence $\ubmY^{(n)}$, denoted as $\ubmY^{(n)}_k$. Second, one needs to simulate a conditional hitting scenario $\Tilde{\Theta}$ and, third, conditionally on the conditional hitting scenario $\Tilde{\Theta}$, one needs to simulate the first $k$ components from the $L(\Tilde{\Theta})$ random sequences $\lc\lc\bmZ_{j}^{(n,l)}\rc_{j\in\N}\rc_{1\leq l\leq L(\Tilde{\Theta})}$, denoted as $\ubmZ^{(n,l)}_k$, respectively. We discuss the simulation of each of these objects separately in the following paragraphs.

\textbf{$(I)$ Simulation of $\ubmY^{(n)}_k$:\ } To simulate $\ubmY^{(n)}_k$, one might think of first simulating $\overline{\mu}_n$ and then, conditionally on $\overline{\mu}_n$, simulating a random vector from $\minid(\overline{\mu}_n)$. However, as explained above, this is quite challenging since simulating the IDEM $\overline{\mu}_n$ is non-trivial and simulation from $\minid(\overline{\mu}_n)$ for changing realizations of $\overline{\mu}_n$ is also difficult. Therefore, it is generally quite difficult to use this procedure to generate samples from  $\ubmY^{(n)}_k$. Fortunately, exact simulation schemes for min-id processes have been developed in \cite{brueckexactsim2022}. The main advantage of this simulation algorithm is that it does not require to simulate from $\overline{\mu}_n$ to obtain $\ubmY^{(n)}_k$. Instead, it uses concepts from extreme value theory to iteratively simulate $\bmY^{(n)}_1,\ldots,\bmY^{(n)}_k$. Essentially, the simulation algorithm only requires that one is able to simulate from the probability measure corresponding to the appropriately normalized version of $\otimes_{i=1}^k \minid(\eta) \lc \big\{ \cdot \cap \{ \ubmx_k\mid x_{i,j}<a\} \big\}\rc \overline{\nu}_n(\rmd\eta)=:Q^{(i,j)}_{k,a}(\cdot)$ for every $(i,j)\in d\times k$ and $a<\infty$. Note that this still involves sampling a random exponent measure $\eta$, but this is usually simpler than sampling $\overline{\mu}_n$, since $Q_{k,a}^{(i,j)}$ is basically a mixture of min-id distributions with mixing measure of the form $\overline{\nu}_n(\rmd\eta)$. For example, when simulating from the posterior of a univariate NTR prior, simulation from $Q_{k,a}^{(i,j)}$ boils down to simulating $2$-dimensional random vectors instead of a CRM $\overline{\mu}_n$ with infinitely many atoms, see the discussion in \cite[Remark 7]{brueckexactsim2022}. Intuitively, the simulation algorithm only requires to simulate the few relevant atoms $\ubmy^{(e)}$ of the PRM $\sum_{e\in\N} \delta_{\ubmy^{(e)}}$ with intensity $\lambda_\alpha+\lambda_{\overline{\nu}_n}$ such that there is some $(i,j)\in d\times k$ such that $Y^{(n)}_{i,j}=y^{(e)}_{i,j}$. Thus, one can use the simulation algorithm of \cite{brueckexactsim2022} to simulate from $\ubmY^{(n)}_k$ and we refer the reader directly to the paper for the details.

\textbf{$(II)$ Simulation of $\Tilde{\Theta}\sim \Theta\mid\ubmX_n$:\ } Simulating from the conditional hitting scenario is quite involved, since there is a very large number of states that the random partition might assume and there is no simple formula for its distribution. For exponent measures of max-stable distributions which have a density w.r.t.\ the Lebesgue measure an approximate simulation algorithm based on a Gibbs sampling procedure has already been proposed in \cite{ribatet2013conditionalsimmaxstable}. Investigating their algorithm in detail reveals that, even though it is explicitly formulated for max-stable distributions, it essentially does not rely on the max-stability assumption. Further, the condition of assuming a density of the exponent measure w.r.t.\ the Lebesgue measure can also be relaxed to exponent measures having a density w.r.t.\ $ \Lambda_{dn}^{(\infty)}$, including those min-id distributions arising from IDEMs satisfying Condition \ref{condabsolutecont}. Essentially, the approach constructs a Gibbs sampler of the conditional hitting scenario, exploiting the conditional distribution of the conditional hitting scenario when one index $(i,j)$ is deleted from the partition. The precise mathematical derivations and a more detailed explanation is deferred to Section \ref{appgibbssampler}.

\textbf{$(III)$ Simulation of $\ubmZ^{(n,l)}_k$:\ } Provided a realization of the conditional hitting scenario $\Tilde{\Theta}$, it remains to simulate from the $L(\Tilde{\Theta})$ random vectors $\ubmZ^{(n,l)}_k$, whose distribution is determined via the conditional distributions $K_{\Tilde{\Theta}_l}(\bmX_{\Tilde{\Theta}_l},\cdot )$ conditioned on $\{\ubmy_{n\setminus\Tilde{\Theta}_l}>\ubmX_{n\setminus\Tilde{\Theta}_l}\}$, as in (\ref{defcondextrseq}). Those distributions are certainly highly model dependent and it is difficult to try to find a general recipe for simulating from $K_{\Tilde{\Theta}_l}$, since $\lambda$, and thus $K_{\Tilde{\Theta}_l}$, could be discrete, continuous or a mixture of both types of distributions.
When the exponent measure $\lambda^{(IJ)}$ has a density w.r.t.\ $ \Lambda_{\vert IJ\vert}^{(\infty)}$, in particular exponent measures arising from IDEMs satisfying Condition \ref{condabsolutecont}, the derivations in Section \ref{seccclosedformreppost} show that a general recipe for the simulation of $\ubmZ^{(n,l)}_k$ can be described as follows. Essentially, one has to simulate a $\overline{\mu}_n^{(l)}$ from a probability measure of the form $ h(\eta,\ubmX_n,\Tilde{\Theta}_l) \nu(\rmd\eta)$ and then, conditionally on $\overline{\mu}_n^{(l)}$, one has to simulate i.i.d.\ random vectors with distribution $\minid\lc \overline{\mu}_n^{(l)}\rc$. Again, it is to be expected that simulating $ \overline{\mu}_n^{(l)}$ is simpler than directly simulating $\overline{\mu}$.

To summarize, in this general simulation framework the simulation of $\ubmY^{(n)}_k$ and $\ubmZ^{(n,l)}_k$ requires the application of the algorithm of \cite{brueckexactsim2022} to simulate from several (simpler) min-id random vectors. In both cases, an application of \cite{brueckexactsim2022} essentially requires to first simulate an $\tilde{\eta}\sim \tilde{h}(\eta)\nu(\rmd\eta)$ for some density $\tilde{h}$ and then to simulate i.i.d.\ random vectors with distribution $\minid(\tilde{\eta})$. It is expected that simulation from $\tilde{h}(\eta)\nu(\rmd\eta)$ can still be non-trivial and it has to be checked case by case if simulation from these random measures is feasible. Furthermore, the simulation from the conditional hitting scenario is difficult and must be treated case-by-case, since there is generally no closed form representation of its probability distribution available.

\subsection{An MCMC scheme under Condition \ref{condabsolutecont}  }
\label{subsecmcmcmscheme}
Under Condition \ref{condabsolutecont}, the density of $\lambda^{(IJ)}$ w.r.t.\ $\Lambda^{(\infty)}_{\vert IJ\vert}$ is available in closed form and one can resort to a fully MCMC-driven approach.
First, Lemma \ref{lemcondhitscenregcond} provides a closed form representation of the conditional distribution of the conditional hitting scenario and therefore the Gibbs sampler from Supplementary Material \ref{appgibbssampler} below can be used for its simulation in step $(II)$. Moreover, Lemma \ref{lemconddistextrfctregcond} applies and the density of each $\ubmZ^{(n,l)}_k$ is known in closed-form. Thus, step $(III)$ above can be replaced by an MCMC scheme whose density is proportional to $\int_{M_d^0}\prod_{j=1}^k f_\eta(\bmx_j)h(\eta,\ubmX_n,\Tilde{\Theta}_l)\nu(\rmd\eta)+K(\ubmX_n, \Tilde{\Theta}_l)$. Moreover, step $(I)$ essentially requires to simulate from $Q_{k,a}^{(i,j)}$ whose density is proportional to $ \int_{M_d^0} \prod_{j=1}^k f_\eta(\bmx_j)\id_{\{x_{i,j}<a\}} \overline{\nu}_n(\rmd\eta) +g^{(d\times k)}(\ubmx_k)\id_{\{x_{i,j}<a\}}$. Thus, steps $(I)$, $(II)$ and $(III)$ can be fully conducted via a (nested) MCMC scheme.
 The bottleneck of this simulation approach is that the MCMC scheme usually involves the evaluation of high dimensional integrals, since $h(\eta,\ubmX_n,\Tilde{\Theta}_l)$ is defined in terms of high-dimensional integrals. In practice, evaluation of these integrals is usually quite difficult and it has to be assessed case-by-case if such evaluations can be conducted with reasonable accuracy. For example, for our model for partially exchangeable data from Section \ref{secmultivmodel} a similar reasoning as for the proof of Lemma \ref{lemdenshierachicalmodel} yields that $h(\eta,\ubmX_n,\Tilde{\Theta}_l)$ is available in closed form and $\int_{M_d^0}\prod_{j=1}^k f_\eta(\bmx_j)h(\eta,\ubmX_n,\Tilde{\Theta}_l)\nu(\rmd\eta)$ can be represented in terms of a $2$-dimensional Lebesgue integral and an expectation w.r.t\ to the law of a CRM. Supplementary Material \ref{appsimpartiallyexdata} shows how to practically apply the suggested MCMC scheme for steps $(I)$, $(II)$ and $(III)$ to our model for partially exchangeable data from Section \ref{secmultivmodel}.

\subsection{Gibbs sampler for conditional hitting scenarios}
\label{appgibbssampler}

In this subsection we extend the Gibbs sampling algorithm of \cite[Section 3]{ribatet2013conditionalsimmaxstable}, which provides a simulation algorithm for the conditional hitting scenarios of max-stable processes whose exponent measure has a Lebesgue density. We mimic their arguments to show that the algorithm can be extended to yield an algorithm for the simulation of the conditional hitting scenario of min-id distributions whose exponent measure has a density w.r.t.\ $\Lambda^{(\infty)}_{nd}$. In essence, the goal is to construct a Gibbs sampler on the space of partitions of $d\times n$ for the conditional hitting scenario $\Tilde{\Theta}\sim \Theta\mid\ubmX_n$ based on
\begin{align}
    \p\lc \Tilde{\Theta}\in \cdot \mid \Tilde{\Theta}_{\setminus(i,j)}=\theta_{\setminus(i,j)} \rc, \label{eqngibbshitscen}
\end{align} 
where $\Tilde{\Theta}_{\setminus(i,j)}$ (resp.\ $\theta_{\setminus(i,j)}$) denotes the partition of $d\times n\setminus \{(i,j)\}$ resulting from deleting the index $(i,j)$ from the partition $\Tilde\Theta$ (resp.\ $\theta$).

The idea is to use (\ref{eqngibbshitscen}) to consecutively update the states, i.e.\ the subset of the partition that $(i,j)$ belongs to, in a classical Gibbs sampling scheme. The conditional distribution in (\ref{eqngibbshitscen}) is convenient in a Gibbs sampling scheme, since there are at most $L(\theta_{\setminus(i,j)})+1\leq nd$ partitions $\Tilde{\theta}$ such that $\Tilde{\theta}_{\setminus(i,j)}=\theta_{\setminus(i,j)}$.

To describe $(\ref{eqngibbshitscen})$ in closed-from let us assume that Condition \ref{condabsolutecont} is satisfied. Then, Lemma \ref{lemcondhitscenregcond} states that the distribution of the conditional hitting scenario $\tau(\ubmX_n,\cdot)$ is given by
    \begin{align*}
       \tau(\ubmX_n,\theta)= C(\ubmX_n)^{-1} \prod_{l=1}^{L(\theta)} v(\theta_l,\ubmX_n)
    \end{align*}
    where 
    $$ v(\theta_l,\ubmX_n)=   \int_{\{  \ubmy_{n\setminus\theta_l}>\ubmX_{n\setminus \theta_l}\}} \lc f^{(d\times n)}+g^{(d\times n)}\rc\lc (\ubmX_n,\underline{\bmy}_n)(\theta_l) \rc  \lc   \Lambda^{(\infty)}_{nd-\vert \theta_l\vert}  \rc  \rmd \lc \ubmy_{n\setminus\theta_l}\rc      $$
    and $C(\ubmX_n)$ is a constant that only depends on $\ubmX_n$,
    recalling the notation $(\ubmX_n,\underline{\bmy}_n)(\theta_l):=\lc \id_{\{ (i,j)\in \theta_{l} \}}X_{i,j}+ y_{i,j}\id_{\{(i,j)\not \in \theta_{l}\}}\rc_{(i,j)\in d\times n}$.
  Therefore, (\ref{eqngibbshitscen}) can be written as 
  \begin{align}
       &P\lc \Tilde{\Theta}=\theta^\star \mid \Tilde{\Theta}_{\setminus(i,j)}=\theta_{\setminus(i,j)} \rc= \frac{ \tau(\ubmX_n,\theta^\star) \id_{\{ \theta^\star_{\setminus(i,j)} =\theta_{\setminus(i,j)} \}} } { \sum_{\tilde{\theta}\in \mathcal{P}_n  }  \tau(\ubmX_n,\tilde{\theta})\id_{\{\tilde{\theta}_{\setminus(i,j)}=\theta_{\setminus(i,j)} \}}  } \nonumber \\
       &\propto \id_{\{ \theta^\star_{\setminus(i,j)} =\theta_{\setminus(i,j)} \}}\prod_{l=1}^{L(\theta^\star)} v(\theta^\star_l,\ubmX_n)  \propto \frac{
 \id_{\{ \theta^\star_{\setminus(i,j)} =\theta_{\setminus(i,j)} \}} \prod_{l=1}^{L(\theta^\star)} v(\theta^\star_l,\ubmX_n)  } { \prod_{l=1}^{L(\theta)} v(\theta_l,\ubmX_n) }  \label{eqnmcmcdens}
  \end{align}
  A simple combinatorial argument shows that on the right-hand side of (\ref{eqnmcmcdens}) at most $4$ factors remain, since there are at most two subsets in the partition $\theta^\star$ that are not in $\theta$ and vice versa. Thus, since we have at most $L(\theta_{\setminus(i,j)})+1$ \textcolor{black}{candidate partitions $\theta^\star$, some of which may have probability zero}, we have at most $4dn$ factors to compute in each step of the Gibbs sampling scheme, which yields a practically feasible way to simulate the conditional hitting scenario. One should note that \textcolor{black}{the bound of $4dn$ concerns factor evaluations, and the number of evaluations can be much smaller when $L(\theta_{\setminus(i,j)})+1\ll dn$}.


\section{Simulation of a model for partially exchangeable data}
\label{appsimpartiallyexdata}
We want to simulate from the posterior of our model for partially exchangeable data with prior 
$$\lc \mu^{(\kappa_i)}_i\rc_{\leqd},\ \lc \mu_i\rc_{\leqd}\ \overset{indep}{\sim} \crm_{\rho_i(a)\mu_0^{(\kappa_0)}(\rmd b)},\ \mu_0\sim \crm_{l(a,b)\rmd (a,b)},$$
which was introduced in Section \ref{secmultivmodel}, assuming that $\mu_0$ has continuous base measure $\alpha_0$ while the $\lc\mu_i\rc_{\leqd}$ do not have a base measure. Further, we assume that we are given observations $\bmX_{IJ}=\lc X_{i,j}\rc_{1\leq j\leq n_i,1\leq i\leq d}$ from a partially exchangeable array, where $n_i$ is the number of observations in group $i$. To avoid extra notation, we still use $\ubmZ^{(n)}$, $\lc \ubmZ^{(n,l)}\rc_{1\leq l\leq L(\Tilde{\Theta})}$ and $\ubmY^{(n)}$ to denote the sequences appearing in the stochastic representation of the posterior of $\ubmX\mid \bmX_{IJ}$ in (\ref{stochrepposterior}), even though we condition on $\bmX_{IJ}$ instead of $\ubmX_n$. Moreover, to define a proper survival model, we will assume throughout this section that the CRMs $(\eta_i)_{0\leq i\leq d}$ and kernels $(\kappa^{(i)})_{0\leq i\leq d}$ are concentrated on $(0,\infty)$, which ensures that $\ubmX\in(0,\infty)^{d\times\N}$.

We want to follow the simulation approach suggested in Section \ref{subsecmcmcmscheme}. For the simulation from the posterior we need to simulate three types of objects: the conditional hitting scenario $\Tilde{\Theta}=(\Tilde{\Theta}_1,\ldots,\Tilde{\Theta}_L)$ with the associated $\lc \ubmZ^{(n,l)}\rc_{1\leq l\leq L}$ as well as the min-id random vector $\ubmY^{(n)}$. Section \ref{seccclosedformreppost} provides closed-form representations of their distribution in terms of the posterior density. However, in practice, these formulas are difficult to evaluate numerically and the bottleneck for simulation from the posterior via MCMC methods is the speed of the evaluation of these densities. Our strategy to tackle this issue is to first identify that all the involved densities share the same ``structure'' and that their fast and accurate numerical approximation boils down to the fast evaluation of the same type of object. This object will be represented in terms of partial derivatives of the Laplace transform of the underlying IDRMs and we will heavily rely on the automatic differentiation libraries of JAX to enable fast numerical evaluation of the partial derivatives and associated integrals. \footnote{The author has also tried  to evaluate the densities with classical numerical integration methods combined with a Monte-Carlo scheme, but this approach turned out to be way too slow in practice.}

Let $\eta$ denote a CRM without base measure and intensity $\gamma(\rmd a,\rmd b)$, let $J_{1},J_2\subset \N$ be some disjoint index sets and let $\bmx_{J_1}\in [0,\infty)^{\vert J_1\vert}$ as well as $\bmx_{J_2}\in[0,\infty]^{\vert J_2\vert}$. The key observation for the following derivations is that the common ``structure'' appearing in all of the densities below will be terms of the form
\begin{align}
   \e\lk \prod_{j\in J_1}e_{\eta,\kappa}(x_{j})\prod_{j\in J_2}E_{\eta,\kappa}(x_{j})\rk, \label{eqncommomstrct} 
\end{align}
where we define
$$ e_{\eta,\kappa}(x):= \int_0^\infty\kappa(x,y)\eta(\rmd y)\exp\lc - \int_0^{x}\int_0^\infty \kappa(s,y)\eta(\rmd y) \rmd s\rc $$
and
$$ E_{\eta,\kappa}(x):= \exp\lc - \int_0^{x}\int_0^\infty \kappa(s,y)\eta(\rmd y) \rmd s\rc. $$
There is no simple way to compute these expectations for a given CRM $\eta$, but recognizing that $e_{\eta,\kappa}(x)=-\frac{\partial}{\partial x}E_{\eta,\kappa}(x)$ we can rewrite  
$$\e\lk \prod_{j\in J_1}e_{\eta,\kappa}(x_{j})\prod_{j\in J_2}E_{\eta,\kappa}(x_{j})\rk=\e\lk \prod_{j\in J_1} \frac{-\partial}{\partial x_{j}} E_{\eta,\kappa}(x_{j})\prod_{j\in J_2}E_{\eta,\kappa}(x_{j})\rk.$$
Now, if we were allowed to interchange expectation and the derivatives, we would obtain
$$\e\lk \prod_{j\in J_1} \frac{-\partial}{\partial x_{j}} E_{\eta,\kappa}(x_{j})\prod_{j\in J_2}E_{\eta,\kappa}(x_{j})\rk=\lc\prod_{j\in J_1} \frac{-\partial}{\partial x_{j}}\rc \e\lk \prod_{j\in J_1\cup J_2}E_{\eta,\kappa}(x_{j})\rk,$$
with 
$$ \e\lk \prod_{j\in J_1\cup J_2}E_{\eta,\kappa}(x_{j})\rk=:L_{\kappa,\gamma(\rmd a,\rmd b)}(\bmx_{J_1},\bmx_{J_2}) $$
denoting the Laplace transform of the random vector
$$\lc  \lc \int_0^{x_{j}}\int_0^\infty \kappa(s,y)\eta(\rmd y)\rmd s\rc_{j\in J_1},  \lc \int_0^{x_{j}} \int_0^\infty \kappa(s,y)\eta(\rmd y)\rmd s\rc_{j\in J_2}\rc $$
evaluated at $\bm 1=(1)_{j\in J_1\cup J_2}$.
By Proposition \ref{proptrafos}, $L_{\kappa,\gamma(\rmd a,\rmd b)}$ is available in closed form and thus an evaluation of (\ref{eqncommomstrct}) boils down to calculating its derivatives.

In the argument above, we have assumed that we are allowed to interchange expectation and differentiation. This is however non-trivial as $x\mapsto E_{\eta,\kappa}(x)$ is only differentiable Lebesgue-almost everywhere since $x\mapsto\int_0^\infty \kappa(x,y)\eta(\rmd y)$ is, in general, not continuous. For example, when $\kappa(s,y)=\id_{\{s\geq y\}}$ is the Dykstra-Laud kernel then $\int_0^\infty \kappa(x,y)\eta(\rmd y)=\eta([0,x])$, which is càdlàg but not continuous. Therefore, an application of Leibniz' rule is precluded to interchange the expectation and differentiation. Nevertheless, the following lemma shows that we can interchange expectation and differentiation under general regularity conditions.

\begin{lem}
\label{lemderivrepdens}
Let $\eta$ denote a CRM on $[0,\infty)$ with intensity $\gamma(\rmd a,\rmd b)$ and let $\kappa$ denote a kernel which satisfies Condition \ref{condkernel}. Let $J_1,J_2\subset \N$ denote two disjoint finite index sets and let $\bmx_{J_1}=\lc x_j\rc_{j\in J_1}\in[0,\infty)^{\vert J_1\vert}$ and $\bmx_{J_2}:=\lc x_j\rc_{j\in J_2}\in[0,\infty]^{\vert J_2\vert}$. We have 
    \begin{align*}
    & \e\lk   \prod_{j\in J_1} e_{\eta,\kappa}(x_j)  \prod_{j\in J_2} E_{\eta,\kappa}(x_j)  \rk =\lc  \prod_{j\in J_1} \frac{-\partial}{\partial x_j}\rc  L_{\kappa,\gamma(\rmd a,\rmd b)}(\bmx_{J_1},\bmx_{J_2})
    \end{align*}
and for $\vert J_1\vert>1$ we have
    \begin{align*}
        &(-1)^{\vert J_1\vert}\lc  \prod_{j\in J_1} \frac{\partial}{\partial x_j}\rc  \log\lc L_{\kappa, \gamma(\rmd a,\rmd b)} \lc \bmx_{J_1},\bmx_{J_2})\rc \rc \\
        &= \int_0^\infty \int_0^\infty  \prod_{j\in J_1} e_{a\delta_b,\kappa}(x_j) \prod_{j\in J_2} E_{a\delta_b,\kappa}(x_j)  \gamma (\rmd a,\rmd b).
    \end{align*}
    $\lc \otimes_{k=1}^{\vert J_1\vert} \Lambda \rc \otimes \lc\otimes_{k=1}^{\vert J_2\vert } \Lambda^{(\infty)}\rc$-almost everywhere.

\end{lem}

Lemma \ref{lemderivrepdens} will be the central ingredient in developing a fast simulation routine for the posterior as it will allow us to calculate complex expectations via automatic differentiation of the Laplace transforms of univariate IDEMs, which are usually available in closed form due to Proposition \ref{proptrafos}.

\subsection{A convenient representation of the posterior densities}
Let us now use the results of Lemma \ref{lemderivrepdens} to represent the posterior densities in a convenient form. To simplify the exposition we assume that for $1\leq i\leq d$ we have $\int_0^\infty \kappa_i(s,y)\eta_i(\rmd y) \rmd s=\infty$ almost surely and we let $(\eta_i)_{\leqd}\overset{\text{indep}}{\sim} CRM_{\rho_i(a)\rmd  (a_0\delta_{b_0})^{(\kappa_0)}\rmd b}$ be fixed. 
First, we obtain from (\ref{def_f_IJ}) that the density of $\lambda^{(IJ)}$ is given by
\begin{align*}
        &f^{(IJ)}\lc \bmx_{IJ}\rc=\sumd   \id_{\{ x_{k,j}=\infty \forall (k,j)\in IJ, k\not=i\}}  \int_0^\infty \int_0^\infty  \lk \prod_{j:(i,j)\in IJ} e_{a\delta_b,\kappa_i}(x_{i,j})  \rk  \rho_i(a)   \rmd a \alpha_0^{(\kappa_0)}(\rmd b)  \\
        &+  \id_{\{\bmx_{IJ}\in\R^{\vert IJ\vert}\}}\int_{0}^\infty \int_{0}^\infty  \prodd   \e_{\eta_i}\Bigg[   \prod_{j:(i,j)\in IJ}  e_{\eta_i,\kappa_i}(x_{i,j}) \Bigg]    l(a_0,b_0) \rmd a_0\rmd b_0 ,
\end{align*} 
since for every non-zero $\eta$ we have $\int_0^\infty \eta([0,s])\rmd s=\infty$ and $g(\infty)=\infty\delta_\infty$.

\textbf{Conditional density of the conditional hitting scenario $\Tilde{\Theta}$:}
By Lemma \ref{lemcondhitscenregcond}, the distribution of the conditional hitting scenario $\Tilde\Theta$ is proportional to
\begin{align*}
\prod_{l=1}^{L(\Tilde{\Theta})}  \int_{\{  \ubmy_{IJ\setminus\Tilde{\Theta}_l}>\ubmX_{IJ\setminus\Tilde{\Theta}_l}\}} f^{(IJ)}\lc (\bmXIJ,\bmy_{IJ})(\Tilde{\Theta}_l) \rc  \lc   \Lambda^{(\infty)}_{|IJ|-\vert \Tilde{\Theta}_l\vert}  \rc  \rmd \lc \ubmy_{IJ\setminus\Tilde{\Theta}_l}\rc ,
\end{align*}
recalling the notation $(\bmXIJ,\bmy_{IJ})(\Tilde{\Theta}_l):=\lc \id_{\{ (i,j)\in \Tilde{\Theta}_{l} \}}X_{i,j}+ y_{i,j}\id_{\{(i,j)\not \in \Tilde{\Theta}_{l}\}}\rc_{(i,j)\in IJ}$. Next, (\ref{eqnmcmcdens}) from Supplementary Material \ref{appgibbssampler} gives that the conditional distribution of the conditional hitting scenario without index $(i,j)$ $P(\Tilde{\Theta}=\theta^\star\mid \Tilde{\Theta}_{\setminus (i,j)}=\tilde{\theta}_{\setminus (i,j)})$ is proportional to
$$ \id_{\{ \theta^\star_{\setminus(i,j)} =\tilde\theta_{\setminus(i,j)} \}}\prod_{l=1}^{L(\theta^\star)} v(\theta^\star_l,\bmXIJ)  \propto \frac{
 \id_{\{ \theta^\star_{\setminus(i,j)} =\tilde\theta_{\setminus(i,j)} \}} \prod_{l=1}^{L(\theta^\star)} v(\theta^\star_l,\bmXIJ)  } { \prod_{l=1}^{L(\tilde\theta)} v(\tilde\theta_l,\bmXIJ) }  $$
 where the factors $v(\theta_l,\bmXIJ)$ are given by 
\begin{align}
    & v(\theta_l,\bmXIJ)= \nonumber \\
    &\sumd   \id_{\{ \theta_l\subset i\times n\}}  \int_0^\infty \int_0^\infty  \lk \prod_{j:(i,j)\in \theta_l} e_{a\delta_b}(X_{i,j}) \prod_{j:(i,j)\in IJ\setminus \theta_l} E_{a\delta_b}(X_{i,j})  \rk  \rho_i(a)   \rmd a \alpha_0^{(\kappa_0)}\rmd b   \nonumber\\
&+  \int_{0}^\infty \int_{0}^\infty  \prodd \e_{\eta_i } \Bigg[  \prod_{j:(i,j)\in \theta_{l}}    e_{\eta_i} (X_{i,j}) \prod_{j:(i,j)\in IJ\setminus\theta_l}   E_{\eta_i}(X_{i,j}) \Bigg] l(a_0,b_0)\rmd a_0\rmd b_0 . \ \label{eqndesncondhitscen}
\end{align}

\textbf{Density of $\ubmZ^{(n,l)}$:}
Next, we determine the density of $\ubmZ^{(n,l)}$ from Lemma \ref{lemconddistextrfctregcond}, up to a normalizing constant.
\begin{lem}
\label{lemconddensextrseqpartexch}
   In the setting of Lemma \ref{lemconddistextrfctregcond} and conditionally on $(\bmXIJ,\Tilde{\Theta})$ the density of $\lc\bm Z^{(n,l)}_j\rc_{1\leq j\leq k}=:\ubmZ_k^{(n,l)}$  is proportional to
\begin{align}
&p_{\ubmZ^{(n,l)}_k}(\ubmx_k,\Tilde{\Theta}_l,\bmXIJ):= \label{eqndenscondextrseq} \\
&\sumd  \id_{\{ \ubmx_{k\setminus i\times k}=\binfty , \Tilde{\Theta}_l \subset i\times n\}} \int_0^\infty \int_0^\infty \Bigg( \prod_{j=1}^k e_{a\delta_b,\kappa_i}(x_{i,j})    \prod_{j:(i,j)\in \Tilde{\Theta}_l} e_{a\delta_b,\kappa_i}(X_{i,j}) \nonumber \\
&\hspace{3cm}\prod_{j: (i,j)\in IJ\setminus\Tilde{\Theta}_l} E_{a\delta_b,\kappa_i}(X_{i,j}) \Bigg)  \rho_i(a)\rmd a \alpha_0^{(\kappa_0)} (\rmd b) \nonumber\\
&+  \int_{0}^\infty \int_{0}^\infty  \prodd \e_{\eta_i,\kappa_i } \Bigg[  \prod_{j=1}^k e_{\eta_i,\kappa_i} (x_{i,j}) \prod_{j:(i,j)\in \Tilde{\Theta}_{l}}    e_{\eta_i,\kappa_i} (X_{i,j}) \nonumber\\
&\hspace{3cm} \prod_{j:(i,j)\in IJ\setminus\Tilde{\Theta}_l}   E_{\eta_i,\kappa_i}(X_{i,j}) \Bigg] l(a_0,b_0)\rmd a_0\rmd b_0  \nonumber
\end{align}
for every $1\leq l\leq L$.
\end{lem}

\textbf{Density of the exponent measure of $\ubmY^{(n)}$:}
From Supplementary Material \ref{appintrominprocess} and (\ref{defconddistexpmeasure}) it follows that, by definition of the conditional distribution of the exponent measure, for an exchangeable min-id sequence $\ubmY^{(n)}$ with exponent measure $\bar{\lambda}$ we can represent
$$\bar\lambda(A)=\int_\R \bar{K}_{i,j}\lc y_{i,j},\{  \ubmy\in A\}\rc \bar\lambda^{(i,j)}(\rmd y_{i,j}) $$
for every $(i,j)\in d\times\N$, where $\bar{K}_{i,j}$ denotes the conditional distribution of the exponent measure $\bar\lambda$ at location $(i,j)$ and $\bar{\lambda}_{i,j}$ is the exponent measure of $\bmY^{(n)}_{i,j}$.
Thus, given a PRM $N_{i,j}:=\sum_{k\in\N} \delta_{y^{(k)}_{i,j}}$ with intensity $\bar\lambda^{(i,j)}$ we can draw independent sequences $\ubmy^{(k)}\sim \bar{K}_{i,j}(y^{(k)}_{i,j},\cdot)$ and obtain that 
$$ \sum_{k\in\N} \delta_{\ubmy^{(k)}} $$
is a PRM with intensity $\bar\lambda$. Simulation of the univariate PRM $N_{i,j}$ is simple when $\bar\lambda^{(i,j)}((\cdot,\infty))$ can be evaluated, which in turn can be derived from (\ref{defpostexpmeasure}) and (\ref{defpostlevymeasure}). Further, due to partial exchangeability, $\bar{K}_{i,j}(y,\cdot)=\bar{K}_{i,1}(y,\cdot)$. 
Thus, we can simulate $\bar\lambda\lc\{ \ubmy\mid  y_{i,j}>a\}\rc$ exactly by first simulating the finitely many atoms $y_{i,j}^{(k)}>a$ and then, for each $y_{i,j}^{(k)}$ drawing independent realizations $\ubmy^{(k)}\sim K_{i,1}(y_{i,j}^{(k)},\cdot)$, see Appendices \ref{appsimfromposterior} and \ref{appsiminprocess} for more details. 

In our case, the exponent measure of $\ubmY^{(n)}_k$ has conditional distributions $\bar{K}_{i,1}(y,\ubmy_{k\setminus (i,1)}\in\cdot)$ whose density is proportional to
\begin{align}
    &p_{\ubmY^{(n)}_k,i}(y,\ubmy_{k\setminus (i,1)},\bmXIJ):= \nonumber\\
    & \id_{\{ \ubmx_{y\setminus i\times k}=\binfty \}} \int_0^\infty \int_0^\infty \prod_{j=1}^k e_{a\delta_b,\kappa_i}(y_{i,j})  \Bigg( \prod_{j:(i,j)\in IJ} E_{a\delta_b,\kappa_i}(X_{i,j}) \Bigg)  \rho_i(a)\rmd a \alpha_0^{(\kappa_0)}( \rmd b) \nonumber\\
    &+  \int_{0}^\infty \int_{0}^\infty  \prodd \e_{\eta_i,\kappa_i } \Bigg[ \prod_{j=1}^k e_{\eta_i,\kappa_i} (y_{i,j}) \prod_{j:(i,j)\in IJ}   E_{\eta_i,\kappa_i}(X_{i,j}) \Bigg] l(a_0,b_0)\rmd a_0\rmd b_0  ,\label{eqndenspostminidseq}
\end{align}
where we set $y_{i,1}:=y$ to keep the notation concise.

\begin{rem}
\label{rempracticaltrick}
From a practical perspective, it is convenient to observe that the density $p_{\ubmY^{(n)}_k,i}(y,\ubmy_{k\setminus (i,1)},\bmXIJ)$ from (\ref{eqndenspostminidseq}) is of the same type as the density of $\ubmZ^{(n,l)}_k$. In fact, we can construct a ``fictive'' min-id sequence with corresponding ``fictive'' hitting scenario and show that we can simulate from $\bar{K}_{i,1}(y,\ubmy_{k\setminus (i,1)}\in\cdot)$ by simulating the corresponding random vector of the type $\ubmZ_k^{(n,l)}$. This allows us to only implement the simulation of random vectors of the type $\ubmZ_k^{(n,l)}$ to also simulate from $\bar{K}_{i,1}(y,\ubmy_{k\setminus (i,1)}\in\cdot)$ (resp.\ $\ubmY^{(n)}_k$).

To understand how to construct the ``fictive'' min-id sequence, hitting scenario and random vector of type $\ubmZ^{(n,l)}_k$ let $n$ be such that $IJ\subset d\times (n-1)$ and define $n_+:=\{n,n+1,\ldots \}$. Further, define the ``reduced'' min-id sequence $\underline{\bm A}$ with index set $IJ \cup d\times n_+$ via
$$\underline{\bm A}:=\lc X_{i,j}\rc_{(i,j)\in IJ\cup d\times n_+}.$$
Then, for $i\in \{1,\ldots,d\}$, define a ``fictive'' hitting scenario of $IJ\cup (i,n)$ by 
$$\Tilde\Theta^{(i)}:=\lc\Tilde\Theta^{(i)}_1,\Tilde\Theta^{(i)}_2\rc:=\big((i,n),IJ \big).$$  
Let $\mathbf{K}_{i,n}(y,\cdot)$ denote the conditional distribution of the exponent measure of $\underline{\bm A}_{IJ\cup  d \times n_+}$ given $A_{i,n}=y$. Let $\underline{\bm{\mathsf{Z}}}^{(\Tilde\Theta^{(i)}_1)}\sim \mathbf{K}_{i,n}(y,\cdot) $, which has index set $IJ\cup d\times n_+$ and coincides with $A_{i,n}=y$ on $ (i,n)$.
Similarly to (\ref{defcondextrseq}), define 
$$\underline{\bm{\mathsf{Z}}}^{(n,1)}\sim \lc \bm{\mathsf{Z}}^{(\Tilde{\Theta}^{(i)}_1)}_{n+j}\rc _{j\geq 1}   \mid \underline{\bm{\mathsf{Z}}}^{(\Tilde\Theta^{(i)}_1)}_{IJ}>\bmXIJ.$$
Then, the density of $\underline{\bm{\mathsf{Z}}}^{(n,1)}_{ k\setminus (i,1)}$ is given by $\int_0^\infty p_{\underline{\bm{\mathsf{Z}}}^{(n,1)}_{k}}(\ubmy_k,(i,n),\bm A_{IJ\cup(i,n)}) \rmd y_{i,1} $ due to Lemma \ref{lemconddensextrseqpartexch} and it is easily seen that this coincides with $p_{\ubmY^{(n)}_k,i}(y,\ubmy_{k\setminus (i,1)},\bmXIJ)$.

Thus, to simulate from $\bar{K}_{i,1}(y,\ubmy_{k\setminus (i,1)}\in\cdot)$ we can simulate $\underline{\bm{\mathsf{Z}}}^{(n,1)}_k$ for the fictive min-id sequence $\underline{\bm A}$ with fictive conditional hitting scenario $\Tilde{\Theta}^{(i)}$ and set component $(i,1)$ of $\underline{\bm{\mathsf{Z}}}^{(n,1)}_{k}$ equal to $y$, exploiting partial exchangeability.
\end{rem}

\textbf{Summary:} To conclude this subsection, let us summarize the main message. It is possible to simulate $\Tilde{\Theta}$, $\lc \ubmZ^{(n,l)}\rc_{1\leq l\leq L(\Tilde{\Theta})}$ and $\ubmY^{(n)}$ via standard MCMC algorithms, but it requires fast evaluation of the densities (\ref{eqndesncondhitscen}), (\ref{eqndenscondextrseq})  and (\ref{eqndenspostminidseq}). Their efficient evaluation largely depends on the efficient evaluation of terms of the form
\begin{align}
    \e_{\eta } \Bigg[ \prod_{j\in J_1} e_{\eta,\kappa} (y_{j}) \prod_{j\in J_2}   E_{\eta,\kappa}(y_{j}) \Bigg] .\label{eqexpectationtodiff}
\end{align}
Due to Lemma \ref{lemderivrepdens}, these terms may be efficiently evaluated via automatic differentiation of terms of the form $L_{\kappa,\gamma(\rmd a\rmd b)}(\bmx_{J_1},\bmx_{J_2})$, which is the strategy that we will employ in the following.

\subsection{A concrete example}
In this section we want to illustrate that the above-mentioned computational tricks are indeed applicable in practice. For this purpose we focus on the concrete example of a generalized gamma CRM $\mu_0$ on $(0,\infty)$ with intensity
$$l(a,b)\rmd a\rmd b=a^{-1-\sigma} \exp\lc-a\rc \Gamma(1-\sigma)^{-1}\rmd a q(b) \rmd b,$$
where $q(b)$ denotes a non-negative function such that $\int_0^t q(b)\rmd b<\infty$ and $\sigma\in (0,1)$ is a parameter which determines the deviation from the Gamma CRM which is obtained by choosing $\sigma=0$. For $\kappa_0$ we choose the rectangular kernel $\kappa_0(s,y)=(2\tau_0)^{-1}\id_{\{ \vert s-y\vert\leq \tau_0\}}$ for some $\tau_0>0$. Further, the marginal CRMs $\mu_i$ are also based on generalized Gamma CRMs with identical $\sigma$ and Dykstra-Laud kernel, i.e.\ we choose $\rho_i(a):=\rho(a):=a^{-1-\sigma} \exp\lc-a\rc \Gamma(1-\sigma)^{-1}$ and $\kappa_i(s,y):=\kappa(s,y)=\tau_1\id_{\{s\geq y\}}$. 

\begin{lem} 
\label{lemclosedformexamle}
Let $\eta$ denote a CRM with intensity $\rho(a)\rmd a q(b)\rmd b$ and let $x_0:=0<x_1<\ldots<x_m<\infty$. Then
    \begin{align*}
     L_{\kappa,\rho(a)\rmd a q(b)\rmd b}(\ubmx_m)&=\exp\lc - \sum_{j=1}^m\int_{x_{j-1}}^{x_j}  \frac{\lc \tau_1\sum_{k=j}^m (x_k-b)_+ +1\rc^\sigma -1}{\sigma}  q(b)\rmd b \rc.
\end{align*}
\end{lem}

For simplicity, let us set $q(b)=1$ and $\alpha_0=\Lambda(\cdot \cap (0,\infty))$ in the following. Then, $\alpha_0^{(\kappa_0)}\rmd b=\int_0^\infty \kappa_0(b,y)\rmd y\rmd b$ and  $(a_0\delta_{b_0})^{(\kappa_0)}(\rmd b)=a_0\kappa_0(b,b_0)\rmd b$ and we obtain that for every $a_0,b_0>0$
\begin{align}
     &L_{\kappa,\rho(a)\rmd a (a_0\delta_{b_0})^{(\kappa_0)}(\rmd b)}(\ubmx_m) \label{eqnlaplacetrafoexample}\\
     &=\exp\lc - a_0(2\tau_0)^{-1} \sum_{j=1}^m\int_{x_{j-1}}^{x_j} \frac{\lc \tau_1\sum_{k=j}^m (x_k-b) +1\rc^\sigma -1}{\sigma} \id_{\{\vert b-b_0\vert \leq\tau_0\}}\rmd b \rc \nonumber\\
     &=\exp\Bigg( -a_0(2\tau_0)^{-1}  \sum_{j=1}^m\lc  \frac{\lc \tau_1\sum_{k=j}^m (x_k-b) +1\rc^{\sigma+1} }{-(m-j+1)\tau_1(\sigma+1)\sigma} \Bigg\vert_{b=\min\{\max\{x_{j-1},(b_0-\tau_0)_+\};b_0+\tau_0\} }^{b=\max\{\min\{x_{j},b_0+\tau_0\};(b_0-\tau_0)_+\}} \rc \nonumber\\
     &+a_0(2\tau_0)^{-1}\sigma^{-1}\lc \min\{ x_m;b_0+\tau_0\} -(b_0-\tau_0)_+\rc \Bigg) \nonumber
\end{align}

and
\begin{align}
    &-\log\lc L_{\kappa,\rho(a)\rmd a \alpha_0^{(\kappa_0)}\rmd b}(\ubmx_m) \rc \nonumber\\
    &=  \int_0^\infty  \frac{\lc \tau_1\sum_{j=1}^m (x_j-b)_+ +1\rc^\sigma -1}{\sigma} (2\tau_0)^{-1}\min\{b+\tau_0,2\tau_0\}\rmd b, \label{eqnloglaplaceexample}
\end{align}
where $(b_0-\tau_0)_+:=\max\{b_0-\tau_0;0\}$.

Therefore, applying the results from the previous subsection, the densities \linebreak $p_{\ubmY^{(n)}_k,i}(y,\ubmy_{k\setminus (i,1)},\bmXIJ)$, $p_{\bm Z^{(n,l)}}(\ubmx_k,\Tilde{\Theta}_l,\bmXIJ)$ and the factors $v(\theta_l,\bmXIJ)$ can be evaluated via automatic differentiation of (\ref{eqnlaplacetrafoexample}) and (\ref{eqnloglaplaceexample}) in conjunction with standard numerical integration methods as the expectations of the form (\ref{eqexpectationtodiff}) appearing in (\ref{eqndesncondhitscen}), (\ref{eqndenscondextrseq})  and (\ref{eqndenspostminidseq}) can be evaluated. 

It remains to describe the marginal distributions for $\ubmY^{(n)}$ in an implementable form. The following lemma provides this characterization via the marginal exponent measure $\bar\lambda^{(i,j)}$  of $Y^{(n)}_{i,j}$ in terms of expressions that have already been calculated in closed-form above.
\begin{lem}
\label{lempostminiduniv}
Conditionally on $\bmX_{IJ}$, the exponent measure of $Y^{(n)}_{i,j}$ is defined via 
\begin{align*}
 &\bar{\lambda}^{(i,j)}\lc  (x,\infty]^\complement \rc\\
 &=  \int_0^\infty  \frac{\lc \tau_1 (x-b)_+ +\tau_1\sum_{j:(i,j)\in IJ} (X_{i,j}-b)_+ +1 \rc^\sigma -1  }{\sigma} \alpha_0^{(\kappa_0)}\rmd b  \\
 & -\int_0^\infty  \frac{  \lc\tau_1\sum_{j:(i,j)\in IJ} (X_{i,j}-b)_+ +1\rc^\sigma-1}{\sigma} \alpha_0^{(\kappa_0)}\rmd b  \\
 &+ \int_0^\infty \int_0^\infty \prod_{i=1}^k \zeta_k\lc x ,(X_{i,j})_{j:(i,j)\in IJ} \rc   l(a_0,b_0) \rmd a_0\rmd b_0 ,
\end{align*}
where 
\begin{align*}
 &\zeta_k\lc x,(X_{i,j})_{j:(i,j)\in IJ} \rc =\\
 & \id_{\{i=k\}} \bigg(  L_{\rho(a)\rmd a (a_0\delta_{b_0})^{(\kappa_0)}(\rmd b)}\lc (X_{i,j})_{j:(i,j)\in IJ} \rc  -L_{\rho(a)\rmd a (a_0\delta_{b_0})^{(\kappa_0)}(\rmd b)}\lc x ,(X_{i,j})_{j:(i,j)\in IJ} \rc \bigg)\\
 &\ \ +\id_{ \{i\neq k\}}   L_{\rho(a)\rmd a (a_0\delta_{b_0})^{(\kappa_0)}(\rmd b)}\lc (X_{i,j})_{j:(i,j)\in IJ} \rc 
\end{align*}
and we define $L_{\rho(a)\rmd a (a_0\delta_{b_0})^{(\kappa_0)}(\rmd b)}\lc \emptyset \rc:=1$.
\end{lem}

Thus, we may simulate a PRM $N$ with intensity $\bar{\lambda}^{(i,j)}$ by evaluating the integrals appearing in Lemma \ref{lempostminiduniv}, which require essentially the evaluation of the same objects appearing in the densities $p_{\ubmY^{(n)}_k,i}(y,\ubmy_{k\setminus (i,1)},\bmXIJ)$ and $p_{\bm Z^{(n,l)}}(\ubmx_k,\Tilde{\Theta}_l,\bmXIJ)$. 
To summarize, we have now obtained representations of the densities $p_{\ubmY^{(n)}_k,i}(y,\ubmy_{k\setminus (i,1)},\bmXIJ)$, $p_{\bm Z^{(n,l)}}(\ubmx_k,\Tilde{\Theta}_l,\bmXIJ)$, $v(\theta_l,\bmXIJ)$ and the exponent measure $\bar\lambda^{(i,j)}$ which are implementable in closed form, since all remaining expectations can be calculated via automatic differentiation.

\begin{rem}[Proposal distributions and transition densities in Metropolis Hastings]

As $\ubmY^{(n)}$ and $\bmZ^{(n,l)}$ may take infinite values in their components when $\Tilde{\Theta}_l\subset \{i\}\times \N$ for some $i$ it is not obvious how the respective proposal distribution in an MCMC simulation scheme should look like. Thus, we here report our specific choices for the MCMC scheme that has been implemented by us.
For the Gibbs sampler of $\Tilde{\Theta}$ we follow the elaborations in Supplementary Material \ref{appgibbssampler}. For the simulation of $\ubmZ^{(n,l)}$ and $\ubmY^{(n)}$ we opted for the classical Metropolis-Hastings algorithm. We denote the current state by $\bmx$ (resp.\ $\theta$) and the state that we transition to by $\bmx\prime$ (resp.\ $\theta^\prime$).
Our transition distribution is based on lognormal random variables, which have density
 $$g_s(x)=\frac{1}{x\sqrt{2\pi}}\exp\lc -\lc\log(s)-\log(x)\rc^2 /2\rc,$$
 where $s>0$ is a scale parameter.
If $\Tilde\Theta_l$ contains indices from more than one row $i$ $\ubmZ^{(n,l)}_k$ concentrates on $\R^{d\times k}$ and we choose i.i.d.\ lognormal random variables with marginal scale parameter equal to the current value as proposals, which translates to the transition density
$$ g_1(\ubmx_k^\prime,\ubmx_k)=\prod_{i=1}^d\prod_{j=1}^k g_{x_{i,j}}(x^\prime_{i,j})\id_{\{x^\prime_{i,j}<\infty\}} \text{ w.r.t. } \Lambda_{dk}^{(\infty)}.$$
In the case $\Tilde\Theta_l$ contains indices from only one row $\underline{i}$ the distribution of $\bmZ^{(n,l)}_k$ either takes values in $\R^{d\times k}$ or $A_{\underline{i}}=\{ \ubmx_k\in(0,\infty]^{d\times k}\mid x_{\underline{i},j}\in \R\ \forall 1\leq j\leq k\text{ and } x_{i,j}=\infty\ \forall i\not=\underline{i}\}$ and we choose i.i.d.\ lognormal random variables with scale parameter equal to the current value as proposals, which translates to the transition density
\begin{align*}
    &g_{2,\underline{i}}(\ubmx_k^\prime,\ubmx_k)\\
    &=\frac{ \id_{\{\ubmx_k\in\R^{d\times k}\}} }{2}\Bigg(  \prod_{i=1}^d\prod_{j=1}^k g_{x_{i,j}}(x^\prime_{i,j})\id_{\{x^\prime_{i,j}<\infty\}} + \prod_{j=1}^kg_{x_{i,j}}(x^\prime_{i,j})\id_{\{x_{\underline{i},j}<\infty\}} \prod_{i=1,i\not=\underline{i}}^d\prod_{j=1}^k \id_{\{ x^\prime_{i,j}=\infty\}} \bigg)  \\
    &+\frac{ \id_{\{\ubmx_k\in A_{\underline{i}}\}} }{2}\bigg(  \prod_{j=1}^k g_{x_{\underline{i},j}}(x^\prime_{\underline{i},j})\id_{\{x^\prime_{\underline{i},j}<\infty\}}  \prod_{i=1,i\not=\underline{i}}^d\prod_{j=1}^k g_{1}(x^\prime_{i,j}) \id_{\{x^\prime_{i,j}<\infty\}} \\
    &\hspace{3cm}  +  \prod_{j=1}^kg_{x_{\underline{i},j}}(x^\prime_{\underline{i},j})\id_{\{x^\prime_{\underline{i},j}<\infty\}} \prod_{i=1,i\not=\underline{i}}^d\prod_{j=1}^k \id_{\{ x^\prime_{i,j}=\infty\}} \bigg),
\end{align*} 
 w.r.t. $\Lambda_{dk}^{(\infty)}$, where we default to a scale parameter $1$ when the current value is infinite.

A proposal density for $\ubmZ^{(n,l)}_k$ is given by $g_1(\ubmx_k^\prime,\ubmx_k)$ when $\Tilde{\Theta}_l$ contains indices from more than one row, and by $g_{2,\underline{i}}(\ubmx_k^\prime,\ubmx_k)$ when it contains indices only from row $\underline{i}$. Due to Remark \ref{rempracticaltrick} above, the same type of proposal can be used to simulate $\ubmY_k^{(n)}$.
\end{rem}

\subsection{Additional simulation results}
\label{appsimpartiallyexdataaddplots}
We present the detailed parametrization and additional simulation results for our model for partially exchangeable data by continuing Section \ref{secsimpartiallyexdatamaintext} from the main text, which already contains the plot of the posterior survival function in Figure \ref{figsurvfct}. The simulation algorithm is implemented by following the roadmap laid out in the previous subsections. Recall that the observations are given by $\{0.47,0.16,3.01,1.32,0.91,0.17\}$ and $\{0.24,1.42,0.96,1.11,0.14,1.87\}$. The underlying CRMs $(\mu_i)_{0\leq i\leq 2}$ were chosen as generalized gamma CRMs with identical jump intensities. More specifically, we chose $\mu_0$ as a CRM on $(0,\infty)$ with base measure $\alpha_0([0,t])=t$ and intensity 
$\rho(a)\rmd a \rmd b$,
where $\rho(a):=a^{-1-\sigma} \exp\lc-a\rc \Gamma(1-\sigma)^{-1}$ is the jump intensity of a generalized Gamma CRM and $\sigma\in (0,1)$ is a parameter which determines the deviation from the Gamma CRM, which is obtained by choosing $\sigma=0$. As the kernel $\kappa_0$ at the root we chose the rectangular kernel $\kappa_0(s,y)=(2\tau_0)^{-1}\id_{\{ \vert s-y\vert\leq \tau_0\}}$ where for the margins we chose the Dykstra-Laud kernel $\kappa(s,y)=\tau_1\id_{\{s\geq y\}}$. 
To create the plots we fixed $\tau_0=1$, $\tau_1=2/3$ and $\sigma=1/2$ to illustrate the implied prior and posterior distributions.

Since the posterior simulation procedure is computationally demanding, the purpose of the experiment is not to explore many different simulation scenarios, but rather to provide a transparent illustration of the prior-to-posterior update in a representative small-sample setting. We have drawn $1000$ samples from the posterior, which were generated according to the MCMC simulation algorithm detailed above. In particular, we have generated $1000$ conditional hitting scenarios from a single Markov chain, selecting every $10$-th entry after a burn in period of $2000$ observations. Then, for each selected conditional hitting scenario $\Tilde{\Theta}=\lc\Tilde{\Theta}_1,\ldots,\Tilde{\Theta}_L\rc$, we generated $L(\Tilde{\Theta})$ independent Markov chains to generate $\lc\ubmZ^{(6,l)}_1\rc_{1\leq l\leq L(\Tilde{\Theta})}$, using $5000$ iterations in each chain to extract the final observation. Finally, we also simulated $1000$ independent realizations of $\ubmY_1^{(6)}$ via the simulation algorithm from \cite{brueckexactsim2022}, leveraging essentially the same setup as for the simulation of the $\ubmZ^{(6,l)}_1$ following the outline in Remark \ref{rempracticaltrick}. In total, this results in $1000$ serially dependent observations from the posterior, where the serial dependence is inherited from the serial dependence of the simulated conditional hitting scenarios.

Let us first discuss the posterior predictive survival function in Figure \ref{figsurvfct}. The posterior predictive survival function is estimated by the empirical survival function of the posterior predictive samples, immediately providing a corresponding estimate of the posterior predictive cumulative hazard, which is provided in Figure \ref{fighazrate} below. \textcolor{black}{The Dvoretzky-Kiefer-Wolfowitz bands quantify the Monte Carlo
uncertainty of the empirical estimate of the predictive survival
function based on the serially correlated predictive draws and should not be interpreted as posterior credible bands for the random survival
function. Moreover, the bands should only be treated as indicative since the Dvoretzky-Kiefer-Wolfowitz inequality is usually only valid for independent observations. The corresponding bands for the cumulative hazard are
obtained by the transformation $H(t)=-\log S(t)$.} The survival function and the cumulative hazard plots both show that the posterior predictive distribution is able to learn from the data, since it shrinks the posterior towards smaller event times, reflecting the fact that our data mostly contains small event times. The effect is quite strongly visible, even though we only have $6$ observations in each group, which suggests that the posterior quickly adapts to the observations.

\begin{figure}[!htbp]
    \centering
    \includegraphics[width=0.8\textwidth]{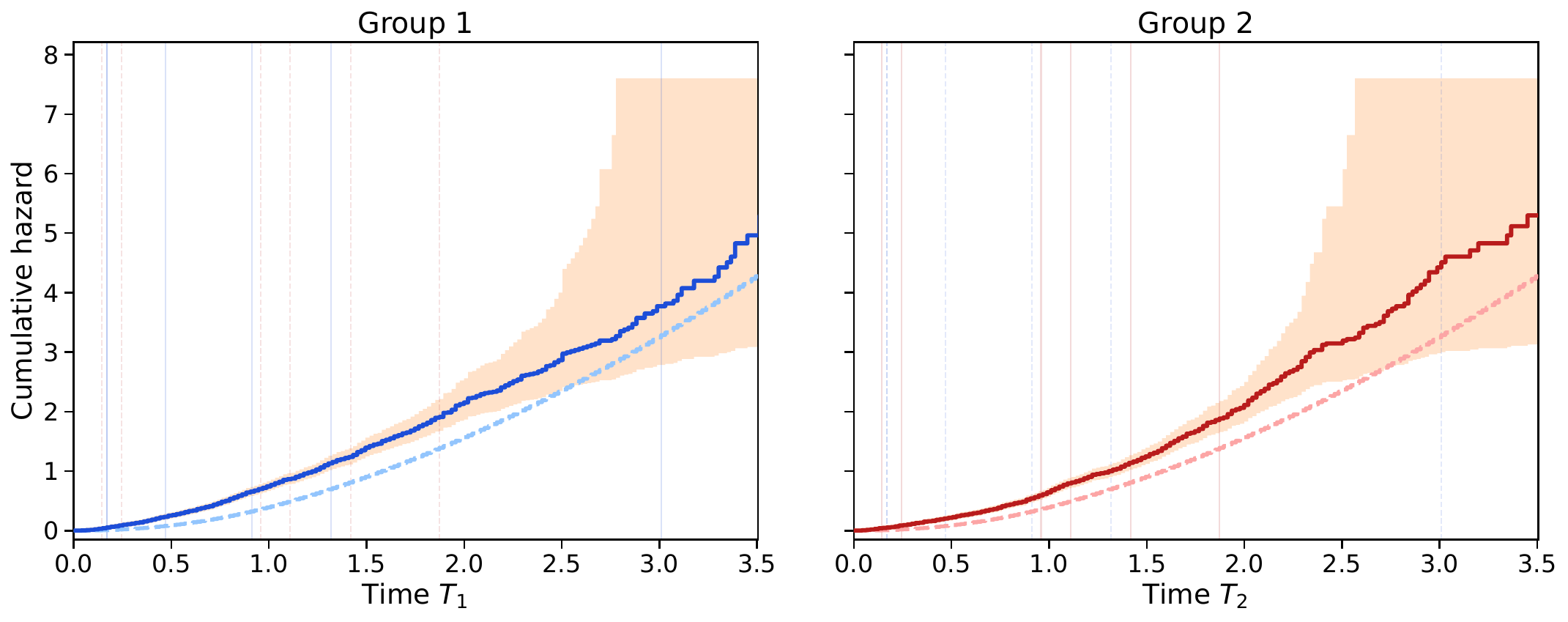}
    \caption{Prior and posterior predictive cumulative hazard functions. Colors (blue,red) distinguish the two groups and the opaque vertical lines mark the observed event times within the groups. The orange shaded regions are the approximative Dvoretzky-Kiefer-Wolfowitz bands quantifying the Monte Carlo uncertainty of the estimated cumulative hazard.}
    \label{fighazrate}
\end{figure}

Another natural theoretical object of interest associated with the model is the posterior behavior of hazard-rate related functionals. Recall that one of the motivations for our model for partially exchangeable data was that an observation should trigger an increase in the hazard rate around the observations, which was theoretically verified in the discussion after Lemma \ref{lemdenshierachicalmodel}. In practice, direct estimation of hazard rates from posterior predictive samples is numerically unstable, since it requires density estimation and division by an estimated survival function. We therefore avoid using estimated hazard rates as the primary numerical diagnostic. Instead, we focus on more stable predictive summaries that can be estimated directly from the samples and visualize the outcome of a local hazard rate increase around observations.
To visualize the local effect of the observations without estimating a hazard rate, we introduce two sample-based diagnostics. The first is the average local predictive mass around the observed event times,
\[
L_i(r)
=
\frac{1}{6}
\sum_{j=1}^{6}
P\left(
|X_{i,7}-t_{i,j}^{\mathrm{obs}}|\leq r
\mid \ubmX_6
\right),
\]
where \(t_{i,1}^{\mathrm{obs}},\ldots,t_{i,6}^{\mathrm{obs}}\) denote the observed event times in group \(i\). The second is the probability that a new predictive draw falls close to at least one observed event time,
\[
M_i(r)
=
P\left(
\min_{1\leq j\leq 6}
|X_{i,7}-t_{i,j}^{\mathrm{obs}}|
\leq r
\mid \ubmX_6
\right).
\]
In contrast to \(L_i(r)\), which averages local probabilities over the observed event times, \(M_i(r)\) treats the union of all neighborhoods around the observed times as a single event. Both diagnostics are computed under the prior and posterior predictive distributions. If the posterior predictive curve lies above the prior predictive curve for small or moderate values of \(r\), this indicates that conditioning on the data increases the predictive mass near the observed event times.
The resulting plots should therefore be interpreted as posterior predictive diagnostics rather than as direct estimates of the posterior hazard rates. They provide a stable visualization of the effect of conditioning on the data in a small-sample partially exchangeable survival model. In particular, they show whether the posterior predictive distribution moves mass toward the observed event times and how the marginal survival behavior differs from the prior distribution. 

\begin{figure}[!htbp]
    \centering
    \includegraphics[width=0.9\textwidth]{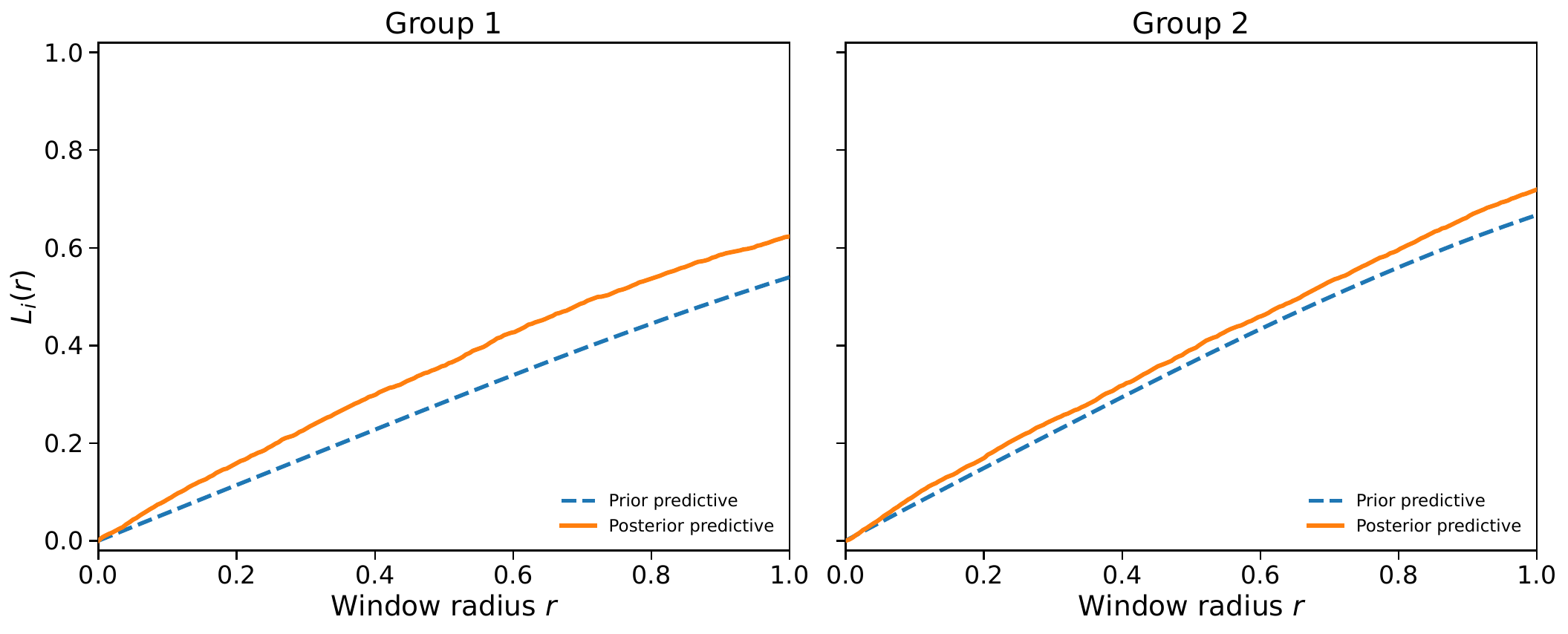}

    \caption{The local predictive mass diagnostics $L_i(r)$ for the prior (dotted,blue) and posterior predictive (solid, orange) distribution.}
    \label{figlocaldiagnosticsmassconc}
\end{figure}

\begin{figure}[!htbp]
    \centering
    \includegraphics[width=0.9\textwidth]{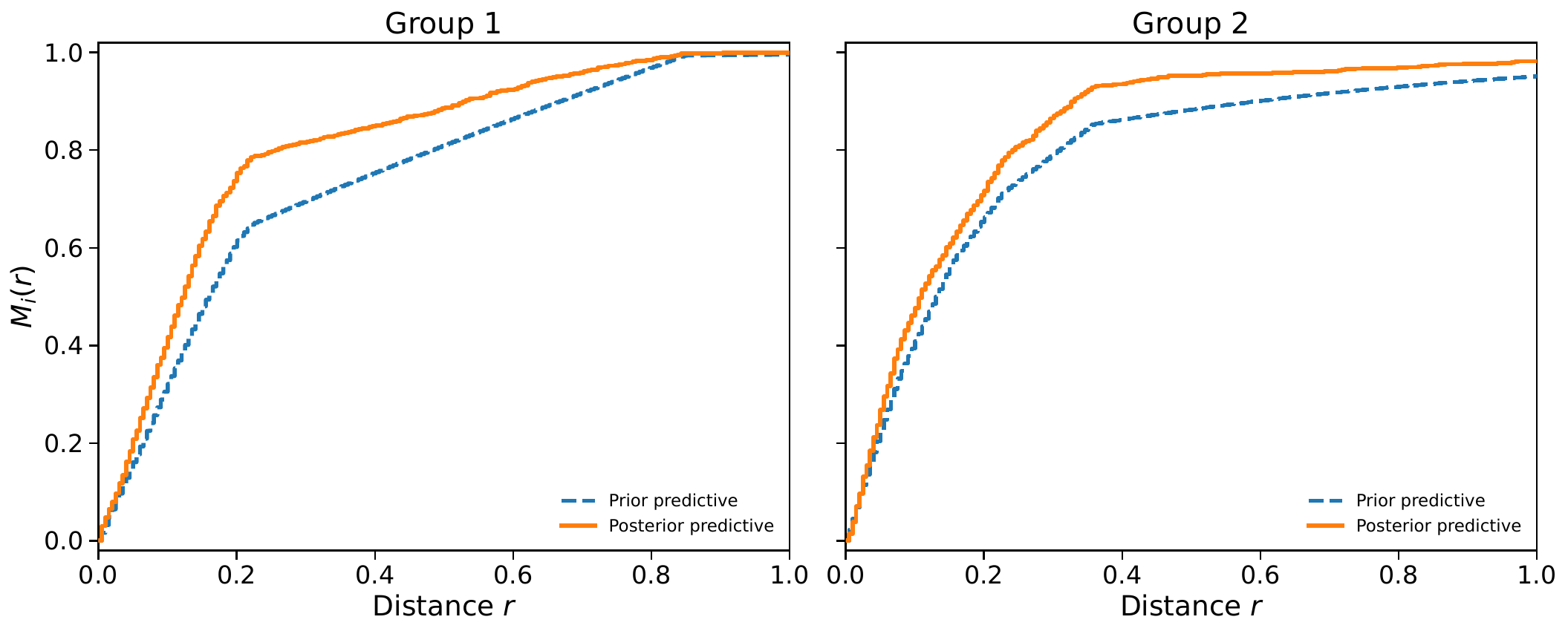}
    \caption{The local predictive nearest-observation probability diagnostics $M_i$ for the prior (dotted,blue) and posterior predictive (solid, orange) distribution.}
    \label{figlocaldiagnosticsclosestobs}
\end{figure}

Figure \ref{figlocaldiagnosticsmassconc} and Figure \ref{figlocaldiagnosticsclosestobs} illustrate that the posterior places significantly more mass close to the observations than the prior, where the majority of the increase is already reached for a small window radius of $0.4$, even though our choice of $\tau_0=1$ implies a hazard rate increase to be distributed in a window of radius $1$ around each observation. Even though the local predictive mass increase is less pronounced in group $2$, the local predictive nearest-observation probability is only slightly less pronounced. This could be due to the fact that the observations are less dispersed in group $2$ compared to group $1$, requiring a smaller adaptation in the posterior to match the observations. However, a theoretical verification of this claim is lacking. To sum up, similarly to the comments above, the plots again suggest that the posterior adapts rather quickly to the observations.

Furthermore, let us illustrate the dispersion of the posterior predictive draws via a scatter plot of the samples. A concentration of the samples around the $45$-degree reference line would indicate a strong dependence between the margins of the posterior samples, indicating that the learned common structure dominates the posterior and vice versa if the observations are scattered. Figure \ref{figposteriorscatter} below shows a strong dispersion of the posterior samples and thus indicates that the marginal effects dominate the posterior. This is in line with results in the literature, see e.g.\ the related model from \cite{camerlengilijoipruenster2021}, where the authors essentially show that posterior consistency is driven by the marginal specifications.
\begin{figure}[!htbp]
    \centering
    \includegraphics[width=0.4\textwidth]{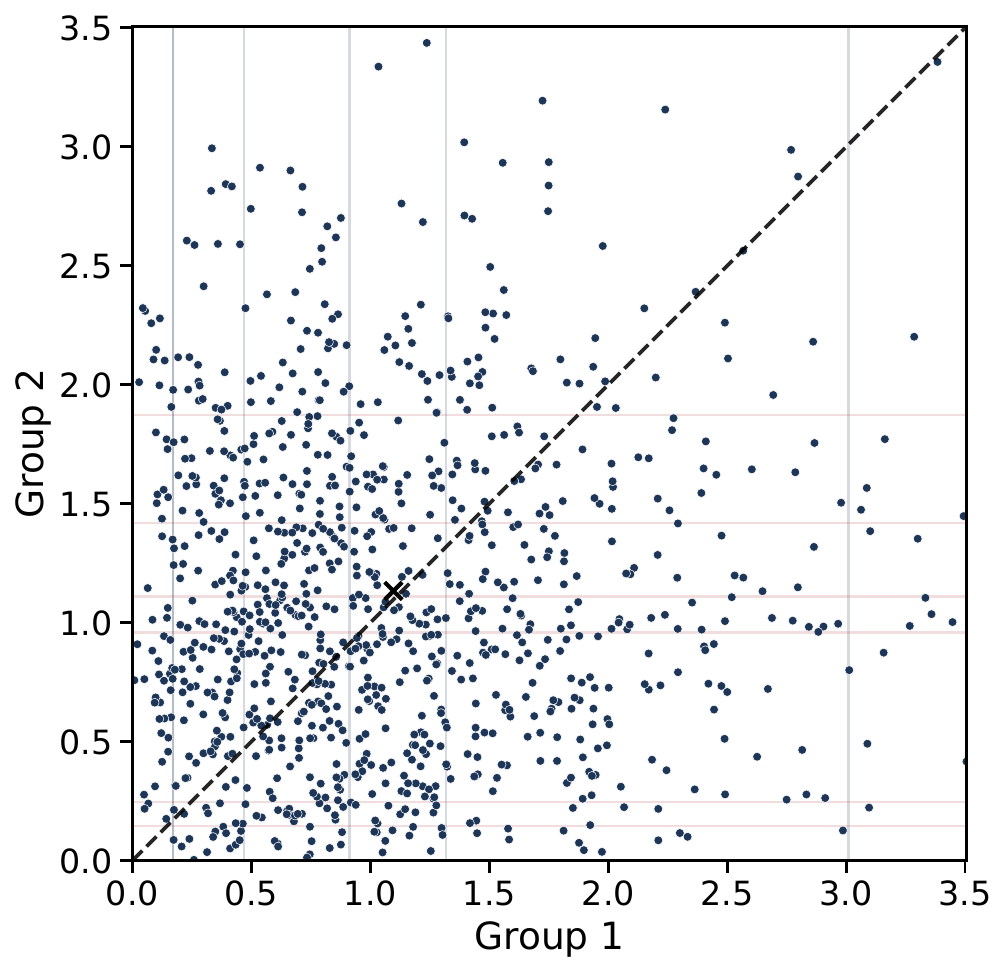}
    \caption{Joint posterior predictive draws for $(X_1, X_2)$ together with the $45$-degree reference line.}
    \label{figposteriorscatter}
\end{figure}

Finally, let us comment on the runtime of the simulation algorithm. In the present implementation, the main bottleneck is not the simulation logic itself but the repeated evaluation of expectations that enter the conditional densities and the probabilities of the conditional hitting scenarios. As explained above, these expectations are computed by automatic differentiation of closed-form Laplace transforms, since a standard Monte-Carlo approach was impractically slow as it required too many draws from the corresponding additive process. However, it is still numerically expensive because the MCMC scheme has to recompute many such quantities repeatedly over the course of the Gibbs updates for the hitting scenario and in the simulation of the min-id distribution.
Moreover, for each expectation which requires differentiation to be evaluated the computational cost increases rather quickly with the number of derivatives. As a consequence, adding an observation enlarges both the dimension of the mixed-derivative calculations and the combinatorial complexity of the conditional hitting scenarios, so that more candidate states and more expensive expectation evaluations have to be handled in each sweep of the algorithm. On a standard desktop computer, this growth was already severe enough that we could not push the experiment far beyond six observations in each margin within a reasonable runtime. To give a rough intuition, a derivative calculation of order $>10$ took about $500$ ms on average. Considering that in each step of the MCMC schemes several such derivatives might need to be calculated, depending on the current draw of the conditional hitting scenario, it becomes apparent that a group size of $10$ would lead to a long runtime for the MCMC algorithm. We  want to stress that this is solely a computational issue and not a theoretical one. A natural direction for future work would therefore be to parallelize these evaluations, for example by using distributed computing on a cluster or GPU-based acceleration for the automatic-differentiation workload. Developing such optimizations is beyond the scope of the present paper and beyond the computational resources and implementation capabilities available to the author.


\section{Additional properties of IDEM priors: prior moments, dependence and moments of mean functionals}
\label{apppriormom}
Let us shortly discuss some of the properties which are implied by an IDEM prior. We express these properties in terms of the Lévy characteristics, since they correspond to the natural description of IDRMs. First, we start by characterizing the prior moments of the random survival function and the cumulative hazard rate. To keep the notation concise, for every $\ubmx_m\in [-\infty,\infty)^{d\times m}, \bm{z}\in \R^m$, denote the $m$-dimensional Laplace transform of $\mu\lc (\cdot,\binfty]^\complement\rc$ as
\begin{align*}
    &L_{(\alpha,\nu)}(\bm{z},\ubmx_m):=\e\lk \exp\lc -\sum_{i=1}^m z_i\mu\lc (\bmx_i,\binfty]^\complement \rc  \rc  \rk \\
    &=\exp\bigg( -\sum_{i=1}^m z_i\alpha\lc (\bmx_i,\binfty]^\complement\rc - \int_{M_d^0}1-\exp\lc -\sum_{i=1}^m z_i\eta \lc (\bmx_i,\binfty]^\complement \rc\rc \nu(\rmd\eta)  \bigg), 
\end{align*}
implicitly assuming that the expectations are well-defined whenever we allow for negative $z_i$. This allows to express the prior moments of the random survival function as well as the moments of the cumulative hazard.
\begin{prop}
\label{proppriormoments}
    Let $\mu$ denote an IDEM with Lévy characteristics $(\alpha,\nu)$, where we denote the survival function associated with $\minid(\mu)$ as $S_\mu(\cdot)$. Then, upon well-definedness of the stated expectations, for every $\ubmx_m\in[-\infty,\infty)^{d\times m} ,\bm{z}\in \R^m$ and $(j_k)_{1\leq k\leq m}\in \N^m$, we have
    \begin{align*}
    &\e\lk \prod_{i=1}^{m} S_\mu(\bmx_i)^{z_i}\rk=L_{(\alpha,\nu)}(\bm{z},\ubmx_m) \text{ and } \\
    &\e\lk  \prod_{i=1}^m \mu\lc (\bmx_i,\binfty]^\complement \rc^{j_i}\rk=(-1)^{\sum_{i=1}^m j_i}\frac{\partial }{\prod_{i=1}^m (\partial z_i)^{j_i}} L_{(\alpha,\nu)}(\bm{z},\ubmx_m)\Big\vert_{\bm z=\bm 0}.
    \end{align*}  
\end{prop}

\begin{ex}
   Proposition \ref{proppriormoments} can be combined with Proposition \ref{proptrafos} to obtain moments of functionals of the form $m_{k,g}:=\e\lk \lc \int_{E_d^\prime} g(\bmx)\mu(\rmd \bmx) \rc^k \rk$. Proposition \ref{proptrafos} yields that $m_{k,g}$ is the $k$-th moment of an infinitely divisible random variable with tractable (univariate) Lévy measure $\upsilon(A):=\nu\lc \{ \eta \mid \int_{E_d^\prime} g(\bmx) \eta(\rmd\bmx)\in A \}\rc$ and (univariate) drift $a=\int_{E_d^\prime}g(\bmx)\alpha(\rmd\bmx)$. Thus, one can differentiate the Laplace transform of $ \int_{E_d^\prime} g(\bmx)\mu(\rmd \bmx)$ $k$-times at $0$ to obtain $m_{k,g}$.
\end{ex}

Up to this point, it is not clear whether an IDEM prior implies positive or negative dependence of $\bmX$ or $S_\mu(\cdot)$. The next proposition will show that an IDEM prior implies positive dependence of many statistical functionals, which is due to association. A random vector $\bm Z$ is called associated if $ \cov\lc g_1(\bm Z),g_2(\bm Z)\rc\geq 0 $ for all functions $g_1,g_2:\R^d\to\R$ which are monotone increasing in each coordinate and for which the covariance is finite. Association is a strong positive dependence property and it is well known that every min-id random vector is associated \cite[Section 5.6]{resnickextreme2013}, which implies that $\ubmX_n$ is associated for every $n\in\N$, conditionally on $\mu$ as well as unconditionally. Here, we will also show that real-valued functionals of $\mu$ are associated and the results are summarized in the following proposition.

\begin{prop}
\label{propassociation}
    Consider a collection of non-negative continuous functions \linebreak$(f_i)_{1\leq i\leq m}$ with compact support on $E_d^\prime$. Then, for every IDEM $\mu$
    $$ \lc \int_{E_d^\prime} f_1(\bmx) \mu(\rmd \bmx), \ldots, \int_{E_d^\prime} f_m(\bmx) \mu(\rmd \bmx) \rc $$
    is associated. Moreover, $\ubmX_n$ is associated for every $n\in\N$, conditionally on $\mu$ and unconditionally.
\end{prop}

An immediate consequence of Proposition \ref{propassociation} is that $\cov\big( \mu\lc (\bmx_2,\binfty]^\complement\rc,$\linebreak $ \mu\lc (\bmx_1,\binfty]^\complement\rc\big)\geq 0$ and $\cov(S_\mu(\bmx_1),S_\mu(\bmx_2))\geq 0$ for all $\bmx_1,\bmx_2\in\R^d$. An abstract implication of Proposition \ref{propassociation} which might be of independent interest beyond the scope of this paper is that non-negative infinitely divisible random vectors are associated, which follows by an appropriate choice of $(f_i)_{1\leq i\leq m}$ and $\mu$, exploiting that association is preserved under weak convergence \cite[Lemma 5.32 iii]{resnickextreme2013}.

Next, we turn to moments of mean functionals of the form 
$$I(f,\mu):=\int_{-\infty}^\infty f(x) \minid(\mu)(\rmd x),$$
where we assume that $d=1$ and that $f:\R\to [0,\infty)$ is monotone. Further, it is easy to see that we can focus on monotone increasing $f$, since we can derive results about monotone decreasing $f$ from the distribution of the monotone increasing function $-f$.

\begin{prop}
\label{propdistrofmeanfct}
Let $\mu$ denote an IDEM on $(-\infty,\infty]$ with Lévy characteristics $(\alpha,\nu)$. Then, for every monotone increasing $f:(-\infty,\infty]\to [0,\infty]$ with $f(\infty)=\infty$ we have $I(f,\mu)\sim I(\textit{id},\mu_f)$, where $\mu_f$ is an IDEM on $[0,\infty]$ with Lévy characteristics $(\alpha_f,\nu_f)$ as defined in Proposition \ref{proptrafos} $(i)$ and $\textit{id}$ denotes the identity function. Moreover, for every $m\in\N$, we have
\begin{align*}
       \e\lk  I(f,\mu)^m\rk&=\int_0^\infty \ldots \int_0^\infty L_{\alpha_f,\nu_f}\big(\bm 1,(t_1,\ldots,t_m)\big) \rmd t_1\ldots\rmd t_m .
\end{align*}    
\end{prop}


\section{Discrete priors: the multivariate NTR prior}

\label{appdiscprior}
Providing explicit and general formulas for the posterior of an IDEM is difficult for priors that concentrate on discrete distributions. The technical reason for this is that the counting measure, the natural dominating measure of discrete distributions, is not $\sigma$-finite and one cannot easily find an analog of Condition \ref{condabsolutecont} for discrete priors.  However, as Condition \ref{condcontmargins} is routinely satisfied even by priors concentrating on discrete distributions, the results of Section \ref{secpostdistr} are applicable and one can aim to find explicit formulas for families of discrete priors that are sufficiently regular. Here, we exemplarily show how to obtain explicit formulas for the posterior in the case of the multivariate NTR (mNTR) prior, i.e.\ CRM priors $\mu$ of the form 
\begin{align}
    \mu(A)=\sum_{k\in \N} a_k \delta_{\bmb_k}(A)=\int_0^\infty \int_{(0,\infty)^d} \textcolor{black}{a}\id_{\{\bmb\in A\}} N(\rmd (a,\bmb)), \label{condmCRM}
\end{align} 
where $N$ is a PRM on $(0,\infty)\times (0,\infty)^d$ with atoms $\lc (a_k,\bmb_k)\rc_{k\in\N}$. Obviously, each margin of $\mu$ is a univariate CRM and thus is a univariate NTR prior.
The mNTR prior is of special interest since it is the natural multivariate extension of univariate NTR priors introduced by \cite{fergusonprioronprob1974,doksumtailfree1974}. The univariate NTR prior has been used as a fundamental building block in many multivariate Bayesian models. Thus, the following closes a long-standing gap in the literature and the results may open the door to more refined multivariate applications of (m)NTR priors.

To simplify our developments we will assume that $\mu$ has no base measure and that the Lévy measure of $\mu$ is infinite and given as the image measure of
\begin{align}
    \lc (0,\infty) \times(0,\infty)^d ,\rho(a,\bmb)\rmd (a,\bmb)\rc;\ (a,\bmb) \mapsto a\delta_{\bmb}\textcolor{black}{+\infty\delta_\binfty}, \label{conddensmNTR}
\end{align} 
where $\rho(a,\bmb)$ is chosen such that $ \int_{\textcolor{black}{(\bm t,\binfty]^\complement}} \int_{0}^\infty \min\{a,1\} \rho( a,\bmb)\rmd ( a,\bmb)<\infty$ for all $\bm t\in \R^d$. \textcolor{black}{This ensures that $\mu$ is a well-defined IDEM. As required for an IDEM prior, we also assume Condition \ref{condsuppminid}. Together with the density assumption (\ref{conddensmNTR}), this implies Condition \ref{condcontmargins}.}

Before we derive the technical results underlying the posterior of an mNTR prior we present the main result of this section: the posterior of an mNTR prior. To this purpose, for any index set $IJ\subset d\times \N$, define $I(IJ):=\{i\mid\exists\ j\text{ s.t.\ } (i,j)\in IJ\}$ with $I(IJ)^\complement=\{1,\ldots,d\}\setminus I(IJ)$. Similarly, define $J(IJ):=\{ j\mid\ \exists\ \leqd \text{ s.t. } (i,j)\in IJ\}$. We will show below that for every conditional hitting scenario $\Tilde\Theta=\lc \Tilde{\Theta}_l\rc_{1\leq l\leq L(\Tilde{\Theta})}$ with positive probability we have that for all $\leqd$ and $1\leq l\leq L(\Tilde{\Theta})$ the vector $\bmX_{\Tilde{\Theta}_l\cap(\{i\}\times \N)}$ is a collection of ties, i.e.\ a vector consisting of identical values. Let us denote these unique values as $\bmX_{\Tilde\Theta_l}(i)$ and define the $d$-dimensional vector 
$$(\bmX_{\Tilde{\Theta}_l},\bm B)(I(\Tilde{\Theta}_l))=\lc \bmX_{\Tilde\Theta_l}(i)\id_{\{i\in I(\Tilde{\Theta}_l)\}}+B_i \id_{\{i\in I(\Tilde{\Theta}_l)^\complement\}}\rc_\leqd$$
for an arbitrary vector $\bm B$ such that $(B_i)_{i \in I(\Tilde{\Theta}_l)^\complement}$ is uniquely defined. To keep the notation concise, let us also denote $\times_{i=1}^k\int_{A_i}=\int_{A_1}\ldots\int_{A_k}$ for $k\in\N$.

\begin{thm}[Posterior of mNTR prior]
\label{thmpostmNTR}
Let $\mu$ denote a CRM without base measure whose intensity measure is given by (\ref{conddensmNTR}) and let $\ubmX$ denote the corresponding exchangeable min-id sequence such that $(\bmX_i)_{i\in\N}$ are conditionally i.i.d.\ with distribution $\minid(\mu)$. Further, let $IJ\subset d\times\N$ denote a finite index set and choose $n\in\N$ such that $IJ\subset d\times n$. Then, conditionally on $\bmXIJ$, the posterior of $\mu$ is given by
$$  \mu\mid \bmXIJ \sim \bar{\mu}_{IJ} +\sum_{l=1}^{L(\Tilde\Theta)} \bar{\mu}_{IJ}^{(l)}$$
where 
\begin{itemize}
    \item[$(i)$] $\bar{\mu}_{IJ}$ is a CRM with intensity $\rho(a,\bmb)\exp\lc-a\sum_{j\in J(IJ)} \id_{\{ b_i\leq X_{i,j}\text{ for some } (i,j)\in IJ\}}\rc\rmd a\rmd \bmb$.
    \item[$(ii)$] $\Tilde{\Theta}$ is distributed according to 
    $$ P(\tilde{\Theta}=\theta)= \tau(\theta,\bmXIJ)=\frac{q_{\theta} \lc \bmXIJ   \rc \id_{\{\Psi(\theta)=\Psi(\bmXIJ)\}} }{\sum_{\overset{\tilde\theta\in \mathcal{P}_{IJ}}{\Psi(\tilde\theta)=\Psi(\bmXIJ)}} q_{\tilde\theta}\lc \bmXIJ   \rc }   $$
    where $\Psi(\theta)$, $\Psi(\bmXIJ)$ and $q_\theta(\cdot)$ are defined in (\ref{defpsioftheta}),(\ref{defmarginalpartsbyobs}) and (\ref{defqdensmntr}) below.
    \item[$(iii)$] Conditionally on $\Tilde{\Theta}$, $\bar{\mu}_{IJ}^{(l)}\sim  A^{(l)}\delta_{(\bmX_{\Tilde{\Theta}_l},\bm B^{(l)})(I(\Tilde{\Theta}_l))}+\infty\delta_\binfty$ are independent for $1\leq l\leq L(\Tilde\Theta)$,
where $(A^{(l)},\bm B^{(l)}):=\lc A^{(l)},\lc (B^{(l)}_i)_{i\in I(\Tilde{\Theta}_l)^\complement}\rc\rc$ has probability density
$$C(\Tilde{\Theta}_l,\bmXIJ)^{-1}\rho_{IJ}^{(l)}\lc a,\lc b_i\rc_{i\in I(\Tilde\Theta_l)^\complement},\bmXIJ\rc \rho\lc a,(\bmX_{\Tilde{\Theta}_l},\bmb)(I(\Tilde{\Theta}_l))\rc  $$
on $(0,\infty)\times \textcolor{black}{(0,\infty)^{\vert I(\Tilde{\Theta}_l)^\complement\vert}}$ with
   \begin{align*}
       &\rho_{IJ}^{(l)}\lc a,\lc b_i\rc_{i\in I(\Tilde\Theta_l)^\complement},\bmXIJ\rc\\
       &:=\otimes_{j=1}^n \lc\minid\lc a\delta_{(\bmX_{\Tilde{\Theta}_l},\bmb)(I(\Tilde{\Theta}_l))}\rc \rc  \lc \big\{ \ubmy_n \ \big\vert \  \bmy_{IJ\setminus\tilde{\Theta}_l}>\bmX_{IJ\setminus\Tilde{\Theta}_l};\ \bmy_{\Tilde{\Theta}_l}=\bmX_{\Tilde{\Theta}_l} \big\} \rc 
   \end{align*}
 and  $C(\Tilde{\Theta}_l,\bmXIJ)$ denoting the corresponding normalizing constant. 
  
\end{itemize}
\textcolor{black}{Conditionally on $\bmXIJ$, $\bar{\mu}_{IJ}$ is independent of the hitting scenario $\Tilde{\Theta}$ and its associated random measures $\lc\bar{\mu}_{IJ}^{(l)}\rc_{1\leq l\leq L(\Tilde{\Theta})}$.}

\end{thm}

Similarly to the posterior of a univariate NTR prior, the posterior of an mNTR prior $\mu$ is the sum of a CRM $\bar\mu_{IJ}$ plus atoms with random weights given by $\bar{\mu}_{IJ}^{(l)}$. However, in contrast to the univariate case, the number of atoms is random and the atoms do not necessarily coincide with the observed locations $\lc \bmX_{IJ\cap (d\times \{j\})}\rc_{1\leq j\leq n}$, since the observations in the $i$-margins of $\bmXIJ$ only function as building blocks for the atoms of the posterior. This is due to the fact that, contrary to the univariate case, there \textcolor{black}{may be} several conditional hitting scenarios with positive probability as it cannot \textcolor{black}{in general} be known by solely observing $\bmXIJ$ whether or not the observations $X_{i,j}$ and $X_{\underline{i},j}$ are induced by the same underlying extremal sequence, see the technical results below for more details. If $d=1$, one recovers the classical result for univariate NTR priors, i.e.\ that the atoms in the posterior appear at each distinct entry of $\lc X_{1,j}\rc_{1\leq j\leq n}$.

Let us mention that the inclusion of a diffuse base measure $\alpha$ into the prior $\mu$ is feasible by going through the same technical steps as below. However, the formulation of the results would then require even heavier notation, which is why we refrained from treating this more general case. Furthermore, the same underlying ideas as below can be applied to CRMs that do not have an intensity of the form (\ref{conddensmNTR}). However, the formulation of the results is then highly dependent on the exact form of the intensity measure and must be adapted case-by-case.

\subsection{The conditional distributions of the exponent measure and hitting scenario of an mNTR prior}
To derive the posterior of an mNTR prior, it is essential to understand the stochastic representation of the corresponding exchangeable min-id sequence $\ubmX$. The same arguments as in Example \ref{exNTRpriors} yield that a stochastic representation of $\ubmX\overset{i.i.d.}{\sim}\minid(\mu)$ is given by
\begin{align}
    \ubmX \sim \min_{k\in\N} \bmx^{(k)}, \label{stochrepmNTR}
\end{align} 
where, conditionally on a PRM $\tilde N=\sum_{k\in\N} \delta_{(a_k,\bmb_k)}$ on $\textcolor{black}{(0,\infty)\times (0,\infty)^d}$, the $\ubmx^{(k)}$ are independent i.i.d.\ sequences with marginal distribution $(1-\exp(-a_k))\delta_{\bmb_k}+\exp(-a_k)\delta_\binfty $. Thus, the components of $\ubmx^{(k)}$ concentrate on the two points $\{\bmb_k,\binfty\}$.
As a consequence, the $i$-margins $(X_{i,j})_{j\in\N}$ of $\ubmX$ can have ties, resembling the familiar tie structure in the observations from a univariate NTR prior. The multivariate analog of this tie structure is encoded by the hitting scenario of $\ubmX$. In contrast to the univariate case, this tie structure cannot be directly read from the data as for $i\neq \underline{i}$ the observations $X_{i,j}$ and $X_{\underline{i},j}$ may originate from the same underlying extremal sequence $\ubmx^{(k)}$ even though $X_{i,j}\neq X_{\underline{i},j}$. However, as one might expect, the tie structure in the $i$-margins of $\ubmX$ will also play a critical role in the posterior of an mNTR prior, since it \textcolor{black}{imposes restrictions on the conditional hitting scenarios with positive probability}.

Before determining the distribution of the conditional hitting scenario, we must first derive the conditional distribution of the exponent measure $\lambda$ of $\ubmX$. From the stochastic representation (\ref{stochrepmNTR}) it directly follows that $\lambda$ is given by
\begin{align}
    \lambda(A)=\int_{(0,\infty)\times(0,\infty)^d} \int_{E_{d\times \N}} \id_{\{ \ubmx\in A\}}\otimes_{j\in\N}\minid(a\delta_{\bm b})\lc \rmd\ubmx \rc \rho(a,\bmb)\rmd a\rmd\bmb, \label{expmeasuremNTR}
\end{align}
for every measurable $A\subset E_{d\times\N}$. From (\ref{expmeasuremNTR}) it immediately follows that $\lambda$ has no mass on sets of the form
$$
\big\{ \ubmx \ \big\vert\ \exists\ j\in \N \text{ and } i\neq \underline{ i}\in\{1,\ldots,d\} \text{ s.t.\ }  x_{i,j}<\infty \text{ and } x_{\underline{i},j}=\infty  \big\}
$$
and
$$
\big\{ \ubmx \ \big\vert\ \exists\ j\neq \underline{j}\in \N \text{ and } i\in\{1,\ldots,d\} \text{ s.t.\ } x_{i,j}, x_{i,\underline{j}}\in \R \text{ and } x_{i,j}\neq x_{i,\underline{j}} \big\}.
$$
In other words, the $j$-th vector $\bmx_j$ in $\ubmx$ must either be finite in every component or be equal to $\binfty$ and all finite values in $(x_{i,j})_{j\in\N}$ must be identical. Therefore, when deriving the conditional distribution $K_{IJ}\lc \bmx_{IJ},\cdot\rc$ of $\lambda$, we focus on $\bmxij$ that complies with these restrictions a priori. More precisely, we only need to derive $K_{IJ}\lc \bmx_{IJ},\cdot\rc$ for
\begin{align*}
    &\bmx_{IJ}\in \mathcal{U}:= \\
    &\big\{ \bmy_{IJ}\in E^\prime_{IJ}  \mid \forall\ j \in J(IJ) : \text{Either } \bmy_{IJ\cap (d\times \{j\})}=\binfty \text{ or } \bmy_{IJ\cap (d\times \{j\}) }\in\R^{\vert IJ\cap (d\times \{j\})\vert} \}\\
    &\cap\{  \bmy_{IJ}\in E^\prime_{IJ}  \mid  \forall\ i\in I(IJ): \text{ All finite components of }\bmy_{IJ\cap (\{i\}\times \N)} \text{ are identical} \}.
\end{align*}
 
 To ease the notation in the following denote the finite values in the $i$-margins of $\bmxij\in \mathcal{U}$ as $\bmxij(i)$, if they exist.  Further, denote
$  I(\bmx_{IJ}):=\big\{i\in\{1,\ldots,d\}\mid x_{i,j}<\infty\text{ for some }(i,j)\in IJ\}$, 
$I(\bmx_{IJ})^\complement=\{1,\ldots,d\}\setminus I(\bmxij)$ and for any $\bmb\in \R^d$ let 
$$(\bmxij,\bmb)(I(\bmxij)):= \lc \bmxij(i)\id_{\{i\in I(\bmxij)\}}+b_i\id_{\{i\in I(\bmxij)^\complement\}}\rc_{\leqd}.$$
Note that in case $\bmx_{IJ}\in \R^{\vert IJ\vert}$ we have that $I(\bmx_{IJ})=I(IJ)=\{i\mid \exists \ j\in\N\text{ such that }(i,j)\in IJ \}$ is simply the collection of $i$-margins that are present in $IJ$, whereas $I(IJ)^\complement=\{i\mid (i,j)\not\in IJ \ \forall j\in\N\}$ is the collection of $i$-margins that are not present in $IJ$.

\begin{prop}
\label{propconddistrmNTR}
Let $\mu$ denote a CRM without base measure whose Lévy measure has intensity (\ref{conddensmNTR}). Let $\ubmX$ denote the corresponding exchangeable min-id sequence with exponent measure $\lambda$. Then, for any finite $IJ\subset d\times \N$, \textcolor{black}{$\lambda^{(IJ)}$-almost every $\bmxij\in \mathcal{U}$} and measurable $A\subset E^\prime_{d\times \N\setminus IJ}$, the conditional probability distribution of $\lambda$ given $\bmxij\in\mathcal{U}$ is given by
        \begin{align*}
       &K_{IJ}\lc \bmx_{IJ},A\rc\\
       &=   \lc \times_{i=1}^{\vert I(\bmxij)^\complement\vert }\int_0^\infty\rc \int_0^\infty \otimes_{j\in\N} \minid\lc a\delta_{(\bmxij,\bmb)(I(\bmxij))}\rc \lc \big\{ \ubmy\mid \ubmy_{\setminus IJ}\in A , \bmy_{IJ}=\bmx_{IJ}\big\}\rc \\
       &\hspace{1cm}\rho\lc a,(\bmxij,\bmb)(I(\bmxij))\rc \rmd a \rmd \lc (b_i)_{i\in I(\bmxij)^\complement} \rc C(\bmxij)^{-1}
    \end{align*}
    where  $C(\bmxij)$ denotes the corresponding normalizing constant. 
\end{prop}

To derive the conditional hitting scenario of an mNTR prior we recall the family of measures 
$\lc\beta_\theta\rc_{\theta\in\mathcal{P}_{IJ}}$ on $\R^{\vert IJ\vert}$ given by 
\begin{align*}
   \beta_\theta (A):=P\lc\bmXIJ\in A, \Theta=\theta\rc, 
\end{align*}
which were already defined in (\ref{defbetameasures}), recalling that $\mathcal{P}_{IJ}$ denotes all possible partitions of $IJ$. Moreover, we recall that $\beta_\theta$ has the representation
\begin{align*}
    &\beta_{\theta}(\rmd\bmxij)= \\
    &\exp\lc -\lambda^{(IJ)}\lc (\bmxij,\binfty]^\complement\rc\rc \prod_{l=1}^{L(\theta)} K_{ \theta_l}\lc \bmx_{\theta_l},\big\{ \bmy_{IJ\setminus \theta_l}>\bmx_{IJ\setminus \theta_l} \big\}\rc   \otimes_{l=1}^{L(\theta)}\lambda^{(\theta_l)}(\rmd \bmx_{\theta_l}),
\end{align*}
where $\lambda^{(\theta_l)}$ denotes the distribution of the $\theta_l$ margin of $\lambda$ given by 
$$ \lambda^{(\theta_l)}(\cdot):=\lambda\lc \{ \ubmx \mid \bmx_{\theta_l}\in\cdot\}\rc. $$
Intuitively, $\beta_\theta$ can be viewed as the (non-normalized) distribution of $\bmX_{IJ}$ when the hitting scenario is fixed to $\theta$. Each $\beta_\theta$ is absolutely continuous w.r.t.\ the distribution of $\bmXIJ$ given by $$\minid(\lambda^{(IJ)})=\sum_{\theta\in\mathcal{P}_{IJ}} \beta_\theta$$
with a density $\rmd\beta_\theta/\rmd \minid(\lambda^{(IJ)})$ on $\R^{\vert IJ\vert}$. By \cite[Theorem 3.2.1]{dombryeyiminko2013regular} the conditional hitting scenario $\Tilde{\Theta}$ follows the distribution
    $$ \tau(\bmXIJ,\theta)=P\lc \Theta=\theta\mid \bmX_{IJ}\rc=\frac{\rmd\beta_{\theta}}{\rmd\minid(\lambda^{(IJ)})}(\bmXIJ). $$
Thus, to derive the distribution of the conditional hitting scenario we have to find $\rmd\beta_\theta/\linebreak\rmd\minid(\lambda^{(IJ)})$. Before turning to this task, we need to introduce further notation. For a fixed partition $\theta$ of $IJ$ define a partition $ \Psi^{(i)}(\theta)$ of $IJ\cap (\{i\}\times\N)$ via
$$\Psi^{(i)}(\theta):=\lc\Psi^{(i)}_l(\theta)\rc_{1\leq l\leq L(\Psi^{(i)}(\theta))}:=\big\{ \theta_l\cap (\{i\}\times\N) \mid 1\leq l\leq L(\theta); \theta_l\cap (\{i\}\times\N)\neq \emptyset\big\},$$
which is the partition of the $i$-margin of $IJ$ that is induced by $\theta$. Further, let 
\begin{align}
    \Psi(\theta):=\lc \Psi^{(i)}(\theta)\rc_{i\in I(IJ)} \label{defpsioftheta}
\end{align} 
denote the collection of the partitions of the $i$-margins of $IJ$ that are induced by $\theta$.

It is crucial to observe that each observation of $\bmXIJ$ uniquely corresponds to a tie structure  $\Psi(\bmXIJ)=\lc \Psi^{(i)}(\bmXIJ)\rc_{\leqd}$, which can be directly read from the data by grouping the indices of all observations in an $i$-margin that take identical values.
Formally, $\Psi^{(i)}(\bmXIJ)$ is the partition of $(\{i\}\times\N)\cap IJ$ induced by the equivalence relation 
\begin{align}
(i,j)\sim ( i,\underline{j})\Leftrightarrow X_{i,j}=X_{i ,\underline{j}}. \label{defmarginalpartsbyobs}
\end{align}

\begin{ex}[\textcolor{black}{Tie structure of an mNTR prior}]
\label{exmntrmarginalties}
\textcolor{black}{Let us illustrate the role of the marginal tie structure in the case
$IJ=2\times2$. First, recall that each extremal sequence $\ubmx^{(k)}$
in (\ref{stochrepmNTR}) takes values in $\{\bmb_k,\binfty\}$, which implies that all its finite
values in the $i$-margin are equal to $b_{k,i}$. Moreover, the density
assumption (\ref{conddensmNTR}) on $\rho$ ensures that $b_{k,i}\neq b_{\underline{k},i}$
almost surely whenever $k\neq\underline{k}$. Thus, $X_{i,1}=X_{i,2}$
if and only if both observations stem from the same extremal sequence,
i.e.\ if $(i,1)$ and $(i,2)$ belong to the same subset of the hitting
scenario $\Theta$. Recalling (\ref{defpsioftheta}) and (\ref{defmarginalpartsbyobs}), we obtain
$$\Psi^{(i)}(\Theta)=\Psi^{(i)}(\bmXIJ),\qquad i=1,2,$$
almost surely. In other words, the partition within each margin can
be directly read from the data by grouping identical values.
However, this does not determine which subsets from different
margins stem from the same extremal sequence, since $b_{k,1}$ and
$b_{k,2}$ need not coincide.}
\end{ex}

From the stochastic representation (\ref{stochrepmNTR}) and by a similar reasoning as for deriving that $\lambda^{(IJ)}$ concentrates on $\mathcal{U}$ one can derive that 
$ \textcolor{black}{\beta_\theta}$ concentrates on $d_{\Psi(\theta)}:=\lc\sum_{i\in I(IJ)}L(\Psi^{(i)}(\theta))\rc$-dimensional \textcolor{black}{subsets} of $\textcolor{black}{(0,\infty)^{\vert IJ\vert}}$ given by
\begin{align*}
\textcolor{black}{\begin{aligned}
G^{\Psi(\theta)}:=\big\{\bmy_{IJ}\in(0,\infty)^{\vert IJ\vert}\ \big\vert\ 
&\text{for all }i\in I(IJ),\ (i,j),(i,\underline j)\in IJ:\\
&y_{i,j}=y_{i,\underline j}\quad\Longleftrightarrow\quad
\exists\,l:\ (i,j),(i,\underline j)\in\Psi_l^{(i)}(\theta)\big\}.
\end{aligned}}
\end{align*}

This can be seen as follows. Fixing the hitting scenario $\theta$ implies that all indices $(i,j)$ corresponding to the same subset $\theta_l$ must stem from the same extremal sequence $\ubmx^{(k)}$ whose only finite atom is $\bmb_{k}=(b_{k,i})_{\leqd}$. Thus, $X_{i,j}=X_{i,\underline{j}}=b_{k,i}$ for all $(i,j),(i,\underline{j})\in\theta_l$. By the same argument one obtains that $X_{i,j}\neq X_{i,\underline{j}}$ whenever $(i,j)$ and $(i,\underline{j})$ are in distinct subsets of $\theta$.
\textcolor{black}{The discussion above implies that $\beta_\theta$ is concentrated on $G^{\Psi(\theta)}$. Thus, whenever $\Psi(\theta)=\Psi(\tilde\theta)$, both $\beta_\theta$ and $\beta_{\tilde\theta}$ are concentrated on the same set.}
The following lemma formally proves these claims and also shows that $\beta_{\theta}$ is singular to $\beta_{\tilde\theta}$ whenever $\Psi(\theta)\neq\Psi(\tilde\theta)$ 

\begin{lem}
\label{lemsuppportbeta}
Let $\theta$ and $\tilde{\theta}$ denote two partitions of $IJ$. We have that $\beta_{\theta}$ is concentrated on $G^{\Psi(\theta)}$. Further, $\beta_{\theta}$ is singular to $\beta_{\tilde\theta}$ when $\Psi(\theta)\neq \Psi(\tilde\theta)$.
\end{lem}

 Somewhat conversely, we will now show that $\beta_{\theta}$ and $\beta_{\tilde\theta}$ are absolutely continuous w.r.t.\ a common dominating measure when $\Psi(\theta)=\Psi(\tilde{\theta})$.

\begin{prop}
\label{propdenshitscenmNTR}
Let $\theta=(\theta_l)_{1\leq l\leq L(\theta)}$ denote a partition of $IJ$ with corresponding induced partitions $\Psi(\theta)=\lc\Psi^{(i)}(\theta)\rc_{i\in I(IJ)}=\lc \lc \Psi^{(i)}_l(\theta)\rc_{1\leq l\leq L(\Psi^{(i)}(\theta))}\rc_{i\in I(IJ)}$.
Then, \textcolor{black}{for every Borel set $A\subset(0,\infty)^{\vert IJ\vert}$,} 
$$\beta_\theta(A)=\int_{A} q_\theta(\bmb_{IJ})\rmd \lc \lc b_i^{(k)} \rc_{1\leq k\leq L(\Psi^{(i)}(\theta))}\rc_{i\in I(IJ)}$$
where
        \begin{align}
            &q_\theta(\bmb_{IJ})\label{defqdensmntr}:=\exp\lc -\lambda^{(IJ)} \lc(\bmb_{IJ},\binfty]^\complement\rc \rc \\
            &\prod_{l=1}^{L(\theta)} \Bigg( \int_{(0,\infty)\times (0,\infty)^{\vert I(\theta_l)^\complement\vert}}     \lc\otimes_{j\in\N} \minid\lc a^{(l)}\delta_{(\bmb_{\theta_l},\bmb^{(l)})(I(\theta_l))}  \rc \rc \nonumber \\
            &\Big( \{ \bmx_{\theta_l}= \bmb_{\theta_l}, \bmx_{IJ\setminus\theta_l}>\bmb_{IJ\setminus\theta_l} \} \Big) \textcolor{black}{\rho\lc a^{(l)},(\bmb_{\theta_l},\bmb^{(l)})(I(\theta_l))\rc}  \rmd\lc a^{(l)},(b_i^{(l)})_{i\in I( \theta_l)^\complement} \rc \Bigg)  \nonumber
        \end{align}
with $ \bmb_{IJ}:=\lc  \sum_{k=1}^{L(\Psi^{(i)}(\theta))} b_i^{(k)}\id_{\{ (i,j)\in \Psi^{(i)}_k\}}\rc_{(i,j)\in IJ}$ and  $\bmb_{\theta_l}:=\bmb_{IJ\cap\theta_l}$.

\end{prop}

\textcolor{black}{Note that, even though $\beta_\theta$ and $\beta_{\tilde{\theta}}$ may have a common dominating measure when $\Psi(\theta)=\Psi(\tilde{\theta})$, their densities may vanish on different subsets of $G^{\Psi(\theta)}$, i.e.\ $\beta_\theta$ and $\beta_{\tilde{\theta}}$ need not be equivalent measures.}

It now follows from Lemma \ref{lemsuppportbeta} that 
$$  \frac{\rmd\beta_\theta}{\rmd \minid\lc\lambda^{(IJ)}\rc}(\bmXIJ)=0\ \minid\lc\lambda^{(IJ)}\rc\text{-almost-surely }  \forall\ \theta \text{ with }\Psi(\theta)\neq\Psi(\bmXIJ).$$
Moreover, as \textcolor{black}{all measures $\beta_\theta$ with $\Psi(\theta)=\Psi(\bmXIJ)$ are concentrated on $G^{\Psi(\bmXIJ)}$}, one can exploit the singularity again to obtain 
\begin{align}
    \frac{\rmd\beta_\theta}{\rmd\minid\lc\lambda^{(IJ)}\rc}(\bmX_{IJ})=\frac{d\beta_\theta}{d\lc\sum_{\underset{\Psi(\tilde\theta)=\Psi(\bmXIJ)}{\tilde\theta\in \mathcal{P}_{IJ}}}\beta_{\tilde\theta}\rc}(\bmX_{IJ})\id_{\{\Psi(\theta)=\Psi(\bmXIJ)\}}. 
\end{align}
\textcolor{black}{Proposition \ref{propdenshitscenmNTR} now implies that all $\beta_\theta$ with $\Psi(\theta)=\Psi(\bmXIJ)$ have the image of the $d_{\Psi(\bmXIJ)}$-dimensional Lebesgue measure on the positive orthant under the corresponding coordinate representation as their common dominating measure.} This allows to determine
$\rmd\beta_\theta /\rmd \minid\lc\lambda^{(IJ)}\rc$ explicitly.

\begin{cor}
\label{cordistcondhitmntr}
Let $IJ\subset d\times\N$ denote a finite set of indices and let $\mu$ denote a CRM without base measure whose intensity measure is given by (\ref{conddensmNTR}). Further, let $\ubmX$ denote the corresponding exchangeable min-id sequence such that $(\bmX_i)_{i\in\N}$ are conditionally i.i.d.\ with distribution $\minid(\mu)$. Then, the distribution of the conditional hitting scenario $\Tilde{\Theta}$ of $IJ$ of an mNTR prior with intensity (\ref{conddensmNTR}) is given by
\begin{align*}
  \tau(\bmXIJ,\theta)=\frac{q_{\theta} \lc \bmXIJ   \rc\id_{\{\Psi(\theta)=\Psi(\bmXIJ)\}} }{\sum_{\overset{\tilde\theta\in \mathcal{P}_{IJ}}{\Psi(\tilde\theta)=\Psi(\bmXIJ)}} q_{\tilde\theta} \lc \bmXIJ   \rc }  .
\end{align*}
\textcolor{black}{ $\bmXIJ$-almost everywhere.}
\end{cor}

\begin{ex}[\textcolor{black}{Conditional hitting scenarios with positive probability}]
\label{exmntrtwobytwo}

{\color{black} Continuing Example \ref{exmntrmarginalties}, assume that the realization of $\bmXIJ$ is given by $\bmXIJ=\lc(1.5,2.3)^\intercal,(3.7,2.3)^\intercal\rc$.
We have
\begin{align*}
\Psi^{(1)}(\bmXIJ)&=\big\{\{(1,1)\},\{(1,2)\}\big\},&
\Psi^{(2)}(\bmXIJ)&=\big\{\{(2,1),(2,2)\}\big\}.
\end{align*}
Hence, $(1,1)$ and $(1,2)$ must belong to distinct subsets of the
conditional hitting scenario, whereas $(2,1)$ and $(2,2)$ must
belong to the same subset. Among the $15$ partitions of $IJ$, only $\big\{\{(1,1)\},\{(1,2)\},\{(2,1),(2,2)\}\big\}, \ 
\big\{\{(1,1),(2,1),(2,2)\},\{(1,2)\}\big\}$ and $
\big\{\{(1,1)\},\{(1,2),(2,1),(2,2)\}\big\}
$
satisfy these restrictions. This corresponds to the concentration of $\beta_\theta$ on $G^{\Psi(\theta)}$
in Lemma \ref{lemsuppportbeta} and the restriction
$\Psi(\theta)=\Psi(\bmXIJ)$ in Corollary \ref{cordistcondhitmntr}.

One should note that the converse does not hold. For example, the partition $\big\{\{(1,1),(2,1),\linebreak (2,2)\},$ $\{(1,2)\}\big\}$ requires an extremal sequence with finite value $1.5$ and $2.3$ in the respective margins, since
it must generate $X_{2,1}=X_{2,2}=2.3$ as well as $X_{1,1}=1.5$.
This would imply $X_{1,2}\leq1.5$, contradicting $X_{1,2}=3.7$.
Thus, this partition also has conditional probability zero, which is in line with $q_\theta(\bmX_IJ)$ from (\ref{defqdensmntr}) vanishing for $\theta=\big\{\{(1,1),(2,1),(2,2)\},\{(1,2)\}\big\}$.
Therefore, the marginal tie structure $\Psi(\bmXIJ)$ allows us to exclude many conditional hitting scenarios with probability $0$, but does not imply that all remaining ones have positive probability. The additional restrictions encoded in the realization $\bmXIJ$ are accounted for by the densities $q_\theta$ in (\ref{defqdensmntr}), for which $q_\theta(\bmXIJ)$ may vanish even when $\Psi(\theta)=\Psi(\bmXIJ)$.
}
\end{ex}

In practice, one needs to compute $\vert\{\theta\in\mathcal{P}_{IJ}\mid \Psi(\theta)= \Psi(\bmXIJ)\}\vert$ many \textcolor{black}{non-negative} probabilities of conditional hitting scenarios, which is usually much smaller than the number of possible partitions $B_{\vert IJ\vert}$, where $B_m$ denotes the $m$-th Bell number. In the univariate case $d=1$, there is only one conditional hitting scenario with positive probability since $\theta=\Psi(\bmXIJ)$ is unique, which explains why there is no analog of the conditional hitting scenario appearing in the classical posterior representations for univariate NTR priors in \cite{fergusonprioronprob1974,doksumtailfree1974}.

Combining the results above one can deduce the posterior of an mNTR prior as stated in Theorem \ref{thmpostmNTR} at the beginning of this section.


\section{Additional examples of IDEM prior}
\label{appexamples}

\begin{ex}[Univariate mixture hazards \cite{camerlengilijoipruenster2021,Dykstralaud1981,loweng1989}, Example \ref{exhierdepmixhaz} continued]
\label{supp:exhierdepmixhaz}
Continuing Example \ref{exhierdepmixhaz} in the case $d=1$, we argue that there is actually one canonical choice of $\mathbb{Y}$ and $\kappa$ which includes all construction possibilities of \textcolor{black}{$\mu_1^{(\kappa_1)}$}. Recall that every univariate measure can be uniquely associated to a non-negative and non-decreasing cumulative distribution function. Thus, every univariate IDRM can be uniquely associated to a non-negative and non-decreasing infinitely divisible càdlàg process $\lc H(t)\rc_{t\geq 0}:=\lc\mu_1^{(\kappa_1)}([0,t])\rc_{t\geq 0}$. \textcolor{black}{For the model of \cite{camerlengilijoipruenster2021}, $H$ has no deterministic component}. In \cite[Proposition 2.16]{brueckmaischerer2023} it is shown that every \textcolor{black}{driftless} non-negative and non-decreasing infinitely divisible càdlàg process\textcolor{black}{ with $H(0)=0$} can be represented as $\lc\sum_{i\in\N} f_i(t)\rc_{t\geq 0}$, where $\lc f_i\rc_\iinn$ are the atoms of a PRM on the (Polish) space $D^0_{\geq 0,\nearrow}([0,\infty))$ of non-negative and non-decreasing c\`adl\`ag functions indexed by $[0,\infty)$ that satisfy $f(0)=0$, which is equipped with the Borel $\sigma$-algebra generated by the Skorohod topology. \textcolor{black}{Thus, there must be a non-negative and non-decreasing infinitely divisible càdlàg process $H(t)\sim \sum_{i\in\N} f_i(t)$, for which  $\mu_1^{(\kappa_1)}([0,t])\sim H(t)$. One can show that the absolute continuity of $\mu_1^{(\kappa_1)}$ implies that $N$ must be concentrated on $\mathbb{Y}^\prime:=D^0_{\geq 0,\nearrow}([0,\infty))\cap\{f\text{ absolutely continuous}\}$ and choosing $\kappa(s,f)=f^\prime(s)$ yields $\mu_1^{(\kappa_1)}([0,t])\sim H(t)\sim \int_{\mathbb{Y}^\prime}\int_0^t \kappa (s,f)\rmd s N(\rmd f)=\sum_{i\in\N} f_i(t)$. Thus, it suffices to use PRMs on $\mathbb{Y}^\prime$ with kernel $\kappa(s,f)=f^\prime(s)$ to reconstruct all univariate models of \cite{camerlengilijoipruenster2021}. \\
Finally, let us remark that also the models of \cite{Dykstralaud1981,loweng1989} are defined in terms of a non-negative and non-decreasing and infinitely divisible càdlàg process. Thus, these models can also be represented in terms of a PRM on $\mathbb{Y}^\prime$ with kernel $\kappa(s,f)=f^\prime(s)$ as illustrated above. }
\end{ex}

\begin{ex}[Dependent NTR priors \cite{epifanilijoi2010,rivapalacioleisen2018}]
\label{exdepNTRpriors}
Extending the framework of NTR priors to the multivariate setting, \cite{epifanilijoi2010,rivapalacioleisen2018} couple vectors of cumulative hazards $\lc \lc H_{i}(t)\rc_{t\geq 0}\rc_{\leqd}$ with the help of Lévy copulas, see e.g.\ \cite{Tankov2016} for a review on Lévy copulas. Essentially, they define marginal non-negative and non-decreasing driftless Lévy processes $\lc\lc\Tilde{H}_i(t)\rc_{\leqd}\rc_{ t\geq 0}$, couple them with a suitable Lévy copula to a multivariate Lévy process, and then apply a uniform deterministic time-change $\gamma(\cdot)$ to the multivariate Lévy process to obtain a multivariate additive process $\lc \lc H_{i}(t)\rc_{\leqd}\rc_{t\geq 0}:=\lc \lc \Tilde{H}_i(\gamma(t))\rc_{\leqd}\rc_{t\geq 0}$. Then, they define a prior on the random vector $\bmX\in (0,\infty)^d$ in terms of its survival functions $S(t_1,\ldots,t_d):=\prodd \exp\lc -H_i(t_i) \rc$.
Thus, $S$ defines a min-id survival function, with exponent measure $ \mu_{\mu_1^{(H_1)},\ldots,\mu_d^{(H_d)}}$, where $ \mu_i^{(H_i)}(-\infty,t]:=H_i(t)$. Moreover, it is easy to see that $\mu$ is an IDEM prior without base measure.

To obtain the Lévy characteristics of $\mu$, consider a PRM $N$ on $(0,\infty)\times(0,\infty)^d$ with intensity $\Tilde{\nu} \rmd(a_1,\ldots,a_d,b)=\gamma^\prime(b)\rho(a_1,\ldots,a_d)\rmd (a_1,\ldots,a_d,b)$, where $\rho$ is the multivariate Lévy density that corresponds to the choice of Lévy copula and marginal Lévy processes. Then, the Lévy measure of $\mu$ is given as the image measure of the map $(a_1,\ldots,a_d,b)\mapsto \mu_{a_1\delta_b,\ldots,a_d\delta_b} $ under $\Tilde{\nu}$. Thus, conditionally on $\mu$, each component of $\bmX$ has the same set of atoms, but the probabilities associated with these atoms are different. 

A similar argument to Example \ref{excondhitscenntr} shows that the conditional hitting scenarios of \cite{epifanilijoi2010,rivapalacioleisen2018} can also be simply read off from the realizations of $\ubmX_n$ by partitioning $d\times n$ by the ties in the margins of the data.
\end{ex}


\section{Proofs of the results from the supplementary material}
\label{appproofssupp}

\subsection*{Proof of Lemma \ref{lemderivrepdens}}
\begin{proof}
We start by proving the first statement.
    Note that for a deterministic discrete measure $\eta$ we have that $E_{\eta,\kappa}:\ x\mapsto \exp\lc - \int_0^x \int_0^\infty \kappa(s,y)\eta(\rmd y)\rmd s \rc$ is not differentiable everywhere. But as it contains a Lebesgue integral it is Lebesgue-almost everywhere differentiable with derivative given Lebesgue-almost everywhere by
    $$ -e_{\eta,\kappa}(x):=-\int_0^\infty\kappa(x,y)\eta(\rmd y)  \exp\lc - \int_0^{x}\int_0^\infty \kappa(s,y)\eta(\rmd y) \rmd s\rc.$$
    Unfortunately, in general this is not enough to apply Leibniz' rule to interchange differentiation and expectation and we have to argue more rigorously. First, observe that for every $\eta$ the function $E_{\eta,\kappa}$ is absolutely continuous as a composition of the Lipschitz function $\exp(-\cdot)\id_{\{\cdot\geq 0\}}$ and the absolutely continuous function $\int_0^x \int_0^\infty \kappa(s,y)\eta(\rmd y)\rmd s$. Thus, the only candidate for its density is $-e_{\eta,\kappa}(x)$ and therefore we have for all $x\in[0,\infty]$
    $$ 1 -\int_0^x e_{\eta,\kappa}(z) \rmd z=E_{\eta,\kappa}(x). $$
    W.l.o.g.\ assuming that $J_1=\{1,\ldots, \vert J_1\vert\}$ we can rewrite
    \begin{align*}
    &\lc \prod_{j=1}^{\vert J_1\vert} -\frac{\partial }{\partial x_j} \rc L_{\kappa,\gamma(\rmd a,\rmd b)}(\bmx_{J_1},\bmx_{J_2}) \\
    &=\lc \prod_{j=1}^{\vert J_1\vert}  -\frac{\partial }{\partial x_j} \rc \int_\Omega \prod_{j=1}^{\vert J_1\vert} \lc 1-\int_0^{x_j} e_{\eta,\kappa}(z_j) \rmd  z_j \rc \prod_{j\in J_2} E_{\eta,\kappa}(x_j) \crm_{\gamma(\rmd a,\rmd b)}(\rmd\eta) \\
    &=\lc \prod_{j=1}^{\vert J_1\vert} -\frac{\partial }{\partial x_j} \rc  \Bigg( \int_\Omega \prod_{j=2}^{\vert J_1\vert}  \lc 1-\int_0^{x_j}  e_{\eta,\kappa}(z_j) \rmd  z_j \rc\prod_{j\in J_2}  E_{\eta,\kappa}(x_j) \crm_{\gamma(\rmd a,\rmd b)}(\rmd\eta)\\
    &-\int_0^{x_1}  \int_\Omega e_{\eta,\kappa}(z_1)\prod_{j=2}^{\vert J_1\vert} \lc 1-\int_0^{x_j}  e_{\eta,\kappa}(z_j) \rmd  z_j \rc \prod_{j\in J_2} E_{\eta,\kappa}(x_j) \crm_{\gamma(\rmd a,\rmd b)}(\rmd\eta)\rmd z_1 \Bigg) , 
    \end{align*}
    where we could apply Fubini due to $e_{\eta,\kappa}\geq 0 $ and $ 1-\int_0^{x_j}  e_{\eta,\kappa}(z_j) \rmd  z_j \geq 0$. By the Lebesgue differentiation theorem the resulting function is now once Lebesgue-almost everywhere differentiable w.r.t.\ $x_1$ and we obtain
    \begin{align*}
      &\frac{-\partial}{\partial x_1} \int_0^{x_1}  \int_\Omega e_{\eta,\kappa}(z_1)\prod_{j=2}^{\vert J_1\vert} \lc 1-\int_0^{x_j}  e_{\eta,\kappa}(z_j) \rmd  z_j \rc \prod_{j\in J_2} E_{\eta,\kappa}(x_j) \crm_{\gamma(\rmd a,\rmd b)}(\rmd\eta)\rmd z_1  \\
      &=  \int_\Omega e_{\eta,\kappa}(x_1) \prod_{j=2}^{\vert J_1\vert}  \lc 1-\int_0^{x_j}  e_{\eta,\kappa}(z_j) \rmd  z_j \rc  \prod_{j\in J_2} E_{\eta,\kappa}(x_j) \crm_{\gamma(\rmd a,\rmd b)}(\rmd\eta) 
    \end{align*} 
    as the $x_1$-Lebesgue-almost everywhere representation of the derivative. Now, we argue iteratively. Fixing an $x_1$ such that the representation holds we repeat the argument to obtain
    \begin{align*}
        &\lc \prod_{j=1}^{\vert J_1\vert} -\frac{\partial }{\partial x_j} \rc L_{\kappa,\gamma(\rmd a,\rmd b)}(\bmx_{J_1},\bmx_{J_2}) =  \int_\Omega \prod_{j=1}^{\vert J_1\vert} e_{\eta,\kappa}(x_j)   \prod_{j\in J_2} E_{\eta,\kappa}(x_j) \crm_{\gamma(\rmd a,\rmd b)}(\rmd\eta)
    \end{align*}
    $(x_1,\ldots,x_{\vert J_1\vert+\vert J_2\vert})$-Lebesgue-almost everywhere, where we use that a Cartesian product of sets has Lebesgue measure $0$ as soon as one of the sets in the product has Lebesgue measure $0$.

    Let us turn to the second statement. Using the same arguments as above and using Proposition \ref{proptrafos}, we get
    \begin{align*}
        & \lc \prod_{j=1}^{\vert J_1\vert} \frac{\partial }{\partial x_j} \rc  -\log\lc L_{\kappa,  \gamma(\rmd a,\rmd b)} \lc \bmx_{J_1},\bmx_{J_2}\rc \rc \\
       &= \lc \prod_{j=1}^{\vert J_1\vert} \frac{\partial }{\partial x_j} \rc   \int_0^\infty \int_0^\infty \lc 1- \prod_{j=1}^{\vert J_1\vert} \lc 1 -\int_0^{x_j} e_{a\delta_b,\kappa}(z_j)\rmd z_j\rc   \prod_{j\in J_2} E_{a\delta_b,\kappa}(x_j)  \rc \gamma(\rmd a,\rmd b) \\
        &=\lc \prod_{j=1}^{\vert J_1\vert} \frac{\partial }{\partial x_j} \rc   \int_0^\infty \int_0^\infty \int_0^{x_1}  e_{a\delta_b,\kappa}(z_1)  \prod_{j=2}^{\vert J_1\vert} \lc 1-\int_0^{x_j} e_{a\delta_b,\kappa}(z_j)\rmd z_j \rc  \\
        &\hspace{1cm}\prod_{j\in J_2}E_{a\delta_b,\kappa}(x_j) \rmd z_1 \gamma(\rmd a,\rmd b) \\
        &+\lc \prod_{j=1}^{\vert J_1\vert} \frac{\partial }{\partial x_j} \rc \int_0^\infty \int_0^\infty \lc 1- \prod_{j=2}^{\vert J_1\vert} \lc 1-\int_0^{x_j} e_{a\delta_b,\kappa}(z_j)\rmd z_j\rc  \prod_{j\in J_2}E_{a\delta_b,\kappa}(x_j) \rc \gamma(\rmd a,\rmd b) ,
    \end{align*}
    where we could split the integral into two integrals because $\vert e_{a\delta_b,\kappa}(z)\vert \leq C_1\min\{a,C_2\}$ for all $z_1$ in a neighborhood of $x_1$ for some constants $C_1,C_2>0$ and $\lim_{a\to 0} (1-e_{a\delta_b,\kappa}(z)) =1$ by our assumptions on $\kappa$. This allows to apply Fubini and the Lebesgue differentiation theorem to obtain the partial derivative w.r.t.\ $x_1$ given by
    \begin{align*}
        \lc \prod_{j=2}^{\vert J_1\vert}\frac{\partial}{\partial x_{j}} \rc   \int_0^\infty \int_0^\infty  e_{a\delta_b,\kappa}(x_1)  \prod_{j=2}^{\vert J_1\vert} \lc 1-\int_0^{x_j} e_{a\delta_b,\kappa}(z_j)\rmd z_j \rc    \prod_{j\in J_2}E_{a\delta_b,\kappa}(x_j)   \gamma(\rmd a,\rmd b)
    \end{align*}
    $x_1$-Lebesgue-almost everywhere. Now, we can iterate the argument above to obtain that
    \begin{align}
        &\lc \prod_{j\in J_1}\frac{\partial}{\partial x_{j}} \rc  -\log\lc L_{\kappa,  \gamma(\rmd a,\rmd b)} \lc \bmx_{J_1},\bmx_{J_2}\rc \rc \nonumber\\
        &= (-1)^{\vert J_1\vert+1}    \int_0^\infty \int_0^\infty  \prod_{j=1}^{\vert J_1\vert} e_{a\delta_b,\kappa}(x_j)    \prod_{j\in J_2}E_{a\delta_b,\kappa}(x_j)   \gamma(\rmd a,\rmd b) \label{eqndiff}
    \end{align}
    $(x_1,\ldots,x_{\vert J-1\vert +\vert J_2\vert})$-Lebesgue-almost everywhere.

    To conclude, we want to mention that for $x\in[0,\infty]$ we can have $E_{\eta,\kappa}(x)=0$ with positive probability. In this case the integral vanishes as well as all its partial derivatives. Thus, we can ignore this case in the calculations above and only focus on those cases where $E_{\eta,\kappa}(x)>0$ almost surely. Moreover, the formulas are also correct when $\int_0^\infty \int_0^\infty \kappa(s,y)\rmd s  \eta(\rmd y)=\infty$ with positive probability or $\int_0^\infty \kappa(s,b)\rmd s=\infty$ on a set of non-zero Lebesgue measure.
\end{proof}

\subsection*{Proof of Lemma \ref{lemconddensextrseqpartexch}}
\begin{proof}
    Take $n$ large enough s.t. $IJ\subset d\times n$.
By Lemma \ref{lemconddistextrfctregcond} and \ref{lemsurvivmodel}, conditionally on the conditional hitting scenario $\Tilde{\Theta}$ and $\bmXIJ$, the survival function of $\lc\bmZ^{(n,l)}_{j}\rc_{1\leq  j\leq k}$ takes the form
\begin{align*}
    &P\lc  \bm Z^{(n,l)}_{j} >\bmx_{j} \text{ for all }1\leq j\leq k \mid \Tilde{\Theta},\bmXIJ\rc  \\
    &= \frac{  \int_{M_d^0} \underset{1\leq  j\leq k}{\otimes} \minid(\eta) \lc  \big \{\ubmy_k\mid \ubmy_k>\ubmx_{k} \big\}  \rc  h(\eta,\bmXIJ,\Tilde{\Theta}_l) \nu(\rmd\eta)  }{C(\bmXIJ,\Tilde{\Theta}_l) },  
\end{align*}
where 
$$h(\eta,\bmXIJ,\Tilde{\Theta}_l) :=  \int_{\{ \ubmy_{IJ\setminus \Tilde{\Theta}_l} >\ubmX_{IJ\setminus\Tilde{\Theta}_l}\} } \prod_{j=1}^n f_\eta\lc (\bmXIJ,\bmy_{IJ})(\Tilde{\Theta}_l)\rc \lc  \Lambda^{(\infty)}_{nd-\vert \Tilde{\Theta}_l\vert} \rc \lc\rmd \ubmy_{n\setminus \Tilde{\Theta}_l}\rc,$$
with $(\bmXIJ,\bmy_{IJ})(\Tilde{\Theta}_l):=\lc \id_{\{ (i,j)\in \Tilde{\Theta}_{l} \}}X_{i,j}+ y_{i,j}\id_{\{(i,j)\not \in \Tilde{\Theta}_{l}\}}\rc_{(i,j)\in IJ}$
and
$$ C(\bmXIJ,\Tilde{\Theta}_l):= \int_{M_d^0} h(\eta,\bmXIJ,\Tilde{\Theta}_l) \nu(\rmd\eta).$$

Note that $h(\eta,\bmXIJ,\Tilde{\Theta}_l)$ can take two forms given by
\begin{align*}
    &h(\otimes_{j=1}^{i-1} \delta_\infty \times (a\delta_b)^{(\kappa_i)}\otimes_{j=i+1}^d \delta_\infty,\bmXIJ,\Tilde{\Theta}_l) \\
    &= \id_{\{  \Tilde{\Theta}_l \subset i\times n\}} \prod_{(i,j)\in \Tilde{\Theta}_l} e_{a\delta_b,\kappa_i}(X_{i,j})   \prod_{j:(i,j)\in IJ\setminus\Tilde{\Theta}_l} E_{a\delta_b,\kappa_i}(X_{i,j})
\end{align*} 
 and
 \begin{align*}
 h(\mu_{\eta^{(\kappa_1)}_1,\ldots,\eta^{(\kappa_d)}_d},\bmXIJ,\Tilde{\Theta}_l) = \prod_{(i,j)\in \Tilde{\Theta}_{l}}    e_{\eta_i,\kappa_i} (X_{i,j}) \prod_{(i,j)\in IJ\setminus\Tilde{\Theta}_l}    E_{\eta_i,\kappa_i}(X_{i,j}). 
 \end{align*}  
 Thus, noticing that
\begin{align*}
     &\int_{M_d^0} \underset{1\leq  j\leq k}{\otimes} \minid(\eta) \lc  \big \{\ubmy_k\mid \ubmy_k>\ubmx_{k} \big\}  \rc  h(\eta,\bmXIJ,\Tilde{\Theta}_l) \nu(\rmd\eta) \\
&=  \int_{\{\ubmy_k\mid \ubmy_k>\ubmx_{k} \big\}} \int_{M_d^0} \prod_{j=1}^k f_\eta(\bmy_j) h(\eta,\bmXIJ,\Tilde{\Theta}_l) \nu(\rmd\eta) \Lambda^{(\infty)}_{kd}(\rmd \ubmy_k)
\end{align*}
and plugging in $f_\eta$ and $h(\eta,\bmXIJ,\Tilde{\Theta}_l) $ yields the result via a simple calculation of the integral.
\end{proof}

\subsection*{Proof of Lemma \ref{lemclosedformexamle}}
\begin{proof}
    Proposition \ref{proptrafos} implies
    \begin{align*}
     &L_{\kappa,\rho(a)\rmd a q(b)\rmd b}(\ubmx_m)\\
     &= \exp\lc - \int_0^\infty \int_0^\infty \lc  1-\exp\lc - a\sum_{j=1}^m \int_0^{x_j} \kappa(s,b)\rmd s \rc \rc \rho(a) q(b)\rmd a\rmd b \rc  \\
     &=\exp\lc - \sum_{j=1}^m\int_{x_{j-1}}^{x_j} \int_0^\infty \lc  1-\exp \lc - a\tau_1\sum_{k=j}^m (x_k-b) \rc \rc \rho(a) q(b)\rmd a\rmd b \rc \\
     &=\exp\lc - \sum_{j=1}^m\int_{x_{j-1}}^{x_j}  \frac{\lc \tau_1\sum_{k=j}^m (x_k-b) +1\rc^\sigma -1}{\sigma}  q(b)\rmd b \rc
\end{align*}
where we have used that for $\sigma\in(0,1)$ the following identity holds
$$ \int_0^\infty (1-\exp(-\lambda x))x^{-1-\sigma}\exp\lc-\beta x \rc\rmd x =\Gamma(1-\sigma)\frac{(\beta+\lambda)^\sigma -\beta^\sigma}{\sigma}. $$
\end{proof}

\subsection*{Proof of Lemma \ref{lempostminiduniv}}
\begin{lem}
\label{lempostminid}
Let $IJ\subset d\times\N$ and $\tilde{IJ}\subset d\times \N$ denote arbitrary index sets. Then, conditionally on $\bmX_{IJ}$, the exponent measure of $\bmY^{(n)}_{\Tilde{IJ}}$ is defined via 
\begin{align*}
 &\bar{\lambda}^{(IJ)}_n\lc  (\bmx_{\tilde{IJ}},\binfty]^\complement \rc\\
 &= \sumd \Bigg( \int_0^\infty  \frac{\lc \tau_1\sum_{j:(i,j)\in\Tilde{IJ}} (x_j-b)_+ +\tau_1\sum_{j:(i,j)\in IJ} (X_{i,j}-b)_+ +1 \rc^\sigma -1  }{\sigma} \alpha_0^{(\kappa_0)}\rmd b  \\
 & -\int_0^\infty  \frac{  \lc\tau_1\sum_{j:(i,j)\in IJ} (X_{i,j}-b)_+ +1\rc^\sigma-1}{\sigma} \alpha_0^{(\kappa_0)}\rmd b \Bigg) \\
 &+ \int_0^\infty \int_0^\infty \prodd \zeta_i\lc (x_{i,j})_{j:(i,j)\in \tilde{IJ}} ,(X_{i,j})_{j:(i,j)\in IJ} \rc   l(a_0,b_0) \rmd a_0\rmd b_0 ,
\end{align*}
where 
\begin{align*}
 &\zeta_i\lc (x_{i,j})_{j:(i,j)\in \tilde{IJ}} ,(X_{i,j})_{j:(i,j)\in IJ} \rc\\
 &= \id_{\{\tilde{IJ}\cap i\times \N\neq \emptyset\}} \bigg(  L_{\rho(a)\rmd a (a_0\delta_{b_0})^{(\kappa_0)}(\rmd b)}\lc (X_{i,j})_{j:(i,j)\in IJ} \rc \\
 &\hspace{2cm} -L_{\rho(a)\rmd a (a_0\delta_{b_0})^{(\kappa_0)}(\rmd b)}\lc (x_{i,j})_{j:(i,j)\in \tilde{IJ}} ,(X_{i,j})_{j:(i,j)\in IJ} \rc \bigg)\\
 &\ \ +\id_{\{\tilde{IJ}\cap i\times \N= \emptyset\}}   L_{\rho(a)\rmd a (a_0\delta_{b_0})^{(\kappa_0)}(\rmd b)}\lc (X_{i,j})_{j:(i,j)\in IJ} \rc 
\end{align*}
and we define $L_{\rho(a)\rmd a (a_0\delta_{b_0})^{(\kappa_0)}(\rmd b)}\lc \emptyset \rc:=1$.
\end{lem}

\begin{proof}[Proof of Lemma \ref{lempostminid}]
Define 
\begin{align*}
 &\zeta_i\lc (x_{i,j})_{j:(i,j)\in \tilde{IJ}} ,(X_{i,j})_{j:(i,j)\in IJ} \rc\\
 &:= \id_{\{\tilde{IJ}\cap i\times \N\neq \emptyset\}} \bigg(  L_{\rho(a)\rmd a (a_0\delta_{b_0})^{(\kappa_0)}(\rmd b)}\lc (X_{i,j})_{j:(i,j)\in IJ} \rc \\
 &\hspace{2cm} -L_{\rho(a)\rmd a (a_0\delta_{b_0})^{(\kappa_0)}(\rmd b)}\lc (x_{i,j})_{j:(i,j)\in \tilde{IJ}} ,(X_{i,j})_{j:(i,j)\in IJ} \rc \bigg)\\
 &\ \ +\id_{\{\tilde{IJ}\cap i\times \N= \emptyset\}}   L_{\rho(a)\rmd a (a_0\delta_{b_0})^{(\kappa_0)}(\rmd b)}\lc (X_{i,j})_{j:(i,j)\in IJ} \rc 
\end{align*}
and $L_{\rho(a)\rmd a (a_0\delta_{b_0})^{(\kappa_0)}(\rmd b)}\lc \emptyset \rc:=1$.
The exponent measure of $\bmY^{(n)}_{\Tilde{IJ}}$ is uniquely defined by expressions of the form 
\begin{align*}
 &\bar{\lambda}^{(IJ)}_n\lc  (\bmx_{\tilde{IJ}},\binfty]^\complement \rc\\
 &\overset{(\ref{defpostexpmeasure})}{=} \sumd \int_0^\infty \int_0^\infty \lc 1- \id_{\{\tilde{IJ}\cap i\times \N\neq \emptyset\}}\prod_{j:(i,j)\in \tilde{IJ}}E_{a\delta_b,\kappa}(x_{i,j})\rc \\
 &\hspace{1cm}\prod_{j:(i,j)\in IJ}E_{a\delta_b,\kappa}(X_{i,j})  \rho(a) \rmd a\alpha_0^{(\kappa_0)}\rmd b \\
 &+ \int_0^\infty \int_0^\infty \prod_{ 1\leq i\leq d} \e_{\eta_i}\lk \lc 1-\id_{\{\tilde{IJ}\cap i\times \N\neq \emptyset\}}\prod_{j:(i,j)\in \tilde{IJ}}E_{\eta_i,\kappa}(x_{i,j})\rc \prod_{j:(i,j)\in IJ}E_{\eta_i,\kappa}(X_{i,j})  \rk\\
 &\hspace{1cm}l(a_0,b_0) \rmd a_0\rmd b_0 \\
 &=\sumd \int_0^\infty \int_0^\infty \lc 1-\id_{\{\tilde{IJ}\cap i\times \N\neq \emptyset\}} \prod_{j:(i,j)\in \tilde{IJ}}E_{a\delta_b,\kappa}(x_{i,j})\rc \\
 &\exp\lc -a\tau_1\sum_{j:(i,j)\in IJ}(X_{i,j}-b)_+\rc a^{-1-\sigma}\exp(-a)\Gamma(1-\sigma)^{-1} \rmd a\alpha_0^{(\kappa_0)}\rmd b \\
 &+ \int_0^\infty \int_0^\infty \prodd \zeta_i \lc (x_{i,j})_{j:j\in \tilde{IJ}} ,(X_{i,j})_{j:j\in IJ} \rc  l(a_0,b_0) \rmd a_0\rmd b_0\\
 &\overset{\star}{=}  \sumd  \int_0^\infty  \frac{\lc \tau_1\sum_{j:j\in \Tilde{IJ}} (x_j-b)_+ +\tau_1\sum_{j:(i,j)\in IJ} (X_{i,j}-b)_+ +1 \rc^\sigma -1  }{\sigma} \alpha_0^{(\kappa_0)}\rmd b  \\
 & -\int_0^\infty  \frac{  \lc\tau_1\sum_{j:(i,j)\in IJ} (X_{i,j}-b)_+ +1\rc^\sigma-1}{\sigma} \alpha_0^{(\kappa_0)}\rmd b  \\
 &+ \int_0^\infty \int_0^\infty \prodd \zeta_i \lc (x_{i,j})_{j:j\in \tilde{IJ}} ,(X_{i,j})_{j:j\in IJ}  \rc l(a_0,b_0) \rmd a_0\rmd b_0,
\end{align*}
where in $\star$ we have again used that
$$ \int_0^\infty (1-\exp(-\lambda x))x^{-1-\sigma}\exp\lc-\beta x \rc\rmd x =\Gamma(1-\sigma)\frac{(\beta+\lambda)^\sigma -\beta^\sigma}{\sigma}. $$

\end{proof}

\begin{proof}[Proof of Lemma \ref{lempostminiduniv}]
Follows directly from Lemma \ref{lempostminid} above.

\end{proof}

\subsection*{Proof of Proposition \ref{propassociation}}

\begin{proof}
    Since max-id random vectors are associated by \cite[Proposition 5.29]{resnickextreme2013} and association is preserved under multiplication with $-1$ we immediately obtain that the min-id random vector $\ubmX_n$ is associated, conditionally on $\mu$ and unconditionally.
    
    To prove the remaining claim, recall that according to \cite[Proposition 5.30]{resnickextreme2013} PRMs $N$ on a complete and separable metric space $\mathbb{Y}$ are associated, which is equivalent to the fact that for every collection of continuous functionals with compact support $\lc g_i\rc_{1\leq i\leq m}$ we have $\lc \int_{\mathbb{Y}} g_1(y) N(\rmd y),\ldots,\int_{\mathbb{Y}} g_m(y) N(\rmd y)\rc$ is an associated random vector. First, we reduce the problem to finite random measures on compact sets, since this will allow us to later invoke simpler arguments about continuous functionals on finite measures on compact sets. Thus, we w.l.o.g.\ consider finite IDRMs on $E_K:=E_d\cap K$, where $K$ denotes the union of the compact supports of the $\lc f_i\rc_{1\leq i\leq m}$. Note that $E_K$ is a complete and separable compact metric space and that $\mu(\cdot\cap K)$ can be viewed as a finite IDRM on $E_K$. Therefore, it is enough to focus on finite random measures on $K$, denoted as $M_d(E_K):=\{\eta(\cdot\cap K) \mid \eta\in M_d \}$, which is again a separable and complete metric space when equipped with the vague topology. Moreover, as $\eta\mapsto \int_{E_K} f_i(\bmx)\eta(\rmd x)$ is a continuous functional w.r.t.\ the vague topology on $M_d(E_K)$, see \cite[p. 111]{kallenberg2017}, and we have that
    \begin{align*}
        \int_{E_K} f_i(\bmx)\mu(\rmd\bmx)&=\int_{E_K} f_i(\bmx)\alpha(\rmd x)+\int_{M_d(E_K)}\int_{E_K} f_i(\bmx)\eta(\rmd x)N(\rmd\eta) \\
        &=\tilde{f}_i(\alpha)+ \int_{M_d(E_K)}\Tilde{f}_i(\eta)N(\rmd\eta),
    \end{align*}
    where $\tilde{f}_i(\eta):=\int_{E_K} f_i(\bmx)\eta(\rmd\bmx)$ and $N$ denotes the PRM on $M_d(E_K)$ such that $\mu(\cdot \cap K) =\alpha(\cdot\cap K)+\int_{M^0_d} \eta(\cdot \cap E_K) N(\rmd\eta)$.
    Thus, we can write $\int_{E_K} f_i(\bmx)\mu(\rmd\bmx)$ as an integral of the PRM $N$ over a non-negative continuous functional $\Tilde{f}_i(\eta)=\int_{E_K} f_i(\bmx)\eta(\rmd x)$ on $M_d(E_K)$ and we can ignore the influence of $\tilde{f}_i(\alpha)$ when considering association, since association of a random vector is preserved under translation.   
    
    The remaining issue is that $\Tilde{f}_i$ does not necessarily have compact support on $M_d(E_K)$, which is why we have to introduce another approximation. Since $K$ is compact the set $M_a:=\{ \mu \mid \mu\lc K\rc< a\}$ is an open set in the vague topology on $M_d(E_K)$, see \cite[Exercise 3.4.11]{resnickextreme2013}, and we have
    $$ \int_{M_d(E_K)}\Tilde{f}_i(\eta) N(\rmd\eta)=\lim_{a\to\infty}  \int_{M_d(E_K)}\Tilde{f}_i(\eta) \id_{\{ \eta\in M_a\}}N(\rmd\eta)$$
    by monotone convergence, using that $\mu(K)$ is finite and that $N$ is $\sigma$-finite. Moreover, $M_a$ is relatively compact w.r.t.\ the vague topology on $M_d(E_K)$, since every continuous function $h$ on $K$ assumes a finite maximum $m_h$ on the compact set $K$. Thus $\sup_{\eta\in M_a} h(\eta)<am_h$ for every continuous function $h$ in $K$ and by \cite[Proposition 3.16]{resnickextreme2013} $M_a$ is relatively compact. Moreover, \cite[Proposition 3.11]{resnickextreme2013} shows that $\id_{\{ \eta\in M_a\}}$ can be approximated from below by a sequence of monotone increasing continuous functions $h_{n,a}\in[0,1]$ such that $\lim_{n\to\infty} h_{n,a}(\eta)=\id_{\{\eta\in M_a\}}$. Thus, by monotone convergence, we have  
    $$ \int_{M_d(E_K)}\Tilde{f}_i(\eta) N(\rmd\eta)=\lim_{a\to\infty} \lim_{n\to\infty}  \int_{M_d(E_K)}\Tilde{f}_i(\eta) h_{n,a}(\eta) N(\rmd\eta).$$
    Now, the function $\Tilde{f}_i(\eta) h_{n,a}(\eta)$ is continuous with support contained in the relatively compact set $M_a$, thus has compact support on $M_d(E_K)$. Therefore, by \cite[Proposition 5.30 iii]{resnickextreme2013}, 
    $$ G_{n,a}:=\lc \tilde{f}_i(\alpha)+\int_{M_d(E_K)} \Tilde{f}_i(\eta)h_{n,a}(\eta) N(\rmd\eta) \rc_{1\leq i\leq m}  $$
    is associated for every $(n,a)$. Since association is preserved under weak convergence by \cite[Lemma 5.32 iii]{resnickextreme2013}, we obtain 
    $$ G:=\lc \int_{E^\prime_d}f_i(\bmx)\mu(\rmd\bmx)\rc_{1\leq i\leq m}=\lc \tilde{f}_i(\alpha)+\int_{M_d(E_K)} \Tilde{f}_i(\eta) N(\rmd\eta) \rc_{1\leq i\leq m}  $$
    is associated if we can show that $G$ can be expressed as a weak limit of a sequence of associated random vectors. First, observe that 
    $$\lim_{n\to\infty} G_{n,a}=\lc  \tilde{f}_i(\alpha)+\int_{M_d(E_K)} \Tilde{f}_i(\eta)\id_{\{\eta\in M_a\}} N(\rmd\eta) \rc_{1\leq i\leq m} =:G_{\infty,a}.$$
    Thus, an application of the dominated convergence theorem implies that the limit of the Laplace transform of $\lim_{n\to\infty} G_{n,a}$ is the Laplace transform of $G_{\infty,a}$. Thus, $G_{\infty,a}$ is the weak limit of $G_{n,a}$ and $G_{\infty,a}$ is associated. Next, since $G=\lim_{a\to\infty}G_{\infty,a}$, another application of the dominated convergence theorem implies that the Laplace transform of $G$ is the limit of the Laplace transform of $\lim_{a\to\infty} G_{\infty,a}$. Thus, $G$ is the weak limit of $G_{\infty,a}$ and $G$ is associated, which proves the claim.
\end{proof}

\subsection*{Proof of Proposition \ref{propdistrofmeanfct}}

\begin{proof}
   Recall that every min-id random variable $X\sim\minid(\mu)$ has the stochastic representation $ X\sim \min_{i\in\N} x_i$ for some PRM $N=\sum_{i\in\N} \delta_{x_i}$ with intensity $\mu$. Thus, $I(f,\mu)$ has the stochastic representation 
    \begin{align*}
        I(f,\mu)&=\e_{X\sim \minid(\mu)}\lk f(X)\rk \sim \e_{N\sim PRM(\mu)}\lk f\lc \min_{i\in\N} x_i\rc\rk\\
        &=\e_{N\sim PRM(\mu)}\lk \min_{i\in\N} f\lc x_i\rc\rk=\e_{N\sim PRM(\mu_f)}\lk \min_{i\in\N}  x_i\rk \\
        &\sim \e_{X\sim \minid(\mu_f)}\lk X\rk = I( \text{id},\mu_f),   
    \end{align*} 
    where we have used that $f$ is monotone increasing and $\mu_f(\cdot)=\mu \lc \{  x\in(-\infty,\infty] \mid f(x)\in \cdot \}\rc$ denotes the IDEM obtained from the image measure of $\mu$ under the transformation $f$. Further, due to Proposition \ref{proptrafos} the Lévy characteristics of $\mu_f$ are $(\alpha_f,\nu_f)$.
    
    Next, let $\eta$ denote an exponent measure on $[0,\infty]$. Then,
   \begin{align*}
            I(\text{id},\eta)&=\int_0^\infty x \minid(\eta)(\rmd x) =\e_{ X\sim\minid(\eta)}\lk  X \rk =\int_0^\infty \exp\Big( -\eta \big( (-\infty, t] \big) \Big) \rmd t,  
        \end{align*}
        where we have used that $\e \lk X\rk=\int_0^\infty 1-P\lc X\leq t \rc \rmd t$ for $X\in[0,\infty]$. Therefore, letting $m\in\N$ we get 
        \begin{align*}
            \e\lk  I(f,\mu)^m \rk&=\e\lk \lc I\lc\text{id},\mu_f\rc \rc^m\rk=\e\bigg[ \bigg(  \int_0^\infty \exp\Big( -\mu_f \big( (-\infty, t] \big) \Big) \rmd t \bigg)^m\bigg] \\
            &=  \e\lk \int_0^\infty \ldots \int_0^\infty  \exp\lc-\sum_{i=1}^m \mu_f \big( (-\infty, t_i]\big)\rc  \rmd t_1\ldots\rmd t_m \rk\\
            &=\int_0^\infty \ldots \int_0^\infty \e\lk \exp\lc -\sum_{i=1}^m \mu_f \big( (-\infty, t_i]\big)\rc \rk  \rmd t_1\ldots\rmd t_m \\
            &=\int_0^\infty \ldots \int_0^\infty L_{\alpha_f,\nu_f}  \big(\bm 1,(t_1,\ldots,t_m)\big)\rmd t_1\ldots\rmd t_m.
        \end{align*}    
\end{proof}

\subsection*{Proof of Proposition \ref{propconddistrmNTR}}
\begin{proof}
    Note that, even though for the posterior distribution we only need to consider $\bmx_{IJ}\in\R^{\vert IJ\vert}$, we must allow for $\bmx_{IJ}\in E^\prime_{\vert IJ \vert}$ to verify that $K_{IJ}(\bmx_{IJ},\cdot)$ is the claimed conditional distribution of $\lambda^{(IJ)}$. Denote $J(IJ)=\{j\mid (i,j)\in IJ\text{ for some }\leqd\}$ and for some $J\subset J(IJ)$ denote $J^\complement=J(IJ)\setminus J$. Define the set
    \begin{align*}
       H(J)&:= \{ \bmy_{IJ}\in E_{\vert IJ\vert}^\prime \mid \forall \leqd :  \infty>y_{i,j}=y_{i,\underline{j}} \text{ for all } (i,j),(i,\underline{j})\in IJ \text{ with } j,\underline{j}\in J \\
    &\hspace{1cm}\text{ and }y_{i,j}=\infty \ \text{ for all } (i,j)\in IJ \text{ with } j\in J^\complement  \} ,  
    \end{align*} 
    which encodes that $\bmy_{IJ\cap (d\times J)}$ is finite with identical values in its $i$-margins and that $\bmy_{IJ\cap (d\times J^\complement)}=\binfty$.
    It is easy to see from (\ref{expmeasuremNTR}) that the exponent measure $\lambda^{(IJ)}$ concentrates on $\cup_{J\subset J(IJ)} H(J)$ and that $ H(J)\cap H(J^\prime)=\emptyset$ whenever $J\neq J^\prime$. 

    We need to verify (\ref{defconddistexpmeasure}) for the claimed representation of $K_{IJ}(\bmxij,\cdot)$. 
    Let $\tilde{\ubmx}:=\ubmx_{\setminus IJ}$ and consider an arbitrary $\mathsf{G}:(-\infty,\infty]^{\vert IJ\vert}\times(-\infty,\infty]^{d\times\N\setminus  IJ} \to [0,\infty)$  that vanishes on $\{\binfty\}\times (-\infty,\infty]^{d\times\N\setminus  IJ} $. Using that $J=\emptyset$ is not relevant since $\mathsf{G}\big(\binfty, \tilde{\ubmx}\big)=0$ we observe
    \begin{align*}
        &\int_{E_{\vert IJ\vert}^\prime}\int_{(-\infty,\infty]^{d\times \N\setminus IJ}} \mathsf{G}(\bmx_{IJ},\Tilde{\ubmx}) K_{IJ}(\bmx_{IJ},\rmd \Tilde{\ubmx}) \lambda^{(IJ)}(\rmd \bmx_{IJ}) \\
        &=\sum_{\emptyset\neq J\subset J(IJ)} \int_{H(J)}  \int_{(-\infty,\infty]^{d\times \N\setminus IJ}} \mathsf{G}(\bmx_{IJ},\Tilde{\ubmx}) K_{IJ}(\bmx_{IJ},\rmd \Tilde{\ubmx}) \lambda^{(IJ)}(\rmd \bmx_{IJ})\\
        &=\sum_{\emptyset\neq J\subset J(IJ)} \int_{(0,\infty)\times \textcolor{black}{(0,\infty)^d}} \int_{H(J)}  \int_{(-\infty,\infty]^{d\times \N\setminus IJ}} \mathsf{G}(\bmx_{IJ},\Tilde{\ubmx}) K_{IJ}(\bmx_{IJ},\rmd \Tilde{\ubmx})\\
        &\hspace{1cm}\lc\otimes_{j\in \N}\minid\lc 
        a\delta_{\bmb} \rc\rc^{(IJ)} (\rmd \bmx_{IJ})  \rho(a,\bmb)\rmd (a,\bmb)\\
        &=\sum_{\emptyset\neq J\subset J(IJ)} \int_{(0,\infty)\times \textcolor{black}{(0,\infty)^d}} \int_{\{ \bmy_{IJ}=\bmb_{IJ}(J)\}}  \int_{(-\infty,\infty]^{d\times \N\setminus IJ}} \mathsf{G}(\bmb_{IJ}(J),\Tilde{\ubmx}) K_{IJ}(\bmb_{IJ}(J),\rmd \Tilde{\ubmx}) \\
        &\hspace{1cm}\lc\otimes_{j\in \N}\minid\lc 
        a\delta_{\bmb} \rc\rc^{(IJ)} ( \rmd \bmx_{IJ})  \rho(a,\bmb)\rmd (a,\bmb)\\
        &=  \sum_{\emptyset\neq J\subset J(IJ)} \int_{(0,\infty)\times \textcolor{black}{(0,\infty)^d}}  \int_{(-\infty,\infty]^{d\times \N\setminus IJ}} \mathsf{G}(\bmb_{IJ}(J),\Tilde{\ubmx}) K_{IJ}(\bmb_{IJ}(J),\rmd \Tilde{\ubmx}) \\
        &\lc\otimes_{j\in \N}\minid\lc  a\delta_{\bmb} \rc\rc^{(IJ)} \lc \{ \bmy_{IJ}=\bmb_{IJ}(J)\}\rc  \rho(a,\bmb)\rmd (a,\bmb) =:\#_1
        \end{align*}
        where $\bmb_{IJ}(J):=\lc b_i\id_{\{ j\in J\}}+ \infty \id_{\{j\in J^\complement\}}\rc_{(i,j)\in IJ}$ and 
        \begin{align}
            \lc\otimes_{j\in \N}\minid\lc  a\delta_{\bmb} \rc\rc^{(IJ)} \label{defminidmarginmeasure}
        \end{align}
        denotes the distribution of the $IJ$-margin of an i.i.d.\ sequence with law $\otimes_{j\in \N}\minid\lc a\delta_{\bmb} \rc$. One should note that those $b_i$ such that 
        $$i\not\in I(IJ,J):=\{i\in\{1,\ldots,d\} \mid (i,j)\in IJ\text{ for some }j\in J\}$$ 
        are integrated out because they do not appear in the function that is integrated. We thus obtain
        \begin{align*}
        &\#_1=\sum_{\emptyset\neq J\subset J(IJ)} \int_{ \textcolor{black}{(0,\infty)^{\vert I(IJ,J)\vert}}} \   \int_{(-\infty,\infty]^{d\times \N\setminus IJ}} \mathsf{G}(\bmb_{IJ}(J),\Tilde{\ubmx}) K_{IJ}(\bmb_{IJ}(J),\rmd \Tilde{\ubmx}) \\
        & \int_0^\infty \times_{i\in I(IJ,J)^\complement } \int_{\textcolor{black}{0}}^\infty    \lc\otimes_{j\in \N}\minid\lc  a\delta_{\bmb} \rc\rc^{(IJ)} \lc \{ \bmy_{IJ}=\bmb_{IJ}(J)\}\rc   \rho(a,\bmb) \rmd (a,\lc b_i\rc_{i\in I(IJ,J)^\complement}) \\
        & \hspace{1cm}  \rmd \lc \lc b_i\rc_{i\in I(IJ,J)}\rc =: \#_2
        \end{align*}
        It remains to observe that 
        $$  \int_0^\infty \times_{i\in I(IJ,J)^\complement } \int_{\textcolor{black}{0}}^\infty   \lc\otimes_{j\in \N}\minid\lc  a\delta_{\bmb} \rc\rc^{(IJ)} \lc \{ \bmy_{IJ}=\bmb_{IJ}(J)\}\rc   \rho(a,\bmb) \rmd (a,\lc b_i\rc_{i\in I(IJ,J)^\complement}) $$
        is exactly equal to the denominator appearing in  $K_{IJ}(\bmb_{IJ}(J),\rmd \Tilde{\ubmx})$ since $I(\bmb_{IJ}(J))=\{i\mid (i,j)\in IJ\text{ for some } j\in J\}=I(IJ,J)$. Thus, the denominator cancels and we get
        \begin{align*}
        &\#_2=\\
        &\sum_{\emptyset\neq J\subset J(IJ)} \int_{ \textcolor{black}{(0,\infty)^{\vert I(IJ,J)\vert}}} \   \int_{(-\infty,\infty]^{d\times \N\setminus IJ}} \mathsf{G}(\bmb_{IJ}(J),\Tilde{\ubmx})  \\
        & \lc \times_{i\in I(\bmb_{IJ}(J))^\complement }\int_0^\infty\rc \int_0^\infty \otimes_{j\in\N} \minid\lc a\delta_{(\bmb_{IJ}(J),\bmb)(I(\bmb_{IJ}(J))}\rc \lc \{ \tilde\ubmy\in \rmd\tilde\ubmx , \bmy_{IJ}=\bmb_{IJ}(J)\}\rc \\
        &\rho(a,(\bmb_{IJ}(J),\bmb)(I(\bmb_{IJ}(J))) \rmd a \rmd \lc (b_i)_{i\in I(\bmb_{IJ}(J))^\complement} \rc\rmd \lc \lc b_i\rc_{i\in  I(\bmb_{IJ}(J))}\rc \\
        &=:\#_3,
        \end{align*}
        where $\minid\lc a\delta_{(\bmb_{IJ}(J),\bmb)(I(\bmb_{IJ}(J))}\rc \lc \{ \tilde\ubmy\in \rmd\tilde\ubmx , \bmy_{IJ}=\bmb_{IJ}(J)\}\rc$ denotes that $\tilde{\ubmx}$ is integrated w.r.t.\ the (non-probability) measure $\otimes_{j\in\N} \minid\lc a\delta_{(\bmb_{IJ}(J),\bmb)(I(\bmb_{IJ}(J))}\rc \linebreak \lc \{\ubmy \mid \tilde{\ubmy}\in \cdot\ , \bmy_{IJ}=\bmb_{IJ}(J)\}\rc$
        Now, we also note that we simply have $(\bmb_{IJ}(J),\bmb)(I(\bmb_{IJ}(J))=\bmb$, which gives 
        \begin{align*}
        \#_3=&\ \sum_{\emptyset\neq J\subset J(IJ)} \int_{ (0,\infty)\times \textcolor{black}{(0,\infty)^d}} \   \int_{(-\infty,\infty]^{d\times \N\setminus IJ}} \mathsf{G}(\bmb_{IJ}(J),\Tilde{\ubmx})   \otimes_{j\in\N} \minid\lc a\delta_{\bmb}\rc \\
        &\hspace{1cm}\lc \{  \tilde\ubmy\in \rmd\tilde{\ubmx} , \bmy_{IJ}=\bmb_{IJ}(J)\}\rc \rho(a,\bmb) \rmd a \rmd \bmb \\
        &=   \int_{ (0,\infty)\times \textcolor{black}{(0,\infty)^d}} \   \int_{(-\infty,\infty]^{d\times \N\setminus IJ}}   \sum_{\emptyset\neq J\subset J(IJ)}  \mathsf{G}(\bmb_{IJ}(J),\Tilde{\ubmx})    \otimes_{j\in\N} \minid\lc a\delta_{\bmb}\rc \\
        &\hspace{1cm}\lc \{ \tilde\ubmy\in \rmd\tilde{\ubmx} , \bmy_{IJ}=\bmb_{IJ}(J)\}\rc \rho(a,\bmb) \rmd a \rmd \bmb \\
        &\overset{\star}{=}   \int_{ (0,\infty)\times \textcolor{black}{(0,\infty)^d}} \   \int_{(-\infty,\infty]^{d\times \N}}    \mathsf{G}(\bmx_{IJ},\Tilde{\ubmx})  \otimes_{j\in\N} \minid\lc a\delta_{\bmb}\rc \lc\rmd\ubmx\rc \rho(a,\bmb) \rmd a \rmd \bmb \\
        &=   \int_{(-\infty,\infty]^{d\times \N}}     \mathsf{G}(\bmx_{IJ},\Tilde{\ubmx})  \lambda(\rmd\ubmx) 
        \end{align*}
        which proves the claim, where in $\star$ we used that $\mathsf{G}\big(\binfty, \tilde{\ubmx}\big)=0$ and that
        \begin{align*}
            &\sum_{\emptyset\neq J\subset J(IJ)} \lc\otimes_{j\in\N}\minid(a\delta_{\bmb})\rc\lc A\cap \{\bmy_{IJ}=\bmb_{IJ}(J)\}\rc\\
            &=\lc\otimes_{j\in\N}\minid(a\delta_{\bmb})\rc\lc A\setminus\{ \bmy_{IJ}=\binfty\}\rc.
        \end{align*} 
        
\end{proof}

\subsection*{Proof of Lemma \ref{lemsuppportbeta}}
\begin{proof}
The central observation for the proof of the claims is that for a Poisson point process $\tilde{N}=\sum_{k\in\N} \delta_{(a_k,\bmb_k)}$ whose intensity is given by (\ref{conddensmNTR}) we have 
$$P\lc \exists\ C\in(0,\infty)\text{ and }\leqd\text{ such that } \tilde{N}\lc(0,\infty)\times \{\bmb\in (0,\infty)^d\mid b_i=C\}\rc>1\rc=0.$$
To see this note that $\tilde{N}^{(i)}:=\sum_{k\in\N} \delta_{b_{k,i}}$ is a Poisson random measure on $(0,\infty)$ with diffuse intensity $$\lc\int_0^\infty \int_{(0,\infty)^{d-1}}\rho(a,\bmb)\rmd a\rmd (b_j)_{1\leq j\leq d,j\neq i}\rc\rmd b_i.$$
Now, it is well known that
\begin{align*}
  0&= P\lc \exists\ C\in(0,\infty)\text{ such that } \tilde{N}^{(i)}\lc\{C\}\rc>1\rc\\
  &= P\lc \exists\ C\in(0,\infty)\text{ such that } \tilde{N}\lc(0,\infty)\times \{\bmb\in (0,\infty)^d\mid b_i=C\}\rc>1\rc
\end{align*}
since $\tilde{N}^{(i)}$ is a simple point process on $(0,\infty)$ as it has a diffuse intensity measure, see \cite[Lemma 3.6]{kallenberg2017}. Thus,  each atom of $N=\sum_{k\in\N} a_k\delta_{\bmb_k}$ has components $(b_{k,i})_{\leqd}$ that are almost surely distinct from all  components of the atoms $(b_{\underline{k},i})_{\leqd}$ whenever $k\neq\underline{k}$.

The claim that $\beta_{\theta}$ is \textcolor{black}{concentrated} on $G^{\Psi(\theta)}$ now immediately follows as the atoms $(a_k,\bmb_k)$ corresponding to the exchangeable sequences $\ubmx^{(k)}$ in (\ref{stochrepmNTR}) correspond to almost surely distinct atoms with distinct components.

To prove the singularity, assume that $\beta_{\theta}(G^\Psi)>0$ for some collection of partitions $\Psi\neq \Psi(\theta)$. Then, there exist $i$, $j\neq \underline{j}$ such that either $(i,j),(i,\underline{j})\in \Psi^{(i)}_k(\theta)$ with $\beta_\theta\lc  y_{i,j}\neq y_{i,\underline{j}}\rc>0$ or $(i,j)\in \Psi^{(i)}_k(\theta)$ and $(i,\underline{j})\in \Psi^{(i)}_{\underline{k}}(\theta)$ for some $k\neq \underline{k}$ with $\beta_\theta\lc y_{i,j}=y_{i,\underline{j}}\rc>0$. In the first case there would exist a sequence $\ubmx^{(k)}$ in (\ref{stochrepmNTR}) such that $\ubmx^{(k)}$ is supported on at least two distinct finite atoms, which has probability $0$. In the second case there would exist $\ubmx^{(k)}$ and $\ubmx^{(\underline{k})}$ for which $b_{k,i}=b_{\underline{k},i}$ for $k\neq \underline{k}$, which has probability $0$. Therefore, singularity of $\beta_{\theta}$ and $\beta_{\tilde\theta}$ follows whenever $\Psi(\theta)\neq\Psi(\tilde\theta)$.
\end{proof}

\subsection*{Proof of Proposition \ref{propdenshitscenmNTR}}
\begin{proof}

 The general idea for the proof is to plug-in Proposition \ref{propconddistrmNTR} into (\ref{defbetameasures}) to verify that $\beta_\theta$ has the claimed density. First, note that $\beta_\theta(\rmd\bmx_{IJ})$ is concentrated in $\textcolor{black}{(0,\infty)^{\vert IJ\vert}}$ and thus we only need to consider $\bmx_{\theta_l}\in\textcolor{black}{(0,\infty)^{\vert\theta_l\vert}}$ when applying Proposition \ref{propconddistrmNTR}, which ensures that $I(\bmx_{\theta_l})=I(\theta_l)$.  Recall that the mass of $\beta_\theta$ is determined via its mass on $G^{\Psi(\theta)}$, which is a $d_{\Psi(\theta)}$-dimensional \textcolor{black}{subset} of $\textcolor{black}{(0,\infty)^{\vert IJ\vert}}$. Then, for $\lc c_{i,k}\rc_{1\leq k\leq L(\Psi^{(i)}),i\in I(IJ)}\in\R^{d_\Psi}$ and recalling the measures $\lc\otimes_{j\in \N}\minid\lc  a\delta_{\bmb} \rc\rc^{(IJ)}$ from (\ref{defminidmarginmeasure}), (\ref{defbetameasures}) implies that $\beta_\theta$ is uniquely determined by expressions of the form 
    \begin{align*}
        &\beta_{\theta}\lc \big\{ \bmx_{IJ}\in G^{\Psi(\theta)}\mid  x_{i,j}\leq c_{i,k}\ \forall (i,j)\in\Psi_k^{(i)}, 1\leq k\leq L(\Psi^{(i)}), i\in I(IJ) \big\}\rc\\
        &=\lc \times_{l=1}^{L(\theta)} \int_{(0,\infty)\times(0,\infty)^d}\int_{(0,\infty)^{\vert \theta_l\vert}} \rc \id_{\{x_{i,j}\leq c_{i,k}\ \forall (i,j)\in\Psi_k^{(i)}, 1\leq k\leq L(\Psi^{(i)}), i\in I(IJ) \}}\\
        &\hspace{1cm}\exp\lc -\lambda^{(IJ)}\lc(\bmxij,\binfty]^\complement\rc \rc \lc \prod_{l=1}^{L(\theta)}K_{\theta_l}\lc \bmx_{\theta_l},\{ \bmy_{IJ\setminus \theta_l}>\bmx_{IJ\setminus \theta_l} \}\rc \rc \\
        &\hspace{1cm}\lc \otimes_{l=1}^{L(\theta)}\lc\otimes_{j\in\N} \minid\lc a^{(l)}\delta_{\bmb^{(l)}}  \rc \rc^{(\theta_l)} \lc \rmd\bmx_{\theta_l} \rc \rho(a^{(l)},\bmb^{(l)})\rmd (a^{(l)},\bmb^{(l)}) \rc \\
         &=  \lc \times_{l=1}^{L(\theta)} \int_{(0,\infty)\times\{\bmb^{(l)}\textcolor{black}{\in(0,\infty)^d}\mid b_i^{(l)}\leq  c_{i,k}\ \forall\ (i,j)\in \Psi^{(i)}_k\subset\theta_l ,1\leq k\leq L(\Psi^{(i)}),i\in I(IJ)\} } \rc \\
         &\hspace{1cm}\lc  \times_{l=1}^{L(\theta)} \int_{\{ \bmx_{\theta_l} \mid x_{i,j}= b^{(l)}_{i}\ \forall (i,j)\in \Psi^{(i)}_k\subset\theta_l ,1\leq k\leq L(\Psi^{(i)}),i\in I(IJ)\}} \rc \\
        &\hspace{1cm}\exp\lc -\lambda^{(IJ)} \lc(\bmxij,\binfty]^\complement\rc \rc    \lc \prod_{l=1}^{L(\theta)} K_{\theta_l}\lc \bmx_{\theta_l},\{ \bmy_{IJ\setminus \theta_l}>\bmx_{IJ\setminus \theta_l} \} \rc \rc \\
        &\hspace{1cm} \lc \otimes_{l=1}^{L(\theta)}\lc\otimes_{j\in\N} \minid\lc a^{(l)}\delta_{\bmb^{(l)}}  \rc \rc^{(\theta_l)}    \lc \rmd\bmx_{\theta_l} \rc\rc \lc \otimes_{l=1}^{L(\theta)}\rho(a^{(l)}, \bmb^{(l)})\rmd (a^{(l)},\bmb^{(l)}) \rc=:\#_1  \\ 
        \end{align*}
        where we have used that $\bmx_{\theta_l}\in \textcolor{black}{(0,\infty)^{\vert\theta_l\vert}}$ and thus its $i$-margins consist of the corresponding values from $\bmb^{(l)}$, which enforces that $b^{(l)}_i\leq c_{i,k}$ for all $i,k$ s.t.\ there exists a $j$ with $(i,j)\in \Psi^{(i)}_k\subset\theta_l$. We further exploit this tie structure to plug-in the corresponding values of $\bmb^{(l)}$ into $\bmx_{IJ\setminus \theta_l}$ and $\bmx_{\theta_l}$ to obtain
        \begin{align*}
            &\#_1= \\
            &\lc \times_{l=1}^{L(\theta)} \int_{(0,\infty)\times\{\bmb^{(l)}\textcolor{black}{\in(0,\infty)^d}\mid b_i^{(l)}\leq  c_{i,k}\ \forall\ (i,j)\in \Psi^{(i)}_k\subset\theta_l ,1\leq k\leq L(\Psi^{(i)}),i\in I(IJ) \} } \rc \\
            &\hspace{1cm}\lc  \times_{l=1}^{L(\theta)} \int_{\{ \bmx_{\theta_l} \mid x_{i,j}= b^{(l)}_{i}\ \forall (i,j)\in \Psi^{(i)}_k\subset\theta_l ,1\leq k\leq L(\Psi^{(i)}),i\in I(IJ) \}} \rc \\
        &\hspace{1cm}\exp\lc -\lambda^{(IJ)} \lc(\tilde\bmb_{IJ},\binfty]^\complement\rc \rc    \lc \prod_{l=1}^{L(\theta)} K_{\theta_l}\lc \tilde\bmb_{\theta_l},\{ \bmy_{IJ\setminus \theta_l}>\tilde\bmb_{IJ\setminus \theta_l} \} \rc \rc \\
        &\hspace{1cm} \lc \otimes_{l=1}^{L(\theta)}\lc\otimes_{j\in\N} \minid\lc a^{(l)}\delta_{\bmb^{(l)}}  \rc \rc^{(\theta_l)} \rc   \lc \rmd\bmx_{\theta_l} \rc \lc \otimes_{l=1}^{L(\theta)}\rho(a^{(l)}, \bmb^{(l)})\rmd (a^{(l)},\bmb^{(l)}) \rc\\
        &=:\#_2  
        \end{align*}
         where $\tilde\bmb_{\theta_l}=\lc b^{(l)}_i\rc_{(i,j)\in\theta_l}$, $\tilde\bmb_{IJ\setminus\theta_l}=\lc \sum_{ k=1}^{L(\theta)} b_i^{(k)}\id_{\{ (i,j)\in \theta_{k}\}} \rc_{(i,j)\in IJ\setminus\theta_l}$ and 
         $\tilde\bmb_{IJ}:=$ \\
         $\lc \sum_{k=1}^{L(\theta)} b_i^{(k)}\id_{\{ (i,j)\in \theta_k\}}\rc_{(i,j)\in IJ}$. Noticing that the integrand only depends on \linebreak 
         $\lc \lc b_i^{(l)}\rc_{ i\in I(\theta_l)}\rc_{ 1\leq l\leq L(\theta)}$
         we may rearrange the terms to integrate out the remaining \linebreak $\lc(b_i^{(l)})_{i\in I(\theta_l)^\complement}\rc_{1\leq l \leq L(\theta)}$ to obtain

         \begin{align*}
             &\#_2=  \lc\times_{l=1}^{L(\theta)}  \int_{\{(b_i^{(l)})_{i\in I( \theta_l)}\mid \textcolor{black}{0<}b_i^{(l)}\leq  c_{i,k}\ \forall\ (i,j)\in \Psi^{(i)}_k\subset\theta_l ,1\leq k\leq L(\Psi^{(i)}),i\in I(IJ) \} }  \rc   \\
             &\hspace{1cm}\exp\lc -\lambda^{(IJ)} \lc(\tilde\bmb_{IJ},\binfty]^\complement\rc \rc \prod_{l=1}^{L(\theta)}\Bigg(  K_{\theta_l}\lc \tilde\bmb_{\theta_l},\{ \bmy_{IJ\setminus \theta_l}>\tilde\bmb_{IJ\setminus \theta_l} \}  \rc  \\
             &\hspace{1cm}\Bigg(  \int_{(0,\infty)\times (0,\infty)^{\vert I(\theta_l)^\complement\vert}} \int_{\{ \bmx_{\theta_l} \mid x_{i,j}= b^{(l)}_{i}\ \forall (i,j)\in \Psi^{(i)}_k\subset\theta_l ,1\leq k\leq L(\Psi^{(i)}),i\in I(IJ) \}}      \\
                &\hspace{1cm} \lc\otimes_{j\in\N} \minid\lc a^{(l)}\delta_{\bmb^{(l)}}  \rc \rc^{(\theta_l)}    \lc \rmd\bmx_{\theta_l} \rc  \rho(a^{(l)}, \bmb^{(l)})  \rmd\lc a^{(l)},(b_i^{(l)})_{i\in I( \theta_l)^\complement} \rc \Bigg) \Bigg)  \\
                &\hspace{1cm}\lc\otimes_{l=1}^{L(\theta)} \rmd ((b_i^{(l)})_{i\in I( \theta_l)}) \rc =:\#_3   
         \end{align*}
        The key step now is to observe that the denominator $C(\tilde{\bmb}_{\theta_l})$ of $K_{\theta_l}\lc \tilde\bmb_{\theta_l},\{ \bmy_{IJ\setminus \theta_l}>\tilde\bmb_{IJ\setminus \theta_l} \}\rc$ is equal to
        \begin{align*}
            &\int_{(0,\infty)\times (0,\infty)^{\vert I(\theta_l)^\complement\vert}} \int_{\{ \bmx_{\theta_l} \mid x_{i,j}= b^{(l)}_{i}\ \forall (i,j)\in \Psi^{(i)}_k\subset\theta_l ,1\leq k\leq L(\Psi^{(i)}),i\in I(IJ) \}}  \\
            &
        \lc \otimes_{j\in\N} \minid\lc a^{(l)}\delta_{\bmb^{(l)}}  \rc\rc^{(\theta_l)} \lc \rmd \bmx_{\theta_l}\rc \rho(a^{(l)}, \bmb^{(l)}) \rmd\lc a^{(l)}, \lc b_i^{(l)}\rc_{i\in I(\theta_l)^\complement}\rc  ,
        \end{align*}   
        since, in the notation of Proposition \ref{propconddistrmNTR}, $I(\tilde\bmb_{\theta_l})=I(\theta_l)$ and $(\tilde\bmb_{\theta_l},\bmb^{(l)})(I(\tilde\bmb_{\theta_l}))=\bmb^{(l)}$  and
        \begin{align*}
           & \int_{\{ \bmx_{\theta_l} \mid x_{i,j}= b^{(l)}_{i}\ \forall (i,j)\in \Psi^{(i)}_k\subset\theta_l ,1\leq k\leq L(\Psi^{(i)}),\leqd \}} 
        \lc \otimes_{j\in\N} \minid\lc a^{(l)}\delta_{\bmb^{(l)}}  \rc\rc^{(\theta_l)} \lc \rmd \bmx_{\theta_l}\rc \\
        &=     \lc \otimes_{j\in\N} \minid\lc a^{(l)}\delta_{\bmb^{(l)}}  \rc\rc\lc \{ \ubmx\mid \bmx_{\theta_l}= \bmb_{\theta_l}\}\rc,
        \end{align*}
        which leads to the cancellation of $C(\tilde{\bmb}_{\theta_l})$ in $\#_3$. Thus, only the numerator of \\
        $K_{\theta_l}\lc \tilde\bmb_{\theta_l},\{ \bmy_{IJ\setminus \theta_l}>\tilde\bmb_{IJ\setminus \theta_l} \}\rc$ remains and we have
        \begin{align*}
            &\#_3=\lc \times_{l=1}^{L(\theta)} \int_{\{(b_i^{(l)})_{i\in I( \theta_l)}\mid \textcolor{black}{0<}b_i^{(l)}\leq  c_{i,k}\ \forall\ (i,j)\in \Psi^{(i)}_k\subset\theta_l ,1\leq k\leq L(\Psi^{(i)}),i\in I(IJ) \} } \rc  \\
             &\hspace{1cm}\exp\lc -\lambda^{(IJ)} \lc(\tilde\bmb_{IJ},\binfty]^\complement\rc \rc \\
             &\hspace{1cm}\prod_{l=1}^{L(\theta)} \Bigg( \int_{(0,\infty)\times (0,\infty)^{\vert I(\theta_l)^\complement\vert}}     \lc\otimes_{j\in\N} \minid\lc a^{(l)}\delta_{(\tilde\bmb_{\theta_l},\bmb^{(l)})(I(\theta_l))}  \rc \rc \\
             &\hspace{1cm}\lc \{ \ubmy \mid \bmy_{\theta_l}= \tilde\bmb_{\theta_l}, \bmy_{IJ\setminus\theta_l}>\tilde\bmb_{IJ\setminus\theta_l} \} \rc  \\
             &\hspace{1cm}\rho(a^{(l)}, \bmb^{(l)})  \rmd\lc a^{(l)},(b_i^{(l)})_{i\in I( \theta_l)^\complement} \rc \Bigg) \lc\otimes_{l=1}^{L(\theta)}  \rmd ((b_i^{(l)})_{i\in I( \theta_l)}) \rc\\
             &=\lc \times_{i=1}^{d} \times_{k=1}^{L(\Psi^{(i)}(\theta))} \int_{\textcolor{black}{(0,\infty)\cap(-\infty,c_{i,k}]}} \rc  q_\theta\lc\bmb_{IJ}\rc \lc\otimes_{i=1}^d\otimes_{k=1}^{L(\Psi^{(i)}(\theta))} \rmd b_i^{(k)} \rc
        \end{align*}
        where  we defined
        \begin{align*}
            &q_\theta(\bmb_{IJ}):=\exp\lc -\lambda^{(IJ)} \lc(\bmb_{IJ},\binfty]^\complement\rc \rc \\
            &\prod_{l=1}^{L(\theta)} \Bigg( \int_{(0,\infty)\times (0,\infty)^{\vert I(\theta_l)^\complement\vert}}     \lc\otimes_{j\in\N} \minid\lc a^{(l)}\delta_{(\textcolor{black}{\bmb_{\theta_l}},\bmb^{(l)})(I(\theta_l))}  \rc \rc \\
            &\Big( \{ \ubmx \mid \bmx_{\theta_l}= \bmb_{\theta_l}, \bmx_{IJ\setminus\theta_l}>\bmb_{IJ\setminus\theta_l} \} \Big) \rho\lc a^{(l)}, (\textcolor{black}{\bmb_{\theta_l}},\bmb^{(l)})(I(\theta_l))\rc  \rmd\lc a^{(l)},(b_i^{(l)})_{i\in I( \theta_l)^\complement} \rc \Bigg) 
        \end{align*}
        and $\bmb_{IJ}:=\lc  \sum_{k=1}^{L(\Psi^{(i)}(\theta))} b_i^{(k)}\id_{\{ (i,j)\in \Psi^{(i)}_k\}}\rc_{(i,j)\in IJ}$, $\bmb_{\theta_l}:=\lc  \sum_{k=1}^{L(\Psi^{(i)}(\theta))} b_i^{(k)}\id_{\{ (i,j)\in \Psi^{(i)}_k\}}\rc_{(i,j)\in \theta_l}$ and $\bmb_{IJ\setminus\theta_l}:=\lc  \sum_{k=1}^{L(\Psi^{(i)}(\theta))} b_i^{(k)}\id_{\{ (i,j)\in \Psi^{(i)}_k\}}\rc_{(i,j)\in IJ\setminus\theta_l}$ are unique representations of the arguments of $q_\theta$ which only depend on $\Psi(\theta)$, implicitly using $\sum_{l=1}^{L(\theta)} \vert I(\theta_l)\vert=\sum_{i\in I(IJ)} L(\Psi^{(i)})$.
    Therefore, $\beta_\theta$ has density $q_\theta$ which has the claimed representation as a $d_\Psi$-dimensional density w.r.t.\ \textcolor{black}{the} Lebesgue measure, which proves the claim.

\end{proof}

\subsection*{Proof of Corollary \ref{cordistcondhitmntr}}
\begin{proof}
    Simply note that for any measurable $A\subset G^{\Psi(\theta)}$ we have
    \begin{align*}
       \beta_\theta(A)&=\int_{A}  q_\theta(\bmb_{IJ})\rmd\lc \lc b_i^{(k)} \rc_{1\leq k\leq L(\Psi^{(i)}(\theta))}\rc_{i\in I(IJ)} \\
       &=\int_{A}  \frac{q_\theta(\bmb_{IJ}) }{\sum_{\overset{\tilde\theta\in \mathcal{P}_{IJ}}{\Psi(\theta)=\Psi(\tilde{\theta})}} q_{\tilde\theta}  (\bmb_{IJ})} \\
        &\hspace{1cm}\Bigg( \sum_{\overset{\tilde\theta\in \mathcal{P}_{IJ}}{\Psi(\theta)=\Psi(\tilde\theta)}} q_{\tilde\theta} (\bmb_{IJ}) \Bigg)\rmd\lc \lc b_i^{(k)} \rc_{1\leq k\leq L(\Psi^{(i)}(\theta))}\rc_{i\in I(IJ)}\\
       &=\int_{A}  \frac{q_\theta(\bmb_{IJ}) }{\sum_{\overset{\tilde\theta\in \mathcal{P}_{IJ}}{\Psi(\theta)=\Psi(\tilde{\theta})}} q_{\tilde\theta}  (\bmb_{IJ})}  \lc \sum_{\overset{\tilde\theta\in \mathcal{P}_{IJ}}{\Psi(\theta)=\Psi(\tilde\theta)}} \beta_{\tilde\theta}(\rmd\bmb_{IJ}) \rc\\
        &=\int_{A}  \frac{q_\theta(\bmb_{IJ}) }{\sum_{\overset{\tilde\theta\in \mathcal{P}_{IJ}}{\Psi(\tilde{\theta})}=\Psi(\bmb_{IJ})} q_{\tilde\theta}  (\bmb_{IJ})} \minid\lc\lambda^{(IJ)}\rc\lc \rmd \bmb_{IJ}\rc,
    \end{align*}
    since $\beta_{\tilde\theta}(A)=0$ for all $\tilde{\theta}$ with $\Psi(\tilde\theta)\neq\Psi(\theta)$ by Lemma \ref{lemsuppportbeta}. Therefore
    $$ \frac{\rmd \beta_{\theta}}{\rmd \minid\lc \lambda^{(IJ)}\rc}(\bmb_{IJ}) = \frac{q_\theta(\bmb_{IJ}) }{\sum_{\overset{\tilde\theta\in \mathcal{P}_{IJ}}{\Psi(\theta)=\Psi(\tilde{\theta})}} q_{\tilde\theta}  (\bmb_{IJ})}   . $$
    Moreover, when $\Psi(\bmXIJ)\neq \Psi(\theta)$ we have $\bmXIJ\not\in G^{\Psi(\theta)}$ $ \minid\lc \lambda^{(IJ)}\rc$-almost-surely and therefore
    $$ \frac{\rmd \beta_{\theta}}{\rmd \minid\lc \lambda^{(IJ)}\rc}(\bmXIJ)=0 \ \minid\lc \lambda^{(IJ)}\rc\text{-almost-surely for all } \bmXIJ\not\in G^{\Psi(\theta)}. $$
    Thus, 
    $$ \frac{\rmd \beta_{\theta}}{\rmd \minid\lc \lambda^{(IJ)}\rc}(\bmXIJ)= \frac{\rmd \beta_{\theta}}{\rmd \minid\lc \lambda^{(IJ)}\rc}(\bmXIJ)\id_{\{\Psi(\theta)=\Psi(\bmXIJ)\}} $$
    $ \minid\lc \lambda^{(IJ)}\rc$-almost-surely, which proves the claim.
\end{proof}

\subsection*{Proof of Theorem \ref{thmpostmNTR}}
\begin{proof}
    Let $n$ be such that $IJ\subset d\times n$. According to Theorem \ref{thmmainresult} the posterior Lévy measure $\bar{\nu}_{IJ}$ of $\bar{\mu}_{IJ}$ is given by
    \begin{align*}
        &\otimes_{j=1}^n \minid(\eta) \lc\{ \ubmy_n\mid \bmy_{IJ}>\bmX_{IJ}\}\rc \nu(\rmd\eta)\\
        &=\otimes_{j=1}^n \minid(\eta) \lc\{ \ubmy_n\mid \bmy_{IJ}>\bmX_{IJ}\}\rc \rho(a,\bmb)\rmd(a,\bmb)\\
        &=\rho(a,\bmb)\exp\lc-a\sum_{j\in J(IJ)} \id_{\{ b_i\leq X_{i,j}\text{ for some } (i,j)\in IJ\}}\rc\rmd a\rmd \bmb
    \end{align*}
    which is the Lévy measure of a CRM with the claimed intensity. The distribution of the conditional hitting scenario is given in Corollary \ref{cordistcondhitmntr}.

    Next, conditionally on the conditional hitting scenario $\Tilde{\Theta}$, the distribution of $\ubmZ^{(n,l)}$ from (\ref{defcondextrseq}) is given by
    \begin{align*}
        &P(\ubmZ^{(n,l)}\in A)=\frac{ K_{\Tilde{\Theta}_l}\lc \bmX_{\Tilde{\Theta}_l},\Big\{ \ubmy\mid \lc\bmy_{n+j}\rc_{j\in\N}\in A ;\ \bmy_{IJ\setminus\tilde{\Theta}_l}>\bmX_{IJ\setminus\Tilde{\Theta}_l}\Big\}\rc}{K_{\Tilde{\Theta}_l}\lc \bmX_{\Tilde{\Theta}_l},\{ \bmy_{IJ\setminus\tilde{\Theta}_l}>\bmX_{IJ\setminus\Tilde{\Theta}_l}\}\rc}
    \end{align*}
    which by Proposition \ref{propconddistrmNTR} is equal to
    \begin{align*}
        &C(\Tilde{\Theta}_l,\bmXIJ)^{-1}\lc \times_{i=1}^{\vert I(\Tilde{\Theta}_l)^\complement\vert }\int_0^\infty\rc \int_0^\infty \otimes_{j\in\N} \minid\lc a\delta_{(\bmX_{\Tilde{\Theta}_l},\bmb)(I(\Tilde{\Theta}_l))}\rc \\
        &\lc \Big\{ \ubmy\mid \lc\bmy_{n+j}\rc_{j\in\N}\in A;\  \bmy_{IJ\setminus\tilde{\Theta}_l}>\bmX_{IJ\setminus\Tilde{\Theta}_l} ;\ \bmy_{\Tilde{\Theta}_l}=\bmX_{\Tilde{\Theta}_l}\Big\}\rc \\
        &\rho(a,(\bmX_{\Tilde{\Theta}_l},\bmb)(I(\Tilde{\Theta}_l))) \rmd a \rmd \lc (b_i)_{i\in I(\Tilde{\Theta}_l)^\complement} \rc \\
        &=  C(\Tilde{\Theta}_l,\bmXIJ)^{-1}\lc \times_{i=1}^{\vert I(\bmX_{\Tilde{\Theta}_l})^\complement\vert }\int_0^\infty\rc \int_0^\infty \otimes_{j\in\N} \minid\lc a\delta_{(\bmX_{\Tilde{\Theta}_l},\bmb)(I(\Tilde{\Theta}_l))}\rc\\
        &\lc \Big\{ \ubmy\mid \lc\bmy_{n+j}\rc_{j\in\N}\in A \Big\}\rc
        \otimes_{j=1}^n \lc\minid(\textcolor{black}{a\delta_{(\bmX_{\Tilde{\Theta}_l},\bmb)(I(\Tilde{\Theta}_l))}})\rc \\
         &\hspace{1cm}\lc \Big\{  \bmy_{IJ\setminus\tilde{\Theta}_l}>\bmX_{IJ\setminus\Tilde{\Theta}_l} ;\ \bmy_{\Tilde{\Theta}_l}=\bmX_{\Tilde{\Theta}_l}\Big\}\rc \\
        &\rho(a,(\bmX_{\Tilde{\Theta}_l},\bmb)(I(\Tilde{\Theta}_l))) \rmd a \rmd \lc (b_i)_{i\in I(\Tilde{\Theta}_l)^\complement} \rc =:\#_1
    \end{align*}
    where the normalizing constant is given by
    \begin{align*}
        &C(\Tilde{\Theta}_l,\bmXIJ)\\
        &:=\lc \times_{i=1}^{\vert I(\Tilde{\Theta}_l)^\complement\vert }\int_0^\infty\rc \int_0^\infty \otimes_{j=1}^n \lc\minid\lc a\delta_{(\bmX_{\Tilde{\Theta}_l},\bmb)(I(\Tilde{\Theta}_l))}\rc \rc  \\
        &\lc \Big\{  \bmy_{IJ\setminus\tilde{\Theta}_l}>\bmX_{IJ\setminus\Tilde{\Theta}_l};\ \bmy_{\Tilde{\Theta}_l}=\bmX_{\Tilde{\Theta}_l} \Big\} \rc
        \rho(a,(\bmX_{\Tilde{\Theta}_l},\bmb)(I(\Tilde{\Theta}_l))) \rmd a \rmd \lc (b_i)_{i\in I(\Tilde{\Theta}_l)^\complement} \rc\\
        &=\lc \times_{i=1}^{\vert I(\Tilde{\Theta}_l)^\complement\vert }\int_0^\infty\rc \int_0^\infty \textcolor{black}{\rho_{IJ}^{(l)}}\lc a,\lc b_i\rc_{i\in I(\bmX_{\theta_l})^\complement},\bmX_{IJ}\rc \\
        &\hspace{1cm} \rho(a,(\bmX_{\Tilde{\Theta}_l},\bmb)(I(\Tilde{\Theta}_l))) \rmd a \rmd \lc (b_i)_{i\in I(\Tilde{\Theta}_l)^\complement} \rc
    \end{align*}
     with
    \begin{align*}
        &\textcolor{black}{\rho_{IJ}^{(l)}}\lc a,\lc b_i\rc_{i\in I(\bmX_{\theta_l})^\complement},\bmX_{IJ}\rc\\
        &:=\otimes_{j=1}^n \lc\minid\lc a\delta_{(\bmX_{\Tilde{\Theta}_l},\bmb)(I(\Tilde{\Theta}_l))}\rc \rc  \lc \big\{  \bmy_{IJ\setminus\tilde{\Theta}_l}>\bmX_{IJ\setminus\Tilde{\Theta}_l};\ \bmy_{\Tilde{\Theta}_l}=\bmX_{\Tilde{\Theta}_l} \big\} \rc .
    \end{align*}
    This allows to rewrite
    \begin{align*}
        \#_1=&  C(\Tilde{\Theta}_l,\bmXIJ)^{-1}\lc \times_{i=1}^{\vert I(\bmX_{\Tilde{\Theta}_l})^\complement\vert }\int_0^\infty\rc \int_0^\infty \otimes_{j\in\N} \minid\lc a\delta_{(\bmX_{\Tilde{\Theta}_l},\bmb)(I(\bmX_{\Tilde{\Theta}_l}))}\rc \lc  A \rc\\
        &\textcolor{black}{\rho_{IJ}^{(l)}}\lc a,\lc b_i\rc_{i\in I(\bmX_{\theta_l})^\complement},\bmXIJ\rc \\
        &\rho(a,(\bmX_{\Tilde{\Theta}_l},\bmb)(I(\Tilde{\Theta}_l))) \rmd a \rmd \lc (b_i)_{i\in I(\Tilde{\Theta}_l)^\complement} \rc.
    \end{align*}
    Thus, $\ubmZ^{(n,l)}$ can be generated as follows:
    \begin{enumerate}
        \item Draw $\lc \textcolor{black}{A^{(l)}},\lc (\textcolor{black}{B_i^{(l)}})_{i\in I(\bmX_{\Tilde{\Theta}_l})^\complement}\rc\rc$ according to the density
        $$C(\Tilde{\Theta}_l,\bmXIJ)^{-1}\textcolor{black}{\rho_{IJ}^{(l)}}\lc a,\lc b_i\rc_{i\in I(\Tilde\Theta_l)^\complement},\bmX_{IJ}\rc \rho(a,(\bmX_{\Tilde{\Theta}_l},\bmb)(I(\Tilde{\Theta}_l))) \rmd a \rmd \lc (b_i)_{i\in I(\Tilde{\Theta}_l)^\complement} \rc. $$

        \item Conditionally on $\lc \textcolor{black}{A^{(l)}},\lc (\textcolor{black}{B_i^{(l)}})_{i\in I(\bmX_{\Tilde{\Theta}_l})^\complement}\rc\rc$, define the random measures \linebreak$\bar\mu_{IJ}^{(l)}:=\textcolor{black}{A^{(l)}}\delta_{(\bmX_{\Tilde\Theta_l},\textcolor{black}{\bm B^{(l)}})(I(\Tilde{\Theta}_l))}$ and draw an i.i.d.\ sequence $\ubmZ^{(n,l)}$ with $\bmZ^{(n,l)}_j\linebreak\sim\minid\lc \bar{\mu}_{IJ}^{(l)} \rc$.
    \end{enumerate}
    Therefore, the posterior of $\mu$ is given by 
    $$ \mu\mid\bmXIJ \sim \bar{\mu}_{IJ}+\sum_{l=1}^{L(\Tilde{\Theta})} \bar{\mu}_{IJ}^{(l)} $$
    as claimed.
\end{proof}


\section{Proofs of results from the main text}
\label{appproofs}

\subsection*{Proof of Corollary \ref{coridexpmisidbyidexpm}}
\begin{proof}
    Infinite divisibility of $\mu^{(i,n)}$ simply follows by identifying the Laplace transform of the $\mu^{(i,n)}$ to have Lévy-Khintchine characteristics $(\alpha/n,\nu/n)$. Since $\mu$ is finite on every localizing set $U_i$, the $\mu^{(i,n)}$ must as well be almost surely finite on every localizing set. Moreover, since the $\mu^{(i,n)}$ are i.i.d., they must satisfy $\mu^{(i,n)}(\binfty)=\binfty$ almost surely. Thus, they are random exponent measures almost surely. Identifying each $\mu^{(i,n)}$ with a version that is an exponent measure for every realization, we obtain the first claim. The second claim simply follows by observing that each measure in $M_d^0$ is the exponent measure of a random vector with $\binfty$ as upper bound of its support, which shows that (\ref{prmrepidrm}) is simply a sum of exponent measures with $\binfty$ as common upper bound of the support of the underlying random vectors.
\end{proof}
\subsection*{Proof of Theorem \ref{thmidexpmimpliesexminid}}

\begin{proof}
   To prove that $\ubmX$ is an exchangeable min-id sequence it is sufficient to prove that its $p$-dimensional margins are min-id, since exchangeability is obvious by construction. To this purpose, let $p\in\N$ and choose a $p$-dimensional margin $IJ=(i_1,j_1),\ldots,(i_p,j_p)$. Denote by $J$ the unique indices in $(j_k)_{1\leq k\leq p}$. Define the set $I(j):=\{ i\mid (i,j)\in IJ\}$ to collect the $i$-indices in IJ with common index $j$, i.e.\ the indices in $I(j)$ correspond to all components of the random vector $\bmX_{j}$ that appear in $\bmX_{IJ}$. The key observation now is that
    \begin{align*}
        S_{IJ}(\bmx_{IJ}):=P\lc \bmX_{IJ}>\bmx_{IJ}\rc&=\e\lk  P\lc \bmX_{IJ}>\bmx_{IJ} \mid \mu \rc  \rk = \e \lk \prod_{j\in J} \exp\lc -\mu(A_j) \rc\rk  ,
    \end{align*}
    where for all $j\in J$ we denote $ A_j:=\{ \bmy\in E_d \mid y_{i}\leq x_{i,j} \text{ for some } i\in I(j)  \}$. Due to the infinite divisibility of $\mu$, we obtain
    that
    \begin{align*}
        &\e \lk \prod_{j\in J} \exp\lc -\mu(A_j) \rc \rk  =\e \lk \exp \lc - \int_{E_d} \lc \sum_{j\in J} \id_{A_j}(\bmx) \rc \mu(\rmd\bmx) \rc\rk \\
        &=\exp\Bigg( - \int_{E_d} \lc \sum_{j\in J} \id_{A_j}(\bmx) \rc \alpha(\rmd\bmx) \\
        &\hspace{1cm}-\int_{M_d^0} 1-\exp\lc -\int_{E_d} \lc\sum_{j\in J} \id_{A_j}(\bmx) \rc \eta(\rmd\bmx) \rc \nu(\rmd \eta)  \Bigg).
    \end{align*}
    To obtain that $\bmXIJ$ is min-id it is sufficient to show that $S_{IJ}^{1/n}(\cdot)$ is a survival function for every $n\in\N$. Since $\mu$ is an IDEM, there exist i.i.d.\ exponent measures $\lc \mu^{(i,n)}\rc_{1\leq i\leq n}$ such that $\mu\sim \sum_{\leqn} \mu^{(i,n)}$, which are themselves IDEMs with Lévy-Khintchine characteristics $(\alpha/n,\nu/n)$ by Corollary \ref{coridexpmisidbyidexpm}. By repeating the calculations above, $\mu^{(i,n)}$ defines the law of an exchangeable sequence with $p$-dimensional marginal survival function $S_{IJ}^{1/n}$. Thus, $S_{IJ}$ is the survival function of a min-id random vector. Therefore, $\ubmX$ is an exchangeable min-id sequence and there must exist an exponent measure $\lambda$ on $(E_{d})^\N$ such that the unconditional distribution of $\ubmX$ is $\minid(\lambda)$. 
    
    It remains to verify that $\lambda$ has the claimed decomposition. From the calculations above \textcolor{black}{and continuity from below} it follows that for every $\ubmx\in[-\infty,\infty)^{\textcolor{black}{d\times\N}}$ and $A_\ubmx:=\{\underline{\bmy}\in \lc E_d\rc^\N\mid y_{i,j}\leq x_{i,j} \text{ for some }$ $\ijdn\}$, $\lambda$ has the representation 
    \begin{align*}
        &\lambda(A_\ubmx)= \int_{E_d} \lc \sum_{j\in\N} \id_{A_j} (\bmx) \rc \alpha(\rmd\bmx) +\int_{M_d^0} 1-\exp\lc -\int_{E_d} \lc\sum_{j\in\N} \id_{A_j}(\bmx) \rc \eta(\rmd\bmx) \rc \nu(\rmd \eta)  \\
        &=\sum_{j\in\N} \alpha(A_j)   +\int_{M_d^0} 1-\prod_{j\in\N}\exp\lc  -\eta({A_j}) \rc \nu(\rmd \eta) \\
        &=\lambda_\alpha(A_\ubmx)   +\int_{M_d^0} 1-\prod_{j\in\N}\text{min-id}(\eta)\lc \{ \bm y\in E_d\mid y_{i}>x_{i,j} \text{ for all }1\leq i\leq d \}\rc \nu(\rmd \eta) \\
        &=\lambda_\alpha(A_\ubmx)   +\int_{M_d^0} 1-\otimes_{j\in\N}\text{min-id}(\eta)\lc \{ \underline{\bm y}\in \lc E_d\rc^\N \mid y_{i,j}> x_{i,j} \text{ for all }\ijdn \}\rc \nu(\rmd \eta) \\
        &=\lambda_\alpha(A_\ubmx)   +\int_{M_d^0} \otimes_{j\in\N}\text{min-id}(\eta)\lc \{ \underline{\bm y}\in \lc E_d\rc^\N \mid y_{i,j}\leq x_{i,j} \text{ for some }\ijdn \}\rc \nu(\rmd \eta) \\
        &=\lambda_\alpha(A_\ubmx)   +\int_{M_d^0} \otimes_{j\in\N} \text{min-id}(\eta) \lc A_\ubmx\rc \nu(\rmd \eta).
    \end{align*}
    where $\lambda_\alpha$ is defined as the measure $\lambda_\alpha=\sum_{j\in\N} \otimes_{i=1}^{j-1} \delta_{\binfty} \otimes \alpha  \otimes_{i=j+1}^\infty \delta_{\binfty}$, which is the exponent measure of an i.i.d.\ sequence of min-id random vectors.
    Now, an inclusion-exclusion type argument provides that 
    $$\lambda=  \lambda_\alpha  +\int_{M_d^0} \otimes_{i\in\N}\minid(\eta) \nu(\rmd \eta)  $$
    on a $\pi$-stable generator of the Borel $\sigma$-algebra of $\textcolor{black}{E^\prime_{d\times \N}}$ as follows. First, by induction, one can show that for all $IJ=(i_k,j_k)_{1\leq k\leq p}\subset d\times \N$ and $\bm a,\bm b\in[-\infty,\infty)^p$ with $\bm a<\bm b$ we have 
    \begin{align*}
    &\lambda\left( \{ \ubmy\in E_{d\times \N} \mid y_{i_k,j_k}\in (a_k , b_k]\ \forall\ 1\leq k\leq p \}\right)
    \\
    &= -\sum_{\bm c\in\prod_{k=1}^p\{a_k,b_k\}} (-1)^{|\{ k\in \{1,\ldots,p\} \mid c_k=b_k\}|}\lambda\left( \{ \ubmy\in E_{d\times \N} \mid  y_{i_k,j_k}\leq c_k \text{ for some }1\leq k\leq p \} \right) ,
    \end{align*} 
    \textcolor{black}{Applying the same inclusion-exclusion identity to $\lambda_\alpha+\lambda_\nu$, the preceding calculations show that both measures agree on all cylinders}
    \[
    \textcolor{black}{\left\{\ubmy\in E^\prime_{d\times\N}\mid y_{i_k,j_k}\in(a_k,b_k]\ \forall\ 1\leq k\leq p\right\},\qquad -\infty\leq a_k<b_k<\infty.}
    \]
    \textcolor{black}{Together with the empty set, these cylinders form a $\pi$-stable generator of the Borel $\sigma$-algebra of $E^\prime_{d\times\N}$. Here, all coordinates not appearing in the cylinder conditions remain unrestricted, including the value $\infty$. Moreover, the cylinders}
    \[
    \textcolor{black}{C_{i,j,m}:=\left\{\ubmy\in E^\prime_{d\times\N}\mid y_{i,j}\leq m\right\},\qquad 1\leq i\leq d,\quad j,m\in\N,}
    \]
    \textcolor{black}{form a countable cover of $E^\prime_{d\times\N}$ and have finite mass under both measures. In particular, the generated $\sigma$-algebra also contains $\{\ubmy\in E^\prime_{d\times\N}\mid y_{i,j}=\infty\}=E^\prime_{d\times\N}\setminus\bigcup_{m\in\N}C_{i,j,m}$. Thus, uniqueness of measures on a generating $\pi$-system yields}
    \[
    \textcolor{black}{\lambda=\lambda_\alpha+\lambda_\nu\quad\text{on }E^\prime_{d\times\N}.}
    \]
    \textcolor{black}{The equality extends to $E_{d\times\N}$, since both measures assign infinite mass to $\{\binfty\}$ by our convention.}

    Concerning the uniqueness, we already know that the law of a min-id sequence is uniquely associated to its exponent measure and that the law of an IDEM is uniquely associated to its Lévy characteristics. Therefore, it solely remains to show that the decomposition of $\lambda$ into $\lambda_\alpha+\lambda_\nu$ is unique, i.e.\ that $\lambda_{\alpha_1}+\lambda_{\nu_1}=\lambda=\lambda_{\alpha_2}+\lambda_{\nu_2}$ implies that $\alpha_1=\alpha_2$ and $\nu_1=\nu_2$. The uniqueness of $\alpha$ simply follows from the fact that $\lambda_\alpha$ is supported on $\mathcal{W}$ with $\lambda_\nu(\mathcal{W})=0$, since for every $\eta$ with $\eta(E^\prime_d)\not=0$ we have $\minid(\eta)(\binfty)\in[0,1)$ which gives
    \begin{align*}
        \lambda_\nu\lc\mathcal{W}\rc&=\int_{M_d^0} \otimes_{j\in\N} \minid(\eta)\lc \{ \ubmx \in E_{d\times \N}\mid \bmx_j=\binfty \text{ for all but one } j\}\rc \nu(\rmd\eta)\\
        &=\int_{M_d^0} \sum_{j\in\N} \minid(\eta)(E_d)\prod_{i\in\N}\minid(\eta)(\binfty) \nu(\rmd\eta)=0.
    \end{align*}
    As $\lambda$ is unique, $\lambda_{\alpha_1}$ and $\lambda_{\alpha_2}$ must coincide on $\mathcal{W}$, and thus $\alpha_1=\alpha_2$. Next, consider $\lambda_{\nu}$ and recall that the distribution of an IDEM is uniquely determined by its Lévy characteristics. Thus, when $\nu_1\not=\nu_2$ there must exist two IDEMs $\mu_1$ and $\mu_2$ which are not identical in law but yield the same exchangeable min-id sequence. To this purpose, recall that an IDRM is uniquely determined by its evaluations on sets of the form $(\bmx,\binfty]^\complement$, where $\bmx\in [-\infty,\infty)^d$. Since $\mu\lc (\bmx,\binfty]^\complement\rc\geq 0$, it is thus sufficient to show that the Laplace transform of $\lc \mu\lc (\bmx_i,\binfty]^\complement\rc \rc_\leqn $ is uniquely determined by $\lambda$, which then implies that $(\alpha_1,\nu_1)$ and $(\alpha_2,\nu_2)$ must coincide. Therefore, for $\lc a_i\rc_\leqn\in \N^n$\textcolor{black}{, writing $m:=\sum_{i=1}^n a_i$,} recall that
    \begin{align*}
       &\e\lk \exp\lc \sumn - a_i\mu\lc (\bmx_i,\binfty]^\complement\rc\rc \rk \\
       &=\exp\lc - \lambda\lc  \lc\times_{i=1}^{a_1} (\bmx_1,\binfty] \times  \ldots\times_{i=1}^{a_n} (\bmx_n,\binfty] \textcolor{black}{\times_{j=m+1}^\infty E_d}\rc^\complement   \rc\rc.
    \end{align*}
    Next, recall that the law of a non-negative random vector is uniquely determined by its Laplace transform on $\N^n$, see e.g.\ \cite[p.\ 8]{kleiberstoyanov2013}. Thus, if there are two IDEMs $\mu_1$ and $\mu_2$ with different Lévy characteristics $(\alpha_1,\nu_1)$ and $(\alpha_2,\nu_2)$ there exists $( a_i)_\leqn \in\N^n$ and $(\bmx_i)_\leqn\in \textcolor{black}{\big([-\infty,\infty)^d\big)^n}$ such that 
    $$\e\lk \exp\lc -\sumn  a_i\mu_1 \lc (\bmx_i,\binfty]^\complement\rc\rc \rk\not=\e\lk \exp\lc -\sumn a_i\mu_2\lc (\bmx_i,\binfty]^\complement\rc\rc \rk.$$
    However, this implies that $\lambda_{\alpha_1}+\lambda_{\nu_1}\not=\lambda_{\alpha_2}+\lambda_{\nu_2}$ and we have a contradiction. Thus, $\lambda$ uniquely corresponds to a tuple of Lévy characteristics $(\alpha,\nu)$ and the law of $\ubmX$ is uniquely determined by $(\alpha,\nu)$. 

    \textcolor{black}{The converse statement that every exchangeable min-id sequence $\ubmX$ with exponent measure $\lambda$ of the form $\lambda=\lambda_\alpha+\lambda_\nu$ for some Lévy characteristics $(\alpha,\nu)$ is uniquely associated to an IDEM follows immediately. By constructing a min-id sequence via the random survival function $\exp\lc-\mu\lc (\cdot,\binfty]^\complement\rc\rc$ the IDEM $\mu$ with Lévy characteristics $(\alpha,\nu)$ generates an exchangeable min-id sequence with the same distribution as $\ubmX$. Moreover, we have already proven above that $(\alpha,\nu)$ is unique as $\lambda$ can only correspond to a unique set of Lévy characteristics, which implies the uniqueness of the law of the corresponding IDEM $\mu$.}
\end{proof}

\subsection*{Proof of Proposition \ref{proptrafos}}

\begin{proof}
\textcolor{black}{First, let us verify that the transformations preserve local finiteness. In $(i)$, for every $\bmy\in\R^d$ we can choose $t\in\R$ such that $g(t,\ldots,t)>\bmy$. Monotonicity then gives $g^{-1}((\bmy,\binfty]^\complement)\subseteq(t,\infty]^\complement$, so $\mu_g((\bmy,\binfty]^\complement)<\infty$. In $(ii)$, local finiteness follows from the assumed boundedness of $g$. In $(iii)$, it follows from $\mu^{(\beta)}((\bmy,\binfty]^\complement)\leq\mu((\bmy,\binfty]^\complement)\beta((\bmy,\binfty]^\complement)<\infty$. Finally, $g(\binfty)=\binfty$ in $(i)$, $g(\binfty)>0$ in $(ii)$ and $\beta(\{\binfty\})>0$ in $(iii)$ preserve the infinite atom at $\binfty$.} It remains to determine the Lévy characteristics in each case. \textcolor{black}{Before deriving the exact form of the Lévy measure on the transformed IDEM in cases $(i)-(iii)$ let us remark that the same argument as above yields that applying either transformation $(i)-(iii)$ to an $\eta\in M_d$ yields a measure in $M_d$. Thus, applying a transformation from $(i)-(iii)$ to a measure in the support $M_d^0$ of $\nu$ again yields a measure in $M_d$, which makes the stated Lévy measures in $(i)-(iii)$ well defined. Moreover, to retain the uniqueness of the Lévy measures of the transformations of $\mu$ below, we need to discard those measures from the Lévy measures which are mapped to the zero measure on $E^\prime_d$, explaining the terms $\eta_g(E^\prime_d)\neq 0$ , $\eta^{(g)}(E^\prime_d)\neq 0$ and $\eta^{(\beta)}(E^\prime_d)\neq 0$ appearing in the respective cases $(i)-(iii)$. Now, let us derive the exact form of the stated Lévy measure in each of the cases $(i)-(iii)$. }
\begin{enumerate}
    \item[$(i)$] Note that $\mu_g$ is the image measure of $\mu$ under $g$. Thus, $\int_{E_d} f(\bmx)\mu_g(\rmd\bmx)=\int_{E_d} f(g(\bmx))\rmd\mu$, which gives us that 
    \begin{align*}
        &\e\lk \exp\lc -\int_{E_d} f(\bmx)\mu_g(\rmd\bmx)\rc \rk\\
        &=\exp\lc -\int_{E_d} f(g(\bmx))\alpha(\rmd\bmx)- \int_{M_d^0} 1-\exp\lc -\int f(g(\bmx))\eta(\rmd\bmx)\rc \nu(\rmd\eta)\rc \\
        &=\exp\lc -\int_{E_d} f(\bmx)\alpha_g(\rmd\bmx)- \int_{M_d^0} 1-\exp\lc- \int f(\bmx)\eta(\rmd\bmx)\rc \nu_g(\rmd\eta)\rc 
    \end{align*}
    for every non-negative measurable $f:E_d\to[0,\infty)$. Thus, $\mu_g$ has characteristic triplet $(\alpha_g,\nu_g)$.
    \item[$(ii)$] We note that $\mu\sim \alpha+\int_{M_d^0} \eta N(\rmd\eta)$ for a PRM $N=\sum_{i\in\N} \delta_{\eta_i}$ on $M_d^0$. Therefore, we have 
    \begin{align*}
        &\mu^{(g)}(A)\\
        &=\int_{A} g(\bmx) \alpha(\rmd \bmx)+\int_{A} g(\bmx)  \lc \suminf \eta_i\rc (\rmd \bmx) = \alpha^{(g)}(A)+ \lim_{n\to\infty}\sumn\int_{A} g(\bmx)   \eta_i(\rmd \bmx) \\
        &=\alpha^{(g)}(A)+ \suminf\int_{A}  \eta_i^{(g)}(\rmd \bmx)=\alpha^{(g)}(A)+\int_{M_d^0} \eta^{(g)}(A) N(\rmd\eta),
    \end{align*}
    where we can interchange the limit and the integration since $\sumn $ $\int_{A} g(\bmx)  \eta_i(\rmd \bmx)$ is a monotone increasing sequence. Note that $\suminf \delta_{\eta_i^{(g)}}$ is again a PRM with intensity measure $\nu\lc \{ \eta\mid \eta^{(g)}\in\cdot\}\rc$, since the map $\eta \mapsto \eta^{(g)}$ is measurable.
    Thus, we get
        \begin{align*}
        &\e\lk \exp\lc -\int_{E_d} f(\bmx)\mu^{(g)}(\rmd\bmx)\rc \rk\\
        &=\exp\lc -\int_{E_d} g(\bmx)f(\bmx)\alpha(\rmd\bmx) \rc \e\lk \exp\lc -\int_{E_d} f(\bmx) \lc \suminf\eta^{(g)}\rc(\rmd\bmx)\rc \rk \\
        &=\exp\lc -\int_{E_d} f(\bmx)\alpha^{(g)}(\rmd\bmx)  -\int_{M_d^0} 1-\exp\lc -\int_{E_d} f(\bmx) \eta(\rmd\bmx)\rc \nu^{(g)}(\rmd\eta)\rc  \\
    \end{align*}
    for every non-negative measurable $f:E_d\to[0,\infty)$, which proves the claim.
    \item[$(iii)$] Observe that 
    \begin{align*}
        \mu^{(\beta)}\lc A\rc&=\int_{A}\alpha\lc (-\binfty,\bmx]\rc  \beta(\rmd\bmx)+\int_{A}\suminf \eta_i\lc (-\binfty,\bmx]\rc  \beta(\rmd\bmx)\\
        &=\alpha^{(\beta)}\lc A\rc  +\suminf \eta_i^{(\beta)}\lc A\rc,
    \end{align*}
    where we can interchange integration and summation since $\eta_i\lc (-\binfty,\bmx]\rc\geq 0$.  Now, similar arguments as in $(ii)$ yield the claim.
\end{enumerate}

\end{proof}

\subsection*{Proof of Proposition \ref{propsubordinationcrm}}

\begin{proof}
\begin{itemize}
    \item[$(i)$]    
    First, when $\mu_0\lc (-\infty,t]\rc<\infty$ almost surely for all $t\in\R$ then 
    \begin{align*}
        P\lc \mu_i\lc (-\infty,t]\rc<\infty \rc&=\e \lk P\lc \mu_i\lc (-\infty,t]\rc<\infty \mid \mu_0\rc\rk\\
        &=\e \lk P\lc \int_{-\infty}^t \int_0^\infty a \rmd N^{(i,\mu_0)}(\rmd a,\rmd b) <\infty \mid \mu_0\rc\rk=1
    \end{align*}
    since (\ref{condcrmpolishpsace}) is almost surely satisfied by $\mu_0$ and where we used the notation that $N^{(i,\mu_0)}$ denotes a PRM on $(0,\infty)\times\R$ with intensity measure $\rho_i(a)\rmd a \mu_0(\rmd b)$. Thus, $\mu_i\lc (-\infty,t]\rc<\infty$ almost surely for all $1\leq i\leq d$ and $t\in\R$.
    
    To prove the remaining claims recall that the sum of independent CRMs is again a CRM. Further, since $\mu_0$ is an IDRM, we can find i.i.d.\ random measures $\lc \mu_0^{(j,n)}\rc_{1\leq j\leq n}$ such that $\mu_0\sim \sumn \mu_0^{(j,n)}$. Now, conditionally on $\lc \mu_0^{(j,n)}\rc_{1\leq j\leq n}$, for every $1\leq j\leq n$ we can construct independent vectors of CRMs $\lc\mu_i^{(j,n)}\rc_{1\leq i\leq d}$ where the entries $\mu_i^{(j,n)}$ are independent CRMs with Lévy characteristics $(\alpha_i/n,\nu_i^{(j,n)})$, where $\nu_i^{(j,n)}$ is the image of $(a,b)\mapsto a\delta_b\textcolor{black}{+\infty\delta_\infty}$ under $\rho_i(a)\rmd a \rmd \mu_0^{(j,n)}$. Since the $\mu_0^{(j,n)}$ are i.i.d.\ we get that $\lc \lc\mu_i^{(j,n)}\rc_{1\leq i\leq d}\rc_{1\leq j\leq n}$ is a collection of $n$ i.i.d.\ vectors of random measures. Moreover, conditionally on $\lc \mu_0^{(j,n)}\rc_{1\leq  j\leq n}$, $\lc \sum_{j=1}^n \mu_i^{(j,n)}\rc_{1\leq i\leq d}$ is distributed as a vector of independent CRMs with respective base measures $\alpha_i$ and Lévy measures given as the image measure of $(a,b)\mapsto a\delta_b\textcolor{black}{+\infty\delta_\infty}$ under $\rho_i\rmd a \sum_{j=1}^n \mu_0^{(j,n)} (\rmd b)$. Thus, since $\mu_0\sim \sum_{j=1}^n \mu_0^{(j,n)}$, we have $\lc \sum_{j=1}^n \mu_i^{(j,n)}\rc_{1\leq i\leq d}\sim \lc \mu_i\rc_{1\leq i\leq d}$.

    For brevity, we omit the details and refer to mimicking the proof of $(iii)$ to obtain the claimed Lévy characteristics of $\mu_i$.

    \item[$(ii)$]
    We need to prove that choosing $\mu_0$ as a CRM implies that $\lc \mu_i\rc_{1\leq i\leq d}$ are also CRMs. Obviously, $\nu_{1,i}$ is the Lévy measure of a CRM. Further, $\nu_0$ is concentrated on measures of the form $\underline{A}:=\{\eta\mid\eta=a\delta_b +\infty\delta_\binfty \text{ for some }(a,b)\in(0,\infty)\times\R\}$.
    In this case, $\crm_{\rho_i(a)\rmd a \eta(\rmd b)}$ is of the form $\crm_{\rho_i(a)\rmd a (a_0\delta_{b_0})(\rmd b)}$ and thus has at most one finite atom. Therefore, $\nu_{2,i}$ also concentrates on measures of the form $\underline{A}$ and the sum $\nu_{1,i}+\nu_{2,i}$ is still a measure concentrating on $\underline{A}$. Along the lines of the proof of \cite[Theorem 3.19]{kallenberg2017} it is now easy to deduce that every Lévy measure which concentrates on $\underline{A}$ is the Lévy measure of a CRM. Thus, the claim is proven.

    \item[$(iii)$] 
    First, it is obvious to see that $\mu_{\mu_1^{(\kappa_1)},\ldots,\mu^{(\kappa_d)}_d}$ as defined in (\ref{defsubordinatedCRMexpmeasure}) is an exponent measure whenever $\mu_i^{(\kappa_i)} \lc (-\infty,t]\rc<\infty$ for all $t\in\R$. Thus, it remains to show its infinite divisibility and for this purpose it is sufficient to derive the Lévy characteristics of $\mu_{\mu_1^{(\kappa_1)},\ldots,\mu^{(\kappa_d)}_d}$. 
    
    First, observe that $\mu_i=\alpha_i+\tilde{\mu}_i$ for all $0\leq i\leq d$, where $\tilde{\mu}_0$ is an IDRM with Lévy characteristics $(0,\nu_0)$ and $\tilde{\mu}_i$ is, conditionally on $\mu_0$, a CRM  without base measure and Lévy measure as the image measure of $(a,b)\mapsto a\delta_b\textcolor{black}{+\infty\delta_\infty}$ under $\rho_i(a)\rmd a \mu_0(\rmd b)$. Note that for every non-negative measurable $f$ on $E^\prime_d$ we have 
    \begin{align*}
        &\int_{E_d^\prime} f(\bmx)\mu_{\mu_1^{(\kappa_1)},\ldots,\mu^{(\kappa_d)}_d}(\rmd\bmx) = \sumd \int_\R f(\bmx^{(i)}(s))\mu_i^{(\kappa_i)}(\rmd s)\\
        &=\sumd \int_{\R} f(\bmx^{(i)}(s))\alpha^{(\kappa_i)}_i(\rmd s)+ \sumd \int_{\R} f(\bmx^{(i)}(s))\kappa_i(s,y)\tilde{\mu}^{(\kappa_i)}_i(\rmd s)\\
        &=\int_{E^\prime_d}f(\bmx) \mu_{\alpha_1^{(\kappa_1)},\ldots,\alpha^{(\kappa_d)}_d}(\rmd\bmx)+\int_{E^\prime_d} f(\bmx)\mu_{\tilde{\mu}_1^{(\kappa_1)},\ldots,\tilde{\mu}^{(\kappa_d)}_d}(\rmd\bmx)
    \end{align*}
    where $\bmx^{(i)}(s)$ has all entries equal to $\infty$ except for the $i$-th entry equal to $s$. Let $\lc \zeta_i\rc_{1\leq i\leq d}$ denote independent univariate CRMs without base measure and Lévy measure as the image measure of $(a,b)\mapsto a\delta_b\textcolor{black}{+\infty\delta_\infty}$ under $\rho_i(a)\rmd a \alpha_0(\rmd b)$. Define
    \begin{align*}
        T_1:=&- \log\lc \e\lk \exp\lc- \int_{E_d^\prime} f(\bmx) \mu_{\zeta^{(\kappa_1)}_1,\ldots,\zeta^{(\kappa_d)}_d}(\rmd\bmx)\rc\rk \rc    \\
        &=-\log\lc  \e\lk \exp\lc- \sumd \int_{\R} f(\bmx^{(i)}(s))  \zeta_i^{(\kappa_i)} (\rmd s) \rc\rk \rc\\
        &=-\log\lc \e\lk  \exp\lc- \sumd \int_\R\lc \int_\R f(\bmx^{(i)}(s))\kappa_i(s,y)\rmd s \rc \zeta_i(\rmd y)\rc \rk \rc\\
        &= \sumd\int_{\R}   \int_0^\infty  1-\exp\lc -\int_\R f(\bmx^{(i)}(s)) (a\delta_b)^{(\kappa_i)}(\rmd s) \rc \rho_i(a)\rmd a  \alpha_0(\rmd b)    \\
        &= \sumd\int_{\R}   \int_0^\infty  1-\exp\lc - \int_{E_d^\prime} f(\bmx)\lc\otimes_{j=1}^{i-1} \delta_\infty \times (a\delta_b)^{(\kappa_i)}\otimes_{j=i+1}^d\rc (\rmd \bmx) \rc \\
        &\hspace{1cm}\rho_i(a)\rmd a  \alpha_0(\rmd b) 
    \end{align*}
    and
    \begin{align*}
        T_2:=& \sumd \int_{\R}  \int_{\R} f(\bmx^{(i)}(s))\kappa_i(s,y) \rmd s \alpha_i(\rmd y) =\int_{E_d^\prime}f(\bmx) \mu_{\alpha_1^{(\kappa_1)},\ldots,\alpha^{(\kappa_d)}_d}(\rmd \bmx) .
    \end{align*}
    We obtain
    \begin{align}
    &-\log\lc \e\lk \exp\lc- \int_{E_d^\prime} f(\bmx)\mu_{\mu^{(\kappa_1)}_1, \ldots,\mu^{(\kappa_d)}_d}(\rmd\bmx)\rc \rk\rc\nonumber \\
    &=-\log\lc \e\lk \exp\lc -\sumd \int_{\R} \int_{\R}  f(\bmx^{(i)}(s)) \kappa_i(s,y)  \rmd s\mu_i(\rmd y) \rc  \rk \rc\nonumber \\
    &=-\log\lc \e\lk \e\lk \exp\lc -\sumd\int_{\R} \lc \int_{\R} f(\bmx^{(i)}(s))\kappa_i(s,y) \rmd s \rc \mu_i(\rmd y) \rc  \ \bigg\vert\ \alpha_0,\tilde{\mu}_0 \rk \rk \rc \nonumber \\
    &=-\log\lc\e\lk \e\lk \exp\lc -\sumd\int_{\R} \lc \int_{\R} f(\bmx^{(i)}(s))\kappa_i(s,y) \rmd s \rc \tilde{\mu}_i(\rmd y) -T_2\rc  \ \bigg\vert\ \alpha_0,\tilde{\mu}_0 \rk \rk \rc \nonumber \\
     &=-\log\lc \e\lk \exp\lc - \sumd \int_{\R}   \int_0^\infty \lc 1-e^{ -\int_{\R} af(\bmx^{(i)}(s))\kappa_i(s,b) \rmd s } \rc \rho_i(a)\rmd a  \tilde{\mu}_0(\rmd b)  \rc   \rk \rc \nonumber \\
     & +T_1 +T_2\nonumber \\
    &=-\log\lc \e\lk \exp\lc - \int_{\R} \lc \sumd \int_0^\infty \lc  1-e^{ -\int_{\R} af(\bmx^{(i)}(s))\kappa_i(s,b) \rmd s } \rc \rho_i(a)\rmd a \rc \tilde{\mu}_0(\rmd b)  \rc  \rk \rc \nonumber \\
    & +T_1 +T_2 \nonumber \\
    &=  \int_{\textcolor{black}{M^0_1}} \lc 1-e^{ -\int_{\R}  \lc \sumd \int_{0}^\infty \lc 1- e^{ -\int_{\R} af(\bmx^{(i)}(s))\kappa_i(s,b)\rmd s } \rc \rho_i(a)\rmd a \rc \eta(\rmd b) } \rc \nu_0(\rmd\eta)   \nonumber \\
    &+T_1+T_2\nonumber \\
    &=  \int_{\textcolor{black}{M^0_1}} \Bigg( 1-e^{ -\sumd\int_{\R}   \int_{0}^\infty \lc 1- e^{ -\int_\R\lc \int_{\R} f(\bmx^{(i)}(s))\kappa_i(s,y)\rmd s \rc (a\delta_b)(\rmd y) } \rc \rho_i(a)\rmd a  \eta(\rmd b) } \Bigg) \nu_0(\rmd\eta) \nonumber \\ 
    & +T_1+T_2\nonumber \\
    &\overset{\star}{=}  \int_{\textcolor{black}{M^0_1}} \lc 1-\e \lk e^{ -\sumd \int_{\R} \lc \int_{\R}  f(\bmx^{(i)}(s))\kappa_i(s,b)\rmd s \rc \eta_i(\rmd b)  }\rk  \rc  \nu_0(\rmd\eta)+ T_1+T_2 \nonumber \\
    &=  \int_{\textcolor{black}{M^0_1}}  \e\lk 1- e^{ -\int_{E_d^\prime} f(\bmx) \mu_{\eta^{(\kappa_1)}_1,\ldots,\eta^{(\kappa_d)}_d} (\rmd\bmx) }\rk   \nu_0(\rmd\eta)  +T_1+T_2 \nonumber \\
    &=  \int_{\textcolor{black}{M^0_1}}\int_{\times_{i=1}^d \textcolor{black}{M_1}}   1- e^{ -\int_{E_d^\prime} f(\bmx) \mu_{\eta^{(\kappa_1)}_1,\ldots,\eta^{(\kappa_d)}_d} (\rmd\bmx) } \otimes_{i=1}^d \crm_{\rho_i(a)\rmd a \eta(\rmd b)}(\rmd \eta_i)  \nu_0(\rmd\eta)\nonumber \\
    &+T_1+T_2 \label{eqnlaplacesubordination}
    \end{align}
    where in $\star$ $\lc \eta_i\rc_{1\leq i\leq d}$ denote independent CRMs without base measure and intensity $\rho_i(a)\rmd a \eta(\rmd b)$ and $\mu_{\eta^{(\kappa_1)}_1,\ldots,\eta^{(\kappa_d)}_d} $ denotes an exponent measure constructed as in (\ref{defsubordinatedCRMexpmeasure}). Thus, combining the above we have that $\mu_{\mu^{(\kappa_1)}_1, \ldots,\mu^{(\kappa_d)}_d}$ has base measure $\mu_{\alpha^{(\kappa_1)}_1,\ldots,\alpha^{(\kappa_d)}_d}$ and Lévy measure $\nu=\nu_1+\nu_2$, where $\nu_1$ is given by 
    $$  \nu_1(\eta \in A)= \sumd \int_{0}^\infty \int_\R  \id_{\{ \otimes_{j=1}^{i-1} \delta_\infty \times (a\delta_b)^{(\kappa_i)}\otimes_{j=i+1}^d \delta_\infty\in A \}} \rho_i(a)\rmd a \alpha_0(\rmd b) $$
    and $\nu_2$ is defined as
    $$ \nu_2(\eta \in A)=\int_{\textcolor{black}{M^0_1}}\int_{\times_{i=1}^d \textcolor{black}{M_1}}  \id_{\big\{ \textcolor{black}{\infty\delta_{\binfty}\neq}\mu_{\eta_1^{(\kappa_1)},\ldots,\eta_d^{(\kappa_d)}} \in A \big\}} \otimes_{i=1}^d \crm_{\rho_i(a)\rmd a \eta(\rmd b)}(\rmd \eta_i) \nu_0(\rmd\eta)$$
   \textcolor{black}{It remains to verify that, for $\nu$-almost every image measure $\eta$, $\eta(U_i)<\infty$ for all $i\in\N$. For the localizing sets $(U_i)_{i\in\N}$ defined in (\ref{deflocseq}), local finiteness of the constructed measure and the Laplace identity (\ref{eqnlaplacesubordination}) with $f=c\id_{U_j}$ imply, as $c\downarrow0$, that the image measures are finite on $U_j$ almost everywhere. Taking $c=1$ and using $1\wedge x\leq (1-e^{-1})^{-1}(1-e^{-x})$ gives $\int_{M_d^0}(1\wedge\eta(U_j))\nu(\rmd\eta)<\infty$ for all $j\in\N$, proving the claims.}
   
    \end{itemize}

\end{proof}

\subsection*{Proof of Theorem \ref{thmdombryeyiminkcond}}
\begin{proof}
    Note that every min-id sequence $\ubmX$ can be viewed as a continuous min-id process with index set $d\times\N$. Condition \ref{condcontmargins} implies that $\ubmX$ has continuous marginal distributions and therefore \cite[Theorem 3.3]{dombryeyiminko2013regular} can be applied to every $\ubmX_{n+k}$, as their results are only formulated for compact index sets. A translation of their formulas for max-id processes to min-id sequences yields the claimed representation.
\end{proof}

\subsection*{Proof of Lemma \ref{lemcondhitseqex}}
{\color{black}
\begin{proof}
We prove that for almost every realization of
$(\ubmX_n,\Tilde{\Theta})$ each $\ubmZ^{(n,l)}$ is exchangeable conditionally on $(\ubmX_n,\Tilde{\Theta})$. To this purpose. we mainly have to show that the probability kernel $K_{IJ}(\bmx_{IJ},\rmd \ubmx_{\setminus IJ})$ is invariant with respect to permutations of $\lc \bmx_{n+j}\rc_{j\in\N}$ for all $n$ such that $IJ\subset d\times n$.

First, fix $IJ\subset d\times n$ and put $p:=\vert IJ\vert$. Let $\pi$
be a permutation of $\N$ with at most finitely many $j\in\N$ for which $\pi(j)\neq j$ and define the permutation
$\pi_n$ of $\N$ by
\begin{align*}
    \pi_n(j) &= j, 1\leq j\leq n \text{ and } \pi_n(n+j)= n+\pi(j),  j\in\N.
\end{align*}
Define the measurable map
$$
T_{\pi_n}:E_{d\times\N}\longrightarrow E_{d\times\N},
\qquad
T_{\pi_n}(\ubmy):=\lc \bmy_{\pi_n(j)}\rc_{j\in\N}.
$$
Since $\ubmX$ is exchangeable,
$$
T_{\pi_n}(\ubmX)\sim\ubmX.
$$
Moreover, $\lambda\circ T_{\pi_n}^{-1}=\lambda$ since the exponent measure of $\ubmX$ is exchangeable as well.

Since $\pi_n$ does not change the first $n$ vectors $T_{\pi_n}$ leaves the coordinates in $IJ$
unchanged and induces a measurable permutation
$T_{\pi_n}^{IJ}$ of
$(-\infty,\infty]^{d\times\N\setminus IJ}$. Now, we define the probability
kernel
$$
K_{IJ}^{\pi}\lc \bmx_{IJ},A\rc
:=
K_{IJ}\lc
\bmx_{IJ},
\lc T_{\pi_n}^{IJ}\rc^{-1}(A)
\rc.
$$
and we claim that $K_{IJ}^{\pi}$ also satisfies the defining disintegration formula (\ref{defconddistexpmeasure}) of $\lambda$
with respect to $\lambda^{(IJ)}$. Observe that for every non-negative
measurable function $\mathsf{G}$ as in (\ref{defconddistexpmeasure}) we have
\begin{align*}
&\int_{E_p^\prime}\int_{(-\infty,\infty]^{d\times\N\setminus IJ}}
\mathsf{G}\lc \bmx_{IJ},\ubmx_{\setminus IJ}\rc
K_{IJ}^{\pi}\lc
\bmx_{IJ},\rmd\ubmx_{\setminus IJ}
\rc
\lambda^{(IJ)}\lc \rmd\bmx_{IJ}\rc\\
&\quad=
\int_{E_p^\prime}\int_{(-\infty,\infty]^{d\times\N\setminus IJ}}
\mathsf{G}\lc
\bmx_{IJ},
T_{\pi_n}^{IJ}(\ubmx_{\setminus IJ})
\rc
K_{IJ}\lc
\bmx_{IJ},\rmd\ubmx_{\setminus IJ}
\rc
\lambda^{(IJ)}\lc \rmd\bmx_{IJ}\rc\\
&\quad=
\int_{E_{d\times\N}}
\mathsf{G}\lc
\bmx_{IJ},
T_{\pi_n}^{IJ}(\ubmx_{\setminus IJ})
\rc
\lambda(\rmd\ubmx) =
\int_{E_{d\times\N}}
\mathsf{G}\lc
\bmx_{IJ},
\ubmx_{\setminus IJ}
\rc
\lambda(\rmd\ubmx),
\end{align*}
where the last equality follows from
$\lambda\circ T_{\pi_n}^{-1}=\lambda$ and the fact that
$T_{\pi_n}$ fixes the coordinates in $IJ$.
Thus, both $K_{IJ}^{\pi}$ and $K_{IJ}$ satisfy
(\ref{defconddistexpmeasure}). 
By the $\lambda^{(IJ)}$-almost sure uniqueness of the disintegration we have
$$
K_{IJ}^{\pi}\lc \bmx_{IJ},\cdot\rc
=
K_{IJ}\lc \bmx_{IJ},\cdot\rc
$$
for $\lambda^{(IJ)}$-almost every
$\bmx_{IJ}$. Thus, we have shown that for every fixed $\pi$,
$$
K_{IJ}\lc
\bmx_{IJ},
\lc T_{\pi_n}^{IJ}\rc^{-1}(A)
\rc
=
K_{IJ}\lc
\bmx_{IJ},A
\rc
$$
for every measurable $A$, outside a $\lambda^{(IJ)}$-null set which
may depend on $\pi$.

Since the collection of finite-support permutations of $\N$ is
countable, the union of these null sets is again a
$\lambda^{(IJ)}$-null set. Hence there exists a measurable
$N_{IJ}\subset E_p^\prime$ with
$
\lambda^{(IJ)}(N_{IJ})=0
$
such that, for every $\bmx_{IJ}\notin N_{IJ}$ the probability measures $K_{IJ}(\bmx_{IJ},\cdot)$ are invariant under every
permutation $\pi$ of $\N$ which only permutes finitely many coordinates. Since exchangeability is determined by
the finite-dimensional distributions, we also obtain that the probability measures $K_{IJ}(\bmx_{IJ},\cdot)$ are invariant under every
permutation $\pi$ of $\N$.

It remains to verify that $(\Tilde{\Theta},\ubmX_n)$-almost surely the kernel $K_{\Tilde{\Theta}_l}(\bmX_{\Tilde{\Theta}_l},\cdot)$ is evaluated on $N_{IJ}^\complement$. To this purpose, fix a
partition
$$
\theta=(\theta_1,\ldots,\theta_{L(\theta)})
$$
of $d\times n$. By \cite[Theorem 3.1]{dombryeyiminko2013regular} (and also recalled in (\ref{defbetameasures})),
$
\beta_\theta(\cdot)=P\lc \ubmX_n\in\cdot , \Theta=\theta\rc
$ is absolutely continuous w.r.t.\ $\bigotimes_{l=1}^{L(\theta)}
\lambda^{(\theta_l)}(\rmd\bmx_{\theta_l})$.
Therefore, for every $1\leq l\leq L(\theta)$,
$$
\beta_\theta\lc
\bmx_{\theta_l}\in N_{\theta_l}
\rc
=0.
$$

By the definition of the conditional hitting scenario,
$\Tilde{\Theta}$ is drawn conditionally on $\ubmX_n$ according to
$\tau(\ubmX_n,\cdot)$, where
$$
\tau(\ubmX_n,\theta)
=
P\lc
\Theta=\theta\mid\ubmX_n
\rc.
$$
Hence, for every measurable $A\subset\R^{d\times n}$, we use the definition of conditional expectation to obtain
\begin{align*}
P\lc
\ubmX_n\in A,\Tilde{\Theta}=\theta
\rc
&=
\e\lk
\id_{\{\ubmX_n\in A\}}
\tau(\ubmX_n,\theta)
\rk
=
\e\lk
\id_{\{\ubmX_n\in A\}}
P\lc
\Theta=\theta\mid\ubmX_n
\rc
\rk\\
&=
P\lc
\ubmX_n\in A,\Theta=\theta
\rc
=
\beta_\theta(A).
\end{align*}
Consequently, for almost every realization
$(\ubmX_n,\Tilde{\Theta})=(\ubmx_n,\theta)$ and every
$1\leq l\leq L(\theta)$ we have $P(\ubmX_n\in N_{\theta_l},\Tilde{\Theta}_l=\theta_l)=\beta_\theta(\ubmx_n\in N_{\theta_l})=0$, i.e.\ 
$$K_{\Tilde{\Theta}_l}(\bmX_{\Tilde{\Theta}_l},\ubmy_{\setminus\Tilde{\Theta}_l})$$
is exchangeable in the argument $\lc\bmy_{n+j}\rc_{j\in\N}$ $(\ubmX_n,\Tilde{\Theta})$-almost surely. This proves that $\bmZ^{(\Tilde{\Theta}_l)}$ is exchangeable $(\ubmX_n,\Tilde{\Theta})$-almost surely.

It merely remains to prove that $\ubmZ^{(n,l)}$ as defined in (\ref{defcondextrseq}) is exchangeable. Fixing a realization $(\ubmX_n,\Tilde{\Theta})=(\ubmx_n,\theta)$ such that exchangeability holds we get for every measurable
$A\subset E_{d\times\N}$,
\begin{align}
&P\lc
\ubmZ^{(n,l)}\in A
\mid
\ubmX_n=\ubmx_n,\Tilde{\Theta}=\theta
\rc \nonumber \\
&\quad=
\frac{
K_{\theta_l}\lc
\bmx_{\theta_l},
\big\{
\ubmy_{\setminus\theta_l}\mid
\ubmy_{n\setminus\theta_l}>
\ubmx_{n\setminus\theta_l},
\lc \bmy_{n+j}\rc_{j\in\N}\in A
\big\}
\rc
}{
K_{\theta_l}\lc
\bmx_{\theta_l},
\big\{
\ubmy_{\setminus\theta_l}\mid
\ubmy_{n\setminus\theta_l}>
\ubmx_{n\setminus\theta_l}
\big\}\rc
}, \label{eqnperm}
\end{align}
where we only need to consider those $(\ubmx_n,\theta)$ where the denominator is non-zero and implicitly set the probability to $0$ otherwise.
The event
$$
\big\{
\ubmy_{\setminus\theta_l}\mid
\ubmy_{n\setminus\theta_l}>
\ubmx_{n\setminus\theta_l}
\big\}
$$
depends only on the first $n$ vectors and is therefore
invariant under every permutation $\pi$ which only affects coordinates $(n+j)_{j\in\N}$. Since $\bmx_{\theta_l}\notin N_{\theta_l}$,
the numerator and the denominator in (\ref{eqnperm}) are invariant under any permutation $\pi$ which only affects coordinates $(n+j)_{j\in\N}$. Thus, conditionally on
$(\ubmX_n,\Tilde{\Theta})$, the sequence $\ubmZ^{(n,l)}$ is
exchangeable.


\end{proof}
}

\subsection*{Proof of Theorem \ref{thmmainresult}}
\begin{proof}
    Everything except for the independence of $\overline{\mu}_n$ and $S_n$ follows from the discussion above the theorem and Lemma \ref{lemcondhitseqex}. \cite[Theorem 3.3 3.]{dombryeyiminko2013regular} provides that $\underline{\bm Y}^{(n)}$ and $\lc \Tilde{\Theta},\lc\underline{\bm Z}^{(n,l)}\rc_{1\leq l\leq L}\rc$ are independent conditionally on $\ubmX_n$. Since $\overline{\mu}_n$ and $S^{(n,l)}$ can be recovered from $\underline{\bm Y}^{(n)}$ and $\underline{\bm Z}^{(n,l)}$, respectively, we obtain their independence.
\end{proof}

\subsection*{Proof of Lemma \ref{lemconddistexpregcond}}
\begin{proof}
    By definition, $K_{IJ}\lc \bmx_{IJ},\cdot\rc$ is a probability measure on $(-\infty,\infty]^{md-\vert IJ\vert}$ for every $\bmx_{IJ}\in E_{\vert IJ\vert}^\prime$, since $f^{(IJ)}$ and $g^{(IJ)}$ are the appropriate marginalizations of $f^{(d\times m)}$ and $
    g^{(d\times m)}$.   It remains to show that for every $\mathsf{G}:E_{\vert IJ\vert}\times E_{md-\vert IJ\vert}\to [0,\infty)$ which vanishes on $\{\binfty\}\times E_{md-\vert IJ\vert}$ (\ref{defconddistexpmeasure}) holds. We have
    \begin{align*}
       &\int_{E^\prime_{\vert IJ\vert}}\int_{E_{md-\vert IJ\vert}} \mathsf{G}\lc \bmx_{IJ},\ubmx_{m\setminus IJ}\rc K_{IJ}\lc  \bmx_{IJ},\rmd \ubmx_{m\setminus IJ}\rc \lambda^{(IJ)}\lc  \rmd \bmx_{IJ}\rc\\
       &=\int_{E^\prime_{\vert IJ\vert}}\int_{E_{md-\vert IJ\vert}} \mathsf{G}\lc \bmx_{IJ},\ubmx_{m\setminus IJ}\rc \frac{f^{(d\times m)}\lc \ubmx_m \rc+g^{(d\times m)}\lc \ubmx_m\rc}{f^{(IJ)}\lc \bmx_{IJ} \rc+g^{(IJ)}\lc \bmx_{IJ} \rc} \lc \Lambda^{(\infty)}_{md-\vert IJ\vert}\rc (\rmd\ubmx_{m\setminus IJ})\\
       &\ \ \ \ \ \ \lc f^{(IJ)}\lc\bmx_{IJ}\rc+g^{(IJ)}\lc\bmx_{IJ}\rc  \rc\lc\Lambda_{\vert IJ\vert}^{(\infty)}\rc (\rmd \bmx_{IJ})\\
       &=\int_{E_{md}} \mathsf{G}\lc \bmx_{IJ},\ubmx_{m\setminus IJ}\rc
       \lambda^{(d\times m)}(\rmd\ubmx_m),
    \end{align*}
    since $f^{(d\times m)}(\rmd\ubmx_m)+g^{(d\times m)}(\rmd\ubmx_m) \lc\otimes_{i=1}^{md} \Lambda^{(\infty)}\rc(\rmd\ubmx_m)=\lambda^{(d\times m)}(\rmd\ubmx_m)$, using that \linebreak$\mathsf{G}(\binfty,\cdot)=0$. Thus (\ref{defconddistexpmeasure}) holds and $K_{IJ}$ is the regular conditional probability of $\lambda^{(d\times m)}$.
\end{proof}

\subsection*{Proof of Lemma \ref{lemconddistextrfctregcond}}

\begin{proof}
To ease the notation let $m=n+k$ and set $\tilde{\ubmX}_m:=\big( \bmX_j\id_{\{j\leq n\}}\linebreak+\bmx_j\id_{\{n< j\leq m\}} \big)_{1\leq j\leq m}$ and $\hat{\ubmX}_m:=\lc \bmX_j\id_{\{j\leq n\}}+\binfty\id_{\{n< j\leq m\}} \rc_{1\leq j\leq m}$. We use Lemma \ref{lemconddistexpregcond}, Fubini, the representation of $f^{(IJ)}$ in (\ref{def_f_IJ}) and Theorem \ref{thmdombryeyiminkcond} to obtain
    \begin{align*}
    &C(\ubmX_n,\Tilde{\Theta}_l)P\lc  \bm Z^{(n,l)}_{i} >\bmx_{i} \text{ for all }n< i\leq m \mid \Tilde{\Theta},\ubmX_n\rc\\
    &=C(\ubmX_n,\Tilde{\Theta}_l)\frac{ K_{\Tilde{\Theta}_l}\lc \bmX_{\Tilde{\Theta}_l}, \{ \ubmy_{m\setminus\Tilde{\Theta}_l}\in E_{md-\vert \tilde{\Theta}_l\vert } \mid \ubmy_{m\setminus\Tilde{\Theta}_l}>\tilde{\ubmX}_{m\setminus\Tilde{\Theta}_l}\}\rc }
    {K_{\Tilde{\Theta}_l}\lc \bmX_{\Tilde{\Theta}_l},\{\ubmy_{n\setminus \Tilde{\Theta}_l}\in E_{nd-\vert \tilde{\Theta}_l\vert } \mid \ubmy_{n\setminus \Tilde{\Theta}_l}>\ubmX_{n\setminus \Tilde{\Theta}_l}\}\rc}  \\
    &=  \int_{\{ \ubmy_{m\setminus \Tilde{\Theta}_l}>\Tilde{\ubmX}_{m\setminus \Tilde{\Theta}_l}\}} \lc f^{(d\times m)}+g^{(d\times m)}\rc\lc (\tilde{\ubmX}_m,\ubmy_m)(\Tilde{\Theta}_l)\rc  \Lambda^{(\infty)}_{md-\vert \Tilde{\Theta}_l\vert} \lc\rmd \ubmy_{m\setminus \Tilde{\Theta}_l} \rc  \\
    &= \Bigg( \int_{M_d^0}   \int_{\{ \ubmy_{m\setminus \Tilde{\Theta}_l}>\Tilde{\ubmX}_{m\setminus \Tilde{\Theta}_l} \}}  \prod_{i=1}^m f_\eta\lc (\tilde{\ubmX}_m,\ubmy_m)(\Tilde{\Theta}_l)\rc \Lambda^{(\infty)}_{md-\vert \Tilde{\Theta}_l\vert}  \lc\rmd \ubmy_{m\setminus \Tilde{\Theta}_l} \rc  \nu(\rmd\eta) \\
    &+ \int_{\{ \ubmy_{m\setminus \Tilde{\Theta}_l}>\Tilde{\ubmX}_{m\setminus \Tilde{\Theta}_l}\}} g^{(d\times m)}\lc (\tilde{\ubmX}_m,\ubmy_m)(\Tilde{\Theta}_l)\rc  \Lambda^{(\infty)}_{md-\vert \Tilde{\Theta}_l\vert} \lc\rmd \ubmy_{m\setminus \Tilde{\Theta}_l} \rc \Bigg)  \\
    &\overset{\star}{=}  \Bigg( \int_{M_d^0} \int_{\{ \bmy_{n+j}>\bmx_{j}\forall 1\leq j\leq k\}}  \prod_{j=1}^{k} f_\eta(\bmy_{n+j}) \int_{\{ \ubmy_{n\setminus \Tilde{\Theta}_l}>\ubmX_{n\setminus \Tilde{\Theta}_l} \}} \prod_{i=1}^n f_\eta\lc (\bmX_i,\bmy_i)(\Tilde{\Theta}_l)\rc \\
    &\ \ \ \Lambda^{(\infty)}_{nd-\vert \Tilde{\Theta}_l\vert}   \lc\rmd \ubmy_{n\setminus\Tilde{\Theta}_l}\rc  \Lambda^{(\infty)}_{kd}  \lc\rmd\lc \bmy_{n+j}\rc_{1\leq j\leq k} \rc  \nu(\rmd\eta)   \\
    & + \int_{\{ \bmy_{n+j}=\binfty \forall 1\leq j\leq k, \ubmy_{n\setminus \Tilde{\Theta}_l}>\ubmX_{n\setminus \Tilde{\Theta}_l}\}} g^{(d\times m)}\lc (\hat{\ubmX}_m,\ubmy_m)(\Tilde{\Theta}_l)\rc \Lambda^{(\infty)}_{md-\vert \Tilde{\Theta}_l\vert}  \lc\rmd \ubmy_{m\setminus \Tilde{\Theta}_l} \rc   \Bigg) \\
    &\overset{\star}{=}  \Bigg( \int_{M_d^0} \underset{1\leq  j\leq k}{\otimes} \minid(\eta) \lc  \big \{ \ubmy_k\mid \bmy_{j}>\bmx_{j}\ \forall\ 1\leq j\leq k \big\}  \rc  h(\eta,\ubmX_n,\Tilde{\Theta}_l) \nu(\rmd\eta) \\
    &\ \ \ \Lambda^{(\infty)}_{nd-\vert \Tilde{\Theta}_l\vert}   \lc\rmd \ubmy_{n\setminus\Tilde{\Theta}_l}\rc  \Lambda^{(\infty)}_{kd}  \lc\rmd\lc \bmy_{n+j}\rc_{1\leq j\leq k} \rc  \nu(\rmd\eta)   \\
    & + \int_{\{  \ubmy_{n\setminus \Tilde{\Theta}_l}>\ubmX_{n\setminus \Tilde{\Theta}_l}\}} g^{(d\times n)}\lc (\ubmX_n,\ubmy_n)(\Tilde{\Theta}_l)\rc \Lambda^{(\infty)}_{nd-\vert \Tilde{\Theta}_l\vert}  \lc\rmd \ubmy_{n\setminus \Tilde{\Theta}_l} \rc   \Bigg), \\
\end{align*}
where 
\begin{align*}
    C(\ubmX_n,\Tilde{\Theta}_l):=& \int_{M_d^0} h(\eta,\ubmX_n,\Tilde{\Theta}_l) \nu(\rmd\eta)\\
    &+ \int_{\{  \ubmy_{n\setminus \Tilde{\Theta}_l}>\ubmX_{n\setminus \Tilde{\Theta}_l}\}} g^{(d\times n)}\lc (\ubmX_n,\ubmy_n)(\Tilde{\Theta}_l)\rc \Lambda^{(\infty)}_{nd-\vert \Tilde{\Theta}_l\vert}  \lc\rmd \ubmy_{n\setminus \Tilde{\Theta}_l} \rc
\end{align*}
and
$$h(\eta,\ubmX_n,\Tilde{\Theta}_l) :=  \int_{\{ \ubmy_{n\setminus \Tilde{\Theta}_l}>\ubmX_{n\setminus \Tilde{\Theta}_l}  \}} \prod_{i=1}^n f_\eta\lc (\ubmX_n,\ubmy_n)(\Tilde{\Theta}_l)\rc \lc  \Lambda^{(\infty)}_{nd-\vert \Tilde{\Theta}_l\vert} \rc \lc\rmd\ubmy_{n\setminus \Tilde{\Theta}_l} \rc.$$
Note that in $\star$ we have repeatedly used that $\ubmX_n$ is real-valued and that $g^{(d\times m)}$ is concentrated on $\mathcal{W}$. Moreover, $C(\ubmX_n,\Tilde{\Theta}_l)$ is real-valued since $C(\ubmX_n,\Tilde{\Theta}_l)\leq f^{(\Tilde{\Theta}_l)}(\ubmX_{\Tilde{\Theta}_l})+g^{(\Tilde{\Theta}_l)}(\ubmX_{\Tilde{\Theta}_l})$.
\end{proof}

\subsection*{Proof of Lemma \ref{lemcondhitscenregcond}}
\begin{proof}
    The distribution of the conditional hitting scenario can be characterized analytically. Recalling (\ref{defbetameasures}), we define the family of measures $\lc\beta_\theta\rc_{\theta\in\mathcal{P}_n}$ via $\beta_\theta (A):=P\lc\ubmX_n\in A, \Theta=\theta\rc$. Trivially, each $\beta_\theta$ is absolutely continuous w.r.t.\ $\beta:=\sum_{\theta\in\mathcal{P}_n} \beta_\theta$ with a density $d\beta_\theta/d\beta$ on $\R^{d\times n}$. Now, \cite[Theorem 3.2.1]{dombryeyiminko2013regular} provides that 
    $$ \tau(\ubmX_n,\Tilde\Theta)=\frac{d\beta_{\Tilde\Theta}}{d\beta}(\ubmX_n). $$
    and by (\ref{defbetameasures}) $\beta_{\theta}(\rmd\ubmx_n)$ is given by
    $$ \exp\lc -\lambda^{(d\times n)}\lc (\ubmx_n,\binfty]^\complement\rc\rc \prod_{l=1}^{L(\theta)} K_{ \theta_l}\lc \bmx_{\theta_l},\big\{ \ubmy_{n\setminus \theta_l}\mid \ubmy_{n\setminus \theta_l}>\ubmx_{n\setminus \theta_l} \big\}\rc   \otimes_{l=1}^{L(\theta)}\lambda^{(\theta_l)}(\rmd \bmx_{\theta_l}).$$
    Noting that 
    $$\otimes_{l=1}^{L(\Theta)} \lambda^{(\Theta_l)}(\rmd \bmx_{\Theta_l})=\lc\prod_{l=1}^{L(\Theta)} f^{(\Theta_l)}(\bmx_{\Theta_l})+g^{(\Theta_l)}(\bmx_{\Theta_l})\rc  \lc \Lambda_{nd}^{(\infty)} \rc(\rmd\ubmx_n)$$
    we get that $\beta$ and $\beta_\theta$ both have a density w.r.t.\ $\Lambda^{(\infty)}_{nd}$. Thus, $\frac{d\beta_{\theta}}{d\beta}$ is given by
    \begin{align*}
        &\frac{ \prod_{l=1}^{L(\theta)} K_{ \theta_l}\lc \bmx_{\theta_l},\{ \ubmy_{n\setminus \theta_l}>\ubmx_{n\setminus \theta_l}\}\rc \prod_{l=1}^{L(\theta)} \lc f^{(\theta_l)}+g^{(\theta_l)}\rc (\bmx_{\theta_l}) }
    {\sum_{\tilde{\theta}\in\mathcal{P}_n} \prod_{l=1}^{L(\tilde{\theta})} K_{ \Tilde{\theta}_l}\lc \bmx_{\tilde{\theta}_l},\{ \ubmy_{n\setminus \Tilde{\theta}_l}>\ubmx_{n\setminus \Tilde{\theta}_l}\}\rc \prod_{l=1}^{L(\Tilde{\theta})} \lc f^{(\tilde{\theta}_l)}+g^{(\tilde{\theta}_l)}\rc(\bmx_{\tilde{\theta}_l})}\\
    &=     \frac{ \prod_{l=1}^{L(\theta)}  \int_{\{  \ubmy_{n\setminus \theta_l}>\ubmx_{n\setminus \theta_l}\}} \lc f^{(d\times n)}+ g^{(d\times n)}\rc\lc (\ubmx_n,\underline{\bmy}_n)(\theta_l) \rc  \lc  \Lambda^{(\infty)}_{nd-\vert \theta_l\vert}\rc  \rmd \lc \ubmy_{n\setminus \theta_l}\rc     }
    {\sum_{\tilde{\theta}\in\mathcal{P}_n}  \prod_{l=1}^{L(\tilde{\theta})} \int_{\{  \ubmy_{n\setminus \tilde{\theta}_l}>\ubmx_{n\setminus \Tilde{\theta}_l}\}} \lc f^{(d\times n)}+ g^{(d\times n)}\rc\lc (\ubmx_n,\underline{\bmy}_n)(\tilde{\theta}_l) \rc   \lc \Lambda^{(\infty)}_{nd-\vert \tilde{\theta}_l\vert}\rc   \rmd  \lc\ubmy_{n\setminus \Tilde{\theta}_l}\rc   }
    \end{align*}
    where we used the representation of $K_{IJ}$ from Lemma \ref{lemconddistexpregcond} and $(\ubmx_n,\underline{\bmy}_n)(\theta_l)=\linebreak\big( \id_{\{ (i,j)\in \theta_{l} \}}x_{i,j}+ y_{i,j}\id_{\{(i,j)\not \in \theta_{l}\}}\big)_{(i,j)\in d\times n}$.
\end{proof}

\subsection*{Proof of Lemma \ref{lemsurvivmodel}}
\begin{proof}
    First we need to show that $\mu_i^{(\kappa_i)}([0,t))<\infty$ for all $t>0$.
    Note that by Condition \ref{condkernel} we have that for every $t>0$ we can find a $C>0$ such that $\int_0^t \int_0^\infty \kappa_i(s,b)\mu_i(\rmd b)\rmd s =\int_0^t \int_0^C \kappa_i(s,b)\mu_i(\rmd b)\rmd s $. Thus, since $\kappa_i$ is also bounded by Condition \ref{condkernel}, $\mu_i^{(\kappa_i)}([0,t))<\infty$ for all $t>0$ if and only if $\mu_i\lc [0,C)\rc<\infty$ for all $C>0$. It is easy to see that $\mu_i\lc [0,C)\rc<\infty$ if $\mu^{(\kappa_0)}_0\lc [0,C)\rc<\infty$. By the same arguments as above, $\mu^{(\kappa_0)}_0\lc [0,C)\rc<\infty$ if and only if $\int_0^{\bar{C}} \int_0^\infty \min\{a_0,1\} l(a_0,b_0)\rmd a_0\rmd b_0<\infty$ for all $\bar{C}>0$, which is satisfied by Condition \ref{condkernel}.

    Next, we show that $\mu_i^{(\kappa_i)}([0,t))>0$ for all $t>0$. To see that $\mu_0([0,t))>0$ almost surely notice that $\int_0^\infty l(a_0,b_0)\rmd a_0=\infty$ for every $b_0>0$ together with $l(a_0,b_0)>0$ implies that $\mu_0$ has infinitely many non-zero atoms which are dense in $[0,t)$, thus $\mu_0([0,t))>0$ almost surely. This also implies that $\mu_0^{(\kappa_0)}([0,t))>0$ since for every non-zero atom $b$ of $\mu_0$ we use Condition \ref{condkernel} to find a $\Bar{b}>b$ such that $\kappa_0(\tilde{s},b)>0$ for all $\tilde{s}\in(b,\Bar{b})$, which implies $\int_{\{b\}}\int_b^{\Bar{b}} \kappa_0(\tilde{s},\tilde{b}) \rmd\tilde{s} \mu_0(\rmd \tilde{b})=\mu_0(\{b\})\int_b^{\Bar{b}} \kappa_0(\tilde{s},b) \rmd\tilde{s}>0 $, implying $\mu_0^{(\kappa_0)}((b,\Bar{b}))>0$. Therefore, $\mu_i([0,t))>0$ almost surely, since $\int_0^\infty \rho_i(a)\rmd a=\infty$ and, conditionally on $\mu_0$, the PRM representation of $\mu_i$ has infinitely many non-zero atoms in $[0,t)]$. Similarly as before, we can now deduce that $\mu_i^{(\kappa_i)}([0,t))>0$ almost surely. 
    
    Let us show that $\mu_{\mu^{(\kappa_1)}_1,\ldots,\mu^{(\kappa_d)}_d}$ is an IDEM. For $\int_0^\infty \kappa_i(s,b)\rmd s=\infty$ we now obviously have $\lim_{t\to\infty}\mu_i^{(\kappa_i)}([0,t))=\infty$ and $\mu_{\mu^{(\kappa_1)}_1,\ldots,\mu^{(\kappa_d)}_d}$ is an IDEM. If $\mu_0([0,\infty))=\infty$ we have  $\mu_0^{(\kappa_0)}([0,\infty))=\int_0^\infty \int_0^\infty \kappa_0(s,b)\rmd s\mu_0(\rmd b)\geq \int_0^\infty \int_b^{b+\delta} C\rmd s\mu_0(\rmd b)=\infty$. Therefore, $\mu_i([0,\infty))=\infty$ almost surely and the analogous argument yields $\mu_i^{(\kappa_i)}([0,\infty))=\infty$. Thus, $\mu_{\mu^{(\kappa_1)}_1,\ldots,\mu^{(\kappa_d)}_d}$ is an IDEM and  Proposition \ref{proptrafos} in combination with Proposition \ref{propsubordinationcrm} $\textcolor{black}{(iii)}$ provides that the Lévy measure has the claimed form.

    It remains to verify Condition \ref{condabsolutecont}. Since for each $1\leq i\leq d$  $\mu_i^{(\kappa_i)}([0,\infty))=\infty$ we have $\e\lk \exp\lc -\mu^{(\kappa_i)}_i([0,t])\rc\rk=\exp\lc- \int_{M_d^0} \minid(\eta)\lc \{ \bmx\in E_d\mid x_i\leq t \}\rc \nu(\rmd\eta)\rc\to 0$ for $t\to\infty$, i.e. $\lim_{t\to\infty}\int_{M_d^0} \minid(\eta)\lc \{ \bmx\in E_d\mid x_i\leq t \}\rc \nu(\rmd\eta)=\infty$. Furthermore, from the representation of $\nu$ it is easy to see that $\nu$-almost-every $\eta$ has a density w.r.t.\ $\Lambda^{(\infty)}_d$ of the claimed form, which shows that Condition \ref{condabsolutecont} is satisfied.
\end{proof}

\subsection*{Proof of Lemma \ref{lemdenshierachicalmodel}}

\begin{proof}
First, note that the density of $\mu_{\eta^{(\kappa_1)}_1,\ldots,\eta^{(\kappa_d)}_d}$ w.r.t. $\Lambda^{(\infty)}_{d}$ is $\prodd e_{\eta_i,\kappa_i}(x_i)$. We use the general formula for the density $f^{(d\times m)}$ defined in (\ref{def_f_m}) to get
    \begin{align*}
        &f^{(d\times m)}(\ubmx_m)=\sumd   \id_{\{ x_{k,j}=\infty \forall (k,j)\in d\times m, k\not=i\}}\int_{0}^\infty \int_{0}^\infty  \lc \prod_{j=1}^m e_{a\delta_b,\kappa_i}(x_{i,j}) \rc  \rho_i(a)\rmd a \alpha^{(\kappa_0)}_0(\rmd b)\\
        &+ \int_{0}^\infty \int_{0}^\infty \int_{\times_{i=1}^d \textcolor{black}{M_1}}  \prod_{j=1}^m \prodd  e_{\eta_i,\kappa_i}(x_{i,j}) \otimes_{i=1}^d \crm_{\rho_i(a)\rmd a (a_0\delta_{b_0})^{(\kappa_0)}(\rmd b)}(\rmd \eta_i) l(a_0,b_0)\rmd a_0\rmd b_0\\
        &=\sumd   \id_{\{ x_{k,j}=\infty \forall (k,j)\in d\times m, k\not=i\}}\int_{0}^\infty \int_{0}^\infty  \prod_{j=1}^m  e_{a\delta_b,\kappa_i}(x_{i,j})    \rho_i(a)\rmd a \alpha^{(\kappa_0)}_0(\rmd b)\\
        &+  \int_{0}^\infty \int_{0}^\infty  \prodd  \e_{\eta_i\sim \crm_{\rho_i(a)\rmd a (a_0\delta_{b_0})^{(\kappa_0)}(\rmd b)}}\bigg[ \prod_{j=1}^m    e_{\eta_i,\kappa_i}(x_{i,j})  \bigg]   l(a_0,b_0)\rmd a_0\rmd b_0 
    \end{align*}
    Using (\ref{def_f_IJ}), the claimed result immediately follows.
\end{proof}

\section{Technical results on random measures and their support}
\label{appdefnrandmeas}
Given the sequence of measurable sets $(U_i)_{i\in\N}$ defined in (\ref{deflocseq}) we may construct a localizing ring by $\hat{S}:=\cup_{n\in\N} \bc\lc E^\prime_d\rc\cap S_n $, where $S_n:= \cup_{i=1}^n U_i$. According to \cite[Chapter 2]{kallenberg2017} a random measure on (the standard Borel space) $E^\prime_d$ is a random element in the space 
$$\{ \mu \mid \mu \text{ is a measure on $E^\prime_d$ which is finite on every set in } \hat{S}  \},$$
which is equipped with the Borel $\sigma$-algebra generated by all evaluations of measures on Borel sets of $E^\prime_d$. We show that this definition coincides with our definition of a random measure from Definition \ref{defnrandommeasure}.

\begin{prop}
    $\mu$ is a random measure on $E^\prime_d$ with localizing ring $\hat{S}$ in the sense of \cite[Chapter 2]{kallenberg2017} if and only if it can be extended to a random measure on $E_d$ with localizing sets $(U_i)_{i\in\N}$ in the sense of Definition \ref{defnrandommeasure}.
\end{prop}
\begin{proof}
    Note that the $\sigma$-algebra that defines measurability of a random measure on $E^\prime_d$ is identical in both definitions. 
    Thus, it remains to show that the finiteness of $\mu$ on the localizing sets $(U_i)_{i\in\N}$ is equivalent to finiteness of $\mu$ on the localizing ring of $E^\prime_d$.
    
    Obviously, if $\mu$ is a random measure on $E^\prime_d$ with localizing ring $\hat{S}$ we have $\mu\lc U_i\rc<\infty$, since each $U_i$ is measurable and thus an element of the localizing ring. Therefore, a random measure in the sense of \cite[Chapter 2]{kallenberg2017} is a random measure on $E^\prime_d$ according to Definition \ref{defnrandommeasure}.
    
    For the converse, we solely need to show that Definition \ref{defnrandommeasure} implies that $\mu$ is finite on every set $A\in\hat{S}$. By the definition of $\hat{S}$ there exists a Borel set $B$ on $E^\prime_d$ and an $n\in\N$ such that $A=B\cap \cup_{i=1}^n U_i$. Therefore $\mu(A)<\infty$ if $\mu(U_i)<\infty$ for all $i\in\N$, which proves the claim.
    
    Finally, it remains to observe that setting $\mu(\binfty)=\infty$ is a measurable extension of the random measure $\mu$ on $E^\prime_d$ to a random measure on $E_d$. To see this, note that $\{\omega\in\Omega \mid \mu(C)\in D\}=\{\omega\in\Omega \mid \mu(C\setminus\{\binfty\})\in D\}$ for every measurable $C\subset E^\prime_d$ and measurable $D\subset[0,\infty]$. Moreover, $\{\omega\in\Omega \mid \mu(C)\in D\}=\emptyset \id_{\{\infty\not\in D\}}+\Omega\id_{\{\infty\in D\}}$ for every measurable $C\subset E_d$ containing $\binfty$ and every measurable $D\subset[0,\infty]$. Thus, it is now easy to see that the evaluation functional $C\mapsto \mu(C)$ is measurable for every measurable $C\subset E_d$, which implies measurability of $\mu$ as a random measure on $E_d$.
\end{proof}

Further, we show that if an IDEM $\mu$ has deterministic support, we may transfer this property to its corresponding Lévy characteristics.

\begin{lem}
\label{lemsupptrafoidem}
    Assume that an IDEM $\mu$ on $E_d$ with Lévy characteristics $(\alpha,\nu)$  has support that is almost surely contained in a measurable set $S\subseteq E_d$. Then, $\alpha\lc S^\complement\rc=0$ and $\nu$ is concentrated on $\big\{\eta\mid \eta\lc S^\complement\rc=0\big\}$.
\end{lem}

\begin{proof}
We prove the claim by contradiction. To this purpose, define $A:=\big\{ \eta\mid \eta\lc S^\complement\rc >0\big\}$ and assume $\mu\lc S^\complement\rc=0$ almost surely but $\nu(A)>0$. Note that $A$ is measurable since $S$ is measurable. Moreover, if $\nu(A)>0$ there exists $a>0$ s.t. $A_a:=\big\{ \eta\mid \eta\lc S^\complement\rc >a\big\}$ satisfies $\nu\lc A_a\rc>0$. Therefore, since $\mu\sim \alpha +\int_{M^0_d}\eta N(\rmd\eta)=\alpha +\int_{M^0_d\cap A}\eta N(\rmd\eta)+\int_{M^0_d\cap A^\complement}\eta N(\rmd\eta)$, we get $0=\mu\lc S^\complement\rc=\alpha\lc S^\complement\rc +\int_{M^0_d\cap A}\eta\lc S^\complement\rc N(\rmd\eta)> aN\lc A_a\rc$ almost surely. However, since $\nu(A_a)>0$ we have $P(N(A_a)>0)=1-\exp(-\nu (A_a))>0$ which is a contradiction.
\end{proof}

This also has implications for the support of the posterior of an IDEM.
\begin{cor}
     If $\mu$ is an IDEM which is concentrated on some set $S\subset E_d$, then $\overline{\mu}_n$ and $\lc\overline{\mu}_n^{(l)}\rc_{1\leq l\leq L(\Tilde{\Theta})}$ from Theorem \ref{thmconddistmixtureofminid} are concentrated on $S$ as well. Therefore, the posterior of an IDEM $\mu$ inherits the support of the prior.
\end{cor}

For example, when $\minid(\mu)$ concentrates on survival functions of random vectors with independent components, i.e.\ $\mu$ concentrates on $\cup_{i=1}^d \ee_i$, the posterior of $\overline{\mu}_n$ and $\lc\overline{\mu}^{(l)}\rc_{1\leq l\leq L(\Tilde{\Theta})}$ concentrate on  $\cup_{i=1}^d \ee_i$ as well. 

\begin{lem}
\label{lemlevymeassigmafinite}
    The Lévy measure of an IDEM on $E_d$ is $\sigma$-finite.
\end{lem}
\begin{proof}
    Let $A_{i,\epsilon}:=\{\eta\in M^0_d\mid \eta(U_i)\geq \epsilon\}$ and define the countable collection of sets $(A_{i,1/j})_{i,j\in\N}$. Moreover, for every $\eta\in M^0_d$ we have $\eta(U_i)>0$ for some $i\in\N$ because otherwise $\eta(E_d^\prime)=0$. Thus, $M_d^0=\cup_{i,j\in\N} A_{i,1/j}$ and it remains to show $\nu\lc A_{i,1/j} \rc<\infty$ for all $i,j\in\N$. We know that $\int_{M_d^0} \min\{\eta(U_i);1\}\nu(\rmd\eta)<\infty$ for all $i\in\N$ by Theorem \ref{thmidrandmeasurelevykhintchine} and therefore $\nu( A_{i,1/j})=\int_{M_d^0} \id_{\{\eta(U_i)\geq 1/j\}}\nu(\rmd\eta)\leq j\int_{M_d^0} \min\{\eta(U_i);1\}\nu(\rmd\eta)<\infty$, which proves the claim.
\end{proof}

\section{Discussion and Outlook}
\label{secconclusion}
This paper introduces a large class of prior distributions based on IDEMs, which can be seen as a generalization of NTR priors to the multivariate and continuous setting. The framework is shown to be general as it embeds the frameworks of \cite{fergusonprioronprob1974,Dykstralaud1981,loweng1989,hjort1990,epifanilijoi2010,lijoinipoti2014,rivapalacioleisen2018,camerlengilijoipruenster2021}. Further, it is shown how to construct tractable IDEM priors from simple building blocks. As a by-product, the paper also extends the concept of subordination of CRMs by CRMs to subordination of CRMs by IDRMs, which allows to construct vectors of dependent IDRMs with tractable Lévy characteristics and is of independent interest beyond the scope of this paper. In contrast to most of the previous literature, the derivation of the posterior is not based on conditioning on a latent structure, but rather on a de-Finetti type argument, which allows for a closed-form representation of the posterior in terms of the posterior predictive distribution and, under some regularity conditions, also of the IDEM. 

A reader who is familiar with extreme value theory may have already noticed that the whole framework may be re-interpreted in the viewpoint of this field. Essentially, the results of this paper say that it is possible to conduct nonparametric Bayesian inference on the exponent measure of a min-id random vector. Thus, this work can be seen as the natural framework for nonparametric Bayesian inference for min-id distributions in the context of extreme value theory as well and an investigation of this topic is deferred to future research.

Finally, an extension of the current framework to observations from continuous real-valued stochastic processes seems to be in reach, since the results of \cite{dombryeyiminko2013regular} apply to observations from stochastic processes as well. However, there are some technical subtleties that deserve a careful treatment. For example, it remains to find a suitable extension of hitting scenarios, since these are inherently defined for observations at finitely many locations and an analog on the continuous process level has to be found. Moreover, it is not straightforward to formulate a suitable analog of classical finite-dimensional IDRMs in the infinite-dimensional setting, e.g.\ such as infinite-dimensional NTR priors. We leave these issues as an open problem for future investigation.

\section{Overview of notation}
\label{appgloss}
\begingroup
\small
\renewcommand{\arraystretch}{1.15}
\subsection*{Spaces and (index) sets}
\begin{center}
\begin{tabular}{p{0.25\textwidth}p{0.65\textwidth}}
\textbf{Symbol} & \textbf{Meaning} \\
$E_p$ & $(-\infty,\infty]^p$ with $p\in\N\cup\{\infty\}$ \\
$E_p^\prime$ & $E_p\setminus\{\binfty\}$ \\
$d\times\N$ & $\{ (i,j) \mid \leqd, j\in \N\}$ \\
$d\times n$ & $\{1,\ldots,d\}\times\{1,\ldots,n\}$ \\
$IJ$ & $=(i_k,j_k)_{1\leq k\leq p}$\\
$I(IJ)$   & $\{i \mid (i,j) \in IJ \text{ for some } j\in\N\}$ \\
$J(IJ)$   &  $\{j \mid (i,j) \in IJ \text{ for some } \leqd\}$ \\
$M_d$ & $\big\{ \eta \text{ measure on }E_d \text{ s.t. } \eta((\bm x,\binfty]^\complement)<\infty\ \forall \ \bmx\in\R^d,\eta(\binfty)=\infty \big\}$ \\
$M_d^0$ & $M_d\cap \{\eta \mid \eta(E_d^\prime)\neq 0\}$ \\
$\mathcal{P}_n$ ($\mathcal{P}_{IJ}$) & Set of all partitions of $d\times n$ ($IJ$) \\
$\ee_i$ & $\{\bmx\in E_d\mid x_i<\infty, x_j=\binfty\ \forall\ 1\leq j\leq d,\ j\neq i\}$ \\
$\ee_i(t)$ & $\big\{ \bmx\in E_d \mid x_i\leq t, x_j=\infty\ \forall\ 1\leq j\leq d,\ j\neq i \big\} $ \\
$\ee_D$ & $\{ \bmx\in E_d \mid x_i<\infty\ \forall\ i\in D \text{ and }x_j=\infty\ \forall\  1\leq j\leq d,\ j\notin D\}$
\end{tabular}
\end{center}

\subsection*{Sequences, vectors, and min-id objects}
\begin{center}
\begin{tabular}{p{0.25\textwidth}p{0.65\textwidth}}
    \textbf{Symbol} & \textbf{Meaning} \\
$\ubmX$ & $(\bmX_j)_{j\in\N}\in\R^{d\times\N}$ \\
$X_{i,j}$ & Coordinate $i$ of observation $j$ in $\ubmX$\\
$\ubmX_n$, $\bmX_{IJ}$ &  $\ubmX_n=(\bmX_j)_{1\leq j\leq n}$, and $\bmX_{IJ}=\lc X_{i,j}\rc_{(i,j)\in IJ}$. \\
$\ubmX_{\setminus IJ}$ & $(X_{i,j})_{ (i,j)\in d\times \N\setminus IJ}$ \\

$\ubmx$, $\ubmx_n$, $\bmx_{IJ}$, $\ubmx_{\setminus IJ}$ & Deterministic counterpart of $\ubmX$, $\ubmX_n$, $\bmXIJ$ and $\ubmX_{\setminus IJ}$ \\
$\minid(\eta)$ & Law of a min-id sequence or vector with exponent measure $\eta$ \\
$\lambda$ & Exponent measure of the min-id sequence $\ubmX$ which is generated by drawing an IDEM $\mu$ and then drawing i.i.d.\ random vectors $\bmX_j\sim \minid(\mu)$  \\
$\lambda_\alpha$ & The part of $\lambda$ that can be attributed to the base measure $\alpha$ \\
$\lambda_\nu$ & The part of $\lambda$ that can be attributed to the L\'evy measure $\nu$ \\
$\lambda^{(IJ)}$ & Exponent measure of $\bmX_{IJ}$ \\
\end{tabular}
\end{center}

\subsection*{(Random) measures}
\begin{center}
\begin{tabular}{p{0.25\textwidth}p{0.65\textwidth}}
\textbf{Symbol} & \textbf{Meaning} \\
$\mu$ & Infinitely divisible random exponent measure \\
$\alpha$ & Deterministic base measure in the L\'evy--Khintchine representation of $\mu$ \\
$\nu$ & L\'evy measure in the L\'evy--Khintchine representation of $\mu$ on $M_d^0$ \\
$\eta$ & Non-random element of $M_d^0$ \\
$\bar\mu$ & The posterior of $\mu$, i.e.\ $\bar\mu\sim \mu\mid\ubmX_n$ \\
$\overline{\nu}_n$ & Posterior L\'evy measure of $\mu$ \\
$L_{(\alpha,\nu)}(\cdot,\ubmx_m)$ & Multivariate Laplace transform of the vector $(\mu((\bmx_i,\binfty]^\complement))_{1\leq i\leq m}$. \\
$\eta_f(\cdot)$ & Image measure of $\eta$ under $f$, i.e.\ $\eta \lc \{ x\mid f(x) \in \cdot \}\rc$ \\
$\mu_{\mu_1^{(\kappa_1)},\ldots,\mu_d^{(\kappa_d)}}$ & Exponent measure from (\ref{defsubordinatedCRMexpmeasure}) built from smoothed univariate IDEMs $\lc\mu_i\rc_{\leqd}$ \\
$\Lambda^{(\infty)}=\delta_\infty+\Lambda$ & Lebesgue measure augmented by a point mass at $\infty$ \\
$\Lambda_p^{(\infty)}$ & $\otimes_{i=1}^p \Lambda^{(\infty)}$ \\
\end{tabular}
\end{center}

\subsection*{Other}
\begin{center}
\begin{tabular}{p{0.25\textwidth}p{0.65\textwidth}}
\textbf{Symbol} & \textbf{Meaning} \\
$\Theta=\Theta(IJ)=(\Theta_1,\ldots,\Theta_L)$ & Hitting scenario (random partition) of $IJ$ \\
$L=L(\Theta(IJ))$ & Number of sets in the partition $\Theta(IJ)$ \\
$\Tilde{\Theta}=(\tilde\Theta_1,\ldots,\tilde\Theta_L)$ & Conditional hitting scenario, i.e.\ $\Theta(IJ)$ conditioned on the observed sample $\bmXIJ$. \\
$\tau(\bmX_{IJ},\theta)$ & $P(\Tilde{\Theta}=\theta)=P(\Theta=\theta \mid \bmXIJ)$ \\
$K_{IJ}(\bmx_{IJ},\cdot)$ & Conditional distribution of $\lambda$ given $\bmx_{IJ}$ \\
$(\ubmX_{IJ},\ubmy_{IJ})(\Tilde{\Theta}_l)$ & $\lc X_{i,j}\id_{\{(i,j)\in\tilde{\Theta}_l\}}+y_{i,j}\id_{\{(i,j)\not\in\tilde{\Theta}_l\}} \rc_{(i,j)\in IJ}$ \\
$\ubmY^{(n)}$ & Exchangeable min-id sequence in the stochastic representation of the posterior \\
$\ubmZ^{(n)}$ & Exchangeable sequence in the stochastic representation of the posterior which stems from the conditional hitting scenario \\
$\lc\ubmZ^{(n,l)}\rc_{1\leq l\leq L}$ & Stochastic decomposition of the exchangeable sequence $\ubmZ^{(n)}=\min_{1\leq l\leq L}\ubmZ^{(n,l)}$ \\
$\bmX_{\Tilde\Theta_l}(i)$ & Unique value of $\lc X_{i,j}\rc_{(i,j)\in\tilde\Theta_l}$, when it exists \\
$(\bmX_{\Tilde\Theta_l},\bm B)(I(\Tilde\Theta_l))$ & $d$-dimensional vector formed from $X_i$ for $i\in I(\Tilde\Theta_l)$ and filler values from $\bm B$ elsewhere \\
\end{tabular}
\end{center}

\endgroup


\end{document}